\documentclass[11pt,twoside,openany]{book}
\usepackage{caption}
\usepackage[caption=false]{subfig}

\usepackage{amsmath,amsfonts,amsthm,amssymb,graphics,graphicx,mathtools,epstopdf,etoolbox,indentfirst,enumitem,mleftright,xcolor,booktabs,footnote,relsize,microtype}
\usepackage{mathrsfs} 

\usepackage{geometry}
\usepackage[bottom]{footmisc} 

\usepackage{fancyhdr}
\fancypagestyle{plain}{

\fancyhf{}
}

\usepackage[
backend=bibtex,
style=alphabetic,
giveninits=true, 
maxbibnames=99,
minalphanames=3,
date=year
]{biblatex}
\usepackage{hyperref} 
\usepackage{doi} 

\renewbibmacro{in:}{}

\AtEveryBibitem{\clearfield{issue}} 

\newbibmacro{string+doiurl}[1]{%
	\iffieldundef{doi}{%
		\iffieldundef{url}{
			#1%
		}{%
			\href{\thefield{url}}{#1}%
		}%
	}{%
		\href{http://doi.org/\thefield{doi}}{#1}%
	}%
}
\DeclareFieldFormat{title}{\usebibmacro{string+doiurl}{\mkbibemph{#1}}}
\DeclareFieldFormat[article,thesis,incollection,inproceedings,manual]{title}%
{\usebibmacro{string+doiurl}
	{#1} 
}
\usepackage{algorithm} 
\usepackage[noend]{algpseudocode}
\floatname{algorithm}{Protocol} 
\algrenewcommand\alglinenumber[1]{\normalsize #1.} 

\newcounter{algsubstate}

\newenvironment{algsubstates}
  {\setcounter{algsubstate}{0}%
   \renewcommand{\State}{%
     \refstepcounter{algsubstate}%
     \Statex {\normalsize\arabic{ALG@line}.\arabic{algsubstate}.}\kern5pt}
     }
  {}

\makeatletter
\newenvironment{breakablealgorithm}
  {
    ~\\[\parskip]
     \refstepcounter{algorithm}
     \hrule height.8pt depth0pt \kern5pt
     \renewcommand{\caption}[2][\relax]{
       {\raggedright\textbf{\fname@algorithm~\thealgorithm} ##2\par}%
       \ifx\relax##1\relax 
         \addcontentsline{loa}{algorithm}{\protect\numberline{\thealgorithm}##2}%
       \else 
         \addcontentsline{loa}{algorithm}{\protect\numberline{\thealgorithm}##1}%
       \fi
       \kern5pt\hrule\kern2pt
     }
  }{
     \kern2pt\hrule\relax
    ~\\[\parskip]
  }
\makeatother

\newfloat{process}{htbp}{loa}
\floatname{process}{Process} 

\makeatletter
\renewcommand\paragraph{\@startsection{paragraph}{4}{\z@}%
                                    {1.25ex \@plus1ex \@minus.2ex}%
                                    {-.5em}
                                    {\normalfont\normalsize\bfseries}}
\makeatother

\newcommand{\ceil}[1]{\left\lceil #1 \right\rceil}
\newcommand{\floor}[1]{\left\lfloor #1 \right\rfloor}
\newcommand{\ket}[1]{\left| #1 \right>}
\newcommand{\bra}[1]{\left< #1 \right|}
\newcommand{\ketbra}[2]{\ket{#1} \! \bra{#2}}
\newcommand{\pure}[1]{\ketbra{#1}{#1}}
\newcommand{\inn}[2]{\langle#1|#2\rangle} 
\newcommand{\tr}[2][]{\operatorname{Tr}_{#1}\!\left[#2\right]} 
\newcommand{\der}{\mathrm{d}} 
\newcommand{\binh}{h_2} 

\newcommand{\acos}{\cos^{-1}}

\newcommand{\Bt}{B'}
\newcommand{\Btstr}{{\str{B}'}}
\newcommand{\Ct}{C'}
\newcommand{\Ctstr}{{\str{C}'}}
\newcommand{\ctstr}{{\str{c}'}}
\newcommand{\calG}{\mathcal{G}}

\newcommand{\calZ}{\mathcal{Z}}
\newcommand{\cdfBin}[3]{B_{#1,#2}(#3)} 
\newcommand{\cmax}{\mathrm{EC}_\mathrm{leak}}
\newcommand{\constr}{\nu} 
\newcommand{\cont}{\eps_\mathrm{con}}
\newcommand{\cperp}{\beta}
\newcommand{\defvar}{\coloneqq} 
\newcommand{\dop}[1]{\operatorname{S}_{#1}} 
\newcommand{\dtol}{\delta_\mathrm{tol}}
\newcommand{\dsou}{\delta_\mathrm{stat}}
\newcommand{\eps}{\varepsilon}
\newcommand{\expval}[1]{\langle#1\rangle} 
\newcommand{\expvaltr}[1]{\tr{#1 \rho_{AB}}} 
\newcommand{\Fj}{} 
\newcommand{\fmax}{f_\mathrm{max}}
\newcommand{\fmin}{f_\mathrm{min}}
\newcommand{\freq}{\operatorname{freq}}
\newcommand{\fZS}[3]{C_{#1,#2}(#3)}
\newcommand{\g}{g} 
\newcommand{\gF}{\tilde{F}} 
\newcommand{\hash}{\operatorname{hash}}
\newcommand{\Hmin}{H_\mathrm{min}}
\newcommand{\Hmax}{H_\mathrm{max}}
\newcommand{\id}{\mathbb{I}} 
\newcommand{\idmap}{\mathcal{I}} 
\newcommand{\idk}{\mathbb{U}} 
\newcommand{\inX}{\{0,1\}}
\newcommand{\inYg}{\{0,1\}}
\newcommand{\inYt}{\{2,3\}}
\newcommand{\keyw}{\tau} 
\newcommand{\len}{\operatorname{len}}
\newcommand{\lin}{r} 
\newcommand{\lkey}{\ell_\mathrm{key}}
\newcommand{\map}{\mathcal{M}}
\newcommand{\mPA}{\mathcal{M}_\mathrm{PA}}
\newcommand{\Max}{\operatorname{Max}}
\newcommand{\Min}{\operatorname{Min}}
\newcommand{\no}[1]{#1^c} 
\newcommand{\norm}[1]{\left\lVert#1\right\rVert} 
\newcommand{\optdom}{\mathcal{D}}
\newcommand{\Og}{\Omega_\mathrm{g}} 
\newcommand{\Oh}{\Omega_\mathrm{h}} 
\newcommand{\OPE}{\Omega_\mathrm{PE}} 
\newcommand{\OpPE}{\Omega'_\mathrm{PE}} 
\newcommand{\OpPEnot}{\Omega'^c_\mathrm{PE}} 
\newcommand{\p}{p} 
\newcommand{\pBell}{\mu} 
\newcommand{\pct}{\%} 
\newcommand{\pd}{P} 
\newcommand{\perr}{p_\mathrm{err}}
\newcommand{\pg}{P_\mathrm{guess}} 
\newcommand{\pr}[2][]{\mathsf{P}_{#1}\!\left[#2\right]}
\newcommand{\prnoscale}[2][]{\mathsf{P}_{#1}[#2]}
\newcommand{\pvm}{P} 
\newcommand{\povm}{M} 
\newcommand{\q}{q} 
\newcommand{\rightsemicirc}{\put(1.5,2.5){\oval(4,4)[r]}\phantom{\circ}}
\newcommand{\rot}[1]{R_{#1}} 
\newcommand{\semicset}{S_{\rightsemicirc}}
\newcommand{\str}[1]{\mathbf{#1}} 
\newcommand{\smf}[1]{\vartheta_{#1}} 
\newcommand{\suchthat}{\text{ s.t.}} 
\newcommand{\term}[1]{\textbf{#1}}
\newcommand{\thA}{{\theta_A}}
\newcommand{\thB}{{\theta_B}}
\newcommand{\upto}[1]{\left[#1\right]} 
\newcommand{\Var}{\operatorname{Var}}
\newcommand{\wexp}{w_\mathrm{exp}}
\newcommand{\wtrue}{w_\mathrm{true}}
\newcommand{\Yt}{Y'}
\newcommand{\Ytstr}{{\str{Y}'}}

\newcommand{\ecom}{\eps^\mathrm{com}}
\newcommand{\esound}{\eps^\mathrm{sou}}
\newcommand{\ecorr}{\eps^\mathrm{cor}}
\newcommand{\esecr}{\eps^\mathrm{sec}}
\newcommand{\eh}{\eps_\mathrm{h}}
\newcommand{\ecEC}{\eps^\mathrm{com}_\mathrm{EC}}
\newcommand{\ecPE}{\eps^\mathrm{com}_\mathrm{PE}}
\newcommand{\eEA}{\eps_\mathrm{EA}}
\newcommand{\eIID}{\eps_\mathrm{stat}}
\newcommand{\ePA}{\eps_\mathrm{PA}}
\newcommand{\es}{\eps_s}
\newcommand{\ez}{\tilde{\eps}_s} 

\newcommand{\prnonoise}{\mathsf{P}^{\star}}
\newcommand{\prnew}[2][]{\widetilde{\mathsf{P}}_{#1}\!\left[#2\right]}

\newcommand{\qA}{A}
\newcommand{\qB}{B}
\newcommand{\qE}{E}
\newcommand{\allE}{\mathsf{E}}

\usepackage{pifont}

\newcommand{\xmark}{\ding{55}}

\newtheorem*{remark}{Remark}

\newtheorem{theorem}{Theorem}[chapter]
\newtheorem{lemma}{Lemma}[chapter]
\newtheorem{corollary}{Corollary}[chapter]
\newtheorem{fact}{Fact}[chapter]
\theoremstyle{definition} 
\newtheorem{definition}{Definition}[chapter]

\usepackage{makecell}

\newcolumntype{L}[1]{>{\raggedright\let\newline\\\arraybackslash\hspace{0pt}}m{#1}}
\newcolumntype{C}[1]{>{\centering\let\newline\\\arraybackslash\hspace{0pt}}m{#1}}
\newcolumntype{R}[1]{>{\raggedleft\let\newline\\\arraybackslash\hspace{0pt}}m{#1}}

\newcounter{tabrownumcounter}
\newcommand\tabrownum{\stepcounter{tabrownumcounter}\roman{tabrownumcounter}}

\hypersetup{
linktoc = all, 
colorlinks = true,
citecolor= blue,
urlcolor= blue,
linkcolor = black 
}

\newcommand*{\SavedEqref}{}
\let\SavedEqref\eqref
\renewcommand*{\eqref}[1]{%
  \begingroup
    \protect\hypersetup{
      linkcolor=blue,
    }%
    \SavedEqref{#1}%
  \endgroup
}

\begin{document}

\pagenumbering{gobble}

\newgeometry{margin=25mm}

\begin{titlepage}
\begin{center}

\vspace{15mm}

\Huge
\textbf{Prospects for device-independent quantum key distribution}
    
\vspace{15mm}

\LARGE 
A thesis submitted to attain the degree of 

DOCTOR OF SCIENCES of ETH ZURICH 

(Dr.~sc.~ETH Zurich) 

\vspace{15mm}

presented by

Ying Zhe Ernest Tan

\vfill

accepted on the recommendation of

Prof.~Dr.~Renato Renner, examiner

Prof.~Dr.~Norbert L\"{u}tkenhaus, co-examiner

Prof.~Dr.~Marco Tomamichel, co-examiner 

\vspace{10mm}

2021

\vspace{20mm}
\end{center}

\end{titlepage}

\restoregeometry

\pagenumbering{roman}

{\pagestyle{plain} 

\chapter*{Abstract}

Device-independent quantum key distribution (DIQKD) aims to achieve secure key distribution with only minimal assumptions, by basing its security on the violation of Bell inequalities. While this offers strong security guarantees, it comes at the cost of only being able to certify positive keyrates for devices with low levels of noise, which are challenging to implement experimentally. In recent years, the gap between the theoretical requirements and experimental performance was narrowed by several landmark results, making DIQKD a promising avenue for investigation.
In this thesis, we present security proofs for DIQKD protocols based on various techniques that further improve the keyrates and noise tolerance.

In particular, we derive methods for securely computing the keyrates of a wider variety of DIQKD protocols than previously studied. These methods are also able to account for the techniques of \emph{noisy preprocessing} and \emph{random key measurements}, which can improve the performance of DIQKD protocols. Furthermore, we also investigate the possibility of \emph{advantage distillation} in DIQKD, which refers to using two-way communication in the information-reconciliation step of the protocol. With these approaches, it is possible in principle for the devices in existing loophole-free Bell tests to achieve positive asymptotic keyrates in DIQKD, which is to say that a secure key could be produced given sufficiently many protocol rounds.

This naturally raises the question of precisely how many rounds would be required in order for an experiment of finite length to be considered a secure demonstration of DIQKD. To address this, we also perform a finite-size security analysis for protocols based on combinations of some of these techniques, excluding advantage distillation. The analysis is based on a fairly recent result known as the entropy accumulation theorem, which simultaneously accounts for finite-size effects and non-IID attacks on the protocol. We improve on previous security proofs based on this theorem by using tighter finite-size bounds and incorporating some combinations of the techniques mentioned above. 
Our results indicate that while the required sample sizes are still impractically large for some Bell-test implementations, there is still scope for further improvement, and we propose some further protocol modifications towards this end.

\chapter*{Acknowledgements}

I would like to thank Renato Renner for his guidance and advice over the past years. He has been an accommodating and understanding supervisor, always ready to support me in directions I choose to pursue. 
I also thank co-examiners Norbert L\"{u}tkenhaus and Marco Tomamichel, for taking the time to evaluate my thesis, as well as providing helpful feedback on my work.

Also, I was fortunate to have had the opportunity to work with many collaborators 
(Jean-Daniel Bancal, 
Koon Tong Goh, 
Melvyn Ho,
Srijita Kundu,
Charles C.-W. Lim, 
Ignatius William Primaatmaja, 
Nicolas Sangouard,
Valerio Scarani,
Ren\'{e} Schwonnek, 
Pavel Sekatski, 
Jamie Sikora,
Xavier Valcarce,
Ramona Wolf) 
over the course of my research, and I am indebted to them for many valuable discussions. In particular, I thank Nicolas Sangouard and Charles Lim for extending invitations for the initial visits which led to these fruitful opportunities, and Valerio Scarani for providing me with useful advice. 

I would like to express my appreciation to all members of the QIT group for the enlightening and interesting conversations we have had, and the assistance they have given me. Special thanks to Elisa for accommodating my many requests for help with translations (including the thesis abstract), as well as L\'{i}dia and Joe for advice on living in Z\"{u}rich. 

Last but certainly not least, I thank my family for all the help and encouragement they have given me. I am also very grateful to Srijita for being a constant source of support, and making these past years so much better than they would otherwise have been. I could not have done this without all your help. 

\cleardoublepage

\tableofcontents

\clearpage}

\definecolor{darkmagenta}{rgb}{0.85, 0, 0.45}
\hypersetup{linkcolor = darkmagenta}

\pagenumbering{arabic}

\chapter{Introduction}
\label{chap:intro}

In quantum key distribution (QKD), the goal is to extract a shared secret key from correlations obtained by measuring quantum systems. Device-independent (DI) quantum key distribution is based on the observation that when these correlations violate a Bell inequality, a secure key can be extracted even if the users' devices are not fully characterized --- one can prove security with only minimal assumptions on the device behaviour~\cite{BHK05,PAB+09,Sca12,BCP+14}, hence the term ``device-independent''. By working with fewer assumptions, DIQKD unlocks the prospect of achieving secret key distribution with an unprecedented level of security against potential attacks. It also lies at an intersection between quantum foundations and cryptography, by making use of the foundational concept of Bell inequalities to prove security in a cryptographic task.

The strong security guarantees offered by DIQKD come with the drawback that it can only tolerate lower levels of noise as compared to standard QKD. This makes it challenging to demonstrate DIQKD in practice. However, experimental and theoretical developments in recent years have brought the possibility of such a demonstration closer to fruition.

For instance, on the experimental front, the advent of loophole-free Bell tests~\cite{HBD+15,SMC+15,GVW+15,RBG+17} has shown that laboratory devices are now able to achieve Bell violations under stringent conditions. Since the security of DIQKD rests on Bell violations, this signifies an important landmark on the path towards a DIQKD demonstration. In fact, these experimental implementations were soon used to realize some other device-independent protocols~\cite{LZL+18,ZSB+20,LLR+21,SZB+21,LZL+21}, bringing us a step closer towards DIQKD. 

As for theoretical developments, several important results were derived which helped to narrow the gap between the experimental devices and the theoretical noise tolerance. 
In particular, one such result was the \term{entropy accumulation theorem} (EAT)~\cite{DFR20,DF19}. Earlier DIQKD security proofs either required the assumption that the device behaviour across the protocol rounds is independent and identically distributed (IID)~\cite{PAB+09}, or they achieved lower keyrates as compared to the IID case~\cite{PM13,LPT+13,VV14,NBS+18,JMS20,arx_Vid17}. The EAT was used to show~\cite{ARV19} that for many DIQKD protocols, essentially the same keyrates can be attained against general attacks as in the IID case, hence achieving both high keyrates and strong security statements (i.e.~without requiring an IID assumption). It also provided explicit bounds on the finite-size behaviour, hence concretely addressing the question of what sample size would be required for a secure demonstration of DIQKD.

With these developments in mind, the aim of this work is to describe a variety of techniques that we have studied in order to further improve the keyrates and noise tolerance of DIQKD. 
In this chapter, we shall begin by first laying out the core ideas behind DIQKD.
With this framework in mind, we will be better positioned to describe earlier progress on the topic and the contributions presented in this work, which we shall list at the end of this chapter. 

\begin{remark}
After preparation of this thesis, several experimental demonstrations of DIQKD were announced~\cite{arx_NDN+21,arx_ZLR+21,arx_LZZ+21}, with various advantages and disadvantages in each case. We will not discuss them here, directing the interested reader to those works for more information. This thesis will remain focused on describing a number of proof techniques available for DIQKD, instead of analyzing these experiments in detail.
\end{remark}

\section{DIQKD devices}
\label{sec:introdevices}

The basic elements of DIQKD can be outlined as follows. The honest users (Alice and Bob) each hold a device that has the following functionality: it can receive some share of a quantum state from an external source, accept a classical input from the user, then based on this input value, perform some measurement on the quantum state and return the resulting output. 
Without loss of generality, we can suppose the states are supplied by an adversarial third party Eve, who keeps some side-information about the distributed states.
Usually, the devices would need to be used multiple times sequentially, producing a string of outputs corresponding to some string of inputs. (In some protocols~\cite{JMS20,arx_Vid17}, the entire input string is instead supplied to the devices in a single shot, receiving a corresponding string of outputs. This can be referred to as a \term{parallel-input} scenario.) 
We refer to each usage as a round of the protocol. Alice and Bob then perform some public classical communication, after which they will either generate a shared key that is ``secret'' with respect to Eve, or abort the protocol (if their data indicates they cannot generate a shared secret key).

To draw a comparison, such devices are basically similar to those required for standard QKD protocols (formulated as entanglement-based setups rather than prepare-and-measure setups). In that setting, the aim is to generate a secure key despite the fact that the states may be supplied by an untrusted third party Eve. However, an important assumption in standard QKD is that the measurements performed for a given input are trusted --- for instance, in the entanglement-based version of the~\cite{BB84} protocol, or early security proofs for the~\cite{Eke91} protocol, they must be specific Pauli measurements on qubits. 
In that sense, we can refer to such protocols as \term{device-dependent}, since they rely on conditions about the measurement device behaviour.
The goal of device-independent QKD is to find a way to generate a secure key with even weaker assumptions; namely, one allows the measurements to be (almost) untrusted as well.

\subsection{From Bell inequalities to security}
\label{sec:Bellineq}

To gain some initial insight, it is helpful to first restrict to a simpler ``IID'' situation. In device-dependent QKD, this would basically be the scenario referred to as \term{collective attacks}, in which the quantum state supplied to Alice and Bob is assumed to be IID across all the uses of the devices, 
i.e.~for $n$ uses of the devices, the underlying state 
is of the form $\rho_{\qA \qB}^{\otimes n}$ where the $\qA$ and $\qB$ registers are held by Alice and Bob respectively.\footnote{One could instead consider a more general setting where all parties act on the same Hilbert space, and we only require that Alice and Bob's measurements commute. However, this introduces various complications beyond the scope of this work; see e.g.~\cite{arx_SW08,NPA08,arx_JNV20}.} Eve's side-information is modelled by allowing her to keep an arbitrary extension of that state,
which we can assume to be a purification without loss of generality (if the global state is not pure, we can simply give Eve an additional purifying register). Since all purifications are isometrically equivalent, we can focus specifically on purifications of the form $\rho_{\qA \qB \qE}^{\otimes n}$ where $\rho_{\qA \qB \qE}$ is a purification of $\rho_{\qA \qB}$.
In the DI setting, the choice of how to extend the definition of collective attacks is perhaps not entirely unique. However, we choose to follow the approach used in an early work~\cite{PAB+09}: we define it as imposing the same constraint on the states as in device-dependent QKD, but also the additional constraints that in each round, the measurements act only on the state in that round, and that whenever the same input is supplied, the same measurement is performed. In other words, when Alice supplies some input $x\in\mathcal{X}$ and gets an output $a\in\mathcal{A}$, this is described by a POVM element $\pvm_{a|x}$ (acting on register $\qA$ of that round) that is the same in every round. Similarly, Bob's measurements are described by POVM elements $\pvm_{b|y}$ 
(for output $b\in\mathcal{B}$ given input $y\in\mathcal{Y}$) 
acting on register $\qB$ of that round. 
While we have described these as general POVMs, it is often possible to reduce the analysis to projective measurements; we defer the details to Appendix~\ref{app:proj}.

With this scenario in mind, the goal in DIQKD is to produce a secure key with only minimal assumptions on $\pvm_{a|x}$ and $\pvm_{b|y}$. To see why this might even be possible in the first place, we turn to the notion of Bell inequalities~\cite{Bell64}, which are inequalities regarding the input-output distribution produced by the devices. Still focusing on the collective-attacks setting for now, we can consider a single round and let $\pr{ab|xy}$ denote the probability of the parties getting outputs $a$ and $b$ if inputs $x$ and $y$ respectively are supplied. Given this model, it is easy to see that
\begin{align}
\pr{ab|xy} = \tr{\pvm_{a|x}\otimes\pvm_{b|y} \rho_{\qA \qB}}.
\label{eq:qprobs}
\end{align}
A Bell inequality is an inequality that holds for distributions $\pr{ab|xy}$ admitting a \term{local hidden variable} (LHV) model\footnote{There is some variation in terminology here; for instance, they may be referred to as local-realistic models or simply ``classical'' models, and different meanings have been assigned to the terms in different works. However, this is not the focus of our work.}, meaning that they can be written in the form
\begin{align}
\pr{ab|xy} = 
\sum_{\lambda} \pr{a|x\lambda} \pr{b|y\lambda} \pr{\lambda},
\label{eq:LHV}
\end{align}
for some ``hidden'' discrete random variable $\lambda$. (More generally, one could integrate over a continuous $\lambda$, but this makes no difference as long as the input and output sets are finite, because the set of LHV-compatible distributions is convex.)
An important example of a Bell inequality is the {CHSH inequality}~\cite{CHSH69}, for the case where 
$a,b\in\{-1,+1\}$ and $x,y\in\{0,1\}$:
\begin{align}
\expval{A_0 B_0} + \expval{A_0 B_1} + \expval{A_1 B_0} - \expval{A_1 B_1} \leq 2,
\label{eq:CHSHclass}
\end{align}
where 
$\expval{A_x B_y} \defvar \sum_{ab} ab\pr{ab|xy}$ is the correlation between the outputs of measurements $x$ and $y$.
In some contexts, it is easier to label the outputs as $\{0,1\}$ rather than $\{-1,+1\}$ --- we shall specify this when necessary. For brevity, we shall refer to the LHS of~\eqref{eq:CHSHclass} as the \term{CHSH value}.

Famously, Bell inequalities can be violated by the distributions produced by some entangled quantum states and measurements, proving that such distributions cannot be explained by LHV models --- this phenomenon is often referred to as \term{quantum nonlocality}. For the purposes of DIQKD, however, the critical fact of interest is that only entangled states can violate Bell inequalities. Hence if the honest parties can certify that $\pr{ab|xy}$ violates some Bell inequality, then they know $\rho_{\qA\qB}$ must be entangled, 
\emph{regardless of what the measurements were}.
This suggests some possibility of generating a shared secret key from $\rho_{\qA\qB}$ (putting aside the question of whether all entangled states allow key distillation~\cite{GW00}).

Alternatively, there is another perspective that may be of interest from a foundational point of view. It is based on the observation that 
distributions of the form~\eqref{eq:qprobs} produced by quantum theory are \term{non-signalling} (NS), in the following sense. Focusing on e.g.~Alice, for an arbitrary $\pr{ab|xy}$ we can compute the marginal distribution $\pr{a|xy} = \sum_b \pr{ab|xy}$ of Alice's output, which \textit{a priori} could depend on both inputs $xy$. However, the NS conditions are the requirement that this distribution is independent of Bob's input $y$ --- put another way, this means there is a well-defined distribution for Alice's output given her input alone, which we can denote as $\pr{a|x}$.\footnote{Notice that for quantum distributions in particular (as in~\eqref{eq:qprobs}), we have $\pr{a|x}=\tr{\pvm_{a|x} \rho_{\qA} }$, i.e.~this is precisely the distribution given by quantum theory for Alice's measurements on her reduced state.} Analyzing Bob's situation similarly, the NS conditions can be stated as:
\begin{align}
\begin{gathered}
\forall a,x,\quad \pr{a|x} \defvar \pr{a|xy} = \sum_b \pr{ab|xy} \text{ is ``well-defined'' (i.e.~independent of $y$)}, \\
\forall b,y,\quad \pr{b|y} \defvar \pr{b|xy} = \sum_a \pr{ab|xy} \text{ is ``well-defined'' (i.e.~independent of $x$)}.
\end{gathered}
\label{eq:NS}
\end{align}
Returning to the topic of DI cryptography, roughly speaking it can be proven that if $\pr{ab|xy}$ is produced by an underlying non-signalling model (such as quantum theory) and also violates a Bell inequality, then the outputs must contain some ``fundamental randomness'' that cannot be predicted by an adversary. This serves as the first step towards producing a secret key; furthermore, since the proof only relies on the NS conditions, it also covers any potential post-quantum theories\footnote{There are also some other approaches~\cite{CR11} for arguing that no post-quantum theory satisfying some plausible postulates can break the security of quantum cryptography, in some sense. However, interpreting these results can be somewhat delicate, and we will not aim to do so here.} satisfying those conditions. However, there are some subtleties to be aware of, and we defer this discussion to Appendix~\ref{app:foundations} (see also~\cite{woodheadthesisnodoi}).

Hence the idea of a DIQKD protocol can be roughly described as having Alice and Bob collect input-output statistics from their devices in some rounds to estimate the probabilities $\pr{ab|xy}$ (with the IID assumption, these estimates can be arbitrarily accurate given enough rounds), 
then checking whether this distribution violates a Bell inequality ``strongly enough'' for them to generate a secret key. 
Of course, this is still a rather incomplete description (for instance, we have not yet explained exactly how strong a Bell violation is needed to certify secret key generation), but we shall defer further details until
Sec.~\ref{sec:protsketch}--\ref{sec:proofsketch}.

For efficiency in the protocol, Alice and Bob may not need to estimate all the individual probabilities, but rather just one or more ``Bell parameters'', e.g.~the CHSH value.
We can write each such parameter in the form $\constr_j = \sum_{abxy} c^{(j)}_{abxy} \pr{ab|xy}$ for some coefficients $c^{(j)}_{abxy}\in\mathbb{R}$.\footnote{In principle, one could also allow these coefficients to be complex, but this is inconvenient in some of our later analysis, e.g.~the Lagrange dual in Chapter~\ref{chap:singlernd}. Hence we focus only on the real-valued case, since in any case, any Bell parameter based on complex-valued
coefficients could be equivalently described as a pair of Bell parameters with real-valued coefficients instead. Another more general possibility would be to consider nonlinear functions of $\pr{ab|xy}$, but we do not do so in this work.} In terms of the measurements $\pvm_{a|x},\pvm_{b|y}$, this can be viewed as measuring the expectation values of some (hermitian) ``Bell observables'',
\begin{align}
\Gamma_j(\pvm_{a|x},\pvm_{b|y})\defvar \sum_{abxy} c^{(j)}_{abxy} \pvm_{a|x}\otimes\pvm_{b|y}.
\label{eq:bellop}
\end{align}
A special case would be the scenario of \term{nonlocal games}, where Alice and Bob use some specified input distribution and they are considered to win the game if the input-output combinations $abxy$ satisfy some predicate $V(a,b,x,y)$. In that case, the winning probability of such a game is indeed a parameter of the form $\sum_{abxy} c_{abxy} \pr{ab|xy}$, for some coefficients $c_{abxy}$ determined by the input distribution and the predicate $V(a,b,x,y)$.

Going beyond collective attacks, one would aim to prove security without the IID assumption on the state and measurements. In device-dependent QKD, this would be referred to as the scenario of \term{coherent attacks}. For DIQKD, we take this to mean that in each round, the state Eve distributes in each round can have arbitrary structure; furthermore, the measurement performed by each device can act not only on the state it has just received, but also on registers storing (possibly quantum) memory from previous rounds. (If we work in the parallel-input setting instead,
then there is no notion of ``previous'' rounds
--- Eve simply distributes a ``large'' state to the devices in a single shot, and each device performs a ``large'' measurement on its share of the state, which is allowed to depend on the entire input string supplied to it by the corresponding honest party.)
This poses a number of challenges for security proofs; e.g.~we do not even have a specific (single-round) distribution $\pr{ab|xy}$ to speak of that describes all the rounds.
Fortunately, these obstacles can be overcome to some extent, as we shall outline in Sec.~\ref{sec:sketchbeyond}. More subtly, though, the coherent-attacks setting in DIQKD introduces a problem regarding device reuse~\cite{BCK13}, which we shall soon discuss.

\begin{remark}
Some care needs to be taken regarding the collective-attacks assumption --- if one naively designs a protocol taking that assumption for granted, the resulting protocol can be completely insecure against trivial non-IID attacks. For instance, consider a protocol where the probabilities $\pr{ab|xy}$ are only estimated from the first $k$ rounds, with the remaining rounds used for key generation. This can be proven secure (for suitable parameter choices) under the collective-attacks assumption; however, it is obviously insecure in practice since Eve could just behave honestly on the first $k$ rounds and send completely insecure states in the remaining rounds. In light of this, it could be said that the appropriate way to treat the collective-attacks assumption would be to \emph{first} design a protocol, then ``reasonably conjecture'' that collective attacks are (near-)optimal against that protocol, in which case one could proceed with a security proof based on that assumption. For instance, if we instead use a random subset of the rounds (not publicly announced beforehand) to estimate $\pr{ab|xy}$, it is at least less obvious how to outperform collective attacks.
\end{remark}

\subsection{Required assumptions}
\label{sec:assumptions}

While our above discussion indicates that the measurements 
can be ``untrusted'', an important caveat is that we do still need to enforce some constraints on them. Starting with those that were already implicit in the above description, we have assumed that the quantum registers held by Alice and Bob's devices have a tensor-product separation, and the measurements act only on their respective Hilbert spaces (and local memory registers, when coherent attacks are allowed). Also, we have imposed that in each round, the measurement of Alice's device cannot depend on Bob's choice of input, and vice versa (this was implicit in the notation $\pvm_{a|x},\pvm_{b|y}$ itself in the collective-attacks case).\footnote{For coherent attacks, whether each party's measurement can depend on the other party's input in \emph{previous} rounds would depend on the nature of the security proof --- the main proof approach used in this work, based on the EAT, indeed allows for this in a certain sense. In contrast, the parallel-input scenarios of~\cite{JMS20,arx_Vid17} require that the output string does not depend on any part of the other party's input string, so an implementation of those protocols would have to justify that assumption by spacelike separation or shielding. The former would seem to require distributing and/or storing a large number of entangled states simultaneously, and is hence currently impractical.} In the context of loophole-free Bell tests~\cite{HBD+15,SMC+15,GVW+15,RBG+17}, this is enforced by having spacelike separation between the measurements. However, that is a very stringent requirement, and in the context of DIQKD, it may be worth considering whether there are easier ways to enforce this condition --- for instance, by implementing appropriate ``shielding'' measures on the devices, to prevent them from leaking unwanted information. (Some security proofs have been developed that allow limited leakage of the inputs~\cite{SPM13,arx_JK21}, but we will not discuss them in detail here.)

In fact, the question of shielding the devices ties in to another vital assumption, which prevents us from claiming we can prove security for \emph{completely} untrusted devices. Namely, we also need to impose the constraint that the outputs of the devices are not simply leaked to Eve (for instance, if the devices were arbitrarily adversarial, they could simply contain transmitters to broadcast this information). Hence the devices do need to be sufficiently characterized to certify that there is no such leakage. As with the inputs, it may be that this could be achieved using some shielding measures, and if so, it may be expedient to use these measures to justify both of these ``no-leakage'' assumptions. 

However, when coherent attacks are allowed, the issue of ensuring that the outputs are not leaked in \emph{any} sense is more subtle than it appears. In~\cite{BCK13}, a \term{memory attack} was constructed based on this point. The relevant scenario is as follows: suppose that the devices are used for an instance of a DIQKD protocol, and a secure key has been generated. Now, if the devices are reused for a second instance of the protocol (or for that matter, some other protocol), it could be possible that the device outputs in this second instance are correlated with the device outputs in the first instance due to memory effects, and thus (indirectly) with the key generated in the first instance. Since the second instance would typically involve public communication computed based on the device outputs, this can hence leak information about the first instance's key. The attack in~\cite{BCK13} is something of a ``maximally adversarial'' model which leaks specific key bits via the public communication, but in a more general information-theoretic sense, any ``non-negligible'' correlations between the first instance's key and the second instance's public communication could cause some problems in formalizing the implications of the security definitions (we discuss this in more detail at the end of Sec.~\ref{sec:securitydef}).

In order to avoid leakage in this form, one option would be to impose a condition along the lines of requiring that in each instance of a DIQKD protocol, the devices do not measure any registers containing any form of ``secret data'' from previous protocols, such as the outputs or the keys.\footnote{An extreme way to ensure this would be to enforce that the devices are never used again after completing a DIQKD protocol instance, preventing any memories they might retain from being accessed in the future, but this seems impractical.} This would likely also enforce that there are no correlations between the outputs across different protocol instances, which is perhaps somewhat in tension with the idea that we allow the devices to retain memory \emph{within} each instance of the protocol (when allowing coherent attacks in the DIQKD sense). However, one could perhaps argue that given enough time between protocol instances, devices that are not ``actively adversarial'' (informally speaking) would not retain any significant memory effects across instances. Alternatively, some ideas have been proposed to work around this problem by using specialized encodings of the public communication~\cite{MS14}. 

\begin{remark}
This ``leakage via memory effects'' issue only becomes possible due to the combination of (1) allowing coherent attacks, and (2) working in the DI setting for QKD. If we restrict to collective attacks in DIQKD, recall that this inherently involves assuming that the device measurements are acting on the states they have received in that round, which are supplied by Eve and can thus be taken to be independent of any data that Eve does not already know. On the other hand, if we allow coherent attacks but stick to the ``basic'' form of device-dependent QKD, such memory effects are not an issue because the measurements are taken to be fully characterized; in particular, this again usually means they are assumed to genuinely act on the states they have just received in each round. In some ways, this suggests that merely the notion of assuming the measurements act on specific registers may be stronger than it appears, because it could implicitly be claiming that some memory effects are being disallowed. This issue is discussed further in a writeup following this thesis, where we note that allowing memory effects may also invalidate some standard proof techniques for device-dependent QKD.
\end{remark}

In light of the above points, it is worth clarifying that DIQKD currently seems unlikely to guarantee security against truly ``adversarially designed'' devices. Rather, one could take the perspective that it is intended for use with devices that are provided by a trustworthy source, but with less stringent requirements for characterizing or certifying the devices, 
in that one does not have to ensure the devices are performing specific measurements.
For that matter, the assumptions listed above for DIQKD can be viewed as being strictly weaker than those that would be involved in ensuring that the measurements are fully characterized, i.e.~they would have to be checked anyway as part of the latter. DIQKD could also be viewed as helping to improve resistance to QKD hacking techniques~\cite{FQT+07,GLL+11,JSK+15}, some of which (on an abstract level) can be viewed as forcing the devices to perform measurements other than the intended ones. On the other hand, some hacking techniques instead basically have the goal of forcing the devices to leak information about the inputs and/or outputs (for instance, some forms of Trojan-horse attacks~\cite{JSK+15}), and these would still be an issue to be wary of in DIQKD.

Before moving on, we state for completeness a few assumptions that are required for the ``classical'' parts of a DIQKD protocol. Following the presentation in~\cite{ARV19}, these are:
\begin{itemize}
\item The honest parties can (locally) generate trusted randomness that is independent of all other registers in the protocol. (This may be used for choosing their inputs, as well as some other classical processing steps.)
\item The honest parties have trusted post-processing units to perform classical computations.
\item The honest parties have an authenticated public channel to perform all classical communication.
\end{itemize}
There are some possibilities for partially relaxing the first assumption~\cite{Hall11,SPM13,KA20,arx_JK21}, but we leave this for other work. 

Finally, in this work we will assume that all the quantum systems the devices act on are finite-dimensional. However, we allow the dimension to be arbitrary, consistent with the idea that the measurements are ``uncharacterized''. 
This is just to ensure validity of the theorems we use, some of which have not been explicitly proven for  infinite-dimensonal systems, though in any case it seems likely that they should generalize in some form to that scenario.
(While separations between finite- and infinite-dimensional quantum behaviour have been found in the context of Bell violations~\cite{arx_JNV20}, it currently seems implausible that such differences could be exploited for explicit attacks on DIQKD protocols.)

\section{Other DI protocols}
\label{sec:otherDI}

Apart from DIQKD, there are a variety of other protocols that can be performed in a device-independent manner. 
Here we shall
briefly highlight a few of them. To begin with, there is the task of \term{device-independent randomness expansion} (DIRE)~\cite{CK11,arx_colbeckthesis}, where the goal is to start with a short random seed and produce a longer string of random bits, which is secret with respect to some adversary's side-information. Similar to DIQKD, the idea here is that the short random seed can be used to choose inputs for devices of the same form as in DIQKD. If the results violate a Bell inequality, then outputs should have some ``private randomness'' by similar arguments as before, which can (with some classical processing) be extracted into a secret key.

However, DIRE differs from DIQKD in several respects. For instance, DIQKD aims to generate a \emph{shared} secret key between Alice and Bob, but in DIRE the intent is only to generate a single secret key. In line with this, in DIRE it is usually assumed that Alice and Bob have a private communication channel rather than a public one; alternatively, one could take the view that both devices are placed ``in the same lab'' (with appropriate shielding from each other) and the honest user(s) can directly read off the outputs of both devices. 
Also, in DIQKD the honest parties have unlimited sources of local randomness, whereas in DIRE they are limited to the seed randomness, which places some constraints on the input distributions they can use.

Apart from DIRE, there is the slightly different task of \term{device-independent randomness generation}, also called device-independent random number generation (DIRNG). This is very similar to DIRE, but instead of a (private) short random seed, the honest parties are given an unbounded supply of public randomness that is assumed to be independent of the devices, and they aim to generate a secret (i.e.~private) key. This is usually an easier task than DIRE, since there is no limit on the amount of starting randomness they can use. 

Another related but distinct task is \term{device-independent randomness amplification} (DIRA)~\cite{CR12,BRG+16,KA20}, where the honest parties are instead given an unbounded supply of ``imperfect'' randomness,  
in the sense that it may be correlated
with the state supplied to the devices. As compared to DIRE and DIRNG, the analysis of DIRA can potentially be more challenging, because the most basic (and thus better-studied) setting for Bell nonlocality is where the inputs are independent of the state, as is implicit in our equations~\eqref{eq:qprobs}--\eqref{eq:LHV}. While there are indeed some results studying Bell nonlocality in scenarios where correlations are allowed between the state and inputs~\cite{Hall11,KHS+12,PK13,PG16}, additional work is needed to convert this notion to a security proof for DIRA~\cite{CR12,BRG+16,KA20}.

The above tasks are perhaps the most well-studied DI protocols. There are a variety of other protocols that have been proposed, but we shall not attempt to present the full spectrum here. Instead, we simply list a few examples: DI protocols have been developed (with partial security) for weak coin flipping~\cite{SCA+11}, bit commitment~\cite{ACK+14}, random access codes~\cite{CKK+16}, weak string erasure~\cite{KW16}, and XOR oblivious transfer~\cite{arx_KST20}. Also, some varieties of ``certified deletion'' tasks have been studied~\cite{FM18,arx_KT20}, where the goal is for one party to prove that they have deleted some information --- intuitively, this is impossible using classical information because it can always be copied, but the no-cloning property of quantum states suggests the possibility of achieving this with quantum information. One point in common for these protocols is that they are in the setting of two-party cryptography (except for~\cite{arx_KT20}, which additionally includes an adversary Eve), where there are only two parties but either of them can potentially be dishonest. 

Finally, it is worth mentioning the concept of \term{self-testing}~\cite{PR92,BMR92,MY98} (also referred to as \term{rigidity}), which is the remarkable fact that some nonlocal distributions $\pr{ab|xy}$ can only be attained by an essentially unique state (up to local isometries and ancillas), and in some cases a similar statement holds for the measurements as well. ``Robust'' versions of these results are also known, which allow some deviation from the specified distribution. This is a very powerful statement on a theoretical level, since it certifies a particular state and/or measurements based only on the distribution $\pr{ab|xy}$, assuming that distribution can be estimated accurately. However, some self-testing statements are not entirely ``operational'' (though a recent work~\cite{arx_CHM21} provides a framework to address this), and converting self-testing results into security proofs for concrete protocols tends to result in rather loose final bounds~\cite{FM18,arx_KST20}. Still, it can serve as a guide to gain some initial intuition for prospective DI protocols.

\section{Previous work}
\label{sec:history}

We now present a rough overview of past developments in the topic of DIQKD, though it would be impractical to cover all results on this topic, and hence we may focus more on some specific works.
The first proposal which sparked interest in the concept of DIQKD could be considered to be~\cite{BHK05}, which studied the scenario where Eve is restricted only by the non-signalling principle. (Alternatively, one could view the notion of self-testing initially developed in~\cite{PR92,BMR92,MY98} as an important precursor. However, while this concept plays a part in some more recent DI security proofs~\cite{FM18,arx_KST20}, those early results on self-testing did not appear to initiate much further development at the time into full DIQKD protocols.) Subsequently, a security proof was developed in~\cite{ABG+07,PAB+09} 
against quantum adversaries restricted to collective attacks. 
Under that assumption, the authors derived a closed-form bound on the asymptotic keyrate, which was tight for the protocol in question.

The analysis for coherent attacks proved to require substantially more effort. For instance, a security proof for this scenario was developed in~\cite{VV14} (and another in~\cite{LPT+13} for a modified protocol), but the asymptotic keyrates obtained from these approaches are lower than in the collective-attacks scenario. Also, we remark that it seems that the security proofs for DIRNG in~\cite{PM13,NBS+18}
should basically apply to DIQKD as well, though they would again yield a lower asymptotic keyrate compared to collective attacks.\footnote{Furthermore, while this approach does not rely on an IID assumption, there are some technical issues regarding whether it truly covers coherent attacks, in that the authors note it may require an assumption that Eve's state has ``decohered'' to classical side-information by the time the protocol is implemented.} The development of the entropy accumulation theorem in~\cite{DFR20} was what allowed for the first security proof~\cite{ARV19} that achieved the same asymptotic keyrates against coherent attacks as for collective attacks (in protocols where the devices are used sequentially and error correction is performed using one-way communication). Shortly afterwards, a somewhat similar result known as the \term{quantum probability estimation} (QPE) framework~\cite{KZB20,ZKB18,ZFK20} was developed for DIRNG/DIRE, though it does not appear to have been applied yet to DIQKD --- we do not focus on this approach here, leaving it for future work.

The parallel-input scenario (for coherent attacks) appears to be the most challenging so far, which is perhaps to be expected since it allows very general dependencies\footnote{However, it would not be strictly accurate to say that it includes the sequential scenario as a special case, because in the latter we have the aforementioned possibility of past inputs being communicated between the devices after each round, which is not captured in the parallel-input scenario.} across the protocol rounds. Security proofs for this setting were developed in~\cite{JMS20,arx_Vid17}, but they also did not achieve the same asymptotic keyrates as compared to collective attacks, and it is currently an open question whether this is even possible in the parallel-input setting. We remark, however, that this proof approach allowed for an extension to a more general DIQKD scenario where the devices are allowed to leak some (fairly small) amount of information about the inputs~\cite{arx_JK21}. 

\begin{remark}
In device-dependent QKD, the question of security proofs against coherent attacks (for protocols with permutation symmetry, at least) can be resolved by e.g.~applying suitable \term{de Finetti theorems}~\cite{rennerthesis} or the \term{postselection technique}~\cite{CKR09}, which allow reductions to collective attacks. However, the existing versions of such theorems have a dependence on the system dimensions, which is a problem in DIQKD where the system dimensions are unbounded. Hence different techniques were developed for DIQKD, though there is still ongoing work~\cite{AR15,arx_JT21} on the possibility of de Finetti theorems suitable for DIQKD. (Another proof approach~\cite{TL17} for device-dependent QKD is based on \term{entropic uncertainty relations} for smoothed entropies, which yield tighter finite-size bounds. However, this is also not straightforward to generalize to DIQKD, though the approach in~\cite{LPT+13} is based on this underlying idea.)
\end{remark}

While this covers the developments in proof techniques against coherent attacks, another important question is finding protocols that achieve better asymptotic keyrates in the first place, even if one has to make the collective-attacks assumption in doing so.
The early work~\cite{PAB+09} obtained a tight bound on the keyrate for their protocol specifically (which was based only on the CHSH inequality and did not involve various possible techniques for improving the keyrates), but there was little progress in deriving similar bounds for more general situations. 
One general approach~\cite{PAM+10,NPS14,BSS14} (developed for DIRNG rather than DIQKD, but the approach works for both) was based on bounding some guessing-probability values.
However, this approach has the drawback that the resulting bounds on the asymptotic keyrate are not tight, at least for the scenario of collective attacks.
(Though the corresponding state of progress in non-IID security proofs should be kept in mind --- the security proofs in~\cite{PM13,NBS+18} yield asymptotic keyrates that can be expressed in terms of the guessing probability, hence it was fairly natural to study it.) Other than this, progress on this topic appears to have been fairly limited,
apart from the more specialized quantities that were studied in the security proofs of~\cite{LPT+13,VV14,JMS20,arx_Vid17} against coherent attacks, which did not yield very tight bounds.

From the experimental perspective, the realization of loophole-free Bell tests~\cite{HBD+15,SMC+15,GVW+15,RBG+17} could be said to have unlocked the first core requirement for DI protocols, i.e.~achieving a Bell violation. (Though as previously mentioned, in DI protocols it may not be strictly necessary to close the signalling loophole via spacelike separation, so in fact there may be some room for allowing more relaxed requirements on this point.) By combining this milestone with the tighter finite-size keyrate bounds given by the EAT or QPE frameworks, experimental demonstrations of DIRNG/DIRE were subsequently achieved~\cite{LZL+18,ZSB+20,LLR+21,SZB+21,LZL+21}. 
There have also been various proposals for techniques to further improve the performance of Bell-test implementations~\cite{MSS20}.
However, DIQKD has thus far remained slightly out of reach (we give a quick intuitive explanation for this in Sec.~\ref{sec:sketchDIRE}, based on the asymptotic keyrates).
Our goal in this work is to develop methods to try to overcome this difficulty.

\begin{remark}
Recently, another topic of interest has been the question of \emph{upper} bounds on DIQKD keyrates~\cite{KWW20,arx_WDH19,AL21,CFH21,FBL21,arx_KHD21}, including examples of distributions $\pr{ab|xy}$ that violate a Bell inequality but for which the DIQKD keyrate is provably zero for a very large class of possible protocols~\cite{FBL21}. However, we will not attempt to cover this subject in this work.
\end{remark}

\section{Current contributions}

Having established a picture of some of the existing results in DIQKD, we now turn to the contributions that we shall present in the subsequent chapters. 
As a preliminary, in Chapter~\ref{chap:overview} we first give an overview of the important concepts in DIQKD, including an outline of the typical protocol structure and some description of how to construct a security proof. 
We will also introduce some techniques that improve the keyrates and noise tolerances of DIQKD, such as \term{noisy preprocessing}, \term{random key measurements}, and \term{advantage distillation}. Some of these techniques were previously studied in device-dependent QKD, but corresponding proofs for DIQKD were only achieved fairly recently (beginning with the works~\cite{HST+20,SGP+21,TLR20}).
We do not present new results in this chapter, apart from the brief overview of those techniques.

With this framework in mind, in Chapter~\ref{chap:singlernd} we describe several methods to bound the asymptotic keyrates, with the goal being to obtain bounds that are tighter than those based on guessing probability, but also applicable to more general scenarios than just the CHSH-based protocol of~\cite{PAB+09}, including the possibility of incorporating noisy preprocessing and random key measurements. While none of these methods simultaneously achieve the twin goals of tight bounds and applicability to all nonlocality scenarios, they do offer various improvements over the previous approaches, which we shall describe. The results in this chapter are based on the approaches we derived in~\cite{TSG+21,arx_TSB+20}, as well as a short summary of our results from~\cite{HST+20,SBV+21,SGP+21}. (We also briefly describe some independent works~\cite{BFF21,WAP21,arx_BRC21,arx_BFF21,arx_MPW21} with similar aims, some of which achieved tighter bounds.) 

Next, in Chapter~\ref{chap:finite} we present a finite-size security proof based on the EAT, as well as some detailed analysis of how it compares to the results for collective attacks. While the former was already studied in~\cite{ARV19} and our approach is essentially similar, our contribution is to somewhat modify the analysis in order to relax the theoretical guarantees required in the error-correction step. Furthermore, our approach accounts for noisy preprocessing and random key measurements, as well as achieving tighter finite-size bounds by using improved versions of the EAT~\cite{DF19} and some statistical bounds~\cite{LLR+21}. We also present some potential protocol modifications which may improve the keyrates. 
Several of the results in this chapter and the previous one are also relevant for DIRNG/DIRE. 
The results in this chapter are based on~\cite{arx_TSB+20}.

In Chapter~\ref{chap:AD}, we turn to the topic of advantage distillation in DIQKD. 
This refers to the use of two-way communication in place of the one-way error-correction step in typical DIQKD protocols.
As compared to one-way error correction, much fewer proof techniques have been developed in this setting --- for instance, all of the approaches we previously described for analyzing coherent attacks do not seem to be easily applicable when advantage distillation is considered. Therefore, in this chapter we restrict ourselves to the collective-attacks model only, and prove some results in this setting. Specifically, we derive a sufficient condition for the security of a particularly prominent advantage distillation protocol, and compare the resulting noise thresholds against previous ones.
This chapter is based on our results in~\cite{TLR20}.

Finally, we conclude in Chapter~\ref{chap:conclusion} with some discussion of the prospects for DIQKD using these techniques, as well as some further questions that should be worth pursuing.

\chapter{Conceptual overview}
\label{chap:overview}

We shall now present a broad overview of the overall concepts and framework required for DIQKD. This is intended as a somewhat pedagogical description, which may be helpful for better understanding of the results in later chapters. The organization of this chapter is as follows. After listing basic notation and definitions in Sec.~\ref{sec:notation}, we shall introduce in Sec.~\ref{sec:minmaxent} some important entropic quantities and their operational implications. In Sec.~\ref{sec:protsketch}, we shall give a rough outline of typical DIQKD protocols. This lets us lay out in Sec.~\ref{sec:securitydef} the security definitions we shall use, and in Sec.~\ref{sec:proofsketch} we sketch how to prove that a DIQKD protocol satisfies these definitions. Finally, in Sec.~\ref{sec:sketchDIRE} we briefly comment on how these proof techniques apply to DIRNG and DIRE.

\section{Notation and definitions}
\label{sec:notation}

\begin{table}
\caption{List of notation}\label{tab:notation}
\def\arraystretch{1.5} 
\setlength\tabcolsep{.28cm}
\begin{tabular}{c l}
\toprule
\textit{Symbol} & \textit{Definition} \\
\toprule
$\log$ & Base-$2$ logarithm \\
\hline
$H(\cdot)$ & Base-$2$ von Neumann entropy \\
\hline
$D(\cdot \Vert \cdot)$ & Base-$2$ quantum relative entropy \\
\hline
$\norm{\cdot}_p$ & Schatten $p$-norm \\
\hline
$\left|\cdot\right|$ & Absolute value of operator; $\left|M\right| \defvar \sqrt{M^\dagger M}$ \\
\hline
$X\geq Y$ (resp.~$X>Y$) & $X-Y$ is positive semidefinite (resp.~positive definite)\\
\hline
$\dop{=}(A)$ (resp.~$\dop{\leq}(A)$) & Set of normalized (resp.~subnormalized) states on register $A$ \\
\hline
$\idk_A$ & Maximally mixed state on register $A$ \\
\hline
$\delta_{j,k}$ & Kronecker delta \\
\hline
$\floor{\cdot}$ (resp.~$\ceil{\cdot}$) & Floor (resp.~ceiling) function \\
\hline
$\upto{n}$ & Indices from $1$ to $n$, i.e.~$\{1,2,\dots,n\}$ \\
\hline
$A_{\upto{n}}$ & Registers $A_1 \dots A_n$ \\
\hline
$\no{\Omega}$ & Complement (i.e.~negation) of an event $\Omega$ \\
\toprule
\end{tabular}
\def\arraystretch{1}
\end{table}

We list some basic notation in Table~\ref{tab:notation}.
Apart from the notation in 
that table,
we will also need to use some other concepts, which we shall define below, and briefly elaborate on in some cases. 
In this work, we will assume that all systems are finite-dimensional, but we will not impose any bounds on the system dimensions unless otherwise specified. 

\begin{definition}\label{def:freq}
(Frequency distributions) For a string $\str{z}\in\mathcal{Z}^n$ on some alphabet $\mathcal{Z}$, $\freq_{\str{z}}$ denotes the following probability distribution on $\mathcal{Z}$:
\begin{align}
\freq_{\str{z}}(z) \defvar \frac{1}{n} \sum_{j=1}^{n} \delta_{z,z_j}.
\end{align}
\end{definition}

\begin{definition}
(2-universal hashing) A \term{2-universal family of hash functions} is a set $\mathcal{H}$ of functions from a set $\mathcal{X}$ to a set $\mathcal{Y}$, such that if $h$ is drawn uniformly at random from $\mathcal{H}$, then
\begin{align}
\pr{h(x) = h(x')} \leq \frac{1}{|\mathcal{Y}|} \qquad \forall x\neq x'.
\label{eq:2universal}
\end{align}
\end{definition}
Qualitatively, 2-universal hashing is a procedure which is intended to have some properties similar to choosing a random function $\mathcal{X} \to \mathcal{Y}$ to produce the hash. Specifically,~\eqref{eq:2universal} is the statement that the probability of two distinct inputs hashing to the same value (i.e.~having a hash collision) is no higher than that of the latter process. 2-universal hashing plays an important role in several steps of QKD protocols --- most straightforwardly, it is used to verify whether the honest parties have obtained matching keys; however, it is also used at other points, which we shall later explain. 

\newpage
\begin{definition}
A state $\rho \in \dop{\leq}(CQ)$ is said to be a \term{classical-quantum state} (cq state for short) if it is in the form 
\begin{align}
\rho_{CQ} = \sum_c \alpha_c \pure{c} \otimes \sigma_c,
\label{eq:cq}
\end{align}
for some normalized states $\sigma_c \in \dop{=}(Q) $ and weights $\alpha_c \geq 0$.
Analogously, we can define states that are qc, ccq, cqq and so forth. It may be convenient to absorb the weights $\alpha_c$ into the states $\sigma_c$, writing them as subnormalized states $\omega_c = \alpha_c\sigma_c \in \dop{\leq}(Q)$ instead. 
\end{definition}

\begin{definition}\label{def:cond}
(Conditioning on classical events) For a classical-quantum state $\rho \in \dop{\leq}(CQ)$ written in the form
$\rho_{CQ} = \sum_c \pure{c} \otimes \omega_c$ 
for some $\omega_c \in \dop{\leq}(Q)$,
and an event $\Omega$ defined on the register $C$, we will define the following ``conditional states'':
\begin{align}
\rho_{\land\Omega} \defvar \sum_{c\in\Omega} \pure{c} \otimes \omega_c, \qquad\qquad \rho_{|\Omega} \defvar \frac{\tr{\rho}}{\tr{\rho_{\land\Omega}}} \rho_{\land\Omega} = \frac{
\sum_{c} \tr{\omega_c}
}{\sum_{c\in\Omega} \tr{\omega_c}} \rho_{\land\Omega} .
\end{align}
We informally refer to these states as the subnormalized and normalized conditional states respectively (the latter is perhaps a slight misnomer if $\tr{\rho}<1$, but this issue does not arise in this work). The process of taking subnormalized conditional states is commutative and ``associative'', in the sense that for any events $\Omega,\Omega'$ we have $(\rho_{\land\Omega})_{\land\Omega'} = (\rho_{\land\Omega'})_{\land\Omega} = \rho_{\land(\Omega\land\Omega')}$; hence for brevity we will denote all of these expressions as
\begin{align}
\rho_{\land\Omega\land\Omega'} \defvar (\rho_{\land\Omega})_{\land\Omega'} = (\rho_{\land\Omega'})_{\land\Omega} = \rho_{\land(\Omega\land\Omega')}.
\end{align}
On the other hand, some disambiguating parentheses are needed when combined with taking normalized conditional states.
\end{definition}

In light of the preceding two definitions, it is reasonable to write a normalized cq state $\rho \in \dop{=}(CQ)$ in the form
\begin{align}
\rho_{CQ} = \sum_c \pr{c} \pure{c} \otimes \rho_{Q|c},
\label{eq:cqnorm}
\end{align}
where $\pr{c}$ is the probability distribution describing $C$, and $\rho_{Q|c}$ can indeed be interpreted as the normalized conditional state on $Q$ corresponding to the register $C$ taking value $c$, i.e.~$\rho_{Q|c} = \tr[C]{\rho_{|\Omega}}$ where $\Omega$ is the event $C=c$.

\begin{definition}
(Distinguishability measures) For $\rho,\sigma \in \dop{=}(A)$, the \term{trace distance} and \term{fidelity} are respectively
\begin{align}
d(\rho,\sigma) \defvar \frac{1}{2}\norm{\rho-\sigma}_1 , \qquad 
F(\rho,\sigma) \defvar \norm{\sqrt{\rho}\sqrt{\sigma}}_1 .
\end{align}
For $\rho,\sigma \in \dop{\leq}(A)$, the \term{generalized fidelity} is
\begin{align}
\gF(\rho,\sigma) \defvar \norm{\sqrt{\rho}\sqrt{\sigma}}_1 + \sqrt{(1-\tr{\rho})(1-\tr{\sigma})},
\end{align}
and the \term{purified distance} is $\pd(\rho,\sigma)\defvar\sqrt{1-\gF(\rho,\sigma)^2}$. 
If either $\rho$ or $\sigma$ is normalized, the fidelity and generalized fidelity are equal.
\end{definition}

The trace distance and fidelity satisfy the \term{Fuchs--van de Graaf inequalities}, 
\begin{align}
d(\rho,\sigma) + F(\rho,\sigma) \geq 1, \qquad d(\rho,\sigma)^2 + F(\rho,\sigma)^2 \leq 1.
\label{eq:fuchsvdg}
\end{align}

\section{Smoothed entropies and their operational relevance}
\label{sec:minmaxent}

Our security proofs will make use of smoothed min- and max-entropies, 
which we shall now define and list some properties of.
For the definitions, we follow the presentation in~\cite{DFR20,DF19}, which can be shown to be equivalent to the definitions in~\cite{Tom16,TL17}. Note that we will not need to explicitly use these definitions in this work, but we state them for completeness and to ensure consistency with other work.

\begin{definition}
For $\rho\in\dop{\leq}(AB)$, the \term{min- and max-entropies of $A$ conditioned on $B$} are
\begin{align}
\Hmin(A|B)_\rho &\defvar 
-\log 
\min_{\substack{\sigma \in \dop{\leq}(B) \suchthat\\ \ker(\rho_B)\subseteq\ker(\sigma_B)}} 
\norm{\rho_{AB}^\frac{1}{2}
(\id_A \otimes \sigma_{B})
^{-\frac{1}{2}}}_\infty^2,
\\ 
\Hmax(A|B)_\rho &\defvar \log 
\max_{\sigma \in \dop{\leq}(B)} 
\norm{\rho_{AB}^\frac{1}{2}
(\id_A \otimes \sigma_{B})
^\frac{1}{2}}_1^2, 
\end{align}
where in the first equation the 
$(\id_A \otimes \sigma_{B})
^{-\frac{1}{2}}$ term
should be understood in terms of the Moore-Penrose generalized inverse.
In both equations, 
the optimum is indeed attained (see 
e.g.~Sec.~6.1.2 and Sec.~6.1.3
of~\cite{Tom16}), and it can be attained by a normalized state, so $\dop{\leq}(B)$ can be replaced by $\dop{=}(B)$ without loss of generality.
When $\tr{\rho}=1$, these definitions 
respectively coincide with the R\'{e}nyi entropies $\widetilde{H}^\uparrow_\infty$ and $\widetilde{H}^\uparrow_{1/2}$ 
in the notation of~\cite{Tom16}.
\end{definition}

As their names suggest, the min- and max-entropies are respectively lower and upper bounds on the von Neumann entropy:
\begin{align}
\Hmin(A|B)_\rho \leq H(A|B)_\rho \leq \Hmax(A|B)_\rho.
\end{align}
The min-entropy has a simple operational interpretation for normalized cq states, as shown in e.g.~\cite{KRS09} (which also describes operational interpretations of min- and max-entropy that apply to fully quantum systems). Namely, for a 
cq state $\rho \in \dop{=}(CQ)$ in the form~\eqref{eq:cqnorm}, 
we have 
\begin{align}
\Hmin(C|Q)_\rho = -\log \pg(C|Q)_\rho, \text{ where }
\pg(C|Q)_\rho \defvar \max_{\{\Lambda_c\}} \sum_c \pr{c} \tr{
\rho_{Q|c} 
\Lambda_c},
\label{eq:Hminop}
\end{align}
with the maximization taking place over POVMs $\{\Lambda_c\}$ on $Q$. As suggested by the notation, $\pg(C|Q)_\rho$ is just the optimal probability of guessing $C$ given access to system $Q$.
This operational interpretation is often useful when studying DI protocols, as we shall explain in Sec.~\ref{sec:sketchcoll} and~\ref{sec:securityAD}.

Apart from this fairly straightforward relation, there are also more elaborate interpretations of these entropies that are vital for QKD. Informally, the min-entropy $\Hmin(C|Q)_\rho$ of a cq state $\rho$ approximately corresponds to a lower bound on the length of 
secret key that can be produced from the system $C$, against an adversary with access to the system $Q$. As for the max-entropy $\Hmax(C|Q)_\rho$, it approximately corresponds to an upper bound on the minimum number of bits that must be communicated to a party holding the system $Q$, in order to produce a guess for $C$ that is correct with high probability.
These tasks are often referred to as \term{privacy amplification} and \term{one-way error correction} respectively.
However, it turns out that these bounds are much too pessimistic when accounting for the fact that a small amount of error is allowed in these tasks --- instead, we can get near-optimal bounds (that we shall soon describe in Facts~\ref{fact:LHL}--\ref{fact:EC} below) by considering the \emph{smoothed} versions of these entropies, which are defined by taking a optimization over neighbouring\footnote{Earlier versions of the smoothed entropies (e.g.~\cite{rennerthesis}) were defined using trace distance as the smoothing parameter; however, the purified distance is often more convenient for theoretical analysis, and we follow that convention here.} states: 
\begin{definition}
For $\rho\in\dop{\leq}(AB)$ and $\eps\in\left[0,\sqrt{\tr{\rho_{AB}}}\right)$, the \term{$\eps$-smoothed min- and max-entropies of $A$ conditioned on $B$} are
\begin{align}
\Hmin^\eps(A|B)_\rho \defvar
\max_
{\substack{\tilde{\rho} \in \dop{\leq}(AB) \suchthat\\ \pd(\tilde{\rho},\rho)\leq\eps}}
\Hmin(A|B)_{\tilde{\rho}}, 
\qquad
\Hmax^\eps(A|B)_\rho \defvar
\min_
{\substack{\tilde{\rho} \in \dop{\leq}(AB) \suchthat\\ \pd(\tilde{\rho},\rho)\leq\eps}} 
\Hmax(A|B)_{\tilde{\rho}}.
\end{align}
\end{definition}

By construction, we have $\Hmin^\eps(A|B)_\rho \geq \Hmin(A|B)_\rho$ and $\Hmax^\eps(A|B)_\rho \leq \Hmax(A|B)_\rho$, and equality holds for $\eps=0$.
Importantly, the gap in each inequality can be rather significant when $\eps>0$, a fact which is cleanly demonstrated by a result known as the (quantum) \term{asymptotic equipartition property} (AEP). In the form presented in~\cite{DFR20} (other versions were previously derived in e.g.~\cite{TCR09}, but they are slightly less tight in the regimes of interest for our work), the AEP states the following:
\begin{fact}\label{fact:AEP}
(Asymptotic equipartition property, as presented in Corollary~4.10 of~\cite{DFR20}) Consider any $\sigma \in \dop{=}(AB)$, and let 
$\str{A}\str{B}$ denote $n$ copies of the registers $AB$.
Then for any $\eps\in(0,1)$, we have 
\begin{align}
\Hmin^\eps(\str{A}|\str{B})_{\sigma^{\otimes n}} 
> n H(A|B)_\sigma - \sqrt{n} \, (2\log(1+2 \dim(A) ))\sqrt{\log\frac{2}{\eps^2}}, \label{eq:AEPmin} \\
\Hmax^\eps(\str{A}|\str{B})_{\sigma^{\otimes n}} 
< n H(A|B)_\sigma + \sqrt{n} \, (2\log(1+2 \dim(A) ))\sqrt{\log\frac{2}{\eps^2}}, \label{eq:AEPmax}
\end{align}
where $ \dim(A) $ is the dimension of $A$.
\end{fact}

The AEP essentially states that for any nonzero value of the smoothing parameter, the smoothed min-entropy 
of $n$ IID copies of a state $\sigma$ is approximately lower-bounded 
by $n$ times the von Neumann entropy, up to some $O(\sqrt{n})$ corrections which would become relatively negligible at large $n$. In contrast, the unsmoothed min-entropy can be shown to satisfy 
\begin{align}
\Hmin(\str{A}|\str{B})_{\sigma^{\otimes n}} 
= n \Hmin(A|B)_\sigma.
\end{align}
Hence in the large-$n$ limit, the difference between $\Hmin^\eps(\str{A}|\str{B})_{\sigma^{\otimes n}}$ and $\Hmin(\str{A}|\str{B})_{\sigma^{\otimes n}}$ is approximately the difference between $n H(A|B)_\sigma$ and $n \Hmin(A|B)_\sigma$, which can be quite substantial.
A similar result holds for the smoothed max-entropy.
This implies that for privacy amplification and error correction, a significant improvement can be achieved by using the smoothed min- and max-entropies rather than the unsmoothed versions.

With this in mind, let us now turn to precise statements regarding those tasks. 
First, for privacy amplification, the central result is the \term{Leftover Hashing Lemma} (LHL), which we state below in the form given in~\cite{TL17} (obtained via a small modification of the proof in~\cite{rennerthesis}):
\begin{fact}\label{fact:LHL}
(Leftover Hashing Lemma, as presented in Prop.~9 of~\cite{TL17}) Consider any $\sigma \in \dop{\leq}(CQ)$ where $C$ is a classical $n$-bit register. Let 
$\mathcal{H}$ be 
a 2-universal family of hash functions from $\mathbb{Z}_2^n$ to $\mathbb{Z}_2^\ell$, 
and let $H$ be a register of dimension $|\mathcal{H}|$.
Define the state
\begin{align}
\omega_{KCQH} \defvar \mathcal{E}\!\left(\sigma_{CQ} \otimes 
\idk_H
\right),
\end{align}
where the map $\mathcal{E}$ represents the (classical) process of applying the hash function specified in the register $H$ to the register $C$, and recording the output in register $K$.
Then for any $\es \in \left[0,\sqrt{\tr{\sigma_{CQ}}}\right)$, we have
\begin{align}
\frac{1}{2}\norm{\omega_{KQH} - \idk_K \otimes\omega_{QH}}_1 \leq 2^{-\frac{1}{2}(\Hmin^{\es}(C|Q)_\sigma - \ell + 2)} + 2\es.
\label{eq:LHL}
\end{align}
\end{fact}

In the above theorem, we should interpret the register $K$ as the secret key, and the trace-distance term in~\eqref{eq:LHL} quantifies how close it is to a uniformly random value that is independent of the ``side-information'' registers $QH$. 
The bound tells us that as long as we choose the length $\ell$ of the key $K$ to be somewhat less than $\Hmin^\eps(C|Q)_\sigma$, the trace-distance term will be forced to be small. This is more or less the definition of a secret key (see Sec.~\ref{sec:securitydef}), and hence we can interpret it as saying that the procedure specified in the theorem can produce a secret key of length slightly less than $\Hmin^\eps(C|Q)_\sigma$. (Notice that the register $H$ is included in the ``side-information'' term in~\eqref{eq:LHL}, so $H$ can be publicly announced in the process; however, it must be originally independent from the state $\sigma$.) 

Next, for one-way error correction, we have the following result from~\cite{RR12} (similar results were obtained earlier in~\cite{rennerthesis,RW05}, but again, this version is tighter for the parameter regimes in this work):
\begin{fact}\label{fact:EC}
(One-way error correction, as presented in Theorem~1 of \cite{RR12}) Suppose Alice and Bob respectively hold registers $C$ and $Q$ of a cq state $\sigma \in \dop{=}(CQ)$, and the task is for Alice to send $\ell$ bits to Bob such that he can produce a guess for $C$ that is correct with probability at least $1-\eps$. Then for any $\ez\in[0,\eps)$, there exists a protocol which achieves this with
\begin{align}
\ell = 
\Hmax^{\ez}(C|Q)_\sigma + 2\log\frac{1}{\eps-\ez} + 4.
\end{align}
\end{fact}
Technically, there is something of a limitation in the above result, in that the protocol it specifies may be difficult to implement in practice (in contrast, Fact~\ref{fact:LHL} describes a relatively simple protocol for privacy amplification). Fortunately, there do exist practical error-correction protocols with performance somewhat close to this threshold, at least heuristically. We discuss further details in Sec.~\ref{sec:EC} (see also the error-correction protocol implemented in the~\cite{arx_NDN+21} experiment).

\begin{remark}
A subtlety in the above error-correction theorem is that for most results in this vein, the procedure that Alice and Bob actually implement will depend on a full description of the state $\sigma$ (or at least some ``detailed'' information about it) --- it is not sufficient to simply have the value of $\Hmax^{\ez}(C|Q)_\sigma$. As an extreme example to demonstrate this, let $\sigma^{\mathrm{copy}}$ be a state where $Q$ is simply a classical copy of $C$, and let $\sigma^{\mathrm{flip}}$ be a state where $Q$ is the bitwise complement of $C$ (taking $C$ to be a bitstring register). Then we trivially have $\Hmax^{\ez}(C|Q)_{\sigma^{\mathrm{copy}}}=\Hmax^{\ez}(C|Q)_{\sigma^{\mathrm{flip}}} = 0$. 
However, the ``error-correction'' procedure that works for $\sigma^{\mathrm{copy}}$ (trivially, Bob just outputs $Q$ as his guess for $C$, without receiving any communication from Alice) fails completely for $\sigma^{\mathrm{flip}}$. Hence we see that \emph{only} knowing the value of $\Hmax^{\ez}(C|Q)_\sigma$ is insufficient, so some care needs to be taken when applying this theorem.

On that note, the protocol in the above theorem only guarantees that Bob's guess is correct with high probability for the particular state $\sigma$ that the protocol is designed to work on (and possibly some other states, but we put this aside for now). If it turns out that Alice and Bob share some other state, the protocol does not (in fact, quite obviously \emph{cannot}, given the example above) by itself guarantee that Bob's guess is correct with high probability. That particular protocol also does not give any indication when Bob's guess is incorrect. Hence additional steps are often necessary to deal with this when implementing error correction in a full protocol for QKD (in which the states may not behave as intended, due to dishonest behaviour), which we shall discuss when providing a protocol outline in Sec.~\ref{sec:protsketch} below.

{To some extent, these issues are related to the fact that the error-correction protocols we discussed above do not have an abort outcome ``built in'' (we will analyze this as a separate feature in Sec.~\ref{sec:protsketch}). For protocols with an abort outcome, it is possible to make statements that hold more ``universally'' for arbitrary states $\sigma$, but the parameter $\eps$ usually has a different interpretation in that case (see~\cite{rennerthesis} for details) --- roughly, it is a bound on the probability that Bob's guess is wrong \emph{and} the protocol does not abort (a caveat: this is very different from the probability that Bob's guess is wrong \emph{conditioned} on the protocol not aborting; see Sec.~\ref{sec:securitydef} for some difficulties with such bounds).}
\end{remark}

The bounds in the above two theorems can be proven to be almost tight~\cite{RR12}, 
so in some sense the smoothed min- and max-entropies are indeed the ``right'' quantities to describe these tasks.
We can gain some informal intuition of why the concept of smoothing is useful in these tasks as follows (it basically uses the fact that some ``error'' is allowed, i.e.~the RHS of~\eqref{eq:LHL}, or the probability of guessing wrongly in error correction). 
Taking privacy amplification as an example, let us suppose one has already proven an ``unsmoothed LHL'' that states it is possible to produce an $\eps$-secret key of length (approximately) equal to the min-entropy $\Hmin(C|Q)$, where by $\eps$-secret we mean the RHS of~\eqref{eq:LHL} has value $\eps$. 
Now take a state $\sigma$ that instead has \emph{smoothed} min-entropy $\Hmin^{\es}(C|Q)_\sigma = \kappa$. By definition, this means there is some $\es$-close state $\sigma'$ such that $\Hmin(C|Q)_{\sigma'} = \kappa$, and by the ``unsmoothed LHL'', privacy amplification on $\sigma'$ can produce an $\eps$-secret key of approximate length $\kappa$. But since $\sigma$ and $\sigma'$ are $\es$-close, this implies\footnote{Technically, recalling the definitions used in this work, the ``closeness'' measure between $\sigma$ and $\sigma'$ is purified distance, while the secrecy condition is based on trace distance, so this argument is implicitly relying on the fact that the purified distance upper bounds the trace distance~\cite{Tom16}. If one chooses to use other combinations of distance measures, the Fuchs--van de Graaf inequalities can be used to convert the parameters appropriately.} that applying the same procedure to $\sigma$ produces an $(\eps+O(\es))$-secret key of the same length $\kappa$ (since the same procedure was applied). Hence from the ``unsmoothed LHL'', we have obtained an LHL that gives key lengths based on $\es$-smoothed min-entropy instead, at the cost of only an additional $O(\es)$ ``penalty'' to the secrecy --- this is in fact a fairly standard approach to proving the LHL in this form; see e.g.~\cite{rennerthesis,TL17}. Recall from our AEP example that smoothing by just a small value of $\es$ can significantly increase the value of the smoothed min-entropy, so this is a ``useful tradeoff'' to make in practice. Similar reasoning applies for error correction, with the caveat that the error-correction procedure to consider is the optimal one for the state that attains the minimum in the smoothing, \emph{not} the optimal one for the original state.

\section{Typical protocol outline}
\label{sec:protsketch}

We now give an outline of the rough structure that the majority of current DIQKD protocols follow, incorporating some techniques that have recently been proven to improve the keyrates. However, note that there remains some flexibility in the exact manner in which each of these steps is implemented, which leads to subtle differences in protocol performance and security proofs. We shall briefly discuss some of these issues below.
\let\oldthealgorithm\thealgorithm 
\renewcommand{\thealgorithm}{(sketch)}
\begin{savenotes}
\begin{breakablealgorithm}
\caption{} 
\label{prot:outline}
\begin{algorithmic}[1]
\State Alice and/or Bob choose a small subset of protocol rounds that will be used to decide whether to accept or abort, typically referred to as \term{test rounds}. The rest of the rounds shall be referred to as \term{generation rounds}.
\State \textbf{Measurement:} In each round, Alice and Bob's devices each receive some share of a quantum state. Alice and Bob select inputs to their devices according to some distributions, which may depend on whether the round is a test or generation round, 
and record the corresponding device outputs. After all the measurements have been performed, they publicly announce the inputs that were used (we shall denote them as $\str{X}$ and $\str{Y}$).
\State \textbf{Sifting:} Alice and Bob may erase\footnote{For ease of later analysis, here we interpret ``erasing'' to mean overwriting their output with some deterministic value. Other protocols may describe this step in terms of ``discarding'' these rounds instead, but this would result in the number of rounds left after sifting being a random variable, which is inconvenient to describe in our subsequent proofs.} the outputs of some generation rounds (informally, those in which weaker correlations are expected --- the most straightforward example would be erasing all rounds in which the inputs correspond to the honest measurements being in different bases).
\State \textbf{Noisy preprocessing:} For all remaining generation rounds, Alice may add a small amount of trusted noise to her output.
\State \textbf{Error correction:} At this point, Alice holds some string $\str{A}$, and Bob holds some string $\str{B}$. Alice and Bob publicly communicate some bits $\str{L}$ for error correction, in the sense that Bob can use $\str{L}$ and $\str{B}$ to produce a guess $\tilde{\str{A}}$ for Alice's string $\str{A}$. Depending on the nature of the error correction protocol, it may be possible for the protocol to abort at this step. 
\State \textbf{Parameter estimation:} 
Using the test rounds, Alice and/or Bob compute the frequencies of each input-output combination (or more broadly, statistical estimators of one or more parameters of the form described in~\eqref{eq:bellop}). Based on whether these statistics fall within some tolerance range, 
they decide whether to accept or abort the protocol.
\State \textbf{Privacy amplification:} If the protocol has not aborted by this point, Alice and Bob apply a privacy amplification procedure on $\str{A}$ and $\tilde{\str{A}}$ respectively to obtain final keys $K_A$ and $K_B$.
\end{algorithmic}
\end{breakablealgorithm}
\end{savenotes}
\let\thealgorithm\oldthealgorithm 
\addtocounter{algorithm}{-1} 

In the parameter-estimation step of the above sketch, the range of accepted values for the observed statistics would be based on the honest specification, although typically with some tolerance interval (to avoid having the honest implementation abort too frequently from statistical fluctuations). For instance, if the honest devices are expected to achieve a CHSH value of $\constr_\mathrm{exp}$, then typically the parameter-estimation step would accept if the corresponding ``statistical estimate'' (informally speaking) is greater than $\constr_\mathrm{exp}-\dtol$ for some $\dtol>0$. We shall denote the number of rounds in the protocol as a parameter $n$, and denote the length of the final keys produced as $\lkey$ --- in general, the value of $\lkey$ (and possibly some other parameter choices in the protocol) would vary with $n$.

With this general structure in mind, we can define the finite-size \term{keyrate} (for a fixed $n$) and the \term{asymptotic keyrate} of such a protocol to be respectively
\begin{align}
\operatorname{rate}_n \defvar \frac{\lkey}{n},  \qquad \operatorname{rate}_\infty \defvar \lim_{n\to\infty} \operatorname{rate}_n.
\label{eq:defkeyrates}
\end{align}
{(To be more precise, some definitions of the asymptotic keyrate technically require taking the $\eps\to0$ limit of some security parameter $\eps$, but we defer this discussion to Sec.~\ref{sec:sketchasympt}, after introducing the security definitions.)} 

The above outline incorporates two recently studied techniques that can improve the keyrates, which we shall refer to as follows:
\begin{itemize}
\item \textbf{Random key measurements:} 
This refers to having Alice and Bob selecting random inputs in generation rounds, in constrast to more commonly-studied protocols~\cite{PAB+09,ARV19} that used a deterministic input in all generation rounds.\footnote{A notable exception is the protocol in~\cite{JMS20}. However, it achieves this by exploiting the rather specific structure of the nonlocal game chosen for that protocol, in which (for the honest implementation) \emph{every} input pair produces some output bits that are highly correlated. Hence it is not obvious how to generalize this approach to DIQKD protocols based on other nonlocal games.} The latter approach has the advantage of usually making the sifting step redundant --- by always using highly correlated measurements in generation rounds, no rounds need to be erased.\footnote{To be more accurate, the idea of using deterministic generation-round inputs was originally an \emph{improvement} on the initial \cite{BB84}~protocol, precisely because it removes the sifting in generation rounds. It is rather surprising that there are some situations in which this turns out to not be optimal.} However, it was discovered in~\cite{SGP+21} that there are situations in which probabilistic choices of generation-round inputs can yield higher keyrates, even after accounting for the sifting-induced loss.

\item \textbf{Noisy preprocessing:} 
This was previously studied in the context of standard QKD, where it was shown that adding some trusted noise to one party's outputs can in fact increase the keyrates. Recently, the analogous result was proven for DIQKD as well~\cite{HST+20,WAP21,SBV+21}. 

\end{itemize}
\noindent We will return to these techniques in Sec.~\ref{sec:sketchimpr} to give some informal intuition of why they can improve the keyrates. 

Also, note that the outline here is mainly focused on protocols in which the error-correction step consists of Bob producing a guess for Alice's string $\str{A}$ directly. We shall informally refer to such protocols as one-way protocols, because in such protocols the error-correction communication $\str{L}$ typically\footnote{The exception would be protocols such as Cascade, where bits are communicated in both directions but the final goal is still for Bob to produce a guess for $\str{A}$.} consists only of a single string sent from Alice to Bob. However, there can be more elaborate procedures, often involving communication in both directions, in which the goal of that step is for both Alice and Bob to generate some \emph{new} strings that are (ideally) highly correlated, and may be rather different from Alice's original string $\str{A}$. Such procedures may be referred to broadly as {advantage distillation}, and we discuss this in Chapter~\ref{chap:AD}. 

As mentioned above, concrete implementations of DIQKD protocols may differ in the details of various steps in the above sketch. We list some of them as follows:

\begin{itemize}
\item First is the question of how the test rounds are chosen. One approach may be to select a subset of fixed size uniformly at random. Another approach would be to make the test/generation decision for each round in an IID manner. Roughly speaking, the former approach usually performs better when working under the assumption of collective attacks; however, without that assumption, the proof technique used in this work (based on the EAT) appears to require the latter approach. We return to this in more detail in Chapter~\ref{chap:finite}. 

\item 
Another issue is the question of whether \emph{both} Alice and Bob know which rounds are test rounds, which may be needed for them to have the option of using different input distributions in test versus generation rounds. Importantly, this cannot simply be agreed on by public communication before the quantum states are distributed to the devices --- if Eve knows this information, she can simply send honest states in test rounds and completely insecure states in generation rounds. 
A trivial way to address this without imposing additional protocol requirements is simply for one or both of the parties to always use the same input distribution, without knowing which rounds are test rounds until the inputs are publicly announced after the measurements. For simplicity, the main protocol we study in this work 
(Protocol~\ref{prot:DIQKD} in Chapter~\ref{chap:finite}) 
uses this approach. However, this has the drawback that either a particular input for each party must be used with high probability (in which case the parameter estimation must account for this biased input distribution), or there are substantial losses in the sifting step (since Alice and Bob will be measuring in poorly correlated bases for a large fraction of generation rounds, as was the case for e.g.~the initial \cite{BB84}~protocol). We discuss various alternatives in Sec.~\ref{sec:preshared}.

\item 
There are several possible options in the error correction step, not merely in terms of the encoding used to construct $\str{L}$, but rather the requirements and guarantees of the error correction procedure. In this work, we focus on approaches where Bob first produces his guess $\tilde{\str{A}}$ (which is intended to be correct with high probability given the honest devices), and then performs a ``verification'' procedure (see Sec.~\ref{sec:proofsketch}) that ensures that if his guess is wrong, the protocol aborts with high probability, even if the devices are not behaving honestly. On the other hand, some other error correction procedures do not perform a verification step in this fashion, but instead implicitly require the parties to first estimate whether the error rate 
is sufficiently low, and abort if it is not (see e.g.~\cite{rennerthesis}). If one uses such an approach, it may be expedient to swap the order of the error correction and parameter estimation steps, such that the error rate estimation is incorporated into parameter estimation.

\item
The above outline does not explicitly specify how the test-round outputs are communicated to perform parameter estimation. The most straightforward approach is simply to publicly announce these outputs; however, in the non-IID case we cannot rule out the possibility that this communication leaks information about the generation-round outputs, which slightly reduces the final key length (see Sec.~\ref{sec:sketchbeyond}). To reduce this effect, an alternative approach is for Bob's guess $\tilde{\str{A}}$ in the error-correction step to also include the test-round outputs, not just the generation-round outputs. This helps to reduce the amount of leakage, because the error-correction string is typically shorter than the raw data. We mainly focus on this approach in this work, in order to optimize the finite-size keyrates, but it comes at the cost of additional steps in the security proof to account for the possibility that Bob's guess for the test-round data may be wrong. (Also, it precludes the option mentioned above of performing parameter estimation before error correction.) A related question is whether the test rounds themselves are included when performing privacy amplification, the effects of which we briefly discuss in Sec.~\ref{sec:coll}. 
Note, however, that these points only affect the keyrates by an amount on the order of the fraction of test rounds, which typically must be chosen to be small in the non-IID case (see Sec.~\ref{sec:coll}), though it may still be relevant for potential experimental demonstrations.

\item 
Another issue is whether the length of the final key is fixed before the protocol begins. In principle, one could consider a protocol where $\lkey$ is chosen as a function of the statistics computed during the parameter estimation step --- informally, a shorter $\lkey$ might be chosen if the statistics indicate more potential eavesdropping. Similarly, it might be possible to adjust the length of the error-correction string $\str{L}$ depending on some observed estimate of the error rate. However, for simplicity we will focus here on protocols where $\lkey$ is a fixed value, chosen before the protocol begins. (This essentially means we will design the protocol such that it simply aborts if there is not enough entropy to securely produce a key of length $\lkey$, even if it might potentially be possible to produce a shorter key.) In line with this, we will focus on protocols where the length of $\str{L}$ has an upper bound that is fixed beforehand, and cannot be changed even if the error rate in a run of the protocol is higher than expected. (To account for the fact that Bob may not guess $\str{A}$ correctly in that case, recall we will include a step that ensures the protocol aborts with high probability if the guess is wrong.)

\end{itemize}

\subsection{Noisy honest devices}

Note that in the above outline, we have implicitly relied on the idea that some honest implementation has been specified, which produces a shared secret key (with high probability) when everything proceeds as intended. 
We highlight that in general, this honest implementation does not need to be free of noise; for instance, it does not necessarily have to involve distributing a perfect maximally entangled state to the devices.
To avoid ambiguity, we shall use the terms \term{honest} devices and \term{noiseless} devices --- the former refers to the behaviour we expect ``in practice'' when there is no \emph{unintended} deviation from the specified implementation, while the latter refers to the behaviour when the state and measurements are in some idealized ``perfect'' form, free of any form of noise. As a rough rule of thumb, most security proofs only really require a description of the honest devices; the noiseless devices merely serve as a convenient ``intermediate'' step to describe the honest devices.

To give a concrete example, 
many basic DIQKD protocols~\cite{PAB+09,ARV19,HST+20} make use of the CHSH value. In that case, a commonly used simple model for the noiseless devices would be for them to implement IID states and measurements, 
specifically the following ones (for brevity we describe the measurements as hermitian observables rather than specifying the individual projectors, using $X,Y,Z$ to denote the Pauli matrices):
\begin{align}\label{eq:noiselessCHSH}
\begin{array}{l}
\text{State} \\
\ket{\Phi^+}\defvar (\ket{00}\pm\ket{11})/\sqrt{2}, \\
{}
\end{array}
\qquad
\begin{array}{l}
\text{Alice's measurements} \\
x=0:\, Z, \\
x=1:\, X, 
\end{array}
\qquad
\begin{array}{l}
\text{Bob's measurements}\\
y=0:\, (Z+X)/\sqrt{2}, \\
y=1:\, (Z-X)/\sqrt{2} .
\end{array}
\end{align}
This combination of state and measurements is useful because it achieves the maximum possible CHSH value within quantum theory (and due to self-testing properties of the CHSH inequality, is in fact essentially the only one that does so, up to local isometries/ancillas). Strictly speaking, in those protocols Bob usually has an additional measurement ($y=2$), corresponding to the same basis as Alice's $x=0$ measurement (i.e.~a Pauli-$Z$ measurement). However, this is used only to improve the key generation rates and is not involved in the CHSH value.

On the other hand, the honest devices in those protocols are usually taken to be somewhat more ``realistic'' devices that do not perfectly implement the above states and measurements. A simple model for such devices would be \term{depolarizing noise}, 
in which the honest devices are described by a parameter $\q \in [0,0.5]$, producing the output distribution
\begin{align}
\pr{ab|xy} = (1-2\q) \prnonoise[ab|xy] + 2\q\, \frac{1}{|\mathcal{A}| |\mathcal{B}|}.
\label{eq:depol}
\end{align}
Here, $\prnonoise[ab|xy]$ refers to some noiseless distribution, which could for instance be produced by~\eqref{eq:noiselessCHSH}, but more generally can be some other noiseless implementation. 
If the outputs are binary-valued and some input pair $\tilde{x} \tilde{y}$ satisfies $\prnonoise[ab|\tilde{x} \tilde{y}] = \delta_{a,b}/2$ for the noiseless distribution, then the parameter $\q$ is exactly the probability of Alice and Bob obtaining different outputs under the distribution $\pr{ab|\tilde{x} \tilde{y}}$. This is sometimes called the \term{quantum bit error rate} (QBER); however, some care should be taken here because $\q$ may not be equal to the QBER in other situations. While this model is very simple, it can be motivated by noting that in the 2-input 2-output setting, the distributions produced by the devices can always be forced to be in the form~\eqref{eq:depol} with $\prnonoise[ab|xy]$ given by~\eqref{eq:noiselessCHSH}, by applying a ``depolarization''/``twirling'' procedure~\cite{MAG06} (this is not necessarily optimal, but it is always an option). 

We can interpret~\eqref{eq:depol} as saying that with probability $2\q$, the noiseless distribution is replaced by uniformly random outputs.
Another perspective is that the honest devices implement the noiseless measurements (assumed to be Pauli measurements for simplicity), but instead of noiselessly sharing the state $\ket{\Phi^+}$, they share the \term{Werner state}
\begin{align}
\varrho_\q \defvar (1-2\q)\pure{\Phi^+} + 2\q\, \frac{\id}{4}.
\label{eq:werner}
\end{align}

Another basic noise model would be \term{limited detection efficiency}. This is described by specifying some noiseless distribution $\prnonoise[ab|xy]$ and a detection-efficiency parameter $\eta \in [0,1]$. The honest devices in this case produce a distribution which is based on $\prnonoise[ab|xy]$, but where each party's output (for any input) is replaced by an ``erasure'' symbol $\perp$ with probability $1-\eta$ (independently of the other party). This serves as a simplistic model for single-photon-pair implementations, in which each photon is independently lost with probability $1-\eta$
(combining the effects of fibre loss and detector inefficiency into the single parameter $\eta$).

However, for protocols based on CHSH (for instance), one would need to ensure the distribution has only 2 outputs. For such cases in this work, assuming that $\prnonoise[ab|xy]$ is a 2-output distribution with output labels $\{-1,+1\}$, we shall usually use the approach of deterministically mapping the $\perp$ output for the honest devices to (say) the $-1$ output.\footnote{Another option is to postselect only on the rounds where both parties had a successful detection. However, this introduces the Bell-test ``detection loophole'' if implemented naively. There are various approaches towards addressing this loophole, but in this work we stick to the simple solution of forcibly mapping all the outputs to binary values.} This can be alternatively viewed as applying a classical Z-channel to the outputs of $\prnonoise[ab|xy]$, flipping $+1$ to $-1$ with probability $1-\eta$. (In some situations, it may be useful to preserve the $\perp$ output for some input values; we will discuss this when it arises.)

The limited detection efficiency model is more complicated to analyze than depolarizing noise, because for any fixed value of $\eta$, a different choice of noiseless distribution $\prnonoise[ab|xy]$ is needed to yield the highest CHSH value for the honest distribution~\cite{Ebe93}. (While higher CHSH values do not necessarily correspond to better DIQKD keyrates, they can serve as a useful starting guide.) Furthermore, when studying more realistic photonic implementations, the noiseless distribution $\prnonoise[ab|xy]$ would not be produced by two-qubit systems like in~\eqref{eq:noiselessCHSH}, but rather more complicated states in a photonic Fock space. We do not discuss this here, but examples of such analysis can be found in~\cite{TWF+18,HST+20}.

\section{Security definitions}
\label{sec:securitydef}

In this work, we follow the security definition used in e.g.~\cite{ARV19}. 
Qualitatively, the concepts involved in this definition are: \term{completeness}, meaning that the honest devices will accept with high probability, and \term{soundness}, meaning that the devices remain ``secure'' (possibly by aborting) even in the presence of dishonest behaviour. These notions are formalized as follows:

\begin{definition} \label{def:secure}
Consider a DIQKD protocol such that at the end, the honest parties either \term{accept} (producing keys $K_A$ and $K_B$ of length $\lkey$ for Alice and Bob respectively) or \term{abort} (producing an abort symbol $\perp$ for all parties).
It is said to be 
$\ecom$-complete and $\esound$-sound if the following properties hold:
\begin{itemize}
\item (Completeness) The honest implementation aborts with probability at most $\ecom$.
\item (Soundness) For any implementation of the protocol, we have 
\begin{align}\label{eq:sound}
\pr{\mathrm{accept}} \frac{1}{2} \norm{\sigma_{K_A K_B E'} - \left(\frac{1}{2^{\lkey}} \sum_k \pure{kk}_{K_A K_B} \right) \otimes \sigma_{E'}}_1 \leq \esound,
\end{align}
where $\sigma$ denotes the normalized state conditioned on the protocol accepting, and $E'$ denotes all side-information registers available to the adversary at the end of the protocol.
\end{itemize}
\end{definition}

The completeness condition is an intuitively sensible requirement for a protocol, and more formally, it rules out some trivial protocols that fulfill the soundness condition (such as a protocol that simply always aborts, hence achieving $\pr{\mathrm{accept}}=0$ and thus $\esound=0$).
The soundness condition basically states that $\sigma_{K_A K_B E'}$ is close in trace distance (up to a factor of $\pr{\mathrm{accept}}$)\footnote{This prefactor might appear to be an unacceptable weakening of the security definition, since it allows $\sigma_{K_A K_B E'}$ to be extremely different from $2^{-\lkey} \sum_k \pure{kk}_{K_A K_B} \otimes \sigma_{E'}$ when $\pr{\mathrm{accept}}$ is small. However, as we shall shortly discuss, this definition indeed turns out to be sufficient to ensure security when composing QKD with other protocols. Furthermore, any ``reasonable'' security definition regarding the states \emph{conditioned} on accepting must in fact have some dependence on $\pr{\mathrm{accept}}$ (albeit not necessarily in this exact form) --- this can be seen by noting that Eve could supply completely classical states/devices that give her perfect knowledge of all outputs, in which case even the state conditioned on accepting is trivially insecure. (The protocol's accept probability for such a state is presumably minuscule, but is typically still nonzero due to simple statistical fluctuations, so it is still mathematically valid to condition on the accept event.)
Hence it is impossible to impose a ``reasonable'' security statement regarding such conditional states unless some dependence on $\pr{\mathrm{accept}}$ is included. 
} to a state where 
the honest parties' keys are (1)~always equal, and (2)~completely independent of Eve's side-information $E'$. In fact, it is often convenient to split the soundness condition into a pair of slightly simpler conditions roughly corresponding to those two properties. To make this precise, consider the following definitions:

\begin{definition}
A DIQKD protocol as described above is said to be 
$\ecorr$-correct and $\esecr$-secret if
the following properties hold:
\begin{itemize}
\item (Correctness) For any implementation of the protocol, we have
\begin{align}
\pr{K_A \neq K_B \land \mathrm{accept}} \leq \ecorr.
\label{eq:correct}
\end{align}
\item (Secrecy) For any implementation of the protocol, we have 
\begin{align}
\pr{\mathrm{accept}} \frac{1}{2} \norm{\sigma_{K_A E'} - \idk_{K_A} \otimes \sigma_{E'}}_1 \leq \esecr,
\label{eq:secrecy}
\end{align}
where $\sigma$ is as described in Definition~\ref{def:secure}, and $\idk_{K_A}$ denotes the maximally mixed state (i.e.~a uniformly random key for Alice).
\end{itemize}
\end{definition}
\noindent It is quite straightforward to show (see e.g.~\cite{arx_PR14}) that if a DIQKD protocol is both $\ecorr$-correct and $\esecr$-secret, then it is $(\ecorr+\esecr)$-sound. (Note that the secrecy condition only 
involves $K_A$,
but the intuition is that if correctness holds as well, then 
the state is close to one where both $K_A$ and $K_B$ are independent of $E'$, up to the $\pr{\mathrm{accept}}$ prefactor.) 
Hence in a security proof, we can prove correctness and secrecy separately, then add the corresponding parameters to get the value of $\esound$.

Of course, when choosing security definitions, an important consideration should be their operational relevance. The reason we have chosen Definition~\ref{def:secure} is that in the case of device-dependent QKD, it is sufficient to imply \term{composable security}~\cite{MR11,arx_PR14,arx_PR21} (we discuss this in more detail in Appendix~\ref{app:comp}). 
Broadly, this is the notion of ensuring that security still holds when protocols are composed with each other --- the idea is that if a protocol satisfies a composable security definition, then it can be ``safely'' used in place of some other idealized functionality in any larger protocol. For instance, in the case of QKD, the ideal functionality would be (glossing over some details) one that simply outputs a perfect shared secret key whenever it does not abort. 

Here, because Definition~\ref{def:secure} is based on trace distance, one can give an even more concrete operational statement. 
To give a rough overview (see Appendix~\ref{app:comp} or~\cite{arx_PR14,VPdR19,arx_PR21} for more detailed exposition): it can be shown that if a QKD protocol is both $\ecom$-complete and $\esound$-sound, then it is \term{$\eps$-secure} in a composable sense, with $\eps = \ecom + \esound$.
We can then derive the following guarantee: if we use an $\eps$-secure QKD protocol in place of the ideal functionality, then the maximum probability of \emph{any} ``failure'' event (what constitutes a ``failure'' can be arbitrary, as long as it is a well-defined event) cannot increase by more than $\eps$. The core idea behind this interpretation is that the trace distance is operationally related to distinguishing probability. Hence if two states are close in trace distance, then they must be ``difficult to distinguish'', which implies that the probability of any event based on the two states cannot differ substantially between them. However, this is not the only consideration, as the full analysis is more involved, based on the \term{Abstract Cryptography} framework~\cite{MR11} (which allows more general metrics in principle, though they may not have the above operational interpretation).

On the other hand, for DIQKD against coherent attacks, the situation is more complicated because of the memory issue discussed previously. 
If there are any non-negligible correlations between the public communication in later instances and secret keys in earlier instances, this causes problems if one tries to prove that the definition is composable, because information about the earlier secret keys is being leaked. Hence to restore the same operational implications for DIQKD as in device-dependent QKD, one would likely have to impose some kind of constraint on the device memories, as stated previously. 
However, strictly speaking, the Abstract Cryptography framework has not been fully formalized in the context of DIQKD devices (i.e.~devices that perform some ``uncharacterized'' measurements). Still, it seems plausible that it could be formalized in a way that allows us to impose the constraint on the device memory across instances, and in that case prove that the same operational implications hold as in device-dependent QKD. 

\begin{remark}
We highlight that even without constraints on the device memory across protocol instances, we can still prove that the DIQKD protocols we study satisfy Definition~\ref{def:secure} --- there are no issues in proving that it holds as a ``standalone'' security definition. The memory issue only affects the operational {implications} of this definition when protocols are composed, not the question of whether the definition itself is satisfied.
\end{remark}

\section{Security proof outline}
\label{sec:proofsketch}

We now give a rough sketch of how to prove a protocol is secure in the sense defined above, i.e.~that it satisfies completeness and soundness. The latter will be proven by showing that the protocol satisfies correctness and secrecy, as previously mentioned.

\paragraph{Completeness:}
This is usually straightforward to show, as it is only based on the honest behaviour, which is typically IID. Hence as long as there is enough ``tolerance'' in the parameter estimation step, and the error-correction string $\str{L}$ is chosen to be long enough, it is easy to argue that the honest protocol accepts with high probability. (We shall discuss below exactly how long $\str{L}$ needs to be in order to ensure this.)

\paragraph{Correctness:}
This is also straightforward to prove for error-correction procedures that follow the structure we have previously outlined, i.e.~those that involve a step to verify whether Bob's guess is correct. More precisely, this step consists of Alice sending a 2-universal hash of $\str{A}$ (together with her choice of hash function) to Bob, who compares it with the hash of his guess $\tilde{\str{A}}$, and accepts if and only if the hashes match. By the defining property of 2-universal hashing, this ensures that the protocol aborts with high probability if $\str{A} \neq \tilde{\str{A}}$, and it is easy to show that this implies the correctness property is satisfied (since $K_A$ and $K_B$ are produced from $\str{A}$ and $\tilde{\str{A}}$ respectively). Importantly, notice that this argument does not rely on \emph{any} other properties of the error-correction procedure, such as the probability of producing a correct guess when the devices behave dishonestly. Focusing on error-correction procedures that include this ``verification hash'' lets one prove the correctness property in a very simple way, ``decoupled'' from all other properties of the error-correction step.

\paragraph{Secrecy:}
The main challenge lies in proving this property. 
For ease of explanation, we will first describe in Sec.~\ref{sec:sketchcoll} a proof sketch under the collective-attacks assumption, then later outline in Sec.~\ref{sec:sketchbeyond} how that assumption can be dropped. 
With the collective-attacks assumption, the initial quantum state shared between Alice, Bob and Eve before the measurements are performed is of the form $\rho_{\qA \qB \qE}^{\otimes n}$, where $\qA,\qB,\qE$ are quantum registers. In the subsequent analysis, we will be referring to several different registers Alice holds, so to help reduce confusion, we explicitly list the notation here:
\begin{itemize}
\item $\qA$: Single-round quantum register to be measured by Alice
\item $\str{A}$: Full classical output string Alice holds (after noisy preprocessing and sifting)
\item $\hat{A}$: Single-round classical value Alice holds (after noisy preprocessing and sifting)
\item $\hat{A}_x$: Single-round classical value Alice holds (after noisy preprocessing and sifting), conditioned on her choosing input $x$
\end{itemize}

\subsection{Collective-attacks analysis}
\label{sec:sketchcoll}

Consider the point in the protocol just before the privacy amplification step.
The state storing all the registers of interest at this point\footnote{In a minor abuse of notation, we are using $\rho$ to denote both the pre-measurement quantum state and the state at this point --- note that in particular, the only common register between the states in the two situations is Eve's side-information $\str{E}$, which has the same reduced state in both cases, as long as Eve does not perform any operations on it. (We can indeed assume this without loss of generality because for the purposes of the secrecy definition~\eqref{eq:secrecy}, it is always to Eve's advantage to ``preserve'' her side-information as long as possible, i.e.~she only applies reversible operations to it, in which case we might as well assume she does not operate on it.)} is of the form $\rho_{\str{A}\tilde{\str{A}}\str{B}\str{X}\str{Y}\str{L}\str{E}}$, where $\str{E}$ denotes all the quantum side-information that Eve stored. 
Let us assume for simplicity that the protocol is such that 
all test-round outputs are publicly announced for parameter estimation, and that apart from that,
the only publicly communicated information is the error-correction data $\str{L}$ and the inputs $\str{X}\str{Y}$. (Protocols involving other one-way public communication can usually be accounted for with small modifications to this proof sketch.)

Recall that the protocol aborts during parameter estimation if the observed statistics lie outside of some range of accepted values. 
Since we have imposed the collective-attacks assumption, in each round we have some well-defined single-round state $\rho_{\qA \qB \qE}$ and possible measurements $\pvm_{a|x}, \pvm_{b|y}$. This yield some corresponding parameter values $\tr{\Gamma_j(\pvm_{a|x},\pvm_{b|y}) \rho_{\qA \qB}}$ (see~\eqref{eq:bellop}) --- note that these are just abstract values that cannot be ``directly'' observed, but for our subsequent argument we only require that these are well-defined values. 
For the discussion in this section, let us avoid ambiguity by following statistical terminology conventions: the term \term{parameter} will refer to such underlying abstract properties of the process producing the data, while the term \term{statistic} will refer to random variables that can be actually observed (or computed from the observations) in the protocol.

With this in mind, let us suppose that the accept condition in parameter estimation is that for every $j$, the \emph{statistic} corresponding to an ``estimate'' of $\tr{\Gamma_j(\pvm_{a|x},\pvm_{b|y}) \rho_{\qA \qB}}$ (we remain slightly informal about exactly how the statistic is computed) lies inside some interval $[\constr_j^\mathrm{min},\constr_j^\mathrm{max}]$. 
Now choose some value $\dsou>0$ (we shall shortly discuss the significance of this value), and note that there are two exhaustive possibilities for the \emph{parameters} 
$\tr{\Gamma_j(\pvm_{a|x},\pvm_{b|y}) \rho_{\qA \qB}}$:
\begin{enumerate}
\item \label{case:abort} $\tr{\Gamma_j(\pvm_{a|x},\pvm_{b|y}) \rho_{\qA \qB}} \notin [\constr_j^\mathrm{min}-\dsou,\constr_j^\mathrm{max}+\dsou]$ for at least one $j$
\item \label{case:accept} $\tr{\Gamma_j(\pvm_{a|x},\pvm_{b|y}) \rho_{\qA \qB}} \in [\constr_j^\mathrm{min}-\dsou,\constr_j^\mathrm{max}+\dsou]$ for all $j$
\end{enumerate}
Qualitatively, case~\ref{case:accept} is just the statement that the true parameter values lie ``$\dsou$-close to''\footnote{To avoid confusion, note that this sense of ``close to'' should not be conflated with the previously mentioned tolerance for deviations from the honest behaviour (which in this notation is essentially captured by the difference between $\constr_j^\mathrm{min}$ and $\constr_j^\mathrm{max}$). They are distinct parameters --- informally, the former describes deviations from the accepted set of values, while the latter describes deviations from the honest behaviour, or in other words it \emph{defines} the accepted set of values. In our detailed security proofs later, the latter will be quantified by a parameter $\dtol$ (distinct from $\dsou$), which we will use to basically take $\constr_j^\mathrm{min}=\constr_j^{\mathrm{exp}}-\dtol$,  $\constr_j^\mathrm{max}=\constr_j^{\mathrm{exp}}+\dtol$ for some ``expected'' parameter value $\constr_j^{\mathrm{exp}}$.} the range of values accepted in parameter estimation, while case~\ref{case:abort} is simply the complement of this possibility.
In case~\ref{case:abort}, by exploiting the IID assumption it is a straightforward exercise in classical statistics to show that the protocol aborts with high probability (given suitable choices for the number of test rounds and the value of $\dsou$). 
This implies the secrecy definition~\eqref{eq:secrecy} is trivially satisfied in that case, due to the $\pr{\mathrm{accept}}$ prefactor. Hence the rest of this proof sketch is devoted to case~\ref{case:accept}, which is rather more involved (essentially, we will need bounds on the trace-distance term in~\eqref{eq:secrecy} rather than merely $\pr{\mathrm{accept}}$ alone). 
However, we will not immediately need to use the fact that we are in case~\ref{case:accept} --- instead, we shall first find some way to express the key length $\lkey$ in terms of single-round quantities.

\begin{remark}
Note that the cases are defined purely in terms of the \emph{parameters}, not the \emph{statistics}.
It may be tempting to think that the observed statistics are used to deduce which of the two cases holds. 
However, this is not the right way to view the proof structure here --- using those statistics alone, it is not even possible to assign a probability to each case (this would require an explicit (Bayesian) prior on the device behaviours, which is not part of the security definitions we use). Rather, our security argument does not rely on identifying which case holds, but merely on the fact that one of them must hold, and proving that the secrecy definition is satisfied in either case.
\end{remark}

Let us assume for this sketch that the protocol produces the final key using only the generation rounds --- the collective-attacks assumption implies that the publicly announced test-round data is completely ``decoupled'' from the generation rounds, so this allows us to focus our analysis on the latter. We will also assume that the number of generation rounds is fixed at some value $m$, i.e.~the test rounds are selected by taking a uniformly random subset of size $n-m$. Let $\str{A}_g\str{B}_g\str{X}_g\str{Y}_g\str{E}_g$ denote the generation rounds within $\str{A}\str{B}\str{X}\str{Y}\str{E}$. 
Assuming the generation-round inputs are chosen with an IID distribution (note that in particular, this covers protocols that use a fixed input in all such rounds), the state on registers $\str{A}_g\str{B}_g\str{X}_g\str{Y}_g\str{E}_g$ has an IID structure as well: $\rho_{\str{A}_g\str{B}_g\str{X}_g\str{Y}_g\str{E}_g} = \rho_{\hat{A} \hat{B} X Y \qE}^{\otimes m}$, where $\hat{A}$ denotes a single round of Alice's final output string (after sifting and noisy preprocessing), and similarly $\hat{B}$ for Bob (after sifting).

With these points in mind, let us compute an expression for $\lkey$.
The privacy amplification step involves Alice trying to produce a secret key from $\str{A}_g$ against Eve's side-information $\str{X}_g\str{Y}_g\str{L}\str{E}_g$. (As briefly noted previously, the secrecy condition only involves Alice's key, so it suffices to consider her side of the protocol at this step.) Recalling the Leftover Hashing Lemma (Fact~\ref{fact:LHL}), this means that the final key will satisfy the secrecy definition as long as $\lkey$ is chosen such that\footnote{It might seem slightly odd to state an upper bound on $\lkey$, given that here we are constructing secure protocols and hence proving \emph{lower} bounds on achievable keyrates. A perspective that may be helpful is to view~\eqref{eq:roughPA} as saying that we should be choosing $\lkey$ as close to that bound as possible, but no higher, since in that case we would be unable to guarantee security via the Leftover Hashing Lemma.}
\begin{align}
\lkey \lesssim \Hmin^\eps(\str{A}_g|\str{X}_g\str{Y}_g\str{L}\str{E}_g),
\label{eq:roughPA}
\end{align} 
up to some asymptotically negligible corrections, depending on the desired level of security and choice of $\eps$. (A full security proof would need to account for the fact that we condition on the parameter-estimation and error-correction steps accepting, which can affect various entropies, but we put aside this issue for this sketch --- roughly speaking, the effect of conditioning on such events is approximately ``cancelled out'' by the $\pr{\mathrm{accept}}$ prefactor in~\eqref{eq:secrecy}.)

Our proof sketch will hence be based on bounding the RHS of~\eqref{eq:roughPA}. 
The first step is to apply a rather ``coarse'' chain rule for the smoothed min-entropy (assuming $\len(\str{L})$ is a fixed constant, and focusing on one-way error correction for simplicity; see Sec.~\ref{sec:mainprot} for further discussion),\footnote{In our protocol outline, we allowed for protocols where $\str{L}$ can have length \emph{up to} some maximum value $k$, rather than being \emph{exactly} a fixed length. This might appear to introduce a technical loophole of leaking additional information through the length of $\str{L}$. However, for the purposes of this analysis we can handle it by noticing that the set of bitstrings of length up to $k$ has cardinality $2^{k+1}-1$, and can hence can be embedded (injectively) in the set of bitstrings of length exactly $k+1$. Therefore, for such protocols it suffices to subtract $k+1$ instead of $\len(\str{L})$ in~\eqref{eq:roughchainrule}, at least if we do not worry about directions of communication in the two-way case.}
\begin{align}
\Hmin^\eps(\str{A}_g|\str{X}_g\str{Y}_g\str{L}\str{E}_g) \geq \Hmin^\eps(\str{A}_g|\str{X}_g\str{Y}_g\str{E}_g) - \len(\str{L}),
\label{eq:roughchainrule}
\end{align}
i.e.~giving the system $\str{L}$ to Eve causes the conditional smoothed min-entropy to decrease by no more than $\len(\str{L})$ bits. 
We now note that the IID structure on $\str{A}_g\str{X}_g\str{Y}_g\str{E}_g$ lets us apply the AEP to bound the first term:
\begin{align}
\Hmin^\eps(\str{A}_g|\str{X}_g\str{Y}_g\str{E}_g) \geq m H(\hat{A}|XY\qE) - O(\sqrt{m}).
\label{eq:roughHminAEP}
\end{align}

As for the $\len(\str{L})$ term, 
we should allow $\str{L}$ to be long enough such that in the \emph{honest} protocol, Bob's guess for $\str{A}$ is correct with high probability. This is precisely the point addressed by Fact~\ref{fact:EC} regarding error correction, which tells us that the optimal length of $\str{L}$ required for this is approximately 
\begin{align}
\len(\str{L}) \approx \Hmax^\eps(\str{A}_g|\str{B}_g\str{X}_g\str{Y}_g)_\mathrm{hon},
\end{align} 
again up to some corrections\footnote{We shall implicitly absorb the length of the ``verification hash'' into this expression --- note that the required length for that hash only depends on the desired correctness parameter, and hence is a constant independent of $m$.} depending on the desired success probability and choice of $\eps$. Note that this is computed with respect to the honest behaviour, because we are only looking for an estimate of how long we should allow $\str{L}$ to be for that case (recall that we assume the verification of whether $\tilde{\str{A}} \neq \str{A}$, in case of unintended behaviour, takes place as a second step). 
This means we can exploit the IID structure on $\str{A}_g\str{B}_g\str{X}_g\str{Y}_g$ for the honest states to apply the AEP:
\begin{align}
\Hmax^\eps(\str{A}_g|\str{B}_g\str{X}_g\str{Y}_g)_\mathrm{hon} \leq m H(\hat{A}|\hat{B}XY)_\mathrm{hon} + O(\sqrt{m}),
\label{eq:roughHmaxAEP}
\end{align}
yielding an upper bound on the required $\len(\str{L})$ in terms of the single-round state.

Putting together all the above inequalities, we see that the secrecy condition will be satisfied as long as $\lkey$ is chosen such that
\begin{align}
\lkey \lesssim m (H(\hat{A}|XY\qE) - H(\hat{A}|\hat{B}XY)_\mathrm{hon}) - O(\sqrt{m}).
\label{eq:roughlkeyprelim}
\end{align}
While this informal description is mainly to capture the asymptotic behaviour, the constants in the inequalities and approximations we have used above can all be tracked to yield an explicit expression in terms of $m$, hence yielding finite-size bounds.\footnote{In some sense, there are two sources of finite-size effects in this analysis. The first, though we did not describe it in detail, is the ``statistical'' fact that in order to ensure the abort probability in case~\ref{case:abort} is high, the protocol must have sufficiently many test rounds. The second is the use of the AEP in~\eqref{eq:roughHminAEP} and~\eqref{eq:roughHmaxAEP}, i.e.~the ``entropic'' fact that the smoothed min/max-entropy of $m$ copies is in general not exactly equal to $m$ times the von Neumann entropy of one copy, but rather can differ by up to $O(\sqrt{m})$. In this analysis, these effects are treated separately, but it is possible in principle that some analysis which incorporates both these corrections simultaneously might achieve better finite-size bounds. The entropy accumulation theorem we present later indeed incorporates both effects together, although it is not entirely clear whether it achieves tighter bounds in doing so (its main purpose is instead to account for non-IID effects). 
On a different note, this could arguably be a reason why a finite-size experiment achieving a Bell violation may sometimes still be unable to achieve DIRNG/DIRE/DIQKD --- when certifying a Bell violation, one only needs to account for the ``statistical'' effect, while these DI protocols also need to account for the ``entropic'' effect.}
However, this does not yet describe a concrete choice of $\lkey$ that we can specify in a protocol description, since $H(\hat{A}|XY\qE)$ is a value that depends on the (potentially dishonest) states and measurements. To obtain this, we finally use the fact that we have restricted our attention to case~\ref{case:accept}, i.e.~where the states and measurements have parameters 
$\tr{\Gamma_j(\pvm_{a|x},\pvm_{b|y}) \rho_{\qA \qB}}$
sufficiently ``close to'' the accepted range. That means that we can ensure~\eqref{eq:roughlkeyprelim} holds by setting
\begin{align}
\lkey \approx m \left(
\inf_{
\mathcal{S}_{\mathrm{close}}
} H(\hat{A}|XY\qE) - H(\hat{A}|\hat{B}XY)_\mathrm{hon}
\right) - O(\sqrt{m}),
\label{eq:roughlkey}
\end{align}
where $\mathcal{S}_{\mathrm{close}}$ denotes the set\footnote{Pedantically, this collection might be a proper class rather than a set, if one allows states and measurements on arbitrary Hilbert spaces (though if we restrict to finite-dimensional Hilbert spaces, and take only one Hilbert space for each dimension as a representative, it can be argued that this is indeed a valid set via the axiom schema of replacement). In any case, all relevant quantities in the security proof are well-defined even when $\mathcal{S}_{\mathrm{close}}$ is a proper class: note that rigorously speaking, the optimization should be understood as referring to the infimum of a set of real numbers (i.e.~indeed a valid set, since it is a subset of $\mathbb{R}$) defined by an existence quantifier taken over the class of ``allowed'' states and measurements. 
For brevity, however, we continue referring to $\mathcal{S}_{\mathrm{close}}$ (and similar collections) using set terminology and notation.} of single-round states and measurements captured within case~\ref{case:accept}.

Hence we have finally reduced the expression for an achievable length of secret key to~\eqref{eq:roughlkey}, which basically only involves single-round quantities. Note that the $H(\hat{A}|\hat{B}XY)_\mathrm{hon}$ term is easy to compute, since it is simply based on the honest behaviour. The main challenge is the first term,
\begin{align}
\inf_{\mathcal{S}_{\mathrm{close}}} H(\hat{A}|XY\qE).
\label{eq:optsimple}
\end{align}
While this is simple enough to state as an optimization problem, it is highly nontrivial to solve in the DI setting, because no bound is placed on the system dimensions --- this means there are effectively infinitely many optimization variables. Furthermore, the optimization is nonconvex because we need to optimize over both the states and the measurements. Hence there may be local minima in the optimization, which means that even if we assume some upper bound on the system dimensions, we still cannot assume that numerical algorithms will find the true minimum. 
A substantial portion of this thesis (Chapter~\ref{chap:singlernd}) is devoted to new techniques for solving this optimization. 

We remark that the guessing-probability-based approach of~\cite{PAM+10,NPS14,BSS14} can basically be viewed as bounding this optimization via the simple inequality\footnote{When $\hat{A}$ is binary-valued and uniform, then a tighter bound $H(\hat{A}|XY\qE) \geq 2 (1-\pg(\hat{A}|XY\qE))$ holds, via a relation to trace distance derived as Theorem~14 of~\cite{BH09}.} $H(\hat{A}|XY\qE) \geq \Hmin(\hat{A}|XY\qE) = -\log \pg(\hat{A}|XY\qE)$, applying the operational interpretation of min-entropy. The guessing probability 
has the convenient property that when bounding it, Eve's side-information can be taken to be classical without loss of generality, since she needs to measure her state to produce her guess. This allows one to bound it arbitrarily tightly using a hierarchy of semidefinite programming (SDP) relaxations developed in the context of quantum nonlocality, known as the NPA hierarchy~\cite{NPA08} (we use this SDP hierarchy as well in Sec.~\ref{sec:NPO}). Conveniently, this SDP formulation easily accounts for arbitrary Bell inequalities, or combinations of several Bell parameters~\cite{NPS14,BSS14}. However, the resulting bound is not very tight because it computes the min-entropy rather than the von Neumann entropy, and the approaches we present in Chapter~\ref{chap:singlernd} are intended to try to avoid this problem.

\subsection{Asymptotic behaviour}
\label{sec:sketchasympt}

Our above sketch can easily be used to find an expression for the asymptotic keyrate. First note that in the limit of large $n$, it is possible to choose the fraction of test rounds to be arbitrarily small\footnote{The intuition is that (for the IID case at least) the accuracy of parameter estimation when $n$ is large mainly depends on the absolute number of test rounds, rather than their proportion out of the total number of rounds. 
Hence for a desired accuracy level, we can fix the number of test rounds as a constant, in which case the fraction of test rounds approaches zero as $n$ increases. See~\cite{DF19} or Sec.~\ref{sec:finitekeylength} for further discussion.}, so we have $m\approx n$. 
Furthermore, by carefully choosing the tolerances accepted in the protocol as well as the $\dsou$ value (as a function of $n$; see the discussion below Theorem~\ref{th:collective} for explicit expressions, or above~\eqref{eq:asympt} for the EAT version), it is possible to have $\mathcal{S}_{\mathrm{close}}$ converge (informally) to $\mathcal{S}_{\mathrm{exact}}$, where $\mathcal{S}_{\mathrm{exact}}$ denotes the set of single-round states and measurements producing \emph{exactly} the same parameter values $\vec{\constr}$ as the honest devices. (Note that $\mathcal{S}_{\mathrm{exact}}$ is defined by the constraint of producing the same values of $\vec{\constr}$ as the honest devices, not that the states and measurements are exactly the honest ones; the latter would be too trivial.)
In that case, dividing both sides of~\eqref{eq:roughlkey} by $n$ and taking the limit $n\to\infty$ yields
\begin{align}
\operatorname{rate}_\infty = \inf_{\mathcal{S}_{\mathrm{exact}}} H(\hat{A}|XY\qE) - H(\hat{A}|\hat{B}XY)_\mathrm{hon} .
\label{eq:devwin}
\end{align}
This result is sometimes known as the Devetak-Winter bound~\cite{DW05}, though we remark that a different proof technique was used in that work, and it was studying a somewhat simpler context where the state $\rho_{\hat{A}\hat{B}\qE}$ is known \textit{a priori}.\footnote{More specifically:~\cite{DW05} proves that if we are given IID copies of a \emph{known} cqq state $\rho_{\hat{A}\qB\qE}$, then there exists a protocol (using one-way communication) such that we can distill a secret key at asymptotic rate $I(\hat{A}:\qB) - I(\hat{A}:\qE)$, which can be rewritten into the equivalent formula $H(\hat{A}|\qE)-H(\hat{A}|\qB)$, closely resembling the formula we obtained above. However, since their result is based on considering a fixed $\rho_{\hat{A}\qB\qE}$, in their formula it is not necessary to take the infimum in the first term or specify that the second term is computed ``honestly''.} 

Strictly speaking, there is a technicality we glossed over when introducing the asymptotic keyrate, in that the definition~\eqref{eq:defkeyrates} does not explicitly require that the protocol is secure in any sense (the definition is mathematically valid without this requirement, but it would then seem to be a rather useless concept). 
To get a more ``useful'' statement, we require some structure along the following lines: the protocol description involves two parameters $n$ and $\eps$, such that the protocol uses $n$ rounds and is both\footnote{The choice here to use the same $\eps$ in both conditions is slightly arbitrary; e.g.~one could instead require $\ecom$-completeness and $\esound$-soundness with $\ecom+\esound=\eps$ (though the class of protocols fulfilling this version is potentially harder to analyze). Note that it is not enough to just require $\eps$-soundness alone, because the completeness condition is required to rule out the following absurd protocol: Alice and Bob simply use the public classical channel to agree that with probability $\eps$ they will set their ``key'' to be the all-$0$ string of length $n$, and otherwise (with probability $1-\eps$) they abort. This is trivially $\eps$-sound since it has $\pr{\mathrm{accept}}=\eps$, and the key length is always $\lkey=n$ (in fact any arbitrary $\lkey$ is possible with this construction). However, it is clearly not a reasonable protocol.} $\eps$-complete and $\eps$-sound, producing a key of length $\lkey(n,\eps)$ (here we shall explicitly denote the dependence of $\lkey$ on these parameters). With this, we can define the finite-size (and ``finite-security'') keyrate and asymptotic keyrate more precisely as 
\begin{align}
\operatorname{rate}^\eps_n \defvar \frac{\lkey(n,\eps)}{n},  \qquad 
\operatorname{rate}^0_\infty \defvar \lim_{\eps\to 0} \lim_{n\to\infty} \operatorname{rate}^\eps_n.
\end{align}

The $\eps\to0$ limit in the asymptotic keyrate definition reflects the qualitative idea that we want the protocol to achieve ``arbitrarily good security'' given enough rounds.\footnote{It is important to take the limits in this order: if the $\eps\to0$ limit is taken first, typically this would only give a trivial value of zero. This is because for QKD/DIQKD scenarios (and many other cryptography or coding tasks), it is impossible to achieve an arbitrarily $\eps$-secure key of nonzero length using only a fixed number of rounds, and hence $\lim_{\eps\to 0} \operatorname{rate}_{n,\eps} = 0$ holds for each $n$, after which the $n\to\infty$ limit is trivial.}
However, for most security proofs following the approaches we outline in this section, the first limit $\lim_{n\to\infty} \operatorname{rate}_{n,\eps}$ is already a constant independent of $\eps$ (as long as $\eps$ is not $0$ or $1$), as we shall show in more detail later (e.g.~the derivation of~\eqref{eq:asympt} in Chapter~\ref{chap:finite}). Hence for brevity, when discussing the asymptotic keyrates in this thesis, we often omit the $\eps\to0$ limit, as in the earlier definition~\eqref{eq:defkeyrates}.

\subsection{Intuition for keyrate improvements}
\label{sec:sketchimpr}

Before proceeding on, we remark that based on the expression~\eqref{eq:devwin}, we can try to gain some understanding of why noisy preprocessing and random key measurements can improve the keyrate. For noisy preprocessing, observe that adding noise to Alice's outcome would increase both $H(\hat{A}|XY\qE)$ and $H(\hat{A}|\hat{B}XY)_\mathrm{hon}$. However, it turns out that in some situations, the first term increases more than the second term, hence improving the overall keyrate.

As for random key measurements (i.e.~choosing the input $X$ according to some non-deterministic distribution), it is helpful to first expand the expression for $H(\hat{A}|XY\qE)$ slightly. 
Let us focus only on the generation rounds, following the arguments in the previous section. 
Also, we shall assume that protocol is such that the sifting step consists of erasing the outputs of rounds in which $X\neq Y$ (most sifting procedures can be phrased this way with an appropriate input labelling). 
In that case, observe that $H(\hat{A}|\qE; X=x,Y=y) = 0$ whenever $x\neq y$. Furthermore, when $x=y$, we can write $H(\hat{A}|\qE; X=x,Y=y) = H(\hat{A}|\qE; X=x)$, i.e.~the conditioning on $y$ is unnecessary, due to the no-signalling property.\footnote{Slight care is needed here, since it might appear that the sifting procedure (which involves $Y$) could affect the claim of no-signalling. This is addressed by the following argument: let $\rho'_{xy}$ denote the state that would be produced on registers $\hat{A} \qE$ if we perform the measurements for inputs $xy$ followed by noisy preprocessing immediately, \emph{without} the sifting step. We can validly invoke the no-signalling property on this state to claim that $H(\hat{A}|\qE)_{\rho'_{xy}}$ does not depend on $y$. Now observe that for the specific case where $x=y$, nothing is performed during the sifting step in the protocol, and hence the state in the protocol is the same as $\rho'_{xy}$. Therefore the entropy $H(\hat{A}|\qE)$ of that state does not depend on $y$ either.}
For brevity, let $\hat{A}_x$ denote a register that stores Alice's output (after sifting and noisy preprocessing) if she uses input $x$. 
Then we can write
\begin{align}
H(\hat{A}|XY\qE) &= \sum_{xy} \pr{X=x,Y=y} H(\hat{A}|\qE; X=x,Y=y) \nonumber\\
&= \sum_{x} \pr{X=Y=x} H(\hat{A}_x|{\qE}) 
\nonumber\\
&= \pr{X=Y} \sum_{x} \keyw_x H(\hat{A}_x|{\qE}) ,
\label{eq:randomkeyEve}
\end{align}
introducing the coefficients $\keyw_x = \pr{X=x|X=Y}$, which satisfy $\sum_x \keyw_x = 1$.
(If the protocol includes the test rounds in the privacy amplification step, a similar analysis can be performed, 
though the coefficients $\keyw_x$ would have a different form.) 
The prefactor $\pr{X=Y}$ is basically what is often referred to as the \term{sifting factor}.

We can analyze Bob's entropy $H(\hat{A}|\hat{B}XY)_\mathrm{hon}$ similarly as well. For this informal sketch, we restrict our attention to the simple example of depolarizing noise\footnote{Other noise models, most notably limited detection efficiency, may not benefit as much from the random-key-measurements technique; see~\cite{SGP+21}.},
in a scenario where the noiseless distribution is given by the state and measurements in~\eqref{eq:noiselessCHSH}, except that Bob's measurements $y=0,1,2,3$ are instead $Z,X,(Z+X)/\sqrt{2},(Z-X)/\sqrt{2}$ respectively.
For that model, $H(\hat{A}_x|\hat{B}_x)_\mathrm{hon}$ is independent of $x$ and hence we just have $H(\hat{A}|\hat{B}XY)_\mathrm{hon} = \pr{X=Y} H(\hat{A}_0|\hat{B}_0)_\mathrm{hon}$.
In that case, we see that the asymptotic keyrate can be written as
\begin{align}
\operatorname{rate}_\infty = \pr{X=Y} \left(\inf_{\mathcal{S}_{\mathrm{exact}}} \sum_{x} \keyw_x H(\hat{A}_x|{\qE}) - H(\hat{A}_0|\hat{B}_0)_\mathrm{hon} \right).
\label{eq:randomkeyrate}
\end{align}

To show the advantage offered by this approach, we compare~\eqref{eq:randomkeyrate} to a protocol which always uses the input $x=0$ for the generation rounds, where the keyrate would be
\begin{align}
\inf_{\mathcal{S}_{\mathrm{exact}}} H(\hat{A}_0|{\qE}) - H(\hat{A}_0|\hat{B}_0)_\mathrm{hon}.
\label{eq:A0keyrate}
\end{align}
The key observation is that for~\eqref{eq:A0keyrate}, Eve can optimize her attack to gain maximal information for the $x=0$ measurement specifically; in contrast, when random key measurements are used, Eq.~\eqref{eq:randomkeyrate} tells us that Eve is forced to make a ``tradeoff'' across the information she can gain for each of the inputs $x$ (recall that she has to decide on what state to use before learning which input $x$ would be chosen in that round). Of course, there may be DIQKD scenarios where Eve can gain just as much information for all the inputs $x$ as compared to the optimal attack on $x=0$ alone\footnote{This is indeed the case for e.g.~the six-state protocol in standard QKD, so the keyrate for that protocol is not improved by this technique.}, but it turns out that for some simple cases (most notably protocols based on the CHSH value), it is indeed the case that not all $H(\hat{A}_x|{\qE})$ can simultaneously be equal to the worst-case value of $H(\hat{A}_0|{\qE})$. This hence allows for potential improvement of the keyrate. 
\begin{remark}
We have glossed over the fact that the $\pr{X=Y}$ prefactor in~\eqref{eq:randomkeyrate} may reduce the improvement here. However, notice that in the simple noise models where~\eqref{eq:randomkeyrate} holds, this merely rescales the keyrate by a nonzero factor. Our above argument indicates that it is possible to simultaneously have 
\begin{align}
\inf_{\mathcal{S}_{\mathrm{exact}}} \sum_{x} \keyw_x H(\hat{A}_x|{\qE}) - H(\hat{A}_0|\hat{B}_0)_\mathrm{hon} > 0,
\quad
\inf_{\mathcal{S}_{\mathrm{exact}}} H(\hat{A}_0|{\qE}) - H(\hat{A}_0|\hat{B}_0)_\mathrm{hon} < 0, 
\end{align}
in which case the former is still higher than the latter even after rescaling by $\pr{X=Y}$.
\end{remark}

\subsection{Beyond collective attacks}
\label{sec:sketchbeyond}

Moving on to coherent attacks, there are several difficulties to address. First, as mentioned previously, one cannot speak of a single distribution $\pr{ab|xy}$ that holds for all rounds. Because of this, while we can still compute various frequencies using the input-output statistics in the protocol, it is harder to say what these frequencies tell us about the devices themselves. Furthermore, there is also no well-defined notion of a single-round state $\rho_{\qA \qB \qE}$ that we can compute the entropies of. Another issue is that we can no longer assume that the test rounds are decoupled from the generation rounds, so any public communication of the test-round outputs may in principle leak information about the generation rounds --- fortunately, this is not too hard to address (though it can have nontrivial effects on the finite-size keyrates), as we shall shortly explain.

Until fairly recently, there were no proof approaches that yielded the same asymptotic keyrates for coherent attacks as compared to collective attacks. The key theoretical result that allowed this to be achieved was the entropy accumulation theorem, proven in~\cite{DFR20} (and subsequently improved in~\cite{DF19}). The EAT also has the powerful property of simultaneously accounting for the parameter estimation issue.

While there are various technical conditions to address when applying the EAT, we can informally summarize how to apply it in DIQKD.
First, instead of the case analysis considered at the start of Sec.~\ref{sec:sketchcoll}, we split into cases where the parameter-estimation step accepts with probability less than some small value $\eEA$, or greater than $\eEA$.\footnote{A very recent work~\cite{arx_Dup21} indicates that this case analysis can be avoided, by proving a Leftover Hashing Lemma for R\'{e}nyi entropies rather than smoothed min-entropy, but we will not use this in this work.} The first case trivially satisfies the secrecy definition for the same reasons as before. As for the second case, we apply the EAT as follows. Suppose we can find some value $h\in\mathbb{R}$ with the following property: for every round in the protocol, if the devices produce a distribution $\pr{ab|xy}$ in that round that falls within the accepted tolerances in the parameter-estimation step, then $H(\hat{A}|XY\qE) \geq h$. (In contrast to the collective-attacks analysis, this is a more ``abstract'' quantity for now, which is not meant to correspond exactly to the actual devices.)
Informally, the EAT then states that the state conditioned on accepting in the parameter-estimation step satisfies a bound much like the AEP:\footnote{Here, the bound involves all the rounds rather than just the generation rounds, because the test and generation rounds cannot be assumed to be independent. This is also why it depends on $n$ rather than $m$.}
\begin{align}
\Hmin^\eps(\str{A}|\str{X}\str{Y}\str{L} \allE) \geq n h - O(\sqrt{n}),
\label{eq:roughHminEAT}
\end{align}
now denoting Eve's side-information as $\allE$ since it may not have a tensor-product structure. Like the AEP, the EAT gives an explicit expression for the constant in the $O(\sqrt{n})$ term --- it depends on 
the accept probability 
(which one should expect), as well as some other more technical properties of the lower bound on $H(\hat{A}|XY\qE)$ (i.e.~strictly speaking, we do not only need the single value $h$ when applying the EAT, but also somewhat more complicated properties). 

Notice that this statement has already accounted for the statistical estimation issues, because it connects the value $h$ (which was computed based on the abstract distributions for individual rounds) to the parameter-estimation step (which uses the observed input-output statistics in the actual protocol). Also, since the bound is similar to the AEP, we see (skipping some details) that by using this bound~\eqref{eq:roughHminEAT} in place of the AEP bound~\eqref{eq:roughHminAEP}, we can obtain essentially the same asymptotic rate~\eqref{eq:devwin} in the end --- the analysis of the error-correction terms is mostly unchanged since it is based on the honest behaviour. 
Regarding the aforementioned issue about the test rounds, this can be addressed by noting that publicly communicating the test-round outputs cannot decrease the conditional min-entropy by more than the number of communicated bits, due to the same simple chain-rule bound as in~\eqref{eq:roughchainrule}. If the fraction of test rounds is small, then the number of bits required for this is fairly small as well, so it does not affect the keyrates in the asymptotic limit. (Unfortunately, it can have some impact on the finite-size keyrates, the extent of which we will discuss in Sec.~\ref{sec:coll}.)

We can also consider the proof techniques of~\cite{JMS20,arx_Vid17,arx_JK21} for the parallel-input scenario.
To give a highly simplified overview, the approach used in these works is to define a single-round 3-player nonlocal game where Alice and Bob try to win a 2-player nonlocal game and Eve tries to guess their outputs. By applying parallel-repetition theorems or threshold theorems or direct-product theorems, one obtains a bound on the probability of Eve guessing the outputs (conditioned on Alice and Bob winning ``enough'' rounds) when $n$ copies of this game are played in parallel, which can be converted into a bound on the min-entropy via the operational interpretation~\eqref{eq:Hminop}. With this in mind, we can also understand to some extent why this approach yields a looser bound than the collective-attacks bound~\eqref{eq:roughHminAEP} (from the AEP) --- firstly, the single-round analysis is based on guessing probability, which cannot be converted exactly to the von Neumann entropy; secondly, parallel-repetition theorems usually do not certify exactly the same bound (in terms of the constants involved) for $n$ parallel rounds as compared to $n$ IID rounds. 

\section{DIRNG and DIRE}
\label{sec:sketchDIRE}

The above security proofs for DIQKD can more or less be generalized to DIRNG/DIRE. 
While we will not go into detail here, we can see the broad picture --- 
following an analogous argument, one finds that the length of secret key that can be produced is approximately $\Hmin^\eps(\str{A}\str{B}|\str{X}\str{Y} \allE)$, since both parties' outputs can be used and there is no need to implement error correction. (There have been recent proposals~\cite{arx_TSB+20,arx_BRC21} to improve over this result by using a ``seed recovery'' step that we describe in Sec.~\ref{sec:preshared}, but we defer further discussion until that point.) Furthermore, by applying the AEP or the EAT~\cite{ARV19} (depending on the allowed class of attacks), we see that the single-round quantity that we need to analyze in this case is basically $H(\hat{A}\hat{B}|XY\qE)$ instead, i.e.~the entropy of both parties' outputs. Many of the techniques we present in this work can also be applied to this entropy with minimal modifications, so they can serve to prove achievable keyrates for DIRNG/DIRE as well.

From this, we can also see why DIQKD may be more difficult to realize than DIRNG/DIRE. Namely, the asymptotic keyrate expression~\eqref{eq:devwin} for DIQKD has the error-correction term $H(\hat{A}|\hat{B}XY)_\mathrm{hon}$, which is not present in DIRNG/DIRE (as it is only based on a single term, $H(\hat{A}\hat{B}|XY\qE)$). Hence a given set of experimental devices may be able to achieve positive asymptotic keyrates for DIRNG/DIRE, but not for DIQKD. There is also some further advantage for DIRNG/DIRE since the asymptotic keyrate involves the entropy of both parties' outputs, not just Alice's. (This does not consider the possibility of DIQKD advantage-distillation protocols as in Chapter~\ref{chap:AD}, where the asymptotic keyrate is not given by~\eqref{eq:devwin}. However, the keyrates for such protocols are much less well understood apart from some somewhat loose upper bounds, so it is difficult to draw a comparison there.)

\chapter{Entropy bounds}
\label{chap:singlernd}

\newcommand{\pinch}{\bar{\calZ}}
\newcommand{\xo}{0}
\newcommand{\xg}{\bar{x}}
\newcommand{\nat}{\mathrm{nat}}
\newcommand{\lambnew}{\tilde{\lambda}}
\newcommand{\csch}{\operatorname{csch}}

In this chapter, we will develop methods to solve the optimization~\eqref{eq:optsimple} sketched in the previous chapter. We begin by precisely defining the problem we aim to address. After doing so, we describe in Sec.~\ref{sec:lagrange} how to transform it to a dual problem with useful structure, and in Sec.~\ref{sec:entdual} we transform the objective function into a more convenient form. With these general techniques in mind, in Sec.~\ref{sec:qubit} we present an algorithm for computing arbitrarily tight bounds on~\eqref{eq:optsimple} in 2-input 2-output scenarios, making use of a reduction to qubit systems observed in e.g.~\cite{PAB+09}. Finally, in Sec.~\ref{sec:NPO} we describe some techniques that can be applied for more general nonlocality scenarios. The results in this chapter are mostly based on~\cite{arx_TSB+20,TSG+21}, and we closely follow the phrasing and presentation in those works. We also briefly describe our results from~\cite{HST+20,SBV+21,SGP+21}, as well as some independent works~\cite{WAP21,BFF21}.

We highlight that although many of the approaches presented in this chapter are numerical, we are not merely heuristically solving the optimization~\eqref{eq:optsimple} (this runs the risk of getting trapped in local minima). Rather, these approaches we present are designed to yield ``secure'' lower bounds, in the sense that the values returned are guaranteed to be lower bounds on the optimization~\eqref{eq:optsimple}, even if they might not be tight in some cases.

\begin{remark}
After preparation of this thesis, several subsequent works~\cite{arx_BFF21,arx_MPW21} developed new numerical methods to address this optimization. As compared to the approaches presented here, these new approaches are more computationally efficient and yield tighter bounds. We direct the interested reader to those works for more information.
\end{remark}

To allow noisy preprocessing to be described in a relatively simple way, we assume for the purposes of Sec.~\ref{sec:lagrange}--\ref{sec:qubit} that Alice's key-generating measurements only have 2 outcomes, in which case noisy preprocessing simply consists of flipping her output with some probability $\p$ (noisy preprocessing on larger numbers of outcomes become more complicated, as it would correspond to more general stochastic maps rather than a simple probabilistic bitflip\footnote{Even without going to the extent of allowing arbitrary stochastic maps, a slightly more general possibility would be to consider different flipping probabilities for each input $x$, and our analysis in Sec.~\ref{sec:lagrange}--\ref{sec:qubit} can easily be modified to handle this option.
However, for simplicity of notation we will not explicitly denote this possibility.}).
Instead of solving the optimization over the set $\mathcal{S}_{\mathrm{close}}$ (which was rather informally defined thus far), we focus on a slightly simpler and more precisely stated problem: given some choice of noisy-preprocessing value $\p$
and any values $\nu_j \in \mathbb{R}$, compute or lower-bound the function
\begin{align}
\breve{\lin}_\p(\vec{\constr}) \defvar \,
\begin{gathered}
\inf_{\rho_{ABE},\pvm_{a|x},\pvm_{b|y}} \sum_{x} \keyw_x H(
\hat{A}_x
|{E}) \\ 
\suchthat\; \tr{\vec{\Gamma}(\pvm_{a|x},\pvm_{b|y}) \rho_{AB}} = \vec{\constr} \;,
\end{gathered} 
\label{eq:mainopt}
\end{align}
where the various terms should be understood as follows:
\begin{itemize}
\item $\rho_{ABE}$ is a tripartite state of arbitrary dimension.
\item $\pvm_{a|x}$ is a projector\footnote{We can assume all the measurements are projective, by considering a suitably chosen joint Stinespring dilation of all the measurements --- see Appendix~\ref{app:proj} for detailed discussion.} corresponding to Alice's outcome $a$ for input $x$, and $\pvm_{b|y}$ is a projector corresponding to Bob's outcome $b$ for input $y$. For the purposes of this optimization, this means they are orthogonal projectors satisfying the constraints $\sum_a \pvm_{a|x} = \id_{A}$ for every $x$ and $\sum_b \pvm_{b|y} = \id_{B}$ for every $y$. 
\item $\keyw_x$ are non-negative values corresponding to the conditional probabilities described in Sec.~\ref{sec:sketchimpr}. For the purposes of this optimization, they can simply be treated as some arbitrary constants whose values are fixed by the protocol.
\item $\hat{A}_x$ is a register storing the value of Alice's output for measurement input $x$, \emph{after} applying noisy preprocessing.
\item $\vec{\Gamma}(\pvm_{a|x},\pvm_{b|y})$ denotes a tuple of operators of the form $\Gamma_j(\pvm_{a|x},\pvm_{b|y}) = \sum_{abxy} c^{(j)}_{abxy} \pvm_{a|x}\otimes\pvm_{b|y}$ for some constants $c^{(j)}_{abxy} \in \mathbb{R}$ (recall~\eqref{eq:bellop}).
The notation $\tr{\vec{\Gamma}(\pvm_{a|x},\pvm_{b|y}) \rho_{AB}}$ is shorthand for the tuple $\left( \tr{{\Gamma}_1(\pvm_{a|x},\pvm_{b|y}) \rho_{AB}}, \tr{{\Gamma}_2(\pvm_{a|x},\pvm_{b|y}) \rho_{AB}}, \dots \right)$.
\end{itemize}

The optimization~\eqref{eq:optsimple} informally sketched in the previous chapter can be viewed as minimizing $\breve{\lin}_\p(\vec{\constr})$ over all values of $\vec{\constr}$ ``close to'' the accepted set\footnote{An alternate perspective is that the equality constraints in~\eqref{eq:mainopt} can be straightforwardly modified into inequality constraints describing the optimization over $\mathcal{S}_{\mathrm{close}}$. However, it turns out that some of our security arguments, such as those based on the EAT, are in fact easier to describe in terms of the optimization in the form~\eqref{eq:mainopt}.}, so this is essentially sufficient --- we will explain in the next chapter how to convert a solution to the precisely stated optimization~\eqref{eq:mainopt} into a full security proof. (Also, strictly speaking the objective functions of~\eqref{eq:mainopt} and~\eqref{eq:optsimple} differ by a factor of $\pr{X=Y}$; see
\eqref{eq:randomkeyEve}. However, this is simply a constant for the purposes of the optimization, hence we ignore it here.)

We highlight that in the optimization~\eqref{eq:mainopt}, noisy preprocessing is applied only to the terms in the objective function, not in the constraint terms (recall that the operators $\Gamma_j(\pvm_{a|x},\pvm_{b|y})$ were defined without noisy preprocessing). This corresponds to the fact that in the protocol sketch we described, noisy preprocessing is only applied to the generation rounds, not the test rounds. In principle, one could design a protocol where noisy preprocessing is applied in the test rounds, but this appears to complicate the analysis and the potential keyrate improvements are at most on the order of the test-round fraction (this can be seen from our later discussion in Sec.~\ref{sec:coll} regarding the test rounds).

A few observations can be made about the optimization~\eqref{eq:mainopt}. Firstly, we can simply restrict the optimization to pure $\rho_{ABE}$, since if $\rho_{ABE}$ is mixed then we can just give Eve the purifying system. 
Also, if we allow arbitrary values of $\vec{\constr}$ in the expression, the optimization may turn out to be infeasible. Following standard conventions in optimization theory, we take $\breve{\lin}_\p(\vec{\constr}) = +\infty$ in that case, i.e.~the infimum of an empty set is $+\infty$.
Finally, with this convention, note that $\breve{\lin}_\p$ is a convex function (which the~~$\breve{}$~~in the notation is intended to indicate), despite the fact that it is defined via a nonconvex optimization. This follows from the fact Eve can perform ``classical mixtures'' of strategies. More precisely: take any $\eps>0$, any $w\in[0,1]$, and two values $\vec{\constr}_0, \vec{\constr}_1$ such that the optimization~\eqref{eq:mainopt} is feasible. By definition of $\breve{\lin}_\p$, there exists some strategy achieving the value $\vec{\constr}_0$ on the constraints with conditional entropy $\sum_{x} \keyw_x H(\hat{A}_x|{E}) = t_0$ for some value $t_0 \leq \breve{\lin}_\p(\vec{\constr}_0) + \epsilon$. Doing the same for $\vec{\constr}_1$, we now note that if Eve has a classical register $E'$ in the state $w\pure{0}+(1-w)\pure{1}$ and uses it to determine which of those strategies to implement (including providing copies of $E'$ to the users' devices so they can implement the corresponding measurements), this yields a strategy that achieves value $w\vec{\constr}_0 + (1-w)\vec{\constr}_1$ on the constraints, and entropy $\sum_{x} \keyw_x H(\hat{A}_x|{E}E') = w t_0+(1-w)t_1$. By definition of $\breve{\lin}_\p$, this implies 
\begin{align}
\breve{\lin}_\p(w\vec{\constr}_0 + (1-w)\vec{\constr}_1) \leq w t_0+(1-w)t_1 \leq w \breve{\lin}_\p(\vec{\constr}_0) +(1-w)\breve{\lin}_\p(\vec{\constr}_1) + \epsilon,
\end{align}
proving that $\breve{\lin}_\p$ is convex since $\epsilon>0$ was arbitrary.

\begin{remark}
The above convexity argument relies on the fact that we do not bound the system dimensions, because implementing this strategy required the classical ancilla $E'$ to ``keep track of'' the strategy used. If we forcibly impose a bound on the dimensions, this can have the unpleasant effect that $\breve{\lin}_\p$ is no longer a convex function (in fact, even the set of values of $\vec{\constr}$ such that $\breve{\lin}_\p(\vec{\constr})$ is feasible may not be a convex set~\cite{DW15}).
\end{remark}

The fact that $\breve{\lin}_\p$ is convex suggests that there may be some potential in considering techniques from convex optimization, despite the fact that the optimization problem is not itself convex. In particular, an important tool is the \term{Lagrange dual}, which we shall now describe.

\section{Lagrange dual}
\label{sec:lagrange}

\newcommand{\dom}{\operatorname{dom}}
\newcommand{\relint}{\operatorname{relint}}
\newcommand{\lagdual}{s}
\newcommand{\onecon}[1]{#1'}

In general, any constrained optimization is lower bounded by its Lagrange dual, a property known as \term{weak duality}. For the optimization~\eqref{eq:mainopt}, the Lagrange dual takes the form
\begin{align}
\sup_{\vec{\lambda}}
\inf_{\rho_{ABE},\pvm_{a|x},\pvm_{b|y}} \left( \sum_{x} \keyw_x H(\hat{A}_x|{E})\right) - \vec{\lambda} \cdot \left(\tr{\vec{\Gamma}(\pvm_{a|x},\pvm_{b|y}) \rho_{AB}} - \vec{\constr} \right) ,
\label{eq:dual}
\end{align}
where the supremum takes place over all $\vec{\lambda} \in \mathbb{R}^N$ (letting $N$ be the number of constraints).
If this bound is tight (i.e.~equal to 
$\breve{\lin}_\p(\vec{\constr})$), 
this is referred to as \term{strong duality}. It is easily shown that the Lagrange dual is always a convex function of the constraint values $\vec{\constr}$. It turns out that because $\breve{\lin}_\p$ is convex as well, we can show that strong duality holds for almost all values of $\vec{\constr}$, in a sense we shall now make precise.

Since we are applying techniques from convex optimization,
we will follow conventions in that field and take the codomain of $\breve{\lin}_\p$ to be $[0,+\infty]$ (negative values can be ignored because the objective function is non-negative). 
It is also conventional to define the ``domain'' of $\breve{\lin}_\p$ to be the set of values $\vec{\constr}$ such that $\breve{\lin}_\p(\vec{\constr}) < +\infty$. We shall denote this as $\dom(\breve{\lin}_\p)$, and in our case this is exactly equal to the set of $\vec{\constr}$ that are achievable by the states and measurements we consider 
(this follows from the fact that our objective function is always finite).
Finally, for any subset $\mathcal{C}$ of an Euclidean space, let $\relint(\mathcal{C})$ denote its relative interior (basically, its interior relative to its affine hull; see e.g.~\cite{BV04v8} for a full definition).
We shall now show that strong duality holds for all $\vec{\constr}\in \relint(\dom(\breve{\lin}_\p))$.

\begin{lemma}\label{lem:duality}
For all $\vec{\constr}\in \relint(\dom(\breve{\lin}_\p))$, the expression~\eqref{eq:dual} is equal to $\breve{\lin}_\p(\vec{\constr})$.
\end{lemma}
\begin{proof}
Abstractly, $\breve{\lin}_\p(\vec{\constr})$ is a constrained optimization of the form
\begin{align}
\breve{\lin}_\p(\vec{\constr}) = \,
\begin{gathered}
\inf_{x\in D} f(x) \\
\suchthat\; g_j(x) = \constr_j \quad \forall j \;,
\end{gathered} 
\end{align}
where $D$ is some optimization domain (which could in principle be a class rather than a set, though we will continue using set notation for it), 
and $f$ and $g_j$ are functions (or class functions) $D\to\mathbb{R}$. Denoting the corresponding Lagrange dual~\eqref{eq:dual} as $\breve{\lagdual}_\p$, it takes the form
\begin{align}
\breve{\lagdual}_\p(\vec{\constr}) = \sup_{\vec{\lambda}} \inf_{x\in D} \left[f(x) - \vec{\lambda}\cdot\vec{g}(x) + \vec{\lambda}\cdot\vec{\constr}\right].
\end{align}

Take any $\vec{\constr} \in \relint(\dom(\breve{\lin}_\p))$. Our desired claim would follow by showing that $\breve{\lagdual}_\p(\vec{\constr}) \geq \breve{\lin}_\p(\vec{\constr})$ (the opposite inequality always holds since the Lagrange dual is always a lower bound on the optimization).
To prove this, we use the fact that since $\breve{\lin}_\p$ is convex and $\vec{\constr} \in \relint(\dom(\breve{\lin}_\p))$, 
there exists a \term{subgradient} of $\breve{\lin}_\p$ at $\vec{\constr}$ (see e.g.~\cite{SB14} Prop.~2.2.2). Explicitly, this means there exists $\vec{\lambda}^{\star}$ such that 
\begin{align}
\breve{\lin}_\p(\vec{\constr}\,') \geq 
\vec{\lambda}^{\star}\cdot\left(\vec{\constr}\,' - \vec{\constr}\right) 
+ \breve{\lin}_\p(\vec{\constr}) \qquad\forall\, \vec{\constr}\,' 
. 
\label{eq:subgrad}
\end{align}
(Basically, the RHS is an affine lower bound on $\breve{\lin}_\p$ that also touches the graph of $\breve{\lin}_\p$ at $\vec{\constr}$.)
Since the Lagrange dual $\breve{\lagdual}_\p$ is a supremum over all possible $\vec{\lambda}$, it is lower-bounded by setting $\vec{\lambda}=\vec{\lambda}^{\star}$, which yields
\begin{align}
\breve{\lagdual}_\p(\vec{\constr}) 
&\geq \inf_{x\in D} \left[f(x) - \vec{\lambda}^{\star}\cdot\vec{g}(x) + \vec{\lambda}^{\star}\cdot\vec{\constr}\right] \label{eq:infx}\\
&\geq \inf_{\vec{\constr}\,'
} \left[\breve{\lin}_\p(\vec{\constr}\,') - \vec{\lambda}^{\star}\cdot\vec{\constr}\,' + \vec{\lambda}^{\star}\cdot\vec{\constr}\right] \label{eq:infgamma} \\
&\geq \breve{\lin}_\p(\vec{\constr}) ,
\end{align}
where the second line holds\footnote{A more detailed argument shows that in fact we have equality in this line, i.e.~$\eqref{eq:infx} = \eqref{eq:infgamma}$.} because for each $x\in D$ in the optimization~\eqref{eq:infx}, we can set $\vec{\constr}\,'=\vec{g}(x)$ 
to obtain a $\vec{\constr}\,'$ such that $f(x) - \vec{\lambda}^{\star}\cdot\vec{g}(x) \geq \breve{\lin}_\p(\vec{\constr}\,') - \vec{\lambda}^{\star}\cdot\vec{\constr}\,'$, and the third line follows from~\eqref{eq:subgrad}. 
This is the desired result.
\end{proof}

While we have only shown that strong duality holds for $\vec{\constr} \in \relint(\dom(\breve{\lin}_\p))$, this is sufficient to cover practical implementations, since any $\vec{\constr} \in \dom(\breve{\lin}_\p)$ that is not in the relative interior would instead be a relative boundary point of the set of distributions achievable by the states and measurements we consider. This is a set of measure zero that is of little relevance after we take into account the accepted tolerances during parameter estimation. (Alternatively, by putting a particular form of dimension restriction on the set of allowed states and measurements, one can ensure strong duality holds for all $\vec{\constr} \in \dom(\breve{\lin}_\p)$, but we defer the details to~\cite{TSG+21}.)

This is important because the Lagrange dual has some convenient structure, and the fact that strong duality holds means that there is no loss incurred by studying it in place of the original optimization.
We first remark that it is not strictly necessary to solve the optimization over $\vec{\lambda}$ in the Lagrange dual --- since that optimization is a supremum, any choice of $\vec{\lambda}$ yields a valid lower bound on $\breve{\lin}_\p(\vec{\constr})$. (Of course, finding better values of $\vec{\lambda}$ yields tighter bounds, but the point is that it is still a valid lower bound even if $\vec{\lambda}$ is not optimal.) Following up on this, a useful property is that for each choice of $\vec{\lambda}$, we can interpret the resulting lower bound on $\breve{\lin}_\p(\vec{\constr})$ to be a function of the parameters $\vec{\constr}$, in the sense that we can write
\begin{align}
\breve{\lin}_\p(\vec{\constr}) \geq
\vec{\lambda}\cdot\vec{\constr} + c_{\vec{\lambda}}, 
\label{eq:duallin}
\end{align}
where
\begin{align}
c_{\vec{\lambda}} \defvar \inf_{\rho_{ABE},\pvm_{a|x},\pvm_{b|y}} \left(\sum_{x} \keyw_x H(\hat{A}_x|{E})\right) - \vec{\lambda} \cdot \tr{\vec{\Gamma}(\pvm_{a|x},\pvm_{b|y}) \rho_{AB}}.
\label{eq:intercept}
\end{align}
Importantly, the lower bound~\eqref{eq:duallin} is affine (when viewed as a function of $\vec{\constr}$). This is a useful property for the purposes of security proofs based on the EAT, as we shall later see.

\newcommand{\lbnd}{L}
Another useful property is regarding protocol design, as follows. 
Putting aside technical details, recall that the motivation behind the optimization~\eqref{eq:mainopt} is essentially to compute achievable keyrates for a DIQKD protocol 
that estimates various Bell parameters $\tr{\Gamma_j(\pvm_{a|x},\pvm_{b|y}) \rho_{AB}}$, and accepts if the estimates are ``sufficiently close to'' some values $\constr_j$.
For such a protocol\footnote{We shall take the values 
$\vec{\constr}$ 
to have been fixed by the protocol, so we shall treat them as constants for the rest of this discussion, without explicitly denoting the $\vec{\constr}$-dependence of values such as the optimal choice of $\vec{\lambda}$ (if the supremum is attained) and the resulting lower bound $\lbnd$.}, any choice of $\vec{\lambda}$ in the Lagrange dual~\eqref{eq:dual} yields some lower bound $\lbnd$ on $\breve{\lin}_\p(\vec{\constr})$, with the value
\begin{align}
\lbnd &=\inf_{\rho_{ABE},\pvm_{a|x},\pvm_{b|y}} \left( \sum_{x} \keyw_x H(\hat{A}_x|{E})\right) - \vec{\lambda}^{} \cdot \left(\tr{\vec{\Gamma}(\pvm_{a|x},\pvm_{b|y}) \rho_{AB}} - \vec{\constr} \right),
\label{eq:multbnd}
\end{align}
and thus also some lower bound $\tilde{\lbnd}$ on the asymptotic keyrate via~\eqref{eq:devwin}. It now turns out that if the original protocol involves estimating multiple Bell parameters, we can use this value of $\vec{\lambda}$ to construct a new DIQKD protocol that 
also achieves an asymptotic keyrate of at least $\tilde{\lbnd}$, but only needs to estimate a single Bell parameter $\tr{\onecon{\Gamma}(\pvm_{a|x},\pvm_{b|y}) \rho_{AB}}$ (and accepts if the estimate is ``sufficiently close to'' some value $\onecon{\constr}$), while leaving all other aspects of the protocol unchanged. This new protocol could be easier to implement in practice, since it only needs to estimate one parameter.

Explicitly, this is achieved by taking $\onecon{\Gamma}(\pvm_{a|x},\pvm_{b|y}) \defvar \vec{\lambda}^{} \cdot \vec{\Gamma}(\pvm_{a|x},\pvm_{b|y})$ and $\onecon{\constr} \defvar \vec{\lambda}^{} \cdot \vec{\constr}$. To show that the asymptotic keyrate of this new protocol is at least the same value $\tilde{\lbnd}$, it would suffice to show that the minimum value of the optimization
\begin{align}
\begin{gathered}
\inf_{\rho_{ABE},\pvm_{a|x},\pvm_{b|y}} \sum_{x} \keyw_x H(
\hat{A}_x
|{E}) \\ 
\suchthat\; \tr{\onecon{\Gamma}(\pvm_{a|x},\pvm_{b|y}) \rho_{AB}} = \onecon{\constr} 
\end{gathered} 
\label{eq:oneconopt}
\end{align}
is at least $\lbnd$. (The error-correction term $H(\hat{A}|\hat{B}XY)_\mathrm{hon}$ in the keyrate~\eqref{eq:devwin} can be assumed to remain the same in the new protocol, since it is based on the honest behaviour, which was unchanged.) To do so, we simply note that the optimization~\eqref{eq:oneconopt} is lower-bounded by its Lagrange dual, which takes the form
\begin{align}
&\sup_{\lambda}
\inf_{\rho_{ABE},\pvm_{a|x},\pvm_{b|y}} \left( \sum_{x} \keyw_x H(\hat{A}_x|{E})\right) - \lambda \left(\tr{\onecon{\Gamma}(\pvm_{a|x},\pvm_{b|y}) \rho_{AB}} - \onecon{\constr} \right) \nonumber\\
\geq& \inf_{\rho_{ABE},\pvm_{a|x},\pvm_{b|y}} \left( \sum_{x} \keyw_x H(\hat{A}_x|{E})\right) -  \left(\tr{\onecon{\Gamma}(\pvm_{a|x},\pvm_{b|y}) \rho_{AB}} - \onecon{\constr} \right)
\nonumber\\
=& \lbnd,
\end{align}
which is the desired lower bound.\footnote{To consider the ``converse'' question of whether the new protocol can achieve a higher asymptotic keyrate than the original protocol: this is not possible in principle (at least, for keyrates computed via~\eqref{eq:devwin}),
since the true value of the original optimization~\eqref{eq:mainopt} is always higher than or equal to the new optimization~\eqref{eq:oneconopt} --- observe that every feasible point of the former is a feasible point of the latter.
At most, the optimization~\eqref{eq:oneconopt} might perhaps be easier to analyze or more stable under numerical approaches; however, the main advantage offered by the new protocol is likely to still be just the simplified parameter-estimation step (and possibly associated reductions in finite-size effects).}

\section{Entropic dual}
\label{sec:entdual}

We now aim to simplify the $H(\hat{A}_x|E)$ terms in the optimization. 
This can be done by rewriting them based on the approach in~\cite{WLC18}. Specifically, recall that we have assumed Alice's generation-round outputs are binary, and in the noisy preprocessing step, her output is flipped with probability $\p$. In that case, observe that the state produced on $\hat{A}_x E$ by performing Alice's measurement $x$ on $\rho_{ABE}$ and then applying noisy preprocessing could also be obtained by the following process: append an ancilla $T$ in the state
\begin{align}
\ket{\phi_\p}_T \defvar \sqrt{1-\p} \ket{0}_T + \sqrt{\p} \ket{1}_T,
\end{align}
apply a pinching channel $\mathcal{P}(\sigma_T)\defvar\sum_t \pure{t}_T \sigma_T \pure{t}_T$ to $T$, then perform a measurement on $AT$ described by the projectors
\begin{align}
\tilde{\pvm}_{a|x} \defvar \pvm_{a|x}\otimes\pure{0}_T + \pvm_{a\oplus 1|x}\otimes\pure{1}_T,
\label{eq:NPPprojs}
\end{align}
and store the outcome in $\hat{A}_x$ without further processing. However, 
the pinching channel on $T$ can in fact be omitted, because the subnormalized conditional states produced on $E$ would still be the same without it (in the following, we leave some tensor factors of $\id$ implicit for brevity, and consider an arbitrary state $\sigma_{ABET}$ before applying the pinching channel):
\begin{align}
&\tr[ABT]{\tilde{\pvm}_{a|x} 
\,\mathcal{P}_T (\sigma_{ABET})\,
\tilde{\pvm}_{a|x}}\nonumber\\
=&\tr[ABT]{\tilde{\pvm}_{a|x} \left(\sum_{t} 
\pure{t}_T
\sigma_{ABET} 
\pure{t}_T
\right) \tilde{\pvm}_{a|x}}\nonumber\\
=& \sum_{t} \tr[ABT]{(\pvm_{a\oplus t|x}\otimes\pure{t}_T)  \sigma_{ABET}(\pvm_{a\oplus t|x}\otimes\pure{t}_T) }\nonumber\\
=& \sum_{t,t'} \tr[ABT]{(\pvm_{a\oplus t|x}\otimes\pure{t}_T)  \sigma_{ABET}\left(\pvm_{a\oplus t'|x}\otimes\pure{t'}_T\right) }\nonumber\\
=& \tr[ABT]{\tilde{\pvm}_{a|x} \sigma_{ABET}\tilde{\pvm}_{a|x} }.
\end{align}
Hence we no longer consider the pinching channel on $T$, i.e.~we simply study the situation where we immediately perform the projective measurement~\eqref{eq:NPPprojs} on the state $\hat{\rho}_{ABET} \defvar \rho_{ABE} \otimes \pure{\phi_\p}_T$ state and store the outcome in register $\hat{A}_x$. As mentioned previously, we can take $\rho_{ABE}$ to be pure without loss of generality, in which case $\hat{\rho}_{ABET}$ is pure as well\footnote{It was important to remove the pinching channel because if $\hat{\rho}_{ABET}$ were a mixed state, then Eq.~\eqref{eq:Hduality} would be replaced by $D(\hat{\rho}_{ABT} \Vert \calZ_x(\hat{\rho}_{ABT})) = H(\hat{A}_x|EE')$, where $E'$ purifies $\hat{\rho}_{ABET}$. While $H(\hat{A}_x|EE')$ is a valid lower bound on $H(\hat{A}_x|E)$, it also turns out to be a trivial one, because the value of $T$ is ``copied'' into $E'$, 
removing the advantage of noisy preprocessing (in which Eve is not supposed to know whether the bitflip has occurred).}, and must hence obey the following relation (derived in e.g.~Theorem~1 of~\cite{Col12}):
\begin{align}
H(\hat{A}_x|E) = D(\hat{\rho}_{ABT} \Vert \calZ_x(\hat{\rho}_{ABT})) = D(\calG({\rho}_{AB}) \Vert \calZ_x(\calG({\rho}_{AB}))),
\label{eq:Hduality}
\end{align}
where
\begin{align}
\calG(\sigma_{AB}) &\defvar \sigma_{AB} \otimes \pure{\phi_\p}_T, \\
\calZ_x(\sigma_{ABT}) &\defvar \sum_a (\tilde{\pvm}_{a|x}\otimes\id_{B}) \sigma_{ABT} (\tilde{\pvm}_{a|x}\otimes\id_{B}) .
\end{align}
Importantly, this expression for $H(\hat{A}_x|E)$ is entirely in terms of the reduced state $\rho_{AB}$, and is convex with respect to $\rho_{AB}$. 

\section{Qubit reduction}
\label{sec:qubit}

We now specialize to 2-input 2-output scenarios. In such cases, 
we have the following ``qubit reduction''~\cite{PAB+09,HST+20}: if 
one finds
a convex function that lower-bounds the right-hand-side of Eq.~\eqref{eq:mainopt} with its optimization restricted to states $\rho_{ABE}$ of dimension $2\times2\times4$ and Pauli measurements (i.e.~rank-1 projective measurements), 
then this function is also a lower bound on 
$\breve{\lin}_\p$.
Note that there is a subtlety in this condition that there are only 2 inputs and 2 outputs; namely, this condition only needs to apply for the measurements involved in the optimization~\eqref{eq:mainopt}, which are the measurements used in the test rounds. In many DIQKD protocols, there are other inputs that Bob uses only in the generation rounds, in order to have stronger correlations with Alice's output $\hat{A}$ (in other words, keeping the value of $H(\hat{A}|\hat{B}XY)_\mathrm{hon}$ low). In terms of the full protocol, Bob is allowed to have more than 2 inputs, and furthermore the inputs that are only used in generation rounds are allowed to have more than 2 outcomes --- it is only the test rounds that are restricted to be 2-input 2-output.\footnote{A vital element of this argument is the fact that for the proof approaches we have outlined, it suffices to bound the entropy of Alice's outputs only (since the $H(\hat{A}|\hat{B}XY)_\mathrm{hon}$ term in~\eqref{eq:devwin} is instead computed from the honest behaviour); in particular, when solving~\eqref{eq:mainopt} we do not need to consider measurements that Bob only performs in generation rounds. This is not necessarily the case for advantage-distillation protocols, as we discuss in Chapter~\ref{chap:AD}.} 

With this reduction, it is possible to obtain bounds on the optimization~\eqref{eq:mainopt} in various scenarios, as long as they are 2-input 2-output. Note that the basic~\cite{PAB+09} protocol corresponds to taking $\p=0$, $\keyw_x=\delta_{x,0}$ (i.e.~only the $x=0$ input is used for generation rounds), and imposing only the CHSH value as a constraint. We now give a quick overview of the subsequent developments on this topic, all of which are fairly recent. 

Firstly, in~\cite{HST+20} the optimization was solved in closed form\footnote{There was a flaw in the proof that the final bound is convex, which was corrected by~\cite{WAP21}, though it does not change the final expression.} for the case where noisy preprocessing is included, but still only using $x=0$ for key generation and only considering the CHSH value. This was achieved by following a similar approach to~\cite{PAB+09}, though the analysis was substantially more complicated because the noisy preprocessing obstructs some of the simplifications used in~\cite{PAB+09} (for example, being able to explicitly compute the eigenvalues of some conditional states). Explicitly, the final bound obtained is
\begin{align}
\breve{\lin}_\p({\constr}) = 1 - \binh\left(\frac{1+\sqrt{(\constr/2)^2 - 1}}{2}\right) + \binh\left(\frac{1+\sqrt{1 - \p(1-\p) (8-\constr^2)}}{2}\right),
\end{align}
where $\constr$ is the CHSH value. It was also noted in that work that for the limited-detection-efficiency noise model, one can slightly improve the keyrates by using the observation in~\cite{ML12}: namely, instead of forcing \emph{all} the measurements to have binary outcomes, we can keep the $\perp$ outcome for Bob's generation measurement(s).
(This is in line with what we noted earlier, that Bob's measurements for generation rounds are in fact allowed to have more than 2 outcomes.)
This somewhat improves the keyrate by keeping the $H(\hat{A}|\hat{B}XY)_\mathrm{hon}$ term smaller, thanks to the data-processing inequality.

Subsequently, in~\cite{WAP21} tight closed-form bounds were derived for a more general case: again, noisy preprocessing is allowed while still restricting the generation input to $x=0$ only, but now the bound is with respect to a generalized CHSH value,
\begin{align}
\alpha \left(\expval{A_0 B_0} + \expval{A_0 B_1}\right) + \expval{A_1 B_0} - \expval{A_1 B_1},
\label{eq:tiltCHSH}
\end{align}
for $\alpha\in\mathbb{R}$. The proof approach in this work also served to somewhat simplify the analysis in~\cite{PAB+09}. Similar results were obtained independently in~\cite{SBV+21}, though the approach in that work more closely followed that of~\cite{PAB+09}, and was partially numerical in the $|\alpha|<1$ regime.

Separately, the random-key-measurements technique was proposed in~\cite{SGP+21}, corresponding to allowing arbitrary choices of $\keyw_x$. This work did not include noisy preprocessing, and only considered the CHSH value. It was found that by using random key measurements, the loophole-free Bell tests in~\cite{HBD+15,RBG+17} would be able to achieve positive asymptotic keyrates in principle (we will discuss this in more detail in the next chapter, where we also address finite-size effects). In contrast to the closed-form results above, this work used a numerical algorithm to obtain secure lower bounds on~\eqref{eq:mainopt}. It is rather similar to the one that we shall present next; however, it made use of an inequality that is not provably tight in that context, and later results~\cite{WAP21,arx_TSB+20} indicated that indeed there was a small loss introduced by using that inequality. 

Given all these techniques, a natural question is whether they can all be combined and analyzed together. This was the goal of the algorithm we developed in~\cite{arx_TSB+20}. In summary, its advantage compared to the results listed above is that it achieves all of the following simultaneously: it applies to arbitrary 2-input 2-output protocols, it accounts for noisy preprocessing and random key measurements, and it provably converges to a tight bound. The drawback is that it is a rather computationally intensive numerical algorithm (though we stress again that the resulting bounds are secure; it is not a heuristic minimization). We now describe this algorithm in some detail. 
For completeness,
we will also list in Sec.~\ref{sec:depolthresh} the depolarizing-noise thresholds obtained by the approaches of~\cite{HST+20,WAP21,SBV+21,SGP+21,arx_TSB+20}, for easier comparison.
(Very recently, a much more computationally efficient method to solve this optimization under the qubit reduction was developed; refer to~\cite{arx_MPW21} for further details.)

\subsection{Secure numerical algorithm}

As mentioned above, our goal would be to find convex function that lower-bounds the right-hand-side of Eq.~\eqref{eq:mainopt} with its optimization restricted to states $\rho_{ABE}$ of dimension $2\times2\times4$ and Pauli measurements.
We note (by following the same arguments as before, but with the optimization domain restricted) that if we pick some $\vec{\lambda}$ and solving the optimization~\eqref{eq:intercept} for the value $c_{\vec{\lambda}}$ over such states and measurements, we would obtain such a lower bound via Eq.~\eqref{eq:duallin}, which is trivially convex since it is affine. Hence it suffices to study the optimization~\eqref{eq:intercept} restricted to such states and measurements. Furthermore, the resulting bounds are still tight in the same sense as before, in that taking the supremum over choices of $\vec{\lambda}$ yields the convex envelope of this restricted version of the optimization~\eqref{eq:mainopt} (except possibly for $\vec{\constr}$ not in the interior of the quantum-achievable values).

Since we have reduced the analysis to Pauli measurements, it is convenient to use the convention where the measurements have output values in $\{-1,+1\}$ instead of $\mathbb{Z}_2$, and define the corresponding hermitian observables\footnote{This notation is similar to the notation $A_j,B_j$ we will use to denote registers in individual rounds in the next chapter; however, we should not need to use these two notations in close proximity and hence there should be little risk of confusion.} $A_x \defvar \sum_a a \pvm_{a|x}$ and $B_y \defvar \sum_b b \pvm_{b|y}$. For 2-input 2-output scenarios, the full probability distribution $\pr{ab|xy}$ is completely parametrized by the 4 correlators $\expvaltr{(A_x \otimes B_y)}$ and the 4 marginals $\expvaltr{(A_x \otimes \id)}$, $\expvaltr{(\id \otimes B_y)}$ (as long as it satisfies the NS conditions; see Appendix~\ref{app:corrs} for details of this parametrization). Without loss of generality, one can assume that the marginals are zero~\cite{PAB+09,HST+20}\footnote{Essentially, this is because one could consider a virtual \term{symmetrization step} in which Alice and Bob jointly flip their outputs using a uniformly random public bit, forcing their marginals to be zero. As argued in~\cite{SR08,PAB+09,HST+20}, the entropy of their original outputs conditioned on Eve's side-information is equal to the entropy of their symmetrized outputs conditioned on Eve's side-information and the publicly communicated bit, and hence we can analyze the latter instead (note that the virtual symmetrization does not need to be physically performed; it merely serves as an intermediate construction to study the entropies).}, hence it suffices to focus on the correlators $\expvaltr{(A_x \otimes B_y)}$. In terms of the optimization~\eqref{eq:mainopt}, this means that it suffices most generally to consider situations where there are 4 constraints, corresponding to the observables
\begin{align}
\Gamma_{xy}(\pvm_{a|x},\pvm_{b|y}) = A_x \otimes B_y, 
\quad \text{ for } x,y\in\{0,1\}.
\end{align}
Considering other forms of constraints in 2-input 2-output scenarios is essentially equivalent to making specific choices of Lagrange multipliers $\vec{\lambda}$ for this 4-constraint formulation. For instance, imposing a constraint based on the CHSH value
is equivalent to restricting to Lagrange-multiplier combinations of the form $(\lambda_{00},\lambda_{01},\lambda_{10},\lambda_{11}) = (\lambda,\lambda,\lambda,-\lambda)$ for some $\lambda\in\mathbb{R}$.

We still have the freedom to choose the basis in which to express the optimization. Following~\cite{SGP+21}, we can use the measurement axes of Alice's two Pauli measurements to define the $X$-$Z$ plane on her system, taking $A_0=Z$ and $A_1 = \cos(\thA)Z + \sin(\thA)X$ for some $\thA \in [0,\pi]$ 
(values of $\thA$ in $[\pi,2\pi]$ can be brought into this range by rotating our axis choice by $\pi$ around the $Z$-axis). 
Analogously, we can choose a basis for Bob such that $B_0=Z$ and $B_1 = \cos(\thB)Z + \sin(\thB)X$ with $\thB \in [0,\pi]$. (This is a different basis choice from the one in~\cite{PAB+09,HST+20} that allows a reduction to {Bell-diagonal} 
$\rho_{AB}$. That choice involves more parameters for the measurements, but fewer parameters for the state.)

With this in mind, we rewrite the optimization~\eqref{eq:intercept} as 
\begin{align}\label{eq:qubitopt}
\begin{gathered}
\min_{\thA} \min_{\thB} 
\min_{\rho_{AB}} 
F_\mathrm{obj}(\thA,\thB,\rho_{AB}),
\text{ where } \\ 
F_\mathrm{obj}(\thA,\thB,\rho_{AB}) \defvar \left( \sum_{x\in\{0,1\}} \keyw_x D(\calG({\rho}_{AB}) \Vert \calZ_x(\calG({\rho}_{AB})))\right)
- \vec{\lambda} \cdot \tr{\vec{\Gamma}(
\thA,\thB
) \rho_{AB}}.
\end{gathered}
\end{align}
(The minima are attained because the objective function is continuous and the domain is compact.)
Furthermore, the objective function is invariant under the substitutions $\rho_{AB} \to (Y \otimes Y) \rho_{AB} (Y \otimes Y)$ and $\rho_{AB} \to \rho_{AB}^*$, which implies~\cite{PAB+09} we can restrict the optimization to $\rho_{AB}$ that are ``almost'' Bell-diagonal --- specifically, with respect to the Bell basis $\{\ket{\Phi^+}, \ket{\Psi^-}, \ket{\Phi^-}, \ket{\Psi^+}\}$ where $\ket{\Phi^\pm} = (\ket{00}\pm\ket{11})/\sqrt{2}$ and $\ket{\Psi^\pm} = (\ket{01}\pm\ket{10})/\sqrt{2}$, we can take
\begin{align}
\rho_{AB} = 
\begin{pmatrix}
L_{\Phi^+} & \ell_1 & 0 & 0 \\
\ell_1 & L_{\Psi^-} & 0 & 0 \\
0 & 0 & L_{\Phi^-} & \ell_2 \\
0 & 0 & \ell_2 & L_{\Psi^+}
\end{pmatrix},
\label{eq:almostdiag}
\end{align}
for some $L_{\Phi^+} , L_{\Psi^-} , L_{\Phi^-} , L_{\Psi^+} , \ell_1 , \ell_2 \in \mathbb{R}$.

The optimization~\eqref{eq:qubitopt} comprises of minimizations over two measurement angles and the 
Alice--Bob state (restricted to the form~\eqref{eq:almostdiag}).
Our numerical approach is based on the fact that each of these optimizations can be securely lower-bounded individually. The details are somewhat tedious and hence we defer them to Appendix~\ref{app:qubitnumerics}; however, the key idea in each case can be briefly summarized as follows:
\begin{itemize}
\item Minimization over $\thA$: First, we find a continuity bound\footnote{This is computed by using a continuity bound derived in~\cite{SBV+21}, which bounds the maximum difference in von Neumann entropy between two states in terms of the fidelity between the states. To the best of our knowledge, this was a novel result (previous continuity bounds such as the Fannes-Audenaert inequality and its refinements~\cite{Win16} were based on trace distance; however, we found them to yield worse bounds for our purposes).} for $F_\mathrm{obj}(\thA,\thB,\rho_{AB})$ with respect to $\thA$. Then by computing the value of $F_\mathrm{obj}(\thA,\thB,\rho_{AB})$ on a sufficiently fine grid of values for $\thA$, we can use the continuity bound to bound the maximum difference between the true minimum and the minimum value on the grid. 

\item Minimization over $\thB$: We observe that  $F_\mathrm{obj}(\thA,\thB,\rho_{AB})$ is an affine function of the variables $r_Z\defvar\cos(\thB)$, $r_X\defvar\sin(\thB)$, which satisfy the constraints $r_Z^2+r_X^2=1$ and $r_X \geq 0$. These constraints describe a semicircular arc.
By using this fact, it is possible to show that this minimization can be relaxed to a minimization over a finite number of points (which correspond to vertices of a polytope that contains the semicircular arc); an approach which was used in~\cite{SGP+21}.

\item Minimization over $\rho_{AB}$: 
In the previous section, we applied the approach of~\cite{WLC18} to express the objective function in a form that is convex with respect to $\rho_{AB}$. Hence the minimization can be securely lower-bounded using the Frank-Wolfe algorithm~\cite{FW56} (this was the approach applied in~\cite{WLC18}\footnote{Recently, an improvement over this approach was achieved using an interior-point algorithm~\cite{arx_HIL+21}, which was found to converge substantially faster. In principle, this could be used to replace the Frank-Wolfe algorithm here as well, though we leave this for future work.} for device-dependent QKD, where the task is precisely to solve the minimization over the state, for fixed measurements).
\end{itemize}
Slight care is needed to correctly combine these individual approaches, but we discuss this as well in Appendix~\ref{asec:algorithmsummary}. 

In general, this approach faces the difficulty that it needs to optimize over the choice of Lagrange multipliers $\vec{\lambda}$. Given that our approach for solving~\eqref{eq:qubitopt} for a specific choice of Lagrange multipliers is already highly computationally intensive (requiring about $5000$ core-hours to achieve the level of accuracy in the bounds~\eqref{eq:certbnds} below for each value of $\p$), it would be impractical to also optimize over the Lagrange multipliers while doing so. It is more feasible to first optimize the Lagrange multipliers while using a simple heuristic algorithm to estimate the minimizations, then certify the final result using our approach for solving~\eqref{eq:qubitopt}. (This is essentially the same perspective as presented in~\cite{SBV+21}.)

We remark that this approach can also yield arbitrarily tight bounds for DIRE, where (as noted in Chapter~\ref{chap:overview}) the goal would typically be to find lower bounds on (weighted sums of) the two-party entropies $H(\hat{A}_x \hat{B}_y|E)$. This is because by the same arguments as above, we have~\cite{Col12,TSG+21}
\begin{align}
\begin{gathered}
H(\hat{A}_x \hat{B}_y|E) = D({\rho}_{AB} \Vert \pinch_{xy}({\rho}_{AB})),
\\ \text{where }
\pinch_{xy}(\sigma_{AB}) \defvar \sum_a (\pvm_{a|x} \otimes \pvm_{b|y}) \sigma_{AB} (\pvm_{a|x} \otimes \pvm_{b|y}) .
\end{gathered}
\label{eq:entdualDIRE}
\end{align}
(Here we omit the parts corresponding to noisy preprocessing, since it is not relevant for DIRE.) This expression can be bounded in the same way as we have just described above, though the objective function would no longer be affine with respect to $(r_Z,r_X)$, and hence the optimization over Bob's measurements would also have to be approached using a continuity bound. This should yield a substantial improvement over previous results, which simply used the CHSH-based bound of~\cite{PAB+09} to bound the entropy of only one party's outputs~\cite{LLR+21}, or which consider the full distribution and bound the entropy of both outputs but use inequalities that are not tight~\cite{BRC20,TSG+21,BFF21}. In addition, the fact that it allows for the random-key-measurement approach could yield further improvements, though there are some technicalities that we address at the end of Sec.~\ref{sec:preshared}. We remark that this question was recently explored using heuristic numerical methods, in~\cite{arx_BRC21}.

\subsection{Resulting bounds}
\label{sec:1rndresults}

\begin{figure}
\centering
\subfloat[$\p=0$]{
\includegraphics[width=0.49\textwidth]{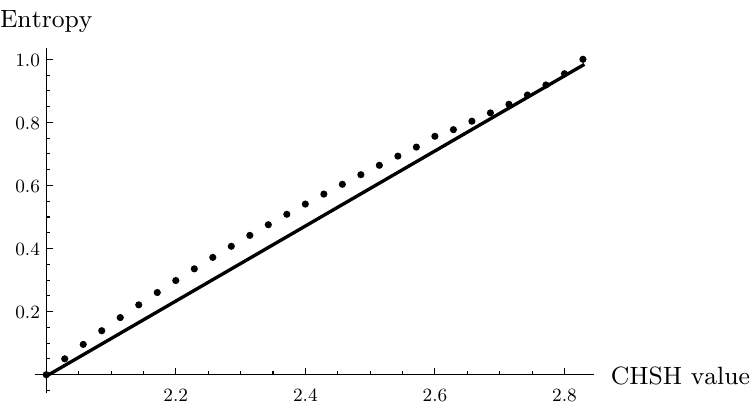}
} 
\subfloat[$\p=0.2$]{
\includegraphics[width=0.49\textwidth]{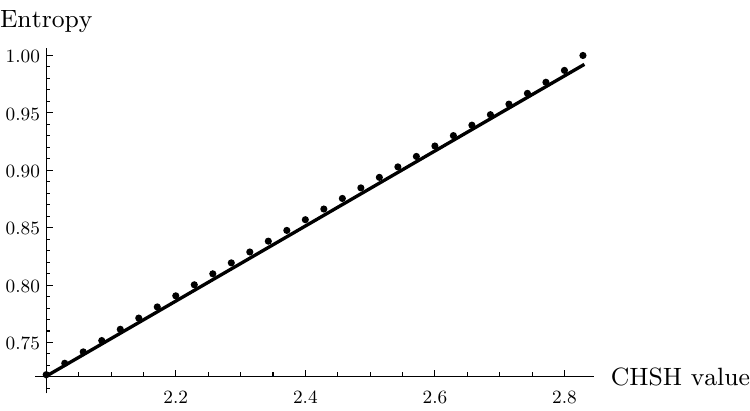}
}
\caption{(From~\cite{arx_TSB+20}) 
The solid lines are the certified lower bounds we obtained (Eq.~\eqref{eq:certbnds}), 
while the points indicate the results of heuristically solving the optimization \eqref{eq:mainopt} over qubit states and measurements (with just the CHSH value as the constraint). 
As previously discussed, the tight bound in each case would be given by the convex envelope of the curve traced out by the points, assuming that the heuristics have found the true minimum. However, we can see that in each case that curve appears to be nonconvex over the interval $[2,2.75]$ (approximately), and its convex envelope would be affine over that interval --- specifically, it would be given by the linear interpolation between the feasible points at the ends of that interval. The certified bound is almost flush with this linear interpolation, indicating that it is basically tight over this interval.
}
\label{fig:entbnds}
\end{figure}

Putting aside the precise details of the approach, we now turn to the results obtained.
We focus on the situation where $\keyw_0=\keyw_1=1/2$ and the only constraint imposed is the CHSH value, because this is the relevant situation for the full protocols studied in the next chapter. More precisely, we consider the described optimization with a single constraint corresponding to the operator
\begin{align}
\Gamma(\thA,\thB) = A_0\otimes B_0 + A_0\otimes B_1 + A_1\otimes B_0 - A_1\otimes B_1,
\end{align}
with the constraint value $\constr$ being the CHSH value. (As previously mentioned, this can be implemented in the formulation where are 4 constraint operators $\Gamma_{x,y}(\thA,\thB)$ by simply restricting to Lagrange-multiplier choices of the form $(\lambda_{00},\lambda_{01},\lambda_{10},\lambda_{11}) = (\lambda,\lambda,\lambda,-\lambda)$.) Each choice of the associated Lagrange multiplier $\lambda$ yields an affine lower bound of the form ${\lambda}{\constr} + c_{{\lambda}}$, as noted in Eq.~\eqref{eq:duallin}.
Importantly, some heuristic computations (also observed in~\cite{SGP+21} for the $\p=0$ case) suggest that the true bound $\breve{\lin}_\p({\constr})$ in this situation is in fact affine over a wide range of CHSH values --- we show this in Fig.~\ref{fig:entbnds}, which displays the results of heuristic minimizations compared to our certified bound in some cases. In particular, this range on which the bound is affine covers all currently experimentally reasonable values. 
This has the implication that there is a single ${\lambda}$ that yields an affine lower bound which is equal to $\breve{\lin}_\p({\constr})$ (i.e.~it is tight) over this entire range; specifically, it is simply the value of ${\lambda}$ corresponding to the gradient of $\breve{\lin}_\p({\constr})$ in this range. This greatly simplifies our task since we only need to solve the optimization for this specific value of ${\lambda}$.

We focused on several values of noisy preprocessing, ranging from $\p=0$ to $\p=0.45$.
In each case, we first solved the minimization~\eqref{eq:intercept} heuristically for some selection of values of ${\lambda}$, in order to estimate the choice of $\lambda$ that yields a tight bound over the range of CHSH values in which $\breve{\lin}_\p({\constr})$ is affine. Then using our algorithm to get certified bounds on the corresponding $c_{\lambda}$ in~\eqref{eq:intercept}, we arrived at the final bounds
\begin{align}
\begin{gathered}
\breve{\lin}_{0}({\constr}) \geq 1.190(\constr-2) - 0.00454, \qquad
\breve{\lin}_{0.2}({\constr}) \geq 0.327(\constr-2) + 0.72063, \\
\breve{\lin}_{0.3}({\constr}) \geq 0.139(\constr-2) + 0.88051, \qquad
\breve{\lin}_{0.4}({\constr}) \geq 0.0341(\constr-2) + 0.97055, \\
\breve{\lin}_{0.45}({\constr}) \geq 0.00855(\constr-2) + 0.992487,
\end{gathered}
\label{eq:certbnds}
\end{align}
some of which are shown in Fig.~\ref{fig:entbnds}. 

\begin{figure}
\centering
\includegraphics[width=0.6\textwidth]{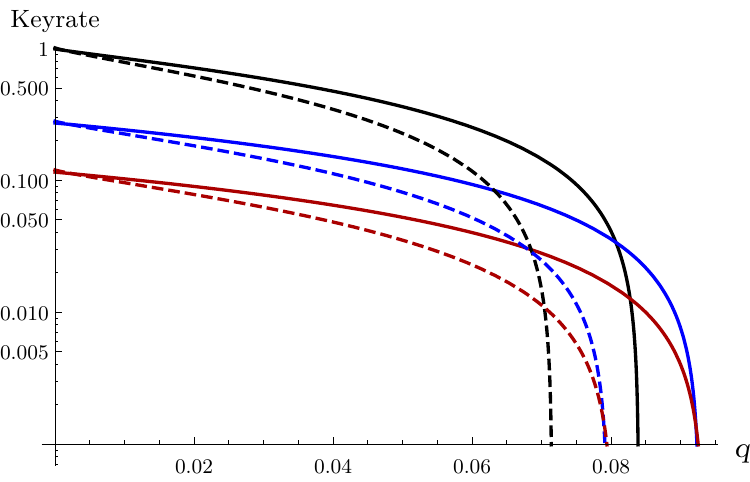}
\caption{(From~\cite{arx_TSB+20}) 
Lower bounds on asymptotic keyrates under depolarizing noise $\q$, on a vertical log scale. The black, blue, and red solid curves show the asymptotic keyrate given by the expression~\eqref{eq:keyratedepol}, for noisy-preprocessing values of $\p=0$, $0.2$ and $0.3$ respectively, based on the bounds~\eqref{eq:certbnds}. 
For comparison, the dashed curves show the corresponding asymptotic keyrates of the protocol in~\cite{HST+20}, which does not use the random-key-measurement method (the $\p=0$ case is equivalent to the~\cite{PAB+09} protocol). The solid curves intersect the horizontal axis at $\q=8.39\pct$, $9.26\pct$ and $9.33\pct$, in order of increasing $\p$. The first value is a minor improvement over~\cite{SGP+21} (despite being effectively the same protocol), likely because our algorithm provably converges to a tight keyrate bound (for fixed $\p$). The last value exceeds all previous bounds~\cite{HST+20,SGP+21,WAP21,SBV+21} for depolarizing-noise tolerance (see Sec.~\ref{sec:depolthresh}).
It is also close to the upper bound of $9.57\pct$ that we derive in Appendix~\ref{asec:optthresh}.}
\label{fig:keyrates}
\end{figure}

So far, these are only the bounds on the conditional entropy against Eve. We now turn to the question of the resulting keyrates. For simplicity, let us focus on the depolarizing-noise scenario, and consider a protocol where the asymptotic keyrate is given by
\begin{align}
\operatorname{rate}_\infty = 
\breve{\lin}_{0}({\constr})
- H(\hat{A}_0|\hat{B}_0)_\mathrm{hon} 
,
\label{eq:keyratedepol}
\end{align}
with $\constr$ still denoting the CHSH value.
This is just the expression~\eqref{eq:randomkeyrate} mentioned in Sec.~\ref{sec:sketchasympt}, taking $\pr{X=Y}=1$ and $\keyw_0=\keyw_1=1/2$. It may seem unclear how to ensure that $\pr{X=Y}=1$ in an actual protocol, but we will address this in the next chapter by explicitly presenting a protocol (Protocol~\ref{prot:preshared}) that achieves this asymptotic keyrate (at least in terms of \emph{net} key generation) by using a pre-shared key. In Fig.~\ref{fig:keyrates}, we plot the keyrates given by applying our bounds~\eqref{eq:certbnds} to this expression.
From the $\p=0.3$ bound, we obtain a depolarizing-noise threshold of $9.33\pct$.
However, we were unable to obtain better thresholds using the higher values of $\p$, for reasons we shall now discuss.

\subsection{Tightness of bounds}

For each of the bounds in Eq.~\eqref{eq:certbnds}, there is a feasible point of the optimization for $c_{\lambda}$ which yields a value within $0.005$ (or less, for higher values of $\p$) of the certified results shown in Eq.~\eqref{eq:certbnds}, so these bounds on entropy are very close to optimal in terms of absolute error.
In terms of the depolarizing-noise thresholds that they yield, taking the convex envelope of some of the feasible points shown in Fig.~\ref{fig:entbnds}
yields the result that the thresholds for $\p=0.2$ and $\p=0.3$ cannot be improved by more than about $0.1$ percentage points, so those thresholds are very close to optimal as well. 
However, larger values of $\p$ face the issue that the asymptotic keyrate becomes extremely low, which makes the horizontal intercept (i.e.~the depolarizing-noise threshold) very sensitive to changes in $\breve{\lin}_{\p}$ --- even a small absolute error in this bound results in a significant change in the threshold value. Therefore, the thresholds we obtained from the certified bounds with $\p=0.4$ and $\p=0.45$ in Eq.~\eqref{eq:certbnds} were only $9.32\pct$ and $9.10\pct$ respectively, worse than the results for $\p=0.3$. Heuristic computations suggest that the true thresholds for those cases might be approximately $9.46\pct$ and $9.50\pct$ respectively, but using our algorithm to certify these values would require it to converge to tolerances that currently appear impractical. Hence a different approach may be needed to find the true thresholds for these values of $\p$. From a practical perspective though, such improvements may be of limited use, because the very low asymptotic rates mean that the finite-size keyrate would likely be zero until extremely large sample sizes.

In any case, we note that for depolarizing noise at least, the threshold value cannot be improved much further by \emph{any} protocol choices within the framework we have presented in this section, e.g.~by using the full distribution as constraints (which also encompasses the use of modified CHSH inequalities~\cite{WAP21,SBV+21}), or adjusting the values of $\keyw_x$. This is essentially because our bounds are very close to the linear interpolation between the points $(2,\binh(\p))$ and $(2\sqrt{2},1)$, as can be seen from Fig.~\ref{fig:entbnds}. Intuitively speaking, the bound on the entropy against Eve in the depolarizing-noise scenario cannot exceed this linear interpolation (because Eve can always perform classical mixtures of strategies in order to attain every point on this linear interpolation), which means that our bounds are close to the highest bounds that are even possible in principle. We discuss this in detail in Appendix~\ref{asec:optthresh}.

\subsection{Comparison of noise thresholds}
\label{sec:depolthresh}

For easier reference, we briefly list the depolarizing-noise thresholds at which the keyrates given by the various approaches become zero. Unless otherwise specified, all these results are based on the CHSH value only.
\def\arraystretch{1.5} 
\begin{center}
\begin{tabular}{L{105mm} C{20mm}}
Basic~\cite{PAB+09} protocol: & $7.14\pct$ \\
Noisy preprocessing only~\cite{HST+20,WAP21}: & $8.08\pct$ \\
Noisy preprocessing and modified CHSH as in~\eqref{eq:tiltCHSH}~\cite{WAP21}: & $8.34\pct$ \\
Random key measurements only\protect\footnotemark~\cite{SGP+21,arx_TSB+20}:& $8.39\pct$ \\
Noisy preprocessing and random key measurements~\cite{arx_TSB+20}: & $9.33\pct$
\end{tabular}
\end{center}
\footnotetext{The value stated here is the one derived in~\cite{arx_TSB+20}, as it is slightly better than the result obtained in~\cite{SGP+21} based on a bound that may not be tight.} 
\def\arraystretch{1}
Also, there is an upper bound of $9.57\pct$ (see Appendix~\ref{asec:optthresh}) on the noise threshold for this family of protocols.

For limited detection efficiency, the situation is more complicated, because of the fact that the choice of noiseless behaviour needs to be tailored as a function of $\eta$. 
To roughly summarize, focusing on the case where the noiseless distribution is produced from a two-qubit state (see~\cite{SBV+21} for results for a more sophisticated photonic model), the best detection efficiency threshold achieved without using any of the techniques described here was $86.5\pct$ (though a number of other values have been cited; see Appendix~\ref{app:detthresh} for a brief summary of the variants).
The results of~\cite{HST+20,WAP21,SBV+21} showed that the threshold can be lowered to $82.6\pct$ by using noisy preprocessing and keeping the $\perp$ outcome for Bob's generation measurement~\cite{ML12}. They also found that using modified CHSH inequalities did not improve this threshold specifically as compared to using the CHSH value (though it gives small improvements in other situations). These works do not implement random key measurements. 
On the other hand,~\cite{SGP+21} found that the random-key-measurements technique on its own does not improve the threshold significantly as compared to the basic~\cite{PAB+09} protocol, although it does somewhat increase the keyrates in the regime where the keyrate is already positive. 

When these techniques are combined, it is unclear how much the detection-efficiency threshold is changed, because of the complexity of optimizing the various parameters. With some heuristic exploration in~\cite{arx_TSB+20}, we were unable to improve the threshold significantly as compared to the~\cite{HST+20,WAP21,SBV+21} values. However, these heuristic results may not have found the true optimal choice of noiseless behaviour --- see also the results of~\cite{arx_BFF21,arx_MPW21}, which found some improvements by considering statistics other than the CHSH value.

\section{General scenarios}
\label{sec:NPO}

We now turn to the question of how to go beyond 2-input 2-output scenarios. Currently, this is still something of an open question --- in~\cite{TSG+21}, we developed an approach to handle this; however, the bound is not tight, and does not consider noisy preprocessing (though in light of this, we do allow more than 2 outputs for the key-generating measurements, in principle). A subsequent work~\cite{BFF21} by other authors has similar limitations, though it is significantly better in terms of computational efficiency. (The bounds in these two works seem to perform well in different noise regimes.) Very recently, the same authors developed a computationally efficient method which converges to arbitrarily tight bounds on this optimization~\cite{arx_BFF21}, and this currently appears to be the state of the art on this topic.
Still,~\cite{TSG+21} was our first attempt to address the optimization~\eqref{eq:mainopt} without going through the indirect approach via min-entropy~\cite{PAM+10,NPS14,BSS14} or the specialized approach of the qubit reduction~\cite{PAB+09}, and we shall thus focus on presenting it here.

To give an overall picture, the basic idea behind this approach (and also that of~\cite{BFF21,arx_BFF21}) is to find a way to lower bound the optimization~\eqref{eq:mainopt} in terms of an expression of the form
\begin{align}
\begin{gathered}
\operatorname*{optimize}_{\rho, \vec{M} } \tr{Q\big(\vec{M}\big)\, \rho } \\ 
\suchthat\; \tr{R_j\big(\vec{M}\big)\, \rho} = 0 \;,
\end{gathered} 
\label{eq:NPO}
\end{align}
where the optimization (which can be a supremum or infimum) takes place over states $\rho$ and operators $M_k$ of arbitrary dimension, and the functions $Q$ and $R_j$ are (noncommutative) polynomials in the operators $M_k$. 
The operators $M_k$ may also be required to satisfy some algebraic constraints such as orthogonality conditions or commutation relations.
Such an optimization may be referred to as a \term{noncommutative polynomial optimization}. 

Importantly, such optimizations in the context of nonlocality have been a subject of some interest, and an approach was developed for relaxing such optimizations to a hierarchy of SDPs, known as the NPA hierarchy~\cite{NPA08}. This is an infinite sequence of SDPs such that each SDP provides a secure bound on the optimization~\eqref{eq:NPO}, and in some scenarios of interest (such as computing the maximum Bell value that can be obtained by quantum systems on arbitrary Hilbert spaces), it can be proven that this sequence of bounds converges to the true optimal value. Hence by reducing our optimization to the form~\eqref{eq:NPO}, we can apply the NPA hierarchy to obtain secure lower bounds on our main optimization. The guessing-probability bounds derived in~\cite{PAM+10,NPS14,BSS14} were based on this SDP hierarchy as well, though as noted in Chapter~\ref{chap:overview}, this does not yield a very tight bound on the von Neumann entropy.

We now turn to the details. 
In this section, one should keep in mind that while we do not have the qubit reduction, we do still impose the condition that all operators are finite-dimensional (to avoid technical issues in some of the theorems we use).\footnote{The NPA hierarchy itself is a notable exception to this restriction, since in fact it converges to the value for infinite-dimensional systems (and strictly speaking, only imposes commutation constraints between the measurements, rather than a tensor-product structure). Hence in principle, the NPA hierarchy may in fact not converge to precisely the value of interest we consider in this work, which is restricted to systems of arbitrary finite dimension, with a tensor-product structure between Alice and Bob. However, since it bounds a larger set than we consider in this work, the resulting bounds are still guaranteed to be secure, which is sufficient for security proofs (though the results may not be tight).} However, the bounds we derive will hold for all finite dimensions, and are independent of the dimension. Also, for the following calculations, it will be easier to work with entropies in base $e$ rather than base $2$. We shall denote this with the subscript ${}_{\nat}$, i.e.~we have $H(\hat{A}_x|E) = (\log e) H_{\nat}(\hat{A}_x|E)$ and so on. Hence in place of~\eqref{eq:mainopt}, we study
\begin{align}
\begin{gathered}
\inf_{\rho_{ABE},\pvm_{a|x},\pvm_{b|y}} \sum_{x} \keyw_x H_{\nat}(
\hat{A}_x
|{E}) \\ 
\suchthat\; \tr{\vec{\Gamma}(\pvm_{a|x},\pvm_{b|y}) \rho_{AB}} = \vec{\constr} \;.
\end{gathered} 
\label{eq:mainoptnat}
\end{align}
Obviously, this optimization is equivalent to~\eqref{eq:mainopt} after rescaling by $\log e$. 

\subsection{Using the Lagrange dual}

We first present an approach based on the Lagrange dual.
As mentioned above, we restrict our attention to the simpler situation where there is no noisy preprocessing. In that case, rather than Eq.~\eqref{eq:Hduality}, we have a simpler relation (again by Theorem~1 of~\cite{Col12}):
\begin{align}\label{eq:Hdualsimple}
\begin{gathered}
H(\hat{A}_x|E) = D({\rho}_{AB} \Vert \pinch_x({\rho}_{AB})),
\\ \text{where }
\pinch_x(\sigma_{AB}) \defvar \sum_a (\pvm_{a|x} \otimes \id_B) \sigma_{AB} (\pvm_{a|x} \otimes \id_B) .
\end{gathered}
\end{align}
Furthermore, for now we will focus on the case where $\keyw_{\xo}=1$ and $\keyw_x=0$ otherwise 
(in short, $\keyw_x=\delta_{x,\xo}$), 
deferring the case of arbitrary $\keyw_x$ to Sec.~\ref{sec:randbasisgen}. 
We prove the following result for this situation:\footnote{For now we only need the $\xg=0$ case of the formula~\eqref{eq:oppoly}, but for later use we shall write this expression for arbitrary $\xg \in \mathcal{X}$.}
\begin{theorem}\label{th:lagbnd}
When noisy preprocessing is not applied, and $\keyw_x=\delta_{x,\xo}$,
the optimization~\eqref{eq:mainoptnat} is lower bounded by
\begin{align}
\sup_{\vec{\lambda}}
\; -\frac{1}{e} 
\left(
\sup_{\sigma_{AB},\pvm_{a|x},\pvm_{b|y}} \tr{\exp\left( \ln(\pinch_{\xo}(\sigma_{AB})) + \vec{\lambda} \cdot\vec{\Gamma}(\pvm_{a|x},\pvm_{b|y}) \right)} 
\right)
+ \vec{\lambda} \cdot \vec{\constr} ,
\label{eq:lagbnd}
\end{align}
and this bound is tight for all $\vec{\constr}$ such that strong duality holds in the Lagrange dual.
Furthermore, this is in turn lower bounded by 
\begin{align}
\sup_{\vec{\lambda}}\; -\frac{1}{e} 
\left(
\sup_{\sigma_{AB},\pvm_{a|x},\pvm_{b|y}} \tr{\sigma_{AB} \, Q_{\xo}(\pvm_{a|x},\pvm_{b|y})}
\right)
+ \vec{\lambda} \cdot \vec{\constr} ,
\label{eq:lagbndGT}
\end{align}
where $Q_{\xo}(\pvm_{a|x},\pvm_{b|y})$ is an operator polynomial defined via the following expression\footnote{The expression $\alpha \csch \alpha$ has a removable discontinuity at $\alpha=0$; it should be understood to be ``filled in'' by replacing it with its limit value, i.e.~we take $0 \csch 0 \defvar \lim_{\alpha\to0} \alpha \csch \alpha = 1$.} (the indices $\vec{a},\vec{a}'$ in the summation here range over the set $\mathcal{A}^{|\mathcal{X}|}$, and similarly the indices $\vec{b},\vec{b}'$ range over the set $\mathcal{B}^{|\mathcal{Y}|}$):
\begin{align}
Q_{\xg}(\pvm_{a|x},\pvm_{b|y}) \defvar 
\sum_{\vec{a}\vec{b}\vec{a}'\vec{b}'} 
e^{
\alpha^+_{\vec{a}\vec{b}\vec{a}'\vec{b}'} 
}
\alpha^-_{\vec{a}\vec{b}\vec{a}'\vec{b}'} \csch \alpha^-_{\vec{a}\vec{b}\vec{a}'\vec{b}'} 
\pinch_{\xg} \left(
\left(\prod_{xy} \pvm_{a'_x|x} \otimes \pvm_{b'_y|y}  \right)^\dagger
\left(\prod_{xy} \pvm_{a_x|x} \otimes \pvm_{b_y|y} \right)
\right),
\label{eq:oppoly}
\end{align}
where
$\alpha^{\pm}_{\vec{a}\vec{b}\vec{a}'\vec{b}'} \defvar \frac{1}{2} \sum_{jxy} \lambda_j \left(c^{(j)}_{a_xb_yxy} \pm c^{(j)}_{a'_xb'_yxy}\right) $.
The two operator products in~\eqref{eq:oppoly} must be taken in the same order over $xy$; any orderings satisfying that condition yield a valid lower bound.
\end{theorem}

This theorem gives us a method to compute secure lower bounds on the optimization, since in the second expression~\eqref{eq:lagbndGT}, for any fixed choice of $\vec{\lambda}$ we can apply the NPA hierarchy to bound the inner optimization in that expression. 
Optimizing over values of $\vec{\lambda}$ yields tighter bounds --- again, it suffices to do this optimization heuristically, since 
every feasible $\vec{\lambda}$ 
yields a secure lower bound. Also, one can experiment with different orderings of the operator product to obtain better bounds. 
Finally, we remark that when implementing this approach, it is useful to make a slight further simplification to~\eqref{eq:oppoly}: if we denote the first index in the operator product as $x^\star y^\star$, then we have (recalling both operator products were to be taken in the same order) 
\begin{align}
&\left(\prod_{xy} \pvm_{a'_x|x} \otimes \pvm_{b'_y|y}  \right)^\dagger
\left(\prod_{xy} \pvm_{a_x|x} \otimes \pvm_{b_y|y}  \right) \nonumber\\
=& \dots \Big(\pvm_{a'_{x^\star}|x^\star} \otimes \pvm_{b'_{y^\star}|y^\star} \Big)
\Big(\pvm_{a_{x^\star}|x^\star} \otimes \pvm_{b_{y^\star}|y^\star} \Big) \dots \nonumber\\
=& \dots \delta_{a'_{x^\star},a_{x^\star}}\delta_{b'_{y^\star},b_{y^\star}}
\Big(\pvm_{a_{x^\star}|x^\star} \otimes \pvm_{b_{y^\star}|y^\star} \Big) \dots , 
\label{eq:oppolysimp}
\end{align}
which lets us eliminate some terms from the sum using the Kronecker deltas 
(specifically, we can omit all terms such that $a'_{x^\star} \neq a_{x^\star}$ or $b'_{y^\star} \neq b_{y^\star}$). 

In the remainder of this section, we present a proof of Theorem~\ref{th:lagbnd}, in a somewhat expository fashion.
We begin by again using the fact that the optimization~\eqref{eq:mainoptnat} is lower-bounded by the Lagrange dual~\eqref{eq:dual} (with the trivial modification of having $H_{\nat}(\hat{A}_x|E)$ in place of $H(\hat{A}_x|E)$).
We put aside for now the supremum over the Lagrange multipliers $\vec{\lambda}$, with the implicit understanding that to get as tight a bound as possible, we will have to optimize over $\vec{\lambda}$ at the end. In that case, the $\vec{\lambda}\cdot\vec{\constr}$ term can be treated as a constant, so we omit it from further study as well. We can now simplify the remaining terms of interest in the Lagrange dual~\eqref{eq:dual} by applying~\eqref{eq:Hdualsimple} and then following the calculation in~\cite{CML16}: 
\begin{align}
&\inf_{\rho_{ABE},\pvm_{a|x},\pvm_{b|y}} 
H_{\nat}(\hat{A}_{\xo}|{E}) - \tr{\vec{\lambda} \cdot \vec{\Gamma}(\pvm_{a|x},\pvm_{b|y}) \rho_{AB}} \nonumber\\
=&\inf_{\rho_{AB},\pvm_{a|x},\pvm_{b|y}} D_{\nat}({\rho}_{AB} \Vert \pinch_{\xo}({\rho}_{AB})) - \tr{\vec{\lambda} \cdot\vec{\Gamma}(\pvm_{a|x},\pvm_{b|y}) \rho_{AB}} \nonumber \\
=& \inf_{\rho_{AB},\sigma_{AB},\pvm_{a|x},\pvm_{b|y}} D_{\nat}({\rho}_{AB} \Vert \pinch_{\xo}({\sigma}_{AB})) - \tr{\vec{\lambda} \cdot\vec{\Gamma}(\pvm_{a|x},\pvm_{b|y}) \rho_{AB}} \nonumber \\
=& \inf_{\sigma_{AB},\pvm_{a|x},\pvm_{b|y}} 
-\frac{1}{e}\tr{\exp\left( \ln(\pinch_{\xo}(\sigma_{AB})) + \vec{\lambda} \cdot\vec{\Gamma}(\pvm_{a|x},\pvm_{b|y}) \right)} ,
\label{eq:ZTFbound}
\end{align}
where the last line uses the fact that the infimum over $\rho_{AB}$ can be solved explicitly~\cite{ZTF14}. This gives us the first claimed lower bound~\eqref{eq:lagbnd} (after restoring the $\vec{\lambda}\cdot\vec{\constr}$ term in the Lagrange dual and moving the $-1/e$ prefactor out of the infimum, flipping it to a supremum in the process).
Furthermore, the series of calculations leading up to~\eqref{eq:ZTFbound} were all equalities, so this bound is tight as long as $\vec{\constr}$ is such that strong duality holds. 

Our next goal is to reduce the bound to an operator polynomial (eventually arriving at the expression~\eqref{eq:lagbndGT}). We begin with an informal observation: notice that in~\eqref{eq:ZTFbound}, it would be convenient if we could split the terms in the operator exponential, in the sense of claiming that
\begin{align}
\tr{\exp(M_1 + M_2 + \dots + M_N)} \overset{?}{=} \tr{\exp(M_1)\exp(M_2)\dots\exp(M_N)}.
\label{eq:matexp}
\end{align}
For instance, $\exp\left( \ln(\pinch_{\xo}(\sigma_{AB})) \right)$ on its own would just simplify to 
$\pinch_{\xo}(\sigma_{AB})$. Also, note that the $\vec{\lambda} \cdot\vec{\Gamma}(\pvm_{a|x},\pvm_{b|y})$ term
is a linear combination of some projectors $\pvm_{a|x} \otimes \pvm_{b|y}$,
and for each fixed input pair $xy$, the operators 
$\{\pvm_{a|x} \otimes \pvm_{b|y}\}_{ab}$ 
form a resolution of the identity (i.e.~they are orthogonal projectors satisfying $\sum_{ab} \pvm_{a|x} \otimes \pvm_{b|y} = \id_{AB}$). This implies that a useful equality holds: for each pair $xy$, and any coefficients $\mu_{ab} \in \mathbb{C}$, we have
\begin{align}
\exp \left(\sum_{ab} \mu_{ab}  \pvm_{a|x} \otimes \pvm_{b|y} \right)
= \sum_{ab} e^{\mu_{ab}} \pvm_{a|x} \otimes \pvm_{b|y} .
\end{align}
Hence if we were able to split the exponential by grouping the terms in $\vec{\lambda} \cdot\vec{\Gamma}(\pvm_{a|x},\pvm_{b|y})$ according to their values of $xy$, the resulting exponentials could be simplified into simple linear combinations of the projectors $\pvm_{a|x} \otimes \pvm_{b|y}$, and their product would be an operator polynomial as we would like to have.

Of course, the equality~\eqref{eq:matexp} does not hold except in some special cases, such as when all the operators commute. Fortunately, some \emph{inequalities} of roughly that form hold, when the RHS is written appropriately. The most well-known of these is likely the Golden--Thompson inequality,
\begin{align}
\tr{\exp(M_1 + M_2} \leq \tr{\exp(M_1)\exp(M_2)}.
\end{align}
However, a generalization of this inequality to arbitrarily many operators was developed fairly recently~\cite{SBT16}:
\begin{fact}\label{fact:MTI}
(Special case of Corollary~3.3 in~\cite{SBT16}) Let $M_1, M_2, \dots, M_N$ be hermitian operators on a finite-dimensional Hilbert space, and define the function $\beta(t) \defvar ({\pi}/{2})(\cosh(\pi t)+1)^{-1}$. Then
\begin{align}
\tr {\exp\left(\sum_j
M_j\right)}
\leq
\int_\mathbb{R}\der t\;\beta(t) \tr{\left| \prod_j \exp \left(\frac{1+it}{2}M_j\right) \right|^2} .
\end{align}
The inequality holds for any ordering of the operator product. 
\end{fact}
Since the operator product ordering is arbitrary, this implies that for any ${j^\star} \in \upto{N}$, we can choose to place $M_{j^\star}$ at the last position in the operator product and slightly simplify the RHS of the above bound\footnote{In~\cite{TSG+21}, some of these expressions were written with the integral inside the trace for ease of presentation. In principle, moving an integral into a trace could cause convergence issues, so we avoid doing so in the description here. Though it can be argued that this is anyway not a problem for this particular theorem: since the operators are finite-dimensional, they can be treated as matrices, and we note that the matrix elements of 
$\left| \prod_j \exp \left(\frac{1+it}{2}M_j\right) \right|^2$ are bounded 
(over $t\in\mathbb{R}$). 
Together with the fact that $\int_\mathbb{R}\der t\;\beta(t)$ is absolutely convergent, this ensures the integral for each matrix element converges.} using trace cyclicity:
\begin{align}
\tr {\exp\left(\sum_j M_j\right)}
&\leq
\int_\mathbb{R}\der t\;\beta(t) \tr{\left| 
\left(
\prod_{j\neq {j^\star}} \exp \left(\frac{1+it}{2}M_j\right)
\right) 
\exp \left(\frac{1+it}{2}M_{j^\star}\right) \right|^2} \nonumber \\
&=
\int_\mathbb{R}\der t\;\beta(t) \tr{\exp(M_{j^\star}) \left| \prod_{j\neq {j^\star}} \exp \left(\frac{1+it}{2}M_j\right) \right|^2} .
\end{align}

In our context, this means we can choose $M_{j^\star} = \ln(\pinch_{\xo}(\sigma_{AB}))$ to get
\begin{align}
& \tr{\exp\left( \ln(\pinch_{\xo}(\sigma_{AB})) + \vec{\lambda} \cdot\vec{\Gamma}(\pvm_{a|x},\pvm_{b|y}) \right)} \nonumber\\
=& \tr{\exp\left( \ln(\pinch_{\xo}(\sigma_{AB})) + \sum_{xy} 
\sum_{ab} \lambnew_{abxy} 
\pvm_{a|x} \otimes \pvm_{b|y} \right)} \text{ where } \lambnew_{abxy} = \sum_j \lambda_j c^{(j)}_{abxy} \in \mathbb{R} \nonumber\\ 
\leq & \int_\mathbb{R}\der t\;\beta(t) \tr{\pinch_{\xo}(\sigma_{AB}) \left| \prod_{xy} \exp \left(\frac{1+it}{2} \sum_{ab} \lambnew_{abxy}  \pvm_{a|x} \otimes \pvm_{b|y} \right) \right|^2}. \nonumber\\
= &\int_\mathbb{R}\der t\;\beta(t) \tr{\sigma_{AB} \pinch_{\xo} \left(\left| \prod_{xy} \exp \left(\frac{1+it}{2} \sum_{ab} \lambnew_{abxy}  \pvm_{a|x} \otimes \pvm_{b|y} \right) \right|^2\right)} 
\text{ since $\pinch_{\xo}$ is self-adjoint} 
\nonumber\\ 
= & \int_\mathbb{R}\der t\;\beta(t) \tr{\sigma_{AB} \pinch_{\xo} \left(\left| \prod_{xy} \sum_{ab} 
\exp\left(\frac{1+it}{2}  \lambnew_{abxy}\right) 
\pvm_{a|x} \otimes \pvm_{b|y}  \right|^2\right)},
\label{eq:beforeexpand}
\end{align}
where in the last line we expanded the operator exponentials using~\eqref{eq:matexp}. 

In principle, this is now an operator polynomial. However, to bound it using the NPA hierarchy in practice, we would need to explicitly evaluate the integrals to obtain the coefficients in the polynomial. We now explain how to evaluate these coefficients such that we arrive at the formula~\eqref{eq:oppoly} --- this is a tedious but conceptually fairly straightforward process.
We begin by expanding
the product $\prod_{xy}$ over the sum $\sum_{ab}$ in~\eqref{eq:beforeexpand}, rewriting the expression as
\begin{align}
& \int_\mathbb{R}\der t\;\beta(t) \tr{\sigma_{AB} \pinch_{\xo} \left(\left| \sum_{
\vec{a}\vec{b}
} \prod_{xy} \exp\left(\frac{1+it}{2}  \lambnew_{a_x b_y xy}\right) \pvm_{a_x|x} \otimes \pvm_{b_y|y}  \right|^2\right)} \nonumber\\ 
= & \int_\mathbb{R}\der t\;\beta(t) \tr{\sigma_{AB} \pinch_{\xo} \left(\left| \sum_{\vec{a}\vec{b}} \exp\left(\sum_{xy} \frac{1+it}{2}  \lambnew_{a_x b_y xy}\right) \prod_{xy} \pvm_{a_x|x} \otimes \pvm_{b_y|y}  \right|^2\right)} \nonumber\\ 
= & \int_\mathbb{R}\der t\;\beta(t) \tr{\sigma_{AB} \pinch_{\xo} \left(\sum_{
\vec{a}\vec{b}\vec{a}'\vec{b}'
} \exp\left(\sum_{xy} \frac{1-it}{2}  \lambnew_{a'_x b'_y xy} + \frac{1+it}{2}  \lambnew_{a_x b_y xy}\right) 
\right.\right. \nonumber\\ 
&\qquad \left.\left.
\left(\prod_{xy} \pvm_{a'_x|x} \otimes \pvm_{b'_y|y}  \right)^\dagger
\left(\prod_{xy} \pvm_{a_x|x} \otimes \pvm_{b_y|y}  \right) \right)} ,
\label{eq:beforeint}
\end{align}
where in the last line both operator products must be taken in the same order over $xy$. Finally, by introducing the values $\alpha^{\pm}_{\vec{a}\vec{b}\vec{a}'\vec{b}'} \defvar \frac{1}{2} \sum_{xy} \left(\lambnew_{a_xb_yxy} \pm \lambnew_{a'_xb'_yxy}\right) = \frac{1}{2} \sum_{jxy} \lambda_j \left(c^{(j)}_{a_xb_yxy} \pm c^{(j)}_{a'_xb'_yxy}\right)$, 
and applying linearity of $\pinch_{\xo}$, $\operatorname{Tr}$ and $\int_\mathbb{R}\der t$
to pull out the summation and scalar coefficients, 
we see that~\eqref{eq:beforeint} is equal to 
\begin{align}
& \sum_{\vec{a}\vec{b}\vec{a}'\vec{b}'} 
\int_\mathbb{R}\der t
\exp\left(
\alpha^+_{\vec{a}\vec{b}\vec{a}'\vec{b}'} + i\alpha^-_{\vec{a}\vec{b}\vec{a}'\vec{b}'} t
\right) 
\beta(t) 
\tr{\sigma_{AB} \pinch_{\xo} \left(
\left(\prod_{xy} \pvm_{a'_x|x} \otimes \pvm_{b'_y|y}  \right)^\dagger
\left(\prod_{xy} \pvm_{a_x|x} \otimes \pvm_{b_y|y}  \right) \right)} \nonumber\\ 
= & \sum_{\vec{a}\vec{b}\vec{a}'\vec{b}'} 
e^{
\alpha^+_{\vec{a}\vec{b}\vec{a}'\vec{b}'} 
}
\alpha^-_{\vec{a}\vec{b}\vec{a}'\vec{b}'} \csch \alpha^-_{\vec{a}\vec{b}\vec{a}'\vec{b}'} 
\tr{\sigma_{AB} \pinch_{\xo} \left(
\left(\prod_{xy} \pvm_{a'_x|x} \otimes \pvm_{b'_y|y}  \right)^\dagger
\left(\prod_{xy} \pvm_{a_x|x} \otimes \pvm_{b_y|y}  \right) \right)} , \label{eq:afterint}
\end{align}
where the integrals were evaluated using $\int_\mathbb{R}\der t \, e^{i \alpha t} \beta(t) = \alpha \csch \alpha$ (note that this formula is consistent under the interpretation that $\alpha \csch \alpha = 1$ at $\alpha=0$). 
This yields the desired bound. Note that the only step of this proof that is not an equality is the generalized Golden--Thompson inequality, i.e.~this is the only step where the bound is not tight.

A slightly curious aspect of this approach is that the final infimum is over a state $\sigma$ that is not directly related to the state $\rho$ in the original optimization~\eqref{eq:mainoptnat}, but rather was introduced as an intermediate variable in the process of deriving~\eqref{eq:ZTFbound}. This results in some small drawbacks, which we shall now discuss and improve on.

\subsection{Using the Gibbs variational principle}
\label{sec:gibbs}

In the above analysis, the bound was tight up until~\eqref{eq:ZTFbound} (as long as strong duality holds). This was the reason we chose to present the above analysis first, since it shows that essentially the only loss of tightness comes from the generalized Golden--Thompson inequality (and when computing the bound in practice, the suboptimalities induced by using a finite NPA hierarchy level). 
Perhaps unexpectedly, however, it turns out that a better final result can be obtained using an approach where the initial steps are not tight by construction. 
Specifically, we prove the following theorem:
\begin{theorem}\label{th:gibbsbnd}
When noisy preprocessing is not applied, and $\keyw_x=\delta_{x,\xo}$,
the optimization~\eqref{eq:mainoptnat} is lower bounded by
\begin{align}
\sup_{\vec{\lambda}}
\; - \ln 
\left(
\sup_{
(\rho_{AB},\pvm_{a|x},\pvm_{b|y}) \in \mathcal{S}_{\vec{\constr}}
} \tr{\exp\left( 
\ln \pinch_{\xo}({\rho}_{AB}) +
\vec{\lambda} \cdot \vec{\Gamma}(\pvm_{a|x},\pvm_{b|y}) 
\right)} 
\right)
+ \vec{\lambda} \cdot \vec{\constr} .
\label{eq:gibbsbnd}
\end{align}
where $\mathcal{S}_{\vec{\constr}}$ denotes the set (or class) of states and measurements satisfying the constraints $\tr{\vec{\Gamma}(\pvm_{a|x},\pvm_{b|y}) \rho_{AB}} = \vec{\constr}$.
This bound is tight for all $\vec{\constr}$ such that strong duality holds in the Lagrange dual.
Furthermore, it is in turn lower bounded by 
\begin{align}
& \sup_{\vec{\lambda}}\; - \ln 
\left(
\sup_{
(\rho_{AB},\pvm_{a|x},\pvm_{b|y}) \in \mathcal{S}_{\vec{\constr}}
} \tr{\rho_{AB} \, Q_{\xo}(\pvm_{a|x},\pvm_{b|y})}
\right)
+ \vec{\lambda} \cdot \vec{\constr} .
\label{eq:gibbsbndGT}
\end{align}
where $Q_{\xo}(\pvm_{a|x},\pvm_{b|y})$ is as defined in~\eqref{eq:oppoly}.
\end{theorem}

As in the previous section, the resulting bound~\eqref{eq:gibbsbndGT} is in a form that can be tackled using the NPA hierarchy (in this case the optimization has 
the additional restriction $(\rho_{AB},\pvm_{a|x},\pvm_{b|y}) \in \mathcal{S}_{\vec{\constr}}$,
but this simply corresponds to imposing noncommutative polynomial constraints $\tr{{\Gamma}_j(\pvm_{a|x},\pvm_{b|y}) \rho_{AB}} = {\constr}_j$).

We now present the proof of this theorem. First, we state some results we shall use:
\begin{fact}\label{fact:Jensen}
(Operator Jensen inequality, as stated in~\cite{HP03})\footnote{There exist 
many other versions of this inequality (and the related Schwarz inequality), with various differences in the theorem conditions,
but we will not need them here.} 
Let $f$ be a continuous real-valued function on an interval $I$. Then $f$ is operator convex if and only if 
for every $N\in\mathbb{N}$
and any Hilbert space $\mathcal{H}$,
we have
\begin{align}
\sum_{j=1}^N Q_j^\dagger f \left(M_j\right) Q_j \geq 
f \left(\sum_{j=1}^N Q_j^\dagger M_j Q_j\right),
\end{align}
for any bounded self-adjoint operators $M_j$ on $\mathcal{H}$ with spectra contained in $I$, and operators $Q_j$ on $\mathcal{H}$ satisfying $\sum_{j=1}^N Q_j^\dagger Q_j = \id$.
\end{fact}
This yields an immediate corollary, using the fact that any completely positive unital map $\phi$ can be expressed in the form $\phi(M) = \sum_{j=1}^N Q_j M Q_j $ for some $Q_j$ satisfying $\sum_{j=1}^N Q_j^\dagger Q_j = \id$:
\begin{corollary}\label{corr:unital}
Let $f$ be a continuous real-valued function on an interval $I$, let $M$ be a bounded self-adjoint operator with spectrum contained in $I$, and let $\phi$ be a completely positive unital map. Then if $f$ is operator convex, we have $\phi(f(M)) \geq f(\phi(M))$.
\end{corollary}

\begin{fact}\label{fact:Gibbs}
(Gibbs variational principle) Let $M$ be a hermitian operator on a finite-dimensional Hilbert space. Then for any quantum state $\rho$,
\begin{align}
\tr{\rho M} - H_{\nat}(\rho) \geq - \ln \tr{e^{-M}}.
\end{align}
\end{fact}

Moving on to the proof, instead of starting with~\eqref{eq:Hdualsimple}, we shall use an intermediate step in the derivation of that relation in~\cite{Col12}, namely
\begin{align}
H(\hat{A}_x|E) = H(\pinch_x({\rho}_{AB})) - H({\rho}_{AB}).
\end{align}
With this equation, we can lower bound each $H_{\nat}(\hat{A}_x|E)$ term in the objective function of~\eqref{eq:mainoptnat} as follows:
\begin{align}
H_{\nat}(\hat{A}_x|E) 
&
= -\tr{\pinch_x({\rho}_{AB}) \ln \pinch_x({\rho}_{AB}) } - H_{\nat}({\rho}_{AB}) \nonumber\\
&= \tr{{\rho}_{AB} \pinch_x^\dagger (-\ln \pinch_x({\rho}_{AB}))} - H_{\nat}({\rho}_{AB}) \nonumber\\
&\geq \tr{{\rho}_{AB} (-\ln \pinch_x^\dagger \pinch_x({\rho}_{AB}))} - H_{\nat}({\rho}_{AB}) 
, \label{eq:afterjensen}
\end{align}
where in the last line we applied the fact that $\rho \geq 0$ together with
Corollary~\ref{corr:unital}
(the conditions are indeed satisfied, because $-\ln$ is operator convex and $\pinch_x^\dagger$ is a completely positive unital map since $\pinch_x$ is a completely positive trace-preserving map).

We now apply a trick that seems slightly arbitrary, though it can be somewhat motivated by the Lagrange dual considered in the previous section: note that for any tuple $\vec{\lambda}$, if we impose that we only consider states and measurements satisfying the constraints in the optimization~\eqref{eq:mainoptnat}, we would trivially have $\vec{\lambda} \cdot \left(\tr{\vec{\Gamma}(\pvm_{a|x},\pvm_{b|y}) \rho_{AB}} - \vec{\constr} \right) = 0$. Therefore, when the constraints are imposed, we can rewrite the last line of~\eqref{eq:afterjensen} as 
\begin{align}
&\tr{{\rho}_{AB} (-\ln \pinch_x^\dagger \pinch_x({\rho}_{AB}))} - H_{\nat}({\rho}_{AB}) - \vec{\lambda} \cdot \left(\tr{\vec{\Gamma}(\pvm_{a|x},\pvm_{b|y}) \rho_{AB}} - \vec{\constr} \right) \nonumber\\
=& \tr{{\rho}_{AB} \left( -\ln \pinch_x^\dagger \pinch_x({\rho}_{AB}) - \vec{\lambda} \cdot \vec{\Gamma}(\pvm_{a|x},\pvm_{b|y}) \right)} - H_{\nat}({\rho}_{AB}) 
+ \vec{\lambda} \cdot \vec{\constr}
\nonumber\\
\geq& - \ln \tr{\exp\left( 
\ln \pinch_x^\dagger \pinch_x({\rho}_{AB}) +
\vec{\lambda} \cdot \vec{\Gamma}(\pvm_{a|x},\pvm_{b|y}) 
\right)} 
+ \vec{\lambda} \cdot \vec{\constr} 
\;\text{ by Fact~\ref{fact:Gibbs}} \nonumber\\
=& - \ln \tr{\exp\left( 
\ln \pinch_x({\rho}_{AB}) +
\vec{\lambda} \cdot \vec{\Gamma}(\pvm_{a|x},\pvm_{b|y}) 
\right)} 
+ \vec{\lambda} \cdot \vec{\constr} ,\label{eq:onetermbnd}
\end{align}
where the simplification in the last line follows from the fact that $\pinch_x$ is self-adjoint and idempotent.

Hence each $H_{\nat}(\hat{A}_x|E)$ term in the objective function of~\eqref{eq:mainoptnat} is lower bounded by~\eqref{eq:onetermbnd}. Also, recall that $\vec{\lambda}$ was arbitrary, which means we can take a supremum over it to obtain tighter lower bounds. Focusing on the case where 
$\keyw_x=\delta_{x,\xo}$, we conclude that the optimization~\eqref{eq:mainoptnat} is lower bounded by:\footnote{Technically, our above reasoning allows the supremum over $\vec{\lambda}_{}$ to be taken inside the infimum here, which could give a better bound in theory (because of the max-min inequality). However, computing this supremum inside the infimum is not possible for the approach we take in this work, so we omit this version.}
\begin{align}
\sup_{\vec{\lambda}}
\inf_{
(\rho_{AB},\pvm_{a|x},\pvm_{b|y}) \in \mathcal{S}_{\vec{\constr}}
} 
- \ln \tr{\exp\left( 
\ln \pinch_{\xo}({\rho}_{AB}) +
\vec{\lambda} \cdot \vec{\Gamma}(\pvm_{a|x},\pvm_{b|y}) 
\right)} 
+ \vec{\lambda} \cdot \vec{\constr} ,
\end{align}
which is equal to the claimed bound~\eqref{eq:gibbsbnd} (after moving the infimum into the $-\ln$ function).

To prove the tightness of this bound, we compare it to the analogous bound~\eqref{eq:lagbnd} in the previous section, which was a tight bound by construction (whenever strong duality holds for $\vec{\constr}$). 
We observe that~\eqref{eq:gibbsbnd} is always \emph{lower} bounded by the tight expression~\eqref{eq:lagbnd} --- this follows from the fact that $\ln x \leq x/e$ 
and
the supremum over states and measurements in~\eqref{eq:gibbsbnd} is over a smaller set. This means that surprisingly, even though many of the inequalities used in deriving~\eqref{eq:gibbsbnd} are not tight \textit{a priori}, it turns out to give a tight bound after all (as long as strong duality holds). Finally, the second bound~\eqref{eq:gibbsbndGT} follows via exactly the same calculations as in the previous section, based on the generalized Golden--Thompson inequality. This completes the proof.

Note that while we have just argued that the bounds~\eqref{eq:lagbnd} and~\eqref{eq:gibbsbnd} in the two theorems are equal (whenever strong duality holds), this is not necessarily the case for the bounds~\eqref{eq:lagbndGT} and~\eqref{eq:gibbsbndGT} that we would compute using the NPA hierarchy. Specifically, by the same reasoning as in our above comparison of~\eqref{eq:lagbnd} and~\eqref{eq:gibbsbnd}, we see that the bound~\eqref{eq:gibbsbndGT} derived here is always at least as good as the analogous bound~\eqref{eq:lagbndGT} in the previous section. However, in this case the latter is not guaranteed to be tight (due to the use of the generalized Golden--Thompson inequality), and hence~\eqref{eq:gibbsbndGT} could in fact be a strict improvement over~\eqref{eq:lagbndGT}. Our explicit computations of these bounds using the NPA hierarchy show that this indeed often turns out to be the case (though the difference may mainly arise from the different sets the optimizations are taken over, rather than the inequality $\ln x \leq x/e$ --- the effect of the latter can typically be ``cancelled out'' in principle by optimizing over the Lagrange multipliers; see~\cite{TSG+21} for further details).

\subsection{Resulting bounds for DIQKD and DIRE}
\label{sec:resultsNPA}

We now turn to the question of the results provided by this method (focusing on~\eqref{eq:gibbsbndGT}, since it is a tighter bound). Unfortunately, it faces the challenge that the operator polynomial~\eqref{eq:oppoly} is generally of very high order. This has the consequence that a high NPA hierarchy level is needed to bound those terms, with the required NPA level increasing rapidly as the number of inputs/outputs increases. Because of this, we were only able to apply this method in its full generality for the basic 2-input 2-output setting, in which case level 6 of the hierarchy is sufficient (for some choices of the operator product order), with the NPA matrix having size $85\times85$ after applying some simplifying techniques (see~\cite{TSG+21}). It is also possible to take some of the $\lambda_j$ parameters to be zero in order to reduce the order of the operator polynomial, allowing us to compute bounds for some scenarios with more inputs and outputs~\cite{TSG+21}, but the results do not appear to be very tight.

For the 2-input 2-output scenario, we implemented this approach using the full output distribution $\pr{ab|xy}$ as constraints (see~\cite{TSG+21} for further details). It turns out that the resulting bounds on $H(\hat{A}_{\xo}|E)$ that we obtain are close to or slightly outperform the result of~\cite{PAB+09}, which is a \emph{tight} bound on that entropy given only the CHSH value. The fact that it slightly outperforms the~\cite{PAB+09} bound in some cases implies that it is not always optimal to bound $H(\hat{A}_{\xo}|E)$ using only the CHSH value, which is not unexpected, but shows that the approach here has potential to provide some improvement in DIQKD keyrates. Roughly speaking, our approach tends to perform well in situations with a moderate to high amount of noise. Speculatively, we might explain this by noting that the main step which introduces a gap is the generalized Golden--Thompson inequality, which is saturated when the operators commute --- it could be that as the amount of noise increases, the operators that attain the optimum become ``closer to commuting'' (this is, of course, a very informal notion for now), making that inequality tighter.

However, even in cases where the bound on $H(\hat{A}_{\xo}|E)$ is better than that of~\cite{PAB+09}, the improvement is fairly small. On the other hand, we found more significant improvements for DIRE instead (still in the 2-input 2-output setting).
While we have only described how to bound $H(\hat{A}_x|E)$ in the above presentation, the computations generalize easily to $H(\hat{A}_x \hat{B}_y|E)$, again by using the relation~\eqref{eq:entdualDIRE} mentioned in the previous approach for the 2-input 2-output scenario.
In contrast to our results for $H(\hat{A}_{\xo}|E)$, we found~\cite{TSG+21} that our bounds on $H(\hat{A}_0 \hat{B}_0|E)$ were a substantial improvement on all earlier results. To some extent, this may be because no tight bounds on $H(\hat{A}_0 \hat{B}_0|E)$ were known, allowing us more scope for improvement. Subsequently, the independent approach developed in~\cite{BFF21} showed further improvements, by optimizing the intended device behaviour to maximize $H(\hat{A}_0 \hat{B}_0|E)$; this was improved on again by their subsequent approach~\cite{arx_BFF21} to compute arbitrarily tight bounds. Heuristic numerical results and conjectured tight bounds for this setting were also proposed in a recent work~\cite{arx_BRC21}, which incorporates the random-key-measurements technique as well. 

\subsection{Accounting for random key measurements}
\label{sec:randbasisgen}

Finally, we discuss the applicability of these approaches when random key measurements are applied.
It might seem at first that the above approaches could directly handle this scenario as well, instead of having to focus on the case where $\keyw_x = \delta_{x,0}$. However, each of them runs into some difficulty. Namely, for the approach based on the Lagrange dual, the computations leading up to~\eqref{eq:ZTFbound} do not hold straightforwardly if there is a sum over multiple terms $D_{\nat}({\rho}_{AB} \Vert \pinch_{\xg}({\rho}_{AB}))$ rather than just a single term --- the problem is that there will be multiple states $\sigma$ to consider when replacing the second argument in each relative entropy term, in which case it is unclear whether the optimization over $\rho$ can be solved. (It is still possible to lower-bound the optimization by splitting the infimum over $(\rho_{AB},\pvm_{a|x},\pvm_{b|y})$ across the terms, allowing each infimum over $\rho$ to be solved individually; however, this is in general not a tight lower bound --- in fact, inspecting the bound produced that way indicates that it completely loses the advantage offered by using random key measurements.)

As for the approach based on the Gibbs variational principle, the arguments up to~\eqref{eq:onetermbnd} were valid for the $H_{\nat}(\hat{A}_x|E)$ terms individually. Hence by following the same line of reasoning in the subsequent arguments (and using a different tuple of coefficients $\vec{\lambda}_x$ for each term), we arrive at this result:\footnote{For clarity, here we switch to using a different index $\xg$ for the summation that corresponds to the summation in the objective function of~\eqref{eq:mainoptnat}, to avoid confusion with the ``purely notational'' indices $x$ in the operators $\pvm_{a|x}$.}
\begin{theorem}\label{th:gibbsrandbasis}
When noisy preprocessing is not applied, 
the optimization~\eqref{eq:mainoptnat} is lower bounded by
\begin{align}
\sup_{\vec{\lambda}_{\xg}}
\inf_{
(\rho_{AB},\pvm_{a|x},\pvm_{b|y}) \in \mathcal{S}_{\vec{\constr}}
} 
\sum_{\xg} \keyw_{\xg}
\left(
- \ln \tr{\exp\left( 
\ln \pinch_{\xg}({\rho}_{AB}) +
\vec{\lambda}_{\xg} \cdot \vec{\Gamma}(\pvm_{a|x},\pvm_{b|y}) 
\right)} 
+ \vec{\lambda}_{\xg} \cdot \vec{\constr} 
\right),
\label{eq:gibbsrandbasis}
\end{align}
where $\mathcal{S}_{\vec{\constr}}$ denotes the set (or class) of states and measurements satisfying the constraints $\tr{\vec{\Gamma}(\pvm_{a|x},\pvm_{b|y}) \rho_{AB}} = \vec{\constr}$.
Furthermore, it is in turn lower bounded by 
\begin{align}
\sup_{\vec{\lambda}_{\xg}}
\inf_{
(\rho_{AB},\pvm_{a|x},\pvm_{b|y}) \in \mathcal{S}_{\vec{\constr}}
} 
\sum_{\xg} \keyw_{\xg}
\left(
- \ln \tr{\rho_{AB} Q_{\xg}(\pvm_{a|x},\pvm_{b|y})} 
+ \vec{\lambda}_{\xg} \cdot \vec{\constr} 
\right),
\label{eq:gibbsrandbasisGT}
\end{align}
where $Q_{\xg}(\pvm_{a|x},\pvm_{b|y})$ is as defined in~\eqref{eq:oppoly}.
\end{theorem}

However, there is still an obstruction to overcome for this case --- the bound~\eqref{eq:gibbsrandbasisGT} is not a noncommutative polynomial optimization, because of the logarithms. (Notice that previously, we could get to the expression~\eqref{eq:gibbsbndGT} by using the monotonicity of $-\ln$ to move the optimization into the logarithm, but this trick does not work here unless we split the infimum across the sum, in which case we incur the same loss of tightness mentioned above.) 

Fortunately, it is possible to work around the issue in this case by using the relation 
\begin{align}
-\ln x = \sup_{\mu>0} \left(-\frac{x}{\mu} - \ln\mu + 1\right),
\end{align}
which can be interpreted geometrically as characterizing the convex function $-\ln$ via the envelope of its tangents (note that $x/\mu + \ln\mu - 1$ is simply the tangent to $\ln x$ at the point $x=\mu$). Therefore, this implies that~\eqref{eq:gibbsrandbasisGT} is equal to
\begin{align}
&\sup_{\vec{\lambda}_{\xg}}
\inf_{
(\rho_{AB},\pvm_{a|x},\pvm_{b|y}) \in \mathcal{S}_{\vec{\constr}}
} 
\sum_{\xg} \keyw_{\xg} \sup_{\mu_{\xg}>0}
\left(
- \frac{1}{\mu_{\xg}}\tr{\rho_{AB} Q_{\xg}(\pvm_{a|x},\pvm_{b|y})} 
- \ln \mu_{\xg} + 1
+ \vec{\lambda}_{\xg} \cdot \vec{\constr} 
\right), \nonumber\\
\geq&
\sup_{\mu_{\xg}>0}
\sup_{\vec{\lambda}_{\xg}}
\inf_{
(\rho_{AB},\pvm_{a|x},\pvm_{b|y}) \in \mathcal{S}_{\vec{\constr}}
} 
\sum_{\xg} \keyw_{\xg}
\left(
- \frac{1}{\mu_{\xg}}\tr{\rho_{AB} Q_{\xg}(\pvm_{a|x},\pvm_{b|y})} 
- \ln \mu_{\xg} + 1
+ \vec{\lambda}_{\xg} \cdot \vec{\constr} 
\right),
\label{eq:gibbsrandbasisfinal}
\end{align}
which \emph{can} indeed be solved as a noncommutative polynomial optimization (again, the suprema over $\mu_{\xg},\vec{\lambda}_{\xg}$ are to be handled heuristically, with more optimal values yielding tighter bounds). The downside is that because we have performed a max-min swap, there may again be some loss of tightness. We find that in practice, computing the bound~\eqref{eq:gibbsrandbasisfinal} for the depolarizing-noise scenario previously considered in Sec.~\ref{sec:qubit} does yield a slightly better bound as compared to the~\cite{PAB+09} expression (for $H(\hat{A}_{\xo}|E)$ alone); however, the improvement is rather small, and is substantially worse than the essentially tight bounds computed in Sec.~\ref{sec:qubit}.

\chapter{Finite-size analysis}
\label{chap:finite}

Having discussed a variety of methods for bounding the asymptotic keyrates, we now turn to the question of studying the finite-size effects. We begin by presenting in Sec.~\ref{sec:mainprot} a detailed description of the main protocol we consider. In Sec.~\ref{sec:finitekeylength}, we state our main theorem regarding this protocol, namely an explicit finite-size expression for the key length. We then give the proof of this theorem in Sec.~\ref{sec:finiteproof}, and discuss the resulting keyrates in Sec.~\ref{sec:finiteplots}. Finally, in Sec.~\ref{sec:mods} we describe various modifications of the main protocol that can improve its keyrates (and may also be useful in improving DIRNG/DIRE keyrates). The results in this chapter are based on~\cite{arx_TSB+20}, and the phrasing and presentation are essentially identical to that work.

\section{Detailed protocol description}
\label{sec:mainprot}

Our results in the previous chapter (specifically Sec.~\ref{sec:1rndresults}) suggest that for 2-input 2-output scenarios under the depolarizing-noise model, considering the CHSH inequality is sufficient to obtain near-optimal results, i.e.~there is not much scope for further improvement by considering modified CHSH inequalities. Hence the protocol we present only considers the CHSH inequality, though in principle the finite-size analysis we perform can be applied to protocols based on arbitrary Bell inequalities, e.g.~following the DIRE analysis in~\cite{BRC20}.

Rather than using the CHSH value in the form~\eqref{eq:CHSHclass}, we shall instead base the protocol on the \term{CHSH game}, which is somewhat easier to describe in a protocol. This is a nonlocal game where 
all inputs and outputs take values in $\{0,1\}$, and Alice and Bob's inputs are chosen independently and uniformly, with the win condition being $a\oplus b = xy$. Expressed in the form mentioned in Sec.~\ref{sec:introdevices} for ``Bell parameters'', the winning probability for this game (given a distribution $\pr{ab|xy}$) is
\begin{align}
w = \sum_{abxy} \frac{1}{4} \delta_{a\oplus b, xy} \pr{ab|xy},
\end{align}
where the factor of $1/4$ comes from the uniform input distribution.
This winning probability $w$ is essentially equivalent to the CHSH value $\constr$, in that they can be converted to each other by the simple relation $\constr = 8w-4$, under the NS conditions~\eqref{eq:NS}. With this game in mind, we can now describe the protocol (though the inputs for Bob will be relabelled somewhat in order to account for generation-round inputs):
\newpage
\begin{savenotes} 
\begin{breakablealgorithm}
\caption{} 
\label{prot:DIQKD}
The protocol is defined in terms of the following parameters (chosen before the protocol begins), which we qualitatively describe:
{\setlist[description]{leftmargin=2cm,labelindent=2cm,itemsep=0mm}
\begin{description}
\item $n$: Total number of rounds
\item $\gamma$: Probability of a test round
\item $\p$: Noisy-preprocessing bias
\item $\cmax$: Bound on ``leakage'' from error correction
\item $\wexp$: Expected winning probability for the (IID) honest devices
\item $\dtol$: Tolerated deviation from expected winning probability $\wexp$
\item $\lkey$: Length of final key
\end{description}}
\noindent The honest behaviour consists of $n$ IID copies of a device characterized by $\wexp$ and an error-correction parameter $h_\mathrm{hon}$ (details in Sec.~\ref{sec:hon}).\\
\begin{algorithmic}[1]
\State\label{step:1rnd}\textbf{Measurement:} For each $j \in \upto{n}$, perform the following steps: 
\begin{algsubstates}
\State Alice and Bob's devices each receive some share of a quantum state.
\State Alice chooses a uniform input $X_j \in \inX$. With probability $\gamma$, Bob chooses a uniform input $Y_j \in \inYt$, otherwise Bob chooses a uniform input $Y_j \in \inYg$.
\State Alice and Bob supply their inputs to their devices, and record the outputs as $A_j$ and $B_j$ respectively. 
\end{algsubstates}
\State Alice and Bob publicly announce their input strings $\str{X}$ and $\str{Y}$. 
\State \textbf{Sifting:} For all rounds such that $Y_j \in \inYg$ and $X_j \neq Y_j$, Alice and Bob overwrite their outputs with $A_j = B_j = 0$.
\State\label{step:NPP}\textbf{Noisy preprocessing:} For all rounds such that $Y_j \in \inYg$ and $X_j = Y_j$, Alice generates a biased random bit $F_j$ with $\pr{F_j=1} = \p$, and overwrites her output $A_j$ with $A_j \oplus F_j$.
\State\label{step:EC}\textbf{Error correction:} Alice and Bob publicly communicate some bits $\str{L}=(\str{L}_\mathrm{EC},\str{L}_\mathrm{h})$ 
for error correction as follows (see Sec.~\ref{sec:EC}): 
\begin{algsubstates}
\State\label{step:guess}Alice and Bob communicate some bits $\str{L}_\mathrm{EC}$ to allow Bob to produce a guess $\tilde{\str{A}}$ for $\str{A}$, such that the set of all possible values for $\str{L}_\mathrm{EC}$ has cardinality upper bounded by $2^{\cmax}$.\footnote{As discussed in Sec.~\ref{sec:proofsketch}, this accommodates the possibility of $\str{L}_\mathrm{EC}$ being of variable length, by choosing $\cmax$ to be $1$ larger than the maximum possible length of $\str{L}_\mathrm{EC}$ (in bits). To be even more precise, if two-way communication is allowed here, one would also need to account for the possible directions and orderings of communications when defining the set of possible values for $\str{L}_\mathrm{EC}$; we will not consider this question further in this work. (Note that~\cite{arx_TSB+20} contains a technical inaccuracy regarding this point, as it simply viewed $\cmax$ as an upper bound on the number of communicated bits, which is not exactly valid in light of this discussion.)}
\State\label{step:hash}Alice computes a 2-universal hash $\str{L}_\mathrm{h} = \hash(\str{A})$ of length $\ceil{\log(1/\eh)}$. She sends $\str{L}_\mathrm{h}$ (and the choice of hash function) to Bob.
\end{algsubstates} 
\State\label{step:PA}\textbf{Parameter estimation:} 
For all $j\in\upto{n}$, Bob sets 
$C_j=\perp$ if $Y_j \in \inYg$; otherwise he sets $C_j=0$ if $\tilde{A}_j \oplus B_j \neq X_j \cdot (Y_j-2)$ and $C_j=1$ if $\tilde{A}_j \oplus B_j = X_j \cdot (Y_j-2)$.
\State\label{step:abort}Bob checks whether $\str{L}_\mathrm{h}=\hash(\tilde{\str{A}})$, as well as whether the value $\str{c}$ on registers $\str{C}$ satisfies $\freq_\str{c}(1)\geq (\wexp-\dtol)\gamma$ and $\freq_\str{c}(0)\leq (1-\wexp+\dtol)\gamma$. If all those conditions hold, Alice and Bob proceed to the next step. Otherwise, the protocol aborts.
\State \textbf{Privacy amplification:} Alice and Bob apply privacy amplification (see Sec.~\ref{sec:PA}) on $\str{A}$ and $\tilde{\str{A}}$ respectively to obtain final keys $K_A$ and $K_B$ of length $\lkey$.
\end{algorithmic}
\end{breakablealgorithm}
\end{savenotes}
\newpage

The rounds in which $Y_j \in \inYt$ will be referred to as test rounds, and the rounds in which $Y_j \in \inYg$ will be referred to as generation rounds (though strictly speaking, the final key in this protocol is obtained from all the rounds, not merely the generation rounds alone).
In each round,
Eve is allowed to hold some extension of the state distributed to the devices. We will use $\allE$ rather than $\str{E}$ to denote the collection of all such quantum side-information she retains over the entire protocol, 
since it may not necessarily have a tensor-product structure. 

We briefly highlight a few aspects of this protocol. Firstly (as mentioned in Sec.~\ref{sec:protsketch}), we do not choose a random subset of fixed size as test rounds, but rather, each round is independently chosen to be a test or generation round, following~\cite{ARV19}. This is in order to apply the entropy accumulation theorem, which holds for processes that can be described using a sequence of maps. Furthermore, the parameter-estimation check is performed on \emph{both} $\freq_\str{c}(1)$ and $\freq_\str{c}(0)$. This was required in order to derive a critical inequality in the security proof (following~\cite{BRC20}), though in some cases it is possible to omit the $\freq_\str{c}(1)$ check (see Eq.~\eqref{eq:fminPEineq} and the subsequent discussion).

We now describe some of the individual steps in more detail.

\subsection{Error correction}\label{sec:EC}

We first discuss Step~\ref{step:hash}, because it will have an impact on our discussion of Step~\ref{step:guess}. Given any $\eh \in (0,1]$, 
if we consider a 2-universal family of hash functions where the output is a bitstring of length $\ceil{\log(1/\eh)}$, then the defining property of 2-universal hashing guarantees that
\begin{align}
\pr{\hash(\str{A})=\hash(\tilde{\str{A}})\middle|\str{A}\neq\tilde{\str{A}}} \leq \eh. \label{eq:hash}
\end{align}
In other words, the probability of getting matching hashes from different strings can be made arbitrarily small, by using sufficiently long hashes. Informally speaking, this gives us some laxity in Step~\ref{step:guess}, because regardless of how much the devices deviate from the honest behaviour, the guarantee~\eqref{eq:hash} will still hold, providing a final ``check'' on how bad the guess $\tilde{\str{A}}$ could be. 
Importantly, our later proof of the \emph{soundness} of the protocol will not rely on any guarantees regarding the procedure in Step~\ref{step:guess} --- only the \emph{completeness} of the protocol (i.e.~the probability that the honest devices mistakenly abort) requires such guarantees. 

We now study Step~\ref{step:guess}. We shall choose $\cmax$ as follows: it is the length of $\str{L}_\mathrm{EC}$ required (let us focus on fixed-length $\str{L}_\mathrm{EC}$ for simplicity) such that given the honest devices, Bob can use $\str{L}_\mathrm{EC}$ and $\str{B}$ to produce a guess $\tilde{\str{A}}$ 
satisfying
\begin{align}
\pr{\str{A} \neq \tilde{\str{A}}}_\mathrm{hon} \leq \ecEC.
\label{eq:comEC}
\end{align}
(In this section, we will use the subscript $_\mathrm{hon}$ to emphasize quantities computed with respect to an honest behaviour.)
We stress that while some preliminary characterization of the devices can be performed beforehand to choose a suitable $\cmax$, this parameter \emph{must not} be changed once the protocol begins. 

As noted back in Sec.~\ref{sec:minmaxent}, the question of what value should be chosen for $\cmax$ in order to achieve a desired $\ecEC$ can be addressed by the protocol in~\cite{RR12}. Explicitly,~\eqref{eq:comEC} can be satisfied as long as we choose $\cmax$ 
such that
\begin{align}
\cmax \geq 
\Hmax^{\ez}(\str{A}|\str{B}\str{X}\str{Y})_\mathrm{hon} + 2\log\frac{1}{\ecEC-\ez} + 4,
\label{eq:optEC}
\end{align}
where $\ez\in[0,\ecEC)$ is to be optimized over.
Since the honest behaviour is IID, 
the max-entropy can be bounded by the AEP~\eqref{eq:AEPmax}:\footnote{In this analysis, we deviated slightly from~\cite{ARV19} by using the error-correction protocol from~\cite{RR12} instead of~\cite{RW05}, and the AEP stated in~\cite{DFR20} rather than~\cite{TCR09}. Both of these yield slight improvements in the bounds (the former at large $n$, and the latter when
$\dim(\hat{A})$ is not too large).
}
\begin{align}
\Hmax^{\ez}(\str{A}|\str{B}\str{X}\str{Y})_\mathrm{hon} \leq n h_\mathrm{hon} + 
\sqrt{n} \, (2\log5)\sqrt{\log\frac{2}{\ez^2}},
\end{align}
where (applying the decomposition $H(Q|Q'C)=\sum_c \pr{c} H(Q|Q';C=c)$ for classical $C$, and using the ``single-round notation'' $\hat{A}\hat{B}XY$ since the behaviour is IID)
\begin{align}
h_\mathrm{hon} \defvar H(\hat{A}|\hat{B} XY)_\mathrm{hon} 
&= \frac{1-\gamma}{4} 
\sum_{z\in\inX} H(\hat{A} | \hat{B};X = Y = z)_\mathrm{hon}  
\nonumber \\
& \qquad 
+ \frac{\gamma}{4} 
\sum_{x\in\inX,y\in\inYt} 
H(\hat{A} | \hat{B};X = x, Y = y)_\mathrm{hon},
\label{eq:ECrate}
\end{align}
where the terms should be understood to refer to the $\hat{A}\hat{B}$ values \emph{after} noisy preprocessing (for the generation rounds). 

However, the protocol achieving the bound~\eqref{eq:optEC} may not be easy to implement. In practice, error-correction protocols typically achieve performance described by
\begin{align}
\cmax \approx \xi(n,\ecEC) n h_\mathrm{hon}, 
\end{align}
where $\xi(n,\ecEC)$ lies between $1.05$ and $1.2$ for ``typical'' values of $n$ and $\ecEC$. (More precise characterizations can be found in~\cite{TMP+17}, which gives for instance an estimate
\begin{align}
\cmax \approx \xi_1 n h_\mathrm{hon} + \tilde{\xi}(\ecEC,h_\mathrm{hon})\sqrt{n},
\end{align}
for a constant $\xi_1$ and a specific function $\tilde{\xi}$.) Furthermore, some protocols used in practice do not have a theoretical bound on $\ecEC$ (for a given $\cmax$), only heuristic estimates. 

Fortunately, as mentioned earlier, the choice of error-correction procedure in Step~\ref{step:guess} will have no effect in our proof of the soundness of Protocol~\ref{prot:DIQKD} (as long as $\cmax$ is a fixed parameter), only its completeness. This means that as long as we are willing to accept heuristic values for $\ecom$, we can use the heuristic values of $\ecEC$ provided by using some ``practical'' error-correction procedure in Step~\ref{step:guess}, and the value of $\esound$ (i.e.~how ``secure'' the protocol is) will be completely unaffected. The critical point to remember is that $\cmax$ is a value to be fixed before the protocol begins, and Alice and Bob \emph{must} only use exactly that many bits (or up to the corresponding maximum allowed number of bits, in the case of variable-length error-correction) when actually executing the protocol. With this in mind, we remark that while we mainly focus on protocols using one-way error correction, this is not quite a strict requirement --- in theory, one could use a procedure involving two-way communication (such as Cascade), as long as $\cmax$ accounts for \emph{all} the communicated bits, not just those sent from Alice to Bob.\footnote{We stress that care must be taken when interpreting this claim in the context of existing results regarding such protocols: for instance, some works on the Cascade protocol only analyze the number of bits sent from Alice to Bob, which would not be valid as an upper bound on $\cmax$ in our context.} Another possibility worth considering might be adaptive procedures that adjust to the noise level encountered during execution of the protocol, rather than the expected noise level (again, making sure to halt if they reach the maximum allowed number of bits corresponding to the pre-chosen $\cmax$ value). 

We remark that in our situation, if we assume that $\prnoscale{\hat{A} = \hat{B} |X = x, Y = y}_\mathrm{hon}$ is the same for all $x\in\inX,y\in\inYt$ (in which case it must equal $\wexp$), then for the $y\in\inYt$ terms in Eq.~\eqref{eq:ECrate} we have 
\begin{align}
H(\hat{A} | \hat{B};X = x, Y = y)_\mathrm{hon} = \binh\left(\prnoscale{\hat{A} \neq \hat{B} |X = x, Y = y}_\mathrm{hon}\right) = \binh(\wexp),
\end{align}
which lies in approximately $[0.600,0.811]$ for $\wexp\in[3/4,{(2+\sqrt{2})}/{4}]$.
If the protocol parameters are such that $\xi(n,\ecEC)\binh(\wexp)$ turns out to be fairly close to $1$,
there is not much loss incurred by simply sending the outputs of the test rounds directly rather than expending the effort to compute appropriate error-correction data.

\subsection{Privacy amplification}\label{sec:PA}

As noted in Sec.~\ref{sec:minmaxent}, privacy amplification can be addressed by the Leftover Hashing Lemma (Fact~\ref{fact:LHL}), which bounds the secrecy of the final key in terms of the difference between the smoothed min-entropy of the device outputs and the length of the key. 
Practically speaking, in the privacy amplification step of the protocol, Alice simply chooses a random function from the 2-universal family and publicly communicates it to Bob, followed by Alice and Bob applying that function to $\str{A}$ and $\tilde{\str{A}}$ respectively. 
As mentioned previously, since the hash choice $H$ is included in the ``side-information'' term in~\eqref{eq:LHL}, it is safe to publicly communicate it.

\subsection{Honest behaviour}
\label{sec:hon}

In general, the honest implementation consists of $n$ IID copies of a device characterized by 2 parameters, $\wexp$ and $h_\mathrm{hon}$. 
We will hence again describe the honest behaviour using single-round registers $\hat{A}\hat{B}XY$. 
$\wexp$ is the probability with which the device wins the CHSH game when supplied with uniformly random inputs $X \in \inX, Y \in \inYt$, and $h_\mathrm{hon}$ is defined in Eq.~\eqref{eq:ECrate}. While $h_\mathrm{hon}$ does not explicitly appear in the protocol description, it is implicitly used to define the parameter $\cmax$, as was described in Sec.~\ref{sec:EC}. (Since $h_\mathrm{hon}$ has a dependence on $\gamma$, strictly speaking it may be more precise to instead view the honest device behaviour as being parametrized by a tuple specifying all the individual entropies in Eq.~\eqref{eq:ECrate}, but for brevity we shall summarize this as the honest behaviour being parametrized by $h_\mathrm{hon}$. In the more specific models of honest devices described below, these entropies are expressed in terms of some simpler parameters.) 

As a simple example, we can take the honest devices to be described by depolarizing noise, with the noiseless distribution $\prnonoise[ab|xy]$ being the following:
the state is $\ket{\Phi^+}$, the measurements for inputs $X\in\inX,Y\in\inYt$ are those in~\eqref{eq:noiselessCHSH} (up to relabelling of Bob's inputs), and the measurements corresponding to Bob's inputs $Y\in\inYg$ are measurements in the same bases as Alice's measurements.
Under this noise model, we have simple expressions for $\wexp$ and $h_\mathrm{hon}$:
\begin{align}
\wexp&=(1-2\q)\frac{2+\sqrt{2}}{4} + 2\q\,\frac{1}{2}, \\
h_\mathrm{hon} & = \frac{1-\gamma}{2}\binh(\p+(1-2\p)\q) + \gamma \binh(\wexp),
\end{align}
where the $\p+(1-2\p)\q$ term is obtained by an explicit computation~\cite{WAP21}.

To focus on more specific experimental implementations, we will also consider the estimates given in~\cite{MvDR+19} for the Bell tests in~\cite{HBD+15} (resp.~\cite{RBG+17}): we characterize them via the parameters $\wexp = 0.797$ (resp.~$0.777$) and $\perr = 0.06$ (resp.~$0.035$), where $\perr$ is a parameter such that the probabilities \emph{before} noisy preprocessing satisfy\footnote{It can be shown that this is equivalent to taking $\prnoscale{\hat{A}\hat{B} |X = Y = z}_\mathrm{hon}$ to be independent of $z$ and taking the marginal distributions $\prnoscale{\hat{A} |X = Y = z}_\mathrm{hon}$, $\prnoscale{\hat{B} |X = Y = z}_\mathrm{hon}$ to be uniform, then defining $\perr \defvar \prnoscale{\hat{A} \neq \hat{B} |X = Y = z}_\mathrm{hon}$.}
\begin{align}
\prnoscale{\hat{A}\hat{B} |X = Y = z}_\mathrm{hon} = 
(1-2\perr)\frac{\delta_{\hat{A},\hat{B}}}{2} + 2\perr\, \frac{1}{4}
& \quad \text{ for all } z\in\inYg,
\end{align}
Furthermore, we take $\prnoscale{\hat{A}=\hat{B} |X = x, Y = y}_\mathrm{hon}$ to be the same for all $x\in\inX,y\in\inYt$ (in which case it must equal $\wexp$). Under this model, the expression~\eqref{eq:ECrate} for $h_\mathrm{hon}$ can be simplified (by the same computation as above):
\begin{align}
h_\mathrm{hon} 
& = \frac{1-\gamma}{2}\binh(\p+(1-2\p)\perr) + \gamma \binh(\wexp).
\end{align}
It was found in~\cite{SGP+21} that the random-key-measurements technique on its own is sufficient to achieve positive asymptotic keyrate for these experimental parameters. One of our goals in this chapter is hence to find what the keyrates would be when taking finite-size effects into account.

\section{Security statement}
\label{sec:finitekeylength}

To describe the length of the final key, we need to introduce some notation and ancillary functions. 
First, we require an \emph{affine} function $\lin_\p$ of the CHSH winning probability $w$, that lower-bounds the entropy minimization problem~\eqref{eq:mainopt} with\footnote{These $\keyw_x$ values reflect the input distributions in the protocol: observe that in the generation rounds, Alice and Bob choose inputs in $\{0,1\}$ uniformly, so the probability of them both choosing input $x$ is the same for both values of $x\in\{0,1\}$; while in the test rounds, Alice chooses her inputs uniformly and no sifting is applied afterwards.} $\keyw_0=\keyw_1=1/2$ and the (only) constraint being the CHSH winning probability.
In other words, $\lin_\p$ is to be an affine lower bound on the function $\breve{\lin}_\p$ defined in~\eqref{eq:mainopt}.
This is basically what we have obtained as the bounds~\eqref{eq:certbnds} (for specific values of $\p$) in the previous chapter, since the CHSH value and CHSH winning probability can be interconverted using $\constr = 8w-4$ (note that this conversion preserves the fact that the bounds~\eqref{eq:certbnds} are affine).
Given such a function $\lin_\p$, we then define the affine function
\begin{align}
\g(w) \defvar \frac{1-\gamma}{2} \lin_\p(w) + \gamma \lin_0(w).
\end{align}
Informally, $\g$ can be interpreted as a lower bound on the entropy ``accumulated'' in one round of the protocol.

Also, for an affine function $f$ defined on all probability distributions on some register $C$, and any subset $\mathcal{S}$ of its domain,
we will define
\begin{align}
\begin{gathered}
\Max(f) \defvar \max_q f(q), 
\qquad 
\Min_{\mathcal{S}}(f) \defvar \inf_{q\in\mathcal{S}} f(q), 
\\
\Var_\mathcal{S}(f) \defvar \sup_{q\in\mathcal{S}} \left(\sum_c q(c) f(\delta_c)^2 - \left(\sum_c q(c) f(\delta_c)\right)^2\right),
\end{gathered}
\end{align}
where $\max_q$ is taken over all distributions on $C$, and $\delta_c$ denotes the distribution with all its weight on the symbol $c$. 

Finally, for brevity we define some compact notation for binomial distributions:
\begin{definition}\label{def:binom}
Let $X\sim\operatorname{Binom}(n,p)$ denote a random variable $X$ following a binomial distribution with parameters $(n,p)$, i.e.~$X$ is the sum of $n$ IID Bernoulli random variables $X_j$ with $\pr{X_j=1}=p$.
We denote the corresponding cumulative distribution function as
\begin{align}
\cdfBin{n}{p}{k} \defvar \pr[X\sim\operatorname{Binom}(n,p)]{X \leq k}.
\end{align}
\end{definition}

\newcommand{\wmin}{w_\mathrm{min}}
\newcommand{\wmax}{w_\mathrm{max}}
With these definitions, we can state the security guarantees of the protocol (see~\cite{arx_TSB+20} for a qualitative explanation of the various parameters):
\begin{theorem}\label{th:DIQKD}
Take any 
$\ecEC,\ecPE,\eEA,\ePA,\eh,\es,\es',\es''\in (0,1]$ 
such that $\es>\es'+2\es''$, and any $\alpha\in(1,2)$, $\alpha'\in(1,1+2/V')$, $\cperp\in[\g(0),\g(1)]$, $\gamma\in(0,1)$, $\p\in[0,1/2]$, where $V'\defvar2\log5$.
Protocol~\ref{prot:DIQKD} is $(\ecEC + \ecPE)$-complete and $(\max\{\eEA, \ePA + 2\es\} + 2\eh)$-sound when performed with any choice of $\cmax$ such that Eq.~\eqref{eq:comEC} holds, and $\dtol,\lkey$ satisfying
\begin{align}
\ecPE &\geq \cdfBin{n}{\gamma\wexp}{\floor{(\wexp-\dtol)\gamma n}} + \cdfBin{n}{
1-\gamma+\gamma\wexp
}{\floor{(1-\gamma+\wexp\gamma-\dtol\gamma)n}},
\label{eq:ecPE}\\
\lkey &\leq n\g(\wexp-\dtol) - n\frac{(\alpha-1)\ln 2}{2}V^2 - n(\alpha-1)^2K_\alpha - n\gamma - n\left(\frac{\alpha'-1}{4}\right)V'^2 
\nonumber \\ &\qquad 
- \frac{\smf{\es'}}{\alpha-1} - \frac{\smf{\es''}}{\alpha'-1} - \left(\frac{\alpha}{\alpha-1}+\frac{\alpha'}{\alpha'-1}-2\right)\log\frac{1}{\eEA} - 3\smf{\es-\es'-2\es''} \nonumber\\&\qquad
- \cmax - \ceil{\log\left(\frac{1}{\eh}\right)} - 2\log\frac{1}{\ePA} + 2,
\label{eq:keylength}
\end{align}
where $\cdfBin{n}{p}{k}$ is the cumulative distribution function of a binomial distribution (Definition~\ref{def:binom}), and
\begin{align}
\begin{aligned}
\smf{\eps} &\defvar \log\frac{1}{1-\sqrt{1-\eps^2}} \leq \log\frac{2}{\eps^2}, 
\\
V &\defvar \sqrt{\Var_{\mathcal{Q}_{f}}(\fmin)+2} + \log
33,
\\
K_\alpha &\defvar \frac{2^{(\alpha-1)(2\log4 + \Max(\fmin)-\Min_{\mathcal{Q}_{f}} (\fmin))} }{6(2-\alpha)^3\ln2}
\ln^3\left(2^{2\log4 + \Max(\fmin)-\Min_{\mathcal{Q}_{f}} (\fmin)} + e^2\right),
\end{aligned}
\label{eq:VK}
\end{align}
with $\fmin$ and $\mathcal{Q}_{f}$ being a function and a set (defined explicitly in Sec.~\ref{sec:fmin}) such that (introducing the notation $\wmax \defvar {(2+\sqrt{2})}/{4}$ and $\wmin \defvar {(2-\sqrt{2})}/{4}$ for the maximum and minimum quantum winning probabilities of the CHSH game):
\begin{align}
\begin{gathered}
\Max(\fmin)
= \frac{1}{\gamma}\g(1) + \left(1-\frac{1}{\gamma}\right)\cperp, 
\qquad 
\Min_{\mathcal{Q}_{f}} (\fmin) = \g(\wmin), \\
\Var_{\mathcal{Q}_{f}}(\fmin) \leq 
\frac{\wmin}{\gamma} 
\min \left\{\Delta_0^2, \Delta_1^2 \right\}
+ \frac{\wmax}{\gamma} 
\max \left\{\Delta_0^2, \Delta_1^2 \right\}
, \text{ where } \Delta_w \defvar \cperp - \g(w).
\end{gathered}
\label{eq:minmaxvarbounds}
\end{align}
\end{theorem}

The variables listed at the start of Theorem~\ref{th:DIQKD} can be considered to be variational parameters that should be chosen to optimize the keyrate as much as possible. 
When optimizing these parameters in practice, we found that the exact expression for $\smf{\eps}$ was numerically unstable, and hence we replaced it with the upper bound of $\log(2/\eps^2)$ (this bound is basically tight at small $\eps$, so it makes little difference). Furthermore, the optimization over $\cperp$ also appears to be somewhat unstable. We found heuristically that the optimal value of $\cperp$ appears to typically be very close to $\g(1)$, and hence for simplicity in some cases we did not optimize over it but instead simply fixed $\cperp=\g(1)$ (or slightly below it, to avoid some instabilities at $\cperp-\g(1)=0$). Finally, we found that direct computation of $\cdfBin{n}{p}{k}$ (e.g.~via the regularized beta function) could sometimes be slow or unstable, and in such cases we followed~\cite{LLR+21} and replaced it with the upper bound in the following theorem:
\begin{fact}\label{fact:ZSbound}\cite{ZS13}
Let
\begin{align}
\fZS{n}{p}{k} \defvar \Phi\left(\operatorname{sign}\left(k-pn\right)\sqrt{2n \,D_{\nat}\!\left({k}/{n} \;\middle\Vert\; p\right)}\right),
\end{align}
where $D_{\nat}(q \Vert p)\defvar q\ln\frac{q}{p} + (1-q)\ln\frac{1-q}{1-p}$ is the base-$e$ relative entropy between the distributions $\{q,1-q\}$ and $\{p,1-p\}$, while $\Phi(z)\defvar \int_{-\infty}^z 
(2\pi)^{-1/2}
e^{-u^2/2} \mathrm{d}u 
=\left(1/2\right)\operatorname{erfc}\!\left(-{z}/{\sqrt{2}}\right)$ is the cumulative distribution function of the standard normal distribution.
Then for any $k\in\upto{n-1}$,
\begin{align}
\fZS{n}{p}{k} \leq \cdfBin{n}{p}{k} \leq \fZS{n}{p}{k+1}. 
\label{eq:ZSbound}
\end{align}
\end{fact}
\noindent Replacing $\cdfBin{n}{p}{k}$ with $\fZS{n}{p}{k+1}$ and computing the latter (which is a Gaussian integral) appeared to be faster and more stable than computing $\cdfBin{n}{p}{k}$ directly. There is little loss incurred by performing this replacement --- the inequalities~\eqref{eq:ZSbound} imply $\fZS{n}{p}{k+1} \leq \cdfBin{n}{p}{k+1}$, so the effect is no larger than replacing $\cdfBin{n}{p}{k}$ by $\cdfBin{n}{p}{k+1}$, which is almost negligible in the parameter regimes studied in this work.

Note that in general, the optimal parameter values (especially for $\alpha$ and $\alpha'$) would depend heavily on $n$.
To get an estimate for the asymptotic scaling of $\lkey$, we can choose all the $\eps$ parameters to take some constant values satisfying the desired completeness and soundness bounds, set $\cperp=\g(1)$, then choose~\cite{DFR20,DF19}
\begin{align}
\begin{aligned}
\alpha-1&= \frac{1}{\sqrt{n}} \sqrt{\frac{2}{V^2 \ln 2} \left(\smf{\es'} + 2\log\frac{1}{\eEA}\right)} , \\
\alpha'-1&= \frac{1}{\sqrt{n}} \sqrt{\frac{4}{V'^2} \left(\smf{\es''} + 2\log\frac{1}{\eEA} \right)} ,
\end{aligned}
\end{align}
taking $n$ to be large enough such that $\alpha \in (1,3/2)$ and $\alpha'\in(1,1+2/V')$. Furthermore, since we picked $\cperp=\g(1)$, we have $\Max(\fmin)-\Min_{\mathcal{Q}_{f}} (\fmin) = \g(1)-\g(\wmin)$, in which case for $\alpha\in(1,3/2)$ we can upper bound $K_\alpha$ with the following constant (independent of $\alpha$, $n$ and $\gamma$):
\begin{align}
\hat{K} \defvar \frac{2^{(2\log4 + \g(1)-\g(\wmin))} }{(3/4)\ln2}
\ln^3\left(2^{2\log4 + \g(1)-\g(\wmin)} + e^2\right) .
\end{align}
With these choices, Eq.~\eqref{eq:keylength} can be satisfied by taking
\begin{align}
\lkey &= \bigg\lfloor
n\g(\wexp-\dtol) - n\frac{(\alpha-1)\ln 2}{2}V^2 - n(\alpha-1)^2 \hat{K} - n\gamma - n\left(\frac{\alpha'-1}{4}\right)V'^2 
\nonumber \\ &\qquad 
- \frac{1}{\alpha-1}\left(\smf{\es'} + 2\log\frac{1}{\eEA} \right) - \frac{1}{\alpha'-1}\left(\smf{\es''} + 2\log\frac{1}{\eEA} \right) + 2\log\frac{1}{\eEA} - 3\smf{\es-\es'-2\es''} \nonumber\\&\qquad
- \cmax - \ceil{\log\left(\frac{1}{\eh}\right)} - 2\log\frac{1}{\ePA} + 2 \bigg\rfloor \nonumber \\
&= \bigg\lfloor n\g(\wexp-\dtol) - \sqrt{n}\,2\sqrt{\frac{V^2 \ln 2}{2} \left(\smf{\es'} + 2\log\frac{1}{\eEA}\right)} - \frac{2}{V^2 \ln 2} \left(\smf{\es'} + 2\log\frac{1}{\eEA}\right)\hat{K} 
- n\gamma \nonumber \\ &\qquad 
- \sqrt{n} \, 2\sqrt{\frac{V'^2}{4} \left(\smf{\es''} + 2\log\frac{1}{\eEA} \right)} + 2\log\frac{1}{\eEA} - 3\smf{\es-\es'-2\es''} \nonumber\\&\qquad 
- \cmax - \ceil{\log\left(\frac{1}{\eh}\right)} - 2\log\frac{1}{\ePA} + 2 \bigg\rfloor \nonumber \\
&= \left\lfloor n\g(\wexp-\dtol)  - n\gamma - O\left(\sqrt{\frac{n}{\gamma}}\right) - O(1) - \cmax \right\rfloor,
\label{eq:estkeylength}
\end{align}
where in the last line the implied constants in the big-$O$ notation are independent of $n$ and $\gamma$ (that line is obtained by noting that $V = O(1/\sqrt{\gamma})$; also, we remark that technically the term involving $\hat{K}$ is of order $O(1/V^2) = O({\gamma})$, but this is anyway upper bounded by $O(1)$ so we have simply absorbed it into the $O(1)$ term).
Furthermore, Eq.~\eqref{eq:ecPE} can be satisfied by choosing $\gamma=2 n^{-1} \dtol^{-2} \log(2/\ecPE)$ (see Eq.~\eqref{eq:chernoffs}, noting that the denominators in the exponents of both terms are trivially upper bounded by $2$), 
in which case by taking e.g.~$\dtol \propto 1/n^{1/3}$ 
we have $\dtol,\gamma = O(1/n^{1/3})$, and substituting this scaling into Eq.~\eqref{eq:estkeylength} then yields
\begin{align}
\lim_{n\to\infty} \frac{\lkey}{n} = 
\frac{1}{2}\left(\lin_\p(\wexp) - 
\sum_{z\in\inX} \frac{1}{2}H(\hat{A} | \hat{B};X = Y = z)_\mathrm{hon}\right)
,
\label{eq:asympt}
\end{align}
taking $\cmax$ according to Eqs.~\eqref{eq:optEC}--\eqref{eq:ECrate}. This is the expected asymptotic result in the sense of the Devetak-Winter expression~\eqref{eq:devwin}, with the prefactor of $1/2$ being due to the sifting.

Given the scaling behaviour shown in Eq.~\eqref{eq:estkeylength}, it can be seen that the optimal values of the various $\eps$ parameters (given some desired values of $\ecom$ and $\esound$) may be of rather different orders of magnitude. This is because some of them appear in $O(1/\sqrt{n})$ corrections to the finite-size keyrate while others appear in $O(1/n)$ corrections. Intuitively speaking, the $\eps$ parameters in the latter can be chosen to be substantially smaller than the former, since the $O(1/n)$ scaling reduces their contribution to the finite-size effects.

\section{Finite-size security proof}
\label{sec:finiteproof}

We now prove that Protocol~\ref{prot:DIQKD} indeed satisfies the security properties claimed in Theorem~\ref{th:DIQKD}.
To do so, we first introduce a virtual protocol that is more convenient to analyze.
For the purposes of understanding this construction, it may be helpful to think of it as being based on a \emph{specific} set of states and measurements that could be occurring in a run of Protocol~\ref{prot:DIQKD} (as opposed to simultaneously considering all possible states and measurements that could be occurring). In particular, this virtual protocol
and the channels $\map_j$ in Sec.~\ref{sec:EAT} should be understood as being constructed in terms of this specific set of states and measurements. Since we will not impose any additional assumptions on these states and measurements beyond those specified by the protocol, this will still yield a valid way for us to prove the desired security properties (in particular the soundness property, which has to be proven for all possible states and measurements that could occur in a run of the protocol).

Consider the state at the end of Step~\ref{step:PA} in Protocol~\ref{prot:DIQKD}. We now describe a virtual protocol\footnote{We stress that this ``protocol'' is not performed in practice 
(and in fact cannot be, since the ``virtual parameter estimation'' step cannot be performed locally by either party). 
However, it produces exactly the same state as Protocol~\ref{prot:DIQKD} on all relevant registers, and can hence be used for the security analysis.} that produces exactly the same state (when it is implemented using the same input state and measurements as those used in a run of Protocol~\ref{prot:DIQKD}), apart from the introduction of two additional registers $\Btstr\Ctstr$.

\let\oldthealgorithm\thealgorithm 
\renewcommand{\thealgorithm}{$1'$}
\begin{savenotes}
\begin{algorithm}[H]
\caption{A virtual protocol}\label{prot:virtual}
\begin{algorithmic}[1]
\State Alice and Bob's devices each receive and store all quantum states that they will subsequently measure.
\State\label{step:EATmap}For each $j \in \upto{n}$, perform the following steps: 
\begin{algsubstates}
\State Alice chooses a uniform input $X_j \in \inX$. With probability $\gamma$, Bob chooses a uniform input $Y_j \in \inYt$, otherwise Bob chooses a uniform input $Y_j \in \inYg$.
\State Alice and Bob supply their inputs to their devices, and record the outputs as $A_j$ and $B_j$ respectively. 
\State Alice and Bob publicly announce their inputs.
\State\textbf{Sifting:} If $Y_j \in \inYg$ and $X_j \neq Y_j$, Alice and Bob overwrite their outputs with $A_j = B_j = 0$.
\State\textbf{Noisy preprocessing:} If $Y_j \in \inYg$ and $X_j = Y_j$, Alice generates a biased random bit $F_j$ with $\pr{F_j=1} = \p$, and overwrites her output $A_j$ with $A_j \oplus F_j$.
\State 
If $Y_j \in \inYg$, Bob sets $\Bt_j = 0$, otherwise Bob sets $\Bt_j = B_j$. 
\State\textbf{Virtual parameter estimation:} 
Set $\Ct_j=\perp$ if $Y_j \in \inYg$; otherwise set $\Ct_j=0$ if $A_j \oplus \Bt_j \neq X_j \cdot (Y_j-2)$ and $\Ct_j=1$ if $A_j \oplus \Bt_j = X_j \cdot (Y_j-2)$.
\end{algsubstates}
\State\textbf{Error correction:} Alice and Bob publicly communicate some bits $\str{L}$ for error correction as previously described, allowing Bob to construct a guess $\tilde{\str{A}}$ for $\str{A}$.
\State\textbf{Parameter estimation:} 
For all $j\in\upto{n}$, Bob sets 
$C_j=\perp$ if $Y_j \in \inYg$; otherwise he sets $C_j=0$ if $\tilde{A}_j \oplus \Bt_j \neq X_j \cdot (Y_j-2)$ and $C_j=1$ if $\tilde{A}_j \oplus \Bt_j = X_j \cdot (Y_j-2)$.
\end{algorithmic}
\end{algorithm}
\end{savenotes}
\let\thealgorithm\oldthealgorithm 
\addtocounter{algorithm}{-1} 

\noindent The key changes as compared to Protocol~\ref{prot:DIQKD} are as follows:
\begin{itemize}
\item All the states that the devices will measure are distributed immediately at the start (note that this is possible because in Protocol~\ref{prot:DIQKD}, the measurement choices $\str{X},\str{Y}$ are not disclosed until all measurements have been performed, and hence the distributed states cannot behave adaptively with respect to the inputs). 
\item The sifting and noisy preprocessing steps are now performed immediately after each measurement, instead of after all measurements are performed. This is to allow us to subsequently apply the EAT.
\item Two additional registers were introduced: $\Btstr$, which is equal to $\str{B}$ on all the test rounds but is otherwise set to $0$, and $\Ctstr$, which is analogous to $\str{C}$ but computed using $\str{A}$ in place of $\tilde{\str{A}}$. These registers were used in a \term{virtual parameter estimation} step.
\item All parameter estimation is performed with $\Btstr$ instead of $\str{B}$ (this substitution has no physical effect since $\Bt_j = B_j$ 
in all rounds used for parameter estimation).
\end{itemize}

Let $\rho$ denote the state on registers ${\str{A}\tilde{\str{A}}\str{B}\Btstr\str{X}\str{Y}\str{L}\str{C}\Ctstr \allE}$ (as well as the choice of hash function in the error-correction step) at the end of Protocol~\ref{prot:virtual}. As mentioned, the reduced state after tracing out $\Btstr \Ctstr$ is exactly the same as that at the end of Step~\ref{step:PA} in Protocol~\ref{prot:DIQKD}. 
Since all subsequent steps in Protocol~\ref{prot:DIQKD} (i.e.~simply the accept/abort check and the privacy amplification) only involve these registers, to prove the security of Protocol~\ref{prot:DIQKD} it suffices to consider the equivalent (apart from $\Btstr\Ctstr$) process where all the steps up to Step~\ref{step:PA} are replaced by this virtual protocol, and then the remaining steps in Protocol~\ref{prot:DIQKD} are performed. With this in mind, let us define the following events on the state $\rho$:
{\setlist[description]{leftmargin=2cm,labelindent=2cm}
\begin{description}
\item $\Og$: $\str{A} = \tilde{\str{A}}$ 
(i.e.~Bob correctly guesses $\str{A}$)
\item $\Oh$: $\hash(\str{A}) = \hash(\tilde{\str{A}})$
\item $\OPE$: $\freq_\str{c}(1)\geq (\wexp-\dtol)\gamma$ and $\freq_\str{c}(0)\leq (1-\wexp+\dtol)\gamma$
\item $\OpPE$: $\freq_{\ctstr}(1)\geq (\wexp-\dtol)\gamma$ and $\freq_{\ctstr}(0)\leq (1-\wexp+\dtol)\gamma$
\end{description}
}
\noindent Note that in terms of these events, the accept condition of the protocol is $\Oh \land \OPE$. With the virtual protocol and the above events in mind, we now turn to proving completeness and soundness of Protocol~\ref{prot:DIQKD}.

\subsection{Completeness}
\label{sec:com}

Completeness is defined entirely with respect to the honest behaviour of the devices, hence all discussion in this section is with respect to the situation where the state $\rho$ described above is the one produced by the honest states and measurements. 
To prove completeness, we simply need to obtain an upper bound on the probability that this honest behaviour yields an abort, i.e.~$\pr{\no{\Oh} \lor \no{\OPE}}_\mathrm{hon}$ (recall we use $\no{\Omega}$ to denote the complement of an event). 
However, we encounter a slight inconvenience here because the event $\no{\OPE}$ involves the register $\str{C}$ produced using Bob's guess $\tilde{\str{A}}$ rather than Alice's actual string $\str{A}$, and there is some small probability that his guess was wrong. To cope with this, we shall break down $\pr{\no{\Oh} \lor \no{\OPE}}_\mathrm{hon}$ into simpler terms that can be bounded in terms of probabilities involving only the ``virtual'' string $\Ctstr$ rather than $\str{C}$, where the former is easier to handle since it is produced from the actual value of $\str{A}$.

We begin by noting that the hashes of $\str{A}$ and $\tilde{\str{A}}$ can only differ if $\str{A}\neq\tilde{\str{A}}$, which is to say that the event $\no{\Oh}$ implies the event $\no{\Og}$. With this, we write
\begin{align}
\pr{\no{\Oh} \lor \no{\OPE}}_\mathrm{hon} 
&\leq \pr{\no{\Og} \lor \no{\OPE}}_\mathrm{hon} \nonumber\\
&= \pr{\no{\Og}}_\mathrm{hon} + \pr{\Og \land \no{\OPE}}_\mathrm{hon}
\end{align}
where in the second line we have partitioned the event $\no{\Og} \lor \no{\OPE}$ into the disjoint events $\no{\Og}$ and $\Og \land \no{\OPE}$.
We shall now upper bound the probabilities of each of these events. 

The $\pr{\no{\Og}}_\mathrm{hon}$ term is straightforward to handle, since by construction the error-correction step ensures that this probability is at most $\ecEC$.
As for the $\pr{\Og \land \no{\OPE}}_\mathrm{hon}$ term, we now make the critical observation that $\Og \land \no{\OPE} = \Og \land \OpPEnot$ (because the event $\Og$ implies that $\str{C}=\Ctstr$). Therefore, we have
\begin{align}
\pr{\Og \land \no{\OPE}}_\mathrm{hon} = \pr{\Og \land \OpPEnot}_\mathrm{hon} \leq \pr{\OpPEnot}_\mathrm{hon},
\end{align}
which is the desired reduction to a term involving $\Ctstr$ rather than $\str{C}$. To explicitly upper bound $\pr{\OpPEnot}_\mathrm{hon}$, observe that under the honest behaviour, the string $\Ctstr$ consists of $n$ IID rounds such that
$\mathrm{Pr}[{\Ct_j=1}]_\mathrm{hon} = \gamma \wexp$ and $\mathrm{Pr}[{\Ct_j=0}]_\mathrm{hon} = \gamma (1-\wexp)$ in each round. Therefore, we have
\begin{align}
\pr{\freq_\str{c}(1) < (\wexp-\dtol)\gamma}_\mathrm{hon}
&\leq 
\pr{\freq_\str{c}(1)n \leq \floor{(\wexp-\dtol)\gamma n}}_\mathrm{hon} \nonumber\\
&= \cdfBin{n}{\gamma\wexp}{\floor{(\wexp-\dtol)\gamma n}}, \label{eq:ecPEbound1}\\
\pr{\freq_\str{c}(0) > (1-\wexp+\dtol)\gamma}_\mathrm{hon} 
&= \pr{\freq_\str{c}(
\neg 0
) \leq 1-(1-\wexp+\dtol)\gamma}_\mathrm{hon} \nonumber\\
&= \pr{\freq_\str{c}(\neg 0) n \leq \floor{
(1-\gamma+\wexp\gamma-\dtol\gamma)
n}}_\mathrm{hon} \nonumber\\
&= \cdfBin{n}{
1-\gamma+\gamma\wexp
}{\floor{(1-\gamma+\wexp\gamma-\dtol\gamma)n} },
\label{eq:ecPEbound0}
\end{align}
where $\neg 0$ represents all symbols other than $0$. 
Hence by the union bound, $\pr{\OpPEnot}_\mathrm{hon}$ is upper bounded by $\ecPE$ is as specified in Eq.~\eqref{eq:ecPE} (i.e.~the sum of the expressions~\eqref{eq:ecPEbound1} and~\eqref{eq:ecPEbound0}). This yields a final upper bound of $\ecEC + \ecPE$ on the probability of the honest behaviour aborting, as desired.

In principle, somewhat simpler expressions could be obtained using the (additive) Chernoff bound, writing $\bar{w}_\mathrm{exp} \defvar 1-\wexp$ for brevity (and assuming $\dtol<\min\{\wexp,\bar{w}_\mathrm{exp}\}$):
\begin{align}
\begin{gathered}
\pr{\freq_\str{c}(1) < (\wexp-\dtol)\gamma}_\mathrm{hon}
\leq e^{-D_{\nat}\left((\wexp-\dtol)\gamma \middle\Vert \wexp\gamma\right)n}
\leq 
e^{-\frac{n \gamma \dtol^2}{2\wexp}}, 
\\
\pr{\freq_\str{c}(0) > (1-\wexp+\dtol)\gamma}_\mathrm{hon} 
\leq e^{-D_{\nat}\left((\bar{w}_\mathrm{exp}+\dtol)\gamma \middle\Vert \bar{w}_\mathrm{exp}\gamma\right)n}
\leq e^{-\frac{n \gamma \dtol^2}{2(\bar{w}_\mathrm{exp}+\dtol)}}
\leq e^{-\frac{n \gamma \dtol^2}{2}}. 
\end{gathered}
\label{eq:chernoffs}
\end{align}
However, it was observed in~\cite{LLR+21} that these bounds are weaker than Fact~\ref{fact:ZSbound} by a significant amount.
As yet another alternative, Hoeffding's inequality yields a bound of $e^{-2n \gamma^2 \dtol^2}$, 
but this is worse than the Chernoff bound whenever $6\gamma \leq 1/\wexp$; furthermore, it does not yield nontrivial bounds if we choose $\gamma\propto1/n$ (for constant $\ecEC,\dtol$), which ought to be possible in principle according to the scaling analysis in~\cite{DF19} (and loosely matches the intuition described in Sec.~\ref{sec:sketchasympt} regarding parameter estimation, except that it gives a constant \emph{expected} number of test rounds rather than an actually constant number).

\subsection{Soundness}
\label{sec:soundness}

Soundness has to be proven against all possible dishonest behaviours (subject to the protocol assumptions as listed in Sec.~\ref{sec:assumptions}). 
In this section, we shall consider any particular state $\rho$ (as defined after the Protocol~\ref{prot:virtual} description) that would be produced by a particular choice of such dishonest behaviour, and prove that it satisfies the soundness condition~\eqref{eq:sound} (for a specific value of $\esound$) regardless of which dishonest behaviour was considered. 
All probabilities are to be understood as being defined with respect to that state $\rho$.

We first note that it is straightforward to show that Protocol~\ref{prot:DIQKD} is $\eh$-correct, since 
\begin{align}
\pr{K_A \neq K_B \land \Oh \land \OPE} 
\leq \pr{K_A \neq K_B \land \Oh} 
\leq \pr{\no{\Og} \land \Oh} 
\leq \pr{\Oh | \no{\Og}}
\leq \eh,
\label{eq:corrproof}
\end{align}
where the last inequality holds by the defining property of 2-universal hashing.

It remains to prove secrecy.
Denote the privacy amplification step in Protocol~\ref{prot:DIQKD} as the map $\mPA$, so the subnormalized state conditioned on the event of Protocol~\ref{prot:DIQKD} accepting can be written as $\mPA(\rho_{\land \Oh \land \OPE})$. Let $H$ be the register storing the choice of hash function used for privacy amplification, and denote $E' \defvar \str{X}\str{Y}\str{L} \allE H$ for brevity. We can now rewrite the secrecy condition as the requirement
\begin{align}
\frac{1}{2}\norm{\mPA(\rho_{\land \Oh \land \OPE})_{K_A E'} - \idk_{K_A} \otimes \mPA(\rho_{\land \Oh \land \OPE})_{E'}}_1 \leq \esecr.
\label{eq:esecrnew}
\end{align}
Now, somewhat similarly to the completeness analysis, we shall find a way to upper bound the left-hand-side of the above expression in terms of $\OpPE$ rather than $\OPE$, as the former is easier to handle.
Specifically, by noting that $\rho_{\land \Oh \land \OPE} = \rho_{\land \Og \land \Oh \land \OPE} + \rho_{\land \no{\Og} \land \Oh \land \OPE}$, we find the following bound using norm subadditivity:
\begin{align}
&\frac{1}{2}\norm{\mPA(\rho_{\land \Oh \land \OPE})_{K_A E'} - \idk_{K_A} \otimes \mPA(\rho_{\land \Oh \land \OPE})_{E'}}_1 \nonumber\\
\leq{}& \frac{1}{2}\norm{\mPA(\rho_{\land \Og \land \Oh \land \OPE})_{K_A E'} - \idk_{K_A} \otimes \mPA(\rho_{\land \Og \land \Oh \land \OPE})_{E'}}_1 \nonumber\\ 
&\qquad + \frac{1}{2}\norm{\mPA(\rho_{\land \no{\Og} \land \Oh \land \OPE})_{K_A E'} - \idk_{K_A} \otimes \mPA(\rho_{\land \no{\Og} \land \Oh \land \OPE})_{E'}}_1 \nonumber\\
\leq{}& \frac{1}{2}\norm{\mPA(\rho_{\land \Og \land \Oh \land \OpPE})_{K_A E'} - \idk_{K_A} \otimes \mPA(\rho_{\land \Og \land \Oh \land \OpPE})_{E'}}_1 
+ \eh, \label{eq:secsplit}
\end{align}
where in the last line we have 
performed the substitution $\Og \land \OPE = \Og \land \OpPE$ in the first term,
and bounded the second term using
\begin{align}
\pr{\no{\Og} \land \Oh \land \OPE} \leq \pr{\no{\Og} \land \Oh} \leq \eh,
\end{align}
as noted in Eq.~\eqref{eq:corrproof}.

We now aim to bound the first term in Eq.~\eqref{eq:secsplit}. To do so, we study three\footnote{This has one more case as compared to the proof sketches in Sec.~\ref{sec:proofsketch}. This is just a technicality to ensure that the smoothing parameter is not too large when applying some theorems regarding subnormalized states.} exhaustive possibilities for the state $\rho$ (the first two are not mutually exclusive, but this does not matter):
{\setlist[enumerate]{leftmargin=2cm,labelindent=2cm}
\begin{enumerate}[label*=Case \arabic*:,ref=\arabic*]
\item \label{case:ECsmall} $\pr{\Og \land \Oh|\OpPE} \leq \es^2$.
\item \label{case:PEsmall} $\pr{\OpPE} \leq \eEA$.
\item \label{case:neither} Neither of the above are true.
\end{enumerate}
}
\noindent In case~\ref{case:ECsmall}, that term is bounded by
\begin{align}
\pr{\Og \land \Oh \land \OpPE} = \pr{\Og \land \Oh |\OpPE} \pr{\OpPE} \leq \es^2.
\label{eq:caseECsmall}
\end{align}
In case~\ref{case:PEsmall}, that term is bounded by
\begin{align}
\pr{\Og \land \Oh \land \OpPE} \leq \pr{\OpPE} \leq \eEA.
\end{align}
The main challenge is case~\ref{case:neither}. To study this case, we first focus on the conditional state $\rho_{|\OpPE}$. Importantly, the relevant smoothed min-entropy of this state can be bounded by using the following theorem, which we prove in Sec.~\ref{sec:EAT} using entropy accumulation:
\begin{theorem}\label{th:rawHmin}
For all parameter values as specified in Theorem~\ref{th:DIQKD}, the state at the end of Protocol~\ref{prot:virtual} satisfies
\begin{align}
\Hmin^{\es}(\str{A}|\str{X}\str{Y} \allE)_{\rho_{|\OpPE}} &>
n\g(\wexp-\dtol) - n\frac{(\alpha-1)\ln 2}{2}V^2 - n(\alpha-1)^2K_\alpha - n\gamma - n\left(\frac{\alpha'-1}{4}\right)V'^2 
\nonumber \\ &\qquad 
- \frac{\smf{\es'}}{\alpha-1} - \frac{\smf{\es''}}{\alpha'-1} - \left(\frac{\alpha}{\alpha-1}+\frac{\alpha'}{\alpha'-1}\right)\log\frac{1}{\pr{\OpPE}} - 3\smf{\es-\es'-2\es''}.
\label{eq:rawHmin}
\end{align}
where $V',\smf{\eps},V,K_\alpha$ are as defined in Theorem~\ref{th:DIQKD}.
\end{theorem}

To relate this to our state of interest, 
we use~\cite{TL17} Lemma~10, which states that for a state $\sigma\in\dop{\leq}(ZZ'QQ')
$ that is classical on registers $ZZ'$, an event $\Omega$ on the registers $ZZ'$, and any $\eps \in [0,\sqrt{\tr{\sigma_{\land\Omega}}})$,
we have
\begin{align}
\Hmin^\eps(ZQ|Q')_{\sigma_{\land\Omega}} \geq \Hmin^\eps(ZQ|Q')_{\sigma} .
\label{eq:condlemma}
\end{align}
In our context, we observe that the probability of the event $\Og \land \Oh$ on the normalized state $\rho_{|\OpPE}$ is $\pr{\Og \land \Oh|\OpPE}$, which is greater than $\es^2$ since we are in case~\ref{case:neither}. Hence the conditions
of the lemma are satisfied (identifying $\str{A}$ with $Z$, $\tilde{\str{A}}$ (and the error-correction hash choice) with $Z'$, $\str{X}\str{Y}\str{L} \allE$ with $Q'$, and leaving $Q$ empty), allowing us to obtain the bound
\begin{align}
\Hmin^{\es}(\str{A}|\str{X}\str{Y}\str{L} \allE)_{\big(\rho_{|\OpPE}\big) {}_{\land \Og \land \Oh} } &\geq \Hmin^{\es}(\str{A}|\str{X}\str{Y}\str{L} \allE)_{\rho_{|\OpPE}} \nonumber\\
&\geq \Hmin^{\es}(\str{A}|\str{X}\str{Y} \allE)_{\rho_{|\OpPE}} 
- \len(\str{L}) \nonumber\\
&\geq \Hmin^{\es}(\str{A}|\str{X}\str{Y} \allE)_{\rho_{|\OpPE}} 
- \cmax - \ceil{\log\left(\frac{1}{\eh}\right)}
,\label{eq:entPA}
\end{align}
where in the second line we have applied a chain rule for smoothed min-entropy (see e.g.~\cite{WTHR11} Lemma~11 or~\cite{Tom16} Lemma~6.8).

Putting together Eq.~\eqref{eq:entPA} and Theorem~\ref{th:rawHmin}, we find that for a key of length $\lkey$ satisfying Eq.~\eqref{eq:keylength}, we have\footnote{Keeping the $\log\left(1/\pr{\OpPE}\right)$ term here yields a slightly tighter result as compared to~\cite{ARV19,MvDR+19}, which instead used $\pr{\OpPE}\leq1$ to write an inequality in place of the equality in the second-last line of Eq.~\eqref{eq:finalPA}.}
\begin{align}
&\frac{1}{2}\left(\Hmin^{\es}(\str{A}|\str{X}\str{Y}\str{L} \allE)_{\big(\rho_{|\OpPE}\big) {}_{\land \Og \land \Oh} } - \lkey  + 2\right) + \log\frac{1}{\pr{\OpPE}} \nonumber\\
\geq{}&
\frac{1}{2} \left(- \left(\frac{\alpha}{\alpha-1}+\frac{\alpha'}{\alpha'-1}\right)\log\frac{1}{\pr{\OpPE}} +
2\log\frac{1}{\ePA} + 
\left(\frac{\alpha}{\alpha-1}+\frac{\alpha'}{\alpha'-1} - 2\right)\log\frac{1}{\eEA} \right) + \log\frac{1}{\pr{\OpPE}} \nonumber\\
={}& \log\frac{1}{\ePA} +\frac{1}{2} \left(\frac{\alpha}{\alpha-1}+\frac{\alpha'}{\alpha'-1}-2\right) \left(\log\frac{1}{\eEA} - \log\frac{1}{\pr{\OpPE}}\right) \nonumber\\
\geq{}& \log\frac{1}{\ePA}, 
\end{align}
where the last line holds because $\pr{\OpPE} \geq \eEA$ in case~\ref{case:neither} (and also $\frac{\alpha}{\alpha-1}+\frac{\alpha'}{\alpha'-1} -2 \geq 2 
> 0
$ since $\alpha,\alpha' \in (1,2)$).

This finally allows us to bound Eq.~\eqref{eq:secsplit}, since
we have $\es^2 \leq \pr{\Og \land \Oh|\OpPE} = \operatorname{Tr}[(\rho_{|\OpPE})_{\land \Og \land \Oh}]$ in case~\ref{case:neither}, and hence we can apply the Leftover Hashing Lemma to see that\footnote{We remark that the events $\OpPE$ and $\Og \land \Oh$ were treated somewhat differently in this analysis, in that 
while we introduced a parameter $\eEA$ to divide the analysis of $\OpPE$ into separate cases, there is no analogous condition for $\Og \land \Oh$ (the condition for case~\ref{case:ECsmall} is not analogous; it is merely a technicality to allow us to apply Eq.~\eqref{eq:condlemma} and the Leftover Hashing Lemma).
This is fundamentally because conditioning on $\OpPE$ via Eq.~\eqref{eq:rawHmin} has a ``larger'' effect on min-entropy compared to conditioning on $\Og \land \Oh$ (via Eq.~\eqref{eq:condlemma} modified for normalized conditional states). When substituting these effects into the final security parameter, one finds that the former worsens the security parameter by an amount that depends on $\log(1/\pr{\OpPE})$, while the latter does not, and hence we required a bound on $\pr{\OpPE}$ in the former. However, see~\cite{arx_Dup21} for a potential approach to avoid this case analysis.}
\begin{align}
&\frac{1}{2}\norm{\mPA(\rho_{\land \Og \land \Oh \land \OpPE})_{K_A E'} - \idk_{K_A} \otimes \mPA(\rho_{\land \Og \land \Oh \land \OpPE})_{E'}}_1 \nonumber\\
={}& \pr{\OpPE} \frac{1}{2}\norm{\mPA((\rho_{|\OpPE})_{\land \Og \land \Oh})_{K_A E'} - \idk_{K_A} \otimes \mPA((\rho_{|\OpPE})_{\land \Og \land \Oh})_{E'}}_1 \nonumber\\
\leq{}& \pr{\OpPE} \left(2^{-\frac{1}{2}\left(\Hmin^{\es}(\str{A}|\str{X}\str{Y}\str{L} \allE)_{\big(\raisebox{1.5\depth}{$\mathsmaller{\rho}$}|\OpPE\big) {\land \Og \land \Oh} } - \lkey + 2\right)} + 2\es \right) \nonumber\\
={}& 2^{-\frac{1}{2}\left(\Hmin^{\es}(\str{A}|\str{X}\str{Y}\str{L} \allE)_{\big(\raisebox{1.5\depth}{$\mathsmaller{\rho}$}|\OpPE\big)\land \Og \land \Oh } - \lkey + 2\right) - \log\left(1/\pr{\OpPE}\right)} + 2\es\pr{\OpPE} 
\nonumber\\
\leq{}& \ePA + 2\es. \label{eq:finalPA}
\end{align}

Since the three possible cases are exhaustive, we conclude that the secrecy condition is satisfied by choosing
\begin{align}
\esecr = \max\{\es^2, \eEA, \ePA + 2\es\} + \eh = \max\{\eEA, \ePA + 2\es\} + \eh.
\end{align}
Recalling that we have already shown the protocol is $\eh$-correct, we finally conclude that it is $(\max\{\eEA, \ePA + 2\es\} + 2\eh)$-sound.

\subsection{Entropy accumulation}\label{sec:EAT}

This section is devoted to the proof of Theorem~\ref{th:rawHmin}. The key theoretical tool in this proof is the entropy accumulation theorem, which we shall now briefly outline in the form stated in~\cite{DF19}. To do so, we shall first introduce \term{EAT channels} and \term{tradeoff functions}.

\begin{definition}\label{def:EATchann}
A \term{sequence of EAT channels} is a sequence 
$\{\map_j\}_{j\in\upto{n}}$ 
where each $\map_j$ is a channel from a register $R_{j-1}$ to registers $D_j S_j T_j R_j$, 
which satisfies the following properties:
\begin{itemize}
\item All $D_j$ are classical registers with a common alphabet $\mathcal{D}$, and all $S_j$ have the same finite dimension.
\item 
For each $\map_j$, the value of $D_j$ is determined from the registers $S_j T_j$ alone. Formally, this means 
$\map_j$ is of the form $\mathcal{P}_j \circ \map'_j$, where $\map'_j$ is a channel from $R_{j-1}$ to $S_j T_j R_j$, and $\mathcal{P}_j$ is a channel from $S_j T_j$ to $D_j S_j T_j$ of the form
\begin{align}
\mathcal{P}_j(\rho_{S_j T_j}) = \sum_{s\in\mathcal{S},t\in\mathcal{T}} (\Pi_{S_j,s} \otimes \Pi_{T_j,t}) \rho_{S_j T_j} (\Pi_{S_j,s} \otimes \Pi_{T_j,t}) \otimes \pure{d(s,t)}_{D_j},
\label{eq:nodisturb}
\end{align}
where $\{\Pi_{S_j,s}\}_{s\in\mathcal{S}}$ and $\{\Pi_{T_j,t}\}_{t\in\mathcal{T}}$ are families of orthogonal projectors on $S_j$ and $T_j$ respectively, and $d: \mathcal{S} \times \mathcal{T} \to \mathcal{D}$ is a deterministic function.
\end{itemize}
\end{definition}

\begin{definition}
Let $\fmin$ be a real-valued affine function defined on probability distributions over the alphabet $\mathcal{D}$. It is called a \term{min-tradeoff function} for a sequence of EAT channels $\{\map_j\}_{j\in\upto{n}}$ if for any distribution $q$ on $\mathcal{D}$, we have 
\begin{align}
\fmin(q) \leq \inf_{\sigma \in \Sigma_j(q)} H(S_j|T_jR)_\sigma \qquad \forall j\in\upto{n},
\end{align}
where $\Sigma_j(q)$ denotes the set of states of the form $(\map_j\otimes\idmap_{R})(\omega_{R_{j-1}R})$ such that the reduced state on $D_j$ has distribution $q$. 
Analogously, a real-valued affine function $\fmax$ defined on probability distributions over $\mathcal{D}$ is called a \term{max-tradeoff function} if 
\begin{align}
\fmax(q) \geq \sup_{\sigma \in \Sigma_j(q)} H(S_j|T_jR)_\sigma \qquad \forall j\in\upto{n}.
\end{align}
The infimum and supremum of an empty set are defined as $+\infty$ and $-\infty$ respectively.
\end{definition}
\noindent With these definitions, we can now state the theorem:
\begin{fact}\label{fact:EAT}
(Entropy accumulation theorem~\cite{DF19,LLR+21}) 
Consider a sequence of EAT channels $\{\map_j\}_{j\in\upto{n}}$ and a normalized state of the form $\rho = (\map_n \circ\dots \circ\map_1 \otimes \idmap_E)\left(\rho^0_{R_0 \allE}\right)$ satisfying the Markov conditions
\begin{align}
I(S_{\upto{j-1}} : T_j | T_{\upto{j-1}} \allE)_\rho 
= 0 \qquad \forall j\in\upto{n}.
\label{eq:markov}
\end{align}

\noindent Let $\fmin$ be a min-tradeoff function for $\{\map_j\}_{j\in\upto{n}}$ and consider any $h\in\mathbb{R}, \eps\in(0,1), \alpha\in(1,2)$. Then for any event $\Omega \subseteq \mathcal{D}^n$ such that 
$\fmin(\freq_{\str{d}}) \geq h$ for all $\str{d}\in\Omega$, we have\footnote{This expression differs slightly from those in~\cite{DF19,LLR+21} because we have not performed the simplifications based on $\smf{\eps} \leq \log(2/\eps^2)$ and $\alpha < 2$ (though the former is only a very small improvement for typical values of $\eps$).
We also remark that in~\cite{LLR+21}, a modification was made to the EAT to improve its dependence on the $\Var$ term, which we have not included here as it involves a further optimization that would significantly increase the complexity of our keyrate computations. 
(In an earlier version of this work, we stated that this modification did not make a difference here as our choice of affine $\fmin$ is essentially equal to the tight bound on $H(S_j|T_jR)$ in our scenario; however, this was in error as this equality does not hold on e.g.~any parts of the domain where $\fmin$ is negative.)}
\begin{align}
\Hmin^{\eps}(\str{S}|\str{T}\allE)_{\rho_{|\Omega}} 
>nh - n\frac{(\alpha-1)\ln 2}{2}V^2 - n(\alpha-1)^2K_\alpha 
- \frac{\smf{\eps}}{\alpha-1} - \frac{\alpha}{\alpha-1}\log\frac{1}{\pr{\Omega}}, 
\end{align}
where $\smf{\eps}$ is as defined in Eq.~\eqref{eq:VK}, and
\begin{align}
\begin{aligned}
V &\defvar \sqrt{\Var_{\mathcal{Q}}(\fmin)+2} + \log
(2\dim(S_j)^2+1),
\\
K_\alpha &\defvar \frac{2^{(\alpha-1)(2\log\dim(S_j) + \Max(\fmin)-\Min_{\mathcal{Q}} (\fmin))} }{6(2-\alpha)^3\ln2}
\ln^3\left(2^{2\log\dim(S_j) + \Max(\fmin)-\Min_{\mathcal{Q}} (\fmin)} + e^2\right),
\end{aligned}
\end{align}
with $\mathcal{Q}$ being the set of all distributions on $\mathcal{D}$ that could be produced by applying some EAT channel to some state.
\end{fact}
\noindent Informally, the Markov conditions
impose the requirement that the register $T_j$ does not ``leak any information'' about the previous registers $S_{\upto{j-1}}$ beyond what is already available from $T_{\upto{j-1}} \allE$. (Without this Markov condition, one could for instance have channels such that $T_j$ simply contains a copy of $S_{j-1}$, in which case it would still be possible to derive a nontrivial min-tradeoff function, but the conclusion of the entropy accumulation theorem would be trivially false.) There is also an analogous EAT statement regarding the max-entropy~\cite{DFR20}, using a max-tradeoff function, though we will be using it slightly differently and will elaborate further on it at that point.

We now describe how the EAT can be used to prove Theorem~\ref{th:rawHmin}. First, note that to prove Theorem~\ref{th:rawHmin} it would be sufficient to consider only the registers $\str{A}\Btstr\str{X}\str{Y}\Ctstr \allE$ of the state $\rho$ (the conditioning event $\OpPE$ is determined by $\Ctstr$ alone). The reduced state on these registers is the same as that at the point when Step~\ref{step:EATmap} of Protocol~\ref{prot:virtual} has finished looping over $j$, since the subsequent steps do not change these registers, and thus we can equivalently study that state in place of $\rho_{|\OpPE}$.
From this point onwards, all smoothed min- or max-entropies refer to that state conditioned on $\OpPE$ (and normalized),
hence for brevity we will omit the subscript specifying the state.

Each iteration of Step~\ref{step:EATmap} of Protocol~\ref{prot:virtual} can be treated as a channel $\map_j$ in a sequence of EAT channels, by considering it to be a channel performing the following operations:
\begin{enumerate}
\item Alice generates $X_j$ as specified in Step~\ref{step:EATmap}.
Conditioned on the value of $X_j$, Alice's device performs some measurement on its share of the stored quantum state $R_{j-1}$ (which includes any memory retained from previous rounds), then performs sifting and noisy preprocessing on the outcome, storing the final result in register $A_j$. 
\item Bob's device behaves analogously, producing the registers $Y_j$ and $\Bt_j$ (we will not need to consider $B_j$).
\item The value of $\Ct_j$ is computed from $A_j \Bt_j X_j Y_j$.
\end{enumerate}
We highlight that in the above description of $\map_j$, the only ``unknowns'' are the measurements it performs on the input state on $R_{j-1}$ --- all other operations are taken to be performed in trusted fashion. 
(This is reasonable because these measurements and the stored states are the only untrusted aspects in the true protocol.) If we had simply considered completely arbitrary channels $\map_j$ producing the respective registers, it would not be possible to 
make a nontrivial security statement about the output.

Identifying $\Ct_j$ with $D_j$, $A_j\Bt_j$ with $S_j$, and $X_jY_j$ with $T_j$ in Definition~\ref{def:EATchann}, we see that these channels $\map_j$ indeed form a valid sequence of EAT channels: $\Ct_j$ is determined from $A_j \Bt_j X_j Y_j$ in the manner specified by Eq.~\eqref{eq:nodisturb}. Additionally, the state they produce always fulfills the Markov conditions, because the values of $X_jY_j$ in each round are generated independently of all preceding registers.

Intuitively, it seems that we could now use the EAT to bound $\Hmin^{\es}(\str{A}|\str{X}\str{Y} \allE)$.
However, there is a technical issue: to apply the EAT, the event $\OpPE$ must be defined entirely in terms of the (classical) registers that appear in the smoothed min-entropy term that we are bounding, which is not a condition satisfied by the registers $\str{A}\str{X}\str{Y}$ alone. This is where the register $\Btstr$ comes into play, following the same approach as~\cite{ARV19}: by a chain rule for the min- and max-entropies (\cite{VDT13} or \cite{Tom16} Eq.~(6.57)), 
we have for any $\es' + 2\es'' < \es$:
\begin{align}
\Hmin^{\es}(\str{A}|\str{X}\str{Y} \allE) &\geq \Hmin^{\es'}(\str{A}\Btstr|\str{X}\str{Y} \allE) - \Hmax^{\es''}(\Btstr|\str{A}\str{X}\str{Y} \allE) - 3\smf{\es-\es'-2\es''} \nonumber\\
&\geq \Hmin^{\es'}(\str{A}\Btstr|\str{X}\str{Y} \allE) - \Hmax^{\es''}(\Btstr|\str{X}\str{Y} \allE) - 3\smf{\es-\es'-2\es''}.
\label{eq:chain}
\end{align}
where the second line holds because the smoothed max-entropy satisfies a data-processing inequality (see e.g.~\cite{Tom16} Theorem~6.2).

The $\Hmax^{\es''}(\Btstr|\str{X}\str{Y} \allE)$ term admits a fairly simple bound, as follows: consider a sequence of EAT channels $\widetilde{\map}_j$ that are identical to $\map_j$ except that they do not produce the registers $A_j\Ct_j$. As before, these maps obey the required Markov conditions.
In addition, recall that for every round the register $\Bt_j$ is set deterministically to $0$ whenever $y\in\inYg$ (which happens with probability $1-\gamma$), hence we always have
\begin{align}
H(\Bt_j|X_j Y_j R)_{(\widetilde{\map}_j\otimes\idmap_{R})(\omega_{R_{j-1}R})} 
= \sum_{y} \pr{Y_j=y} H(\Bt_j|X_j R;Y_j =y)_{(\widetilde{\map}_j\otimes\idmap_{R})(\omega_{R_{j-1}R})} 
\leq \gamma .
\label{eq:fmax}
\end{align}
This means we can apply the max-entropy version of the EAT\footnote{Here we shall use the results from~\cite{DFR20}, because for a constant tradeoff function, this turns out to yield a slightly better bound as compared to the version of the EAT~\cite{DF19} stated here.
Strictly speaking, the reasoning used here is not a direct application of the EAT, because once again, the event $\OpPE$ is not defined on the registers $\Btstr\str{X}\str{Y}$ alone (attempting to address this by including $\str{A}$ in the conditioning registers could result in the Markov conditions not being fulfilled). Fortunately, the bound \eqref{eq:fmax} holds for our maps $\widetilde{\map}_j$ even without a constraint on the output distribution. Hence the reasoning is as follows, in terms of the equations and lemmas in~\cite{DFR20}: first apply Eq.~(32) \emph{without} the event-conditioning term (this is valid since \eqref{eq:fmax} holds without constraints), {then} condition on $\OpPE$ using 
Lemma~B.6 (noting that $\rho_{\Btstr\str{X}\str{Y} \allE} 
= \pr{\OpPE}(\rho_{|\OpPE})_{\Btstr\str{X}\str{Y} \allE} + \pr{\OpPEnot}(\rho_{|\OpPEnot})_{\Btstr\str{X}\str{Y} \allE}$),
and finally apply Lemma~B.10 to obtain Eq.~\eqref{eq:Hmaxbnd}. 
(Alternatively, one could use $\Ctstr$ instead of $\Btstr$. This was done in~\cite{MvDR+19} to slightly improve the bound in the block analysis, but it does not make a difference in our analysis. However, using $\Ctstr$ would seem to make it harder to sharpen the slightly crude bound used to obtain Eq.~\eqref{eq:gproof}.)
} with a \emph{constant} max-tradeoff function of value $\gamma$. Letting $V'=2\log(1+2
\dim(\Bt)
) = 2\log5$, 
this yields the following bound for any $\alpha' \in (1,1+2/V')$:
\begin{align}
\Hmax^{\es''}(\Btstr|\str{X}\str{Y} \allE) < n\gamma + n\left(\frac{\alpha'-1}{4}\right)V'^2 + \frac{\smf{\es''}}{\alpha'-1} + \frac{\alpha'}{\alpha'-1}\log\frac{1}{\pr{\OpPE}}.
\label{eq:Hmaxbnd}
\end{align}

The bulk of our task is to bound the $\Hmin^{\es'}(\str{A}\Btstr|\str{X}\str{Y} \allE)$ term. To do so, we will need an appropriate min-tradeoff function, which we shall now construct.

\subsubsection{Min-tradeoff function}
\label{sec:fmin}

Consider an arbitrary state of the form $(\map_j\otimes\idmap_{R})(\omega_{R_{j-1}R})$.
In this section, all entropies will be computed with respect to this state, and hence for brevity we will omit the subscript specifying the state.

We first note that
\begin{align}
H(A_j \Bt_j | X_j Y_j \Fj R) 
&= \frac{1-\gamma}{4} 
\sum_{z\in\inX} H(A_j | \Fj R;X_j = Y_j = z)  \nonumber \\
& \qquad + \frac{\gamma}{4} 
\sum_{y\in\inYt} \sum_{x\in\inX} 
H(A_j \Bt_j | \Fj R;X_j = x, Y_j = y), \label{eq:1rndmix}
\end{align}
where we have used the fact that $H(A_j \Bt_j | \Fj R;X_j = x, Y_j = y) = 0$ when $(x,y)=(0,1) \text{ or } (1,0)$, and $H(A_j \Bt_j | \Fj R;X_j = x, Y_j = y) = H(A_j | \Fj R;X_j = x, Y_j = y)$ when $(x,y)=(0,0) \text{ or } (1,1)$.

Let $w$ denote the probability that the state wins the CHSH game, conditioned on the game being played. 
Then by applying the simple but somewhat crude bound\footnote{This marks a point where the analysis could be slightly sharpened, in that if we had a tight bound on the ``two-party entropies'' $H(A_j \Bt_j | \Fj R;X_j = x, Y_j = y)$ rather than just the ``one-party entropies'' $H(A_j | \Fj R;X_j = x, Y_j = y)$, we could improve the second term in~\eqref{eq:gproof}. However, note that it would only improve the keyrate by $O(\gamma)$, because of the $\gamma$ prefactor on that term.} $H(A_j \Bt_j | \Fj R;X_j = x, Y_j = y) \geq H(A_j | \Fj R;X_j = x, Y_j = y)$ to the terms in the second sum in Eq.~\eqref{eq:1rndmix}, we get the bound\footnote{Recall that noisy preprocessing is not applied to the rounds with $Y_j \in \inYt$.}
\begin{align}
H(A_j \Bt_j | X_j Y_j \Fj R) 
\geq \frac{1-\gamma}{2} \lin_\p(w) + \frac{\gamma}{2} \sum_{y\in\inYt} \lin_0(w) 
= \g(w).
\label{eq:gproof}
\end{align}

We can now use the function $\g$ to construct a min-tradeoff function $\fmin$, with the domain of $\fmin$ being distributions on $\Ct_j$ (recall that this register is set to $\perp$ if $Y_j\in\inYg$,
and otherwise is set to $0$ or $1$ if the CHSH game is lost or won respectively). First observe that the channel is an \term{infrequent-sampling channel} in the sense described in~\cite{DF19,LLR+21}. By the argument in~\cite{LLR+21}, a valid min-tradeoff function $\fmin$ for the channel is given by 
the (unique) affine function specified by the following values (in a minor abuse of notation, here we interpret $\g$ as a function of a distribution instead of a winning probability):
\begin{align}
\fmin(\delta_c) = 
\begin{cases} 
\frac{1}{\gamma} \g(\delta_c) + \left(1-\frac{1}{\gamma}\right)\cperp & \text{ if } c\neq\perp \\
\cperp & \text{ if } c = \perp
\end{cases}
\, ,
\label{eq:fmin}
\end{align}
where $\cperp\in[\g\left(\delta_0\right),\g\left(\delta_1\right)]$ is a constant that can be chosen to optimize the keyrate. (Intuitively, this function is constructed simply by noting that the maps $\map_j$ can only produce distributions that lie in the slice of the probability simplex specified by the constraint 
$\mathsf{P} 
[\Ct_j = \perp] = 1-\gamma$, and hence the min-tradeoff function is free to take any value for distributions outside of this slice, recalling that we take the infimum of an empty set to be $+\infty$. For distributions within this slice, we know that $\g$ is an affine lower bound on the entropy as a function of the winning probability, and hence we can just set $\fmin$ equal to $\g$ (up to a domain rescaling) on this slice. Any $\fmin$ constructed this way is precisely of the form described in Eq.~\eqref{eq:fmin}, with $\cperp$ being a constant determining its value on all distributions outside of the $\mathsf{P}[\Ct_j = \perp] = 1-\gamma$ slice.)

As shown in~\cite{LLR+21}, the min-tradeoff function constructed this way satisfies
\begin{align}
\begin{gathered}
\Max(\fmin) =
\max\left\{
\frac{1}{\gamma}\Max(\g) + \left(1-\frac{1}{\gamma}\right)\cperp
, \cperp \right\},
\qquad
\Min_{\mathcal{Q}_{f}} (\fmin) = \Min_{\mathcal{Q}_{\g}} (\g), \\
\Var_{\mathcal{Q}_{f}}(\fmin) \leq \sup_{q\in\mathcal{Q}_{\g}} \sum_{c\in\{0,1\}} \frac{q(c)}{\gamma} \left(\cperp - \g(\delta_c)\right)^2,
\end{gathered}
\end{align}
where $\mathcal{Q}_{f}$ denotes
the set of distributions on $\Ct_j$ such that $\mathsf{P}[\Ct_j = \perp] = 1-\gamma$ and $\mathsf{P}[\Ct_j = 1] \in [\gamma\wmin, \gamma\wmax]$, while $\mathcal{Q}_{\g}$ denotes
the set of all distributions on 
the alphabet $\{0,1\}$
such that $\mathsf{P}[1] \in [\wmin,\wmax]$.

For the specific $\g$ and range of $\cperp$ that we consider, 
these expressions simplify to Eq.~\eqref{eq:minmaxvarbounds}, where 
we have solved the optimization $\sup_{q\in\mathcal{Q}_{\g}}$ in the bound on $\Var_{\mathcal{Q}_{f}}(\fmin)$ by observing that it is an affine function of the distribution $q$, and the set $\mathcal{Q}_{\g}$ we use here is essentially
a line segment (in a $1$-dimensional probability simplex).

\subsubsection{Final min-entropy bound}

The event $\OpPE$ is defined by the conditions
$\freq_\ctstr(1)\geq (\wexp-\dtol)\gamma$ and 
$\freq_\ctstr(0)\leq (1-\wexp+\dtol)\gamma$. Hence for all 
$\ctstr \in \OpPE$,
we have (since $\fmin$ is affine)
\begin{align}
\fmin(\freq_\ctstr)
&= \freq_\ctstr(0) \left(\frac{1}{\gamma} \g\left(\delta_0\right) +  \left(1-\frac{1}{\gamma}\right)\cperp\right)
+ \freq_\ctstr(1) \left(\frac{1}{\gamma} \g\left(\delta_1\right) +  \left(1-\frac{1}{\gamma}\right)\cperp\right)
+ \freq_\ctstr(\perp) \cperp \nonumber\\
&= \freq_\ctstr(0) \left(\frac{1}{\gamma} \g\left(\delta_0\right) -\frac{1}{\gamma}\cperp\right)
+ \freq_\ctstr(1) \left(\frac{1}{\gamma} \g\left(\delta_1\right) -\frac{1}{\gamma}\cperp\right)
+ \cperp \nonumber\\
&\geq (1-\wexp+\dtol) \left(\g\left(\delta_0\right) - \cperp\right)
+ (\wexp-\dtol) \left(\g\left(\delta_1\right) - \cperp\right)
+ \cperp \nonumber\\
&= (1-\wexp+\dtol)\g\left(\delta_0\right)
+ (\wexp-\dtol)\g\left(\delta_1\right)
\nonumber\\
&= \g(\wexp-\dtol),
\label{eq:fminPEineq}
\end{align}
where the inequality holds because $\cperp\in[\g\left(\delta_0\right),\g\left(\delta_1\right)]$, and in the last line we use the fact that $\g$ is affine and revert to interpreting it as a function of winning probability. (We remark that if we fix $\cperp=\g(1)$, then in fact this inequality can be derived using only the $\freq_\str{c}(0)$ condition, following~\cite{DF19}. Hence in principle one could sacrifice the option of optimizing $\cperp$ in exchange for reducing the number of checks to perform in the protocol, which improves the completeness parameters.)

Therefore, we can choose $h=\g(\wexp-\dtol)$ in the EAT statement (Fact~\ref{fact:EAT}) to conclude that the state conditioned on $\OpPE$ satisfies
\begin{align}
\Hmin^{\es'}(\str{A}\Btstr | \str{X}\str{Y} \allE) &> n
\g(\wexp-\dtol) 
- n\frac{(\alpha-1)\ln 2}{2}V^2 - n(\alpha-1)^2K_\alpha 
\nonumber \\
&\qquad 
- \frac{\smf{\es'}}{\alpha-1} - \frac{\alpha}{\alpha-1}\log\frac{1}{\pr{\OpPE}}, \label{eq:EATbound}
\end{align}
where $\smf{\eps},V,K_\alpha$ are as defined in Eq.~\eqref{eq:VK}. Putting this together with Eqs.~\eqref{eq:chain} and~\eqref{eq:Hmaxbnd}, we finally obtain the bound in Theorem~\ref{th:rawHmin}.

\begin{figure}
\centering
\subfloat[\cite{HBD+15} parameters]{
\includegraphics[width=0.48\textwidth]{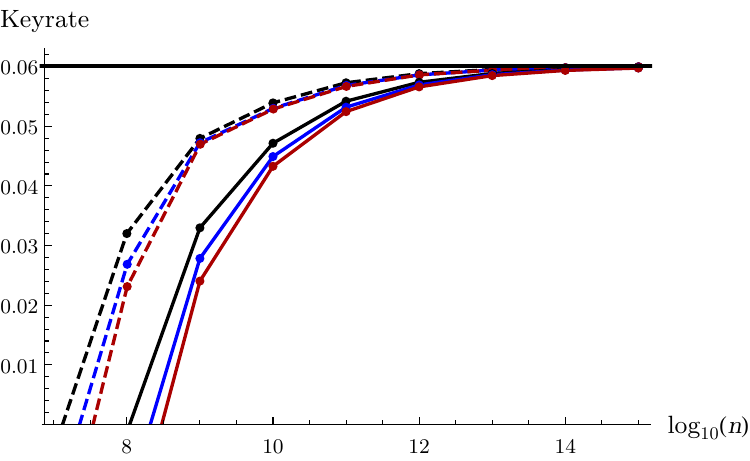}
} 
\subfloat[\cite{RBG+17} parameters]{
\includegraphics[width=0.48\textwidth]{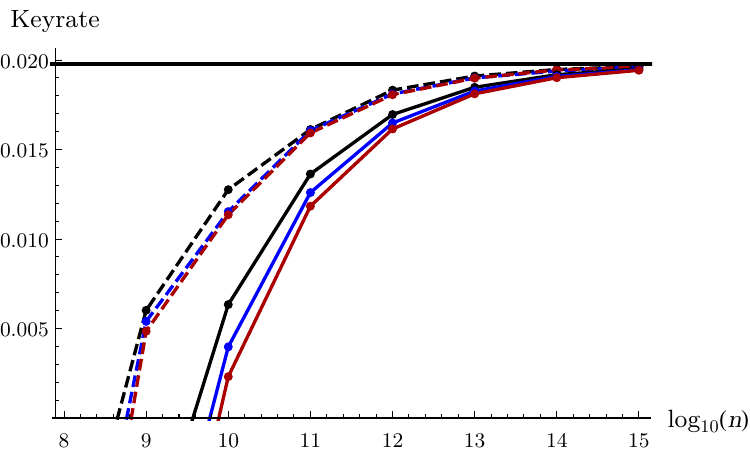}
}
\caption{(From~\cite{arx_TSB+20}) Finite-size keyrates as a function of number of rounds in Protocol~\ref{prot:DIQKD}
without noisy preprocessing ($\p=0$), 
for honest devices following the estimated parameters in~\cite{MvDR+19} for the Bell tests in~\cite{HBD+15} and~\cite{RBG+17}.
The solid curves show the results for coherent attacks (Theorem~\ref{th:DIQKD}), while the dashed curves show the results under the assumption of collective attacks (Theorem~\ref{th:collective}), with the error-correction protocol taken to satisfy Eqs.~\eqref{eq:optEC}--\eqref{eq:ECrate}. The colours correspond to soundness parameters of $\esound=10^{-3}$, $10^{-6}$, and $10^{-9}$ for black, blue, and red respectively, while the completeness parameter is $\ecom=10^{-2}$ in all cases. The horizontal line denotes the asymptotic keyrate. 
All other parameters in Theorems~\ref{th:DIQKD} and \ref{th:collective} were numerically optimized, except $\cperp$. 
}
\label{fig:experiments}
\vspace{2cm}
\includegraphics[width=0.6\textwidth]{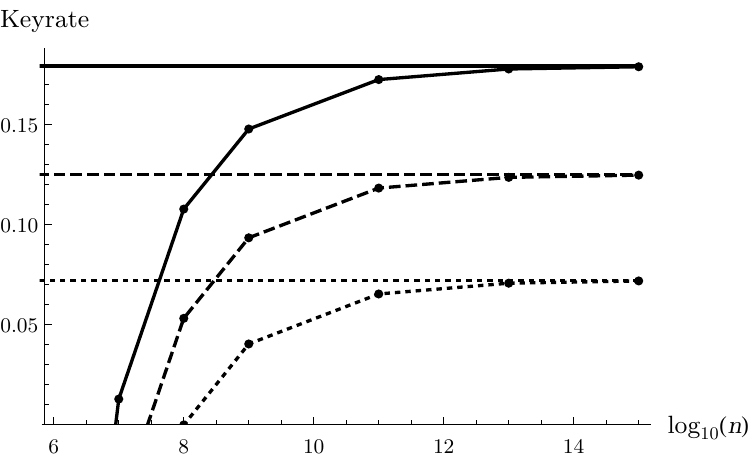}
\caption{(From~\cite{arx_TSB+20}) Finite-size keyrates secure against coherent attacks (Theorem~\ref{th:DIQKD}) as a function of number of rounds in Protocol~\ref{prot:DIQKD}
without noisy preprocessing ($\p=0$), where the honest devices are described by depolarizing noise $\q$ (see Sec.~\ref{sec:hon}). The solid, dashed and dotted curves denote $\q=5\pct$, $6\pct$ and $7\pct$ respectively, with the error-correction protocol taken to satisfy Eqs.~\eqref{eq:optEC}--\eqref{eq:ECrate}. The soundness parameter is $\esound=10^{-6}$ and the completeness parameter is $\ecom=10^{-2}$. The horizontal lines denote the asymptotic keyrates. All other parameters were numerically optimized.
}
\label{fig:depol}
\end{figure}

\section{Resulting finite-size keyrates}
\label{sec:finiteplots}

In Figs.~\ref{fig:experiments} and~\ref{fig:depol}, we plot some finite-size keyrates given by our security proofs, for the various types of honest devices described in Sec.~\ref{sec:hon}. Fig.~\ref{fig:experiments} corresponds to honest devices with performance described by the estimated parameters in~\cite{MvDR+19} for the NV-centre~\cite{HBD+15} and cold-atom~\cite{RBG+17} loophole-free Bell tests. (Noisy preprocessing was not applied for these plots because it appears to only slightly improve the keyrates for those experimental parameters; see Fig.~\ref{fig:protIID} later.)
Fig.~\ref{fig:depol} shows the finite-size keyrates for honest devices subject to depolarizing noise. For reference, we also plot the finite-size keyrates of Protocol~\ref{prot:DIQKD} under the assumption of collective attacks (Theorem~\ref{th:collective} in the subsequent sections).

From Fig.~\ref{fig:experiments}, we see that our security proof requires the \cite{HBD+15} and \cite{RBG+17} experiments to run for approximately $n\sim10^8$ and $10^{10}$ rounds respectively in order to certify a positive finite-size keyrate against coherent attacks, or roughly 1--2 orders of magnitude smaller for collective attacks (we derive the keyrate for that scenario in Sec.~\ref{sec:coll}). While this is a marked improvement over the basic \cite{PAB+09}~protocol (which yields zero asymptotic keyrate for those experiments), these finite-size requirements still appear to be outside the reach of those experimental implementations.
We also see that changing the security parameters by several orders of magnitude only results in fairly small changes to the keyrate, so it appears unlikely that the keyrates could be substantially improved by relaxing these security requirements.

\section{Possible modifications}
\label{sec:mods}

To improve on these results, we now describe some possible modifications of Protocol~\ref{prot:DIQKD}. Namely, Protocol~\ref{prot:preshared} is a modification based on a pre-shared key, which achieves a \emph{net} key generation rate approximately double that of Protocol~\ref{prot:DIQKD}, by overcoming a crucial disadvantage of random-key-measurement protocols (namely, the sifting factor). Also, Protocol~\ref{prot:optcoll} is a modification that is optimized for the collective-attacks assumption --- by changing the protocol itself for this scenario rather than just the security proof, we can further improve the keyrate to some extent. 

\subsection{Coordinating input choices by public communication}
\label{sec:pubchann}

The random-key-measurement protocol has the drawback that the keyrate is effectively halved, since the generation rounds have ``mismatched'' inputs approximately half the time. It would be helpful to find a way to work around this issue. One possible approach could be to observe that in~\cite{ARV19}, it was assumed that the following operations can be performed in each round: the devices receive some shares of a quantum state, then\footnote{The assumption being made here is that since the quantum states are now entirely in the possession of Alice and Bob, the test/generation decision can no longer affect the distributed state --- if the distributed state could depend on whether it is a test or generation round, the protocol would be trivially insecure.} Alice and Bob publicly communicate to come to an agreement on \emph{both} of their input choices, and finally they supply these inputs to their devices. (This was necessary in~\cite{ARV19} because 
Alice and Bob's actions in that DIQKD protocol require both of them to know whether it is a test or generation round. In fact, our analysis can be viewed as the first EAT-based security proof for a ``genuinely sifting-based'' DIQKD protocol, in the sense that Alice and Bob do not coordinate which rounds are test rounds, and simply choose their inputs independently.)
If we assume that this is also possible in our scenario, then Alice and Bob could coordinate their inputs in the generation rounds instead of choosing them independently, thereby avoiding the sifting factor.

Unfortunately, it does not seem clear if such a proposal is entirely plausible in near-term experimental implementations. This is because it relies on the devices being able to store the quantum state for long enough for Alice and Bob to agree on their choice of inputs,
which is potentially challenging for current Bell-test implementations. As an alternative, we propose the following potential modification to the DIQKD protocol in~\cite{ARV19} --- instead of agreeing on the test rounds via public communication, Alice and Bob could use a small amount of pre-shared key to choose which rounds are test rounds, in the same way as in DIRE (for details on the amount of pre-shared key required, see the DIRE protocols in~\cite{ARV19,BRC20} or the discussion in Sec.~\ref{sec:preshared} below). This approach would essentially be a ``key {expansion}'' protocol that requires a small amount of pre-shared key to initialize. We remark that this is not a dramatic change in perspective, because a common method to authenticate channels (namely, message authentication codes) relies on having a small amount of pre-shared key, so the assumed existence of authenticated channels in the DIQKD protocol is likely to require some pre-shared key in any case. 

However, this basic notion cannot immediately be generalized to Protocol~\ref{prot:DIQKD} here, since requiring Alice and Bob to choose uniformly distributed ``matching'' inputs in the generation rounds would require a large amount of pre-shared key (roughly $(1-\gamma)n$ bits), likely exceeding the length of the key generated by the protocol. Fortunately, in the following section, we propose a variation which overcomes this difficulty by ``recovering'' the entropy in the pre-shared key, thereby still achieving net key expansion.

\subsection{Protocol using pre-shared key}
\label{sec:preshared}

Here we describe a variant protocol that avoids the sifting factor \emph{without} requiring the brief quantum storage described above, through the use of a fairly long pre-shared key. (This idea was also independently proposed in~\cite{arx_BRC21}.) The limitation of this variant is that the net increase in secret key is just a (constant) fraction of the amount of pre-shared key; however, the \emph{rate} of net key generation does not have the sifting factor of $1/2$.
Informally, the idea is to simply use the pre-shared key as Alice's input string $\str{X}$, which allows Bob to choose his generation inputs to match Alice's. Just as importantly, this also allows them to (almost) entirely omit the public announcement of their inputs --- hence $\str{X}$ remains private, and with some care it can be incorporated into the final key without losing the entropy it ``contains''. 

We now describe this idea in detail as Protocol~\ref{prot:preshared} below, followed by its security proof. The protocol supposes that Alice and Bob hold a pre-shared (uniform) key of $n$ bits, which we shall simply denote as $\str{X}$, since it will be exactly the string that Alice uses as her inputs. The appropriate value of $\lkey$ to choose will be described later in Theorem~\ref{th:preshared}. 
\begin{savenotes}
\begin{breakablealgorithm}
\caption{} 
\label{prot:preshared}
This protocol proceeds the same way as Protocol~\ref{prot:DIQKD}, except for the following changes:
\begin{itemize}
\item In each round, Alice's input $X_j$ is determined from the pre-shared key $\str{X}$, instead of being generated randomly in that round. Bob's input $Y_j$ is generated as follows: with probability $\gamma$ he chooses a uniformly random $Y_j \in \inYt$, otherwise Bob chooses $Y_j = X_j$. In addition, he generates another register $\Yt_j$ which equals $Y_j$ when $Y_j \in \inYt$ and equals $\perp$ otherwise.
\item Alice and Bob do not publicly announce the strings $\str{X}\str{Y}$. Instead, Bob only announces the string $\Ytstr$.\footnote{It might be possible to consider a slight variant which omits this step. However, knowing $\Ytstr$ allows Alice to compute $\str{Y}$, which may be relevant for error correction since it allows Alice to distinguish the test and generation rounds. In any case, it seems unclear whether the entropy of $\Ytstr$ can be usefully extracted even if it is kept secret, since Alice does not have access to it in that case.} Additionally, the sifting step is unnecessary, since there will be no rounds such that $Y_j \in \inYg$ and $X_j \neq Y_j$.
\item Privacy amplification is performed on the strings $\str{A}\str{X}$ and $\tilde{\str{A}}\str{X}$ instead of $\str{A}$ and $\tilde{\str{A}}$.
\end{itemize}
\end{breakablealgorithm}
\end{savenotes}

To prove the security of this protocol, we can simply follow almost exactly the same security proof as for Protocol~\ref{prot:DIQKD}, with some changes we shall now describe. Firstly, the value of $h_\mathrm{hon}$ (to be used when computing $\cmax$) is replaced by
\begin{align}
\tilde{h}_\mathrm{hon} 
= H(A_j|B_jX_jY_j)_\mathrm{hon} 
&= \frac{1-\gamma}{2} 
\sum_{z\in\inX} H(A_j | B_j;X_j = Y_j = z)_\mathrm{hon}  
\nonumber \\& \qquad 
+ \frac{\gamma}{4} 
\sum_{x\in\inX,y\in\inYt} 
H(A_j| B_j ;X_j = x, Y_j = y)_\mathrm{hon}, \label{eq:ECpreshared}
\end{align}
since the probabilities of $X_j = Y_j = z$ for $z\in\inX$ are now ${(1-\gamma)}/{2}$. (Note that no error-correction information needs to be sent from Alice to Bob regarding $\str{X}$, since both of them have a copy of that string.)

Also, since the strings used in the privacy-amplification step are now $\str{A}\str{X}$ and $\tilde{\str{A}}\str{X}$, this means that we need an equivalent of Eq.~\eqref{eq:EATbound}, with $\Hmin^{\es'}(\str{A}\Btstr\str{X}|\Ytstr E)$ in place of $\Hmin^{\es'}(\str{A}\Btstr | \str{X}\str{Y}E)$. To obtain this, we note that we can simply construct a virtual protocol in the analogous way to Protocol~\ref{prot:virtual}, then consider the same EAT channels $\map_j$ as before, but instead we shall identify $\Ct_j$ with $D_j$, $A_j\Bt_jX_j$ with $S_j$, and $\Yt_j$ with $T_j$ in Definition~\ref{def:EATchann}.
The Markov conditions are again fulfilled, since $\Yt_j$ is generated by trusted randomness in each round and independent of all previous data. To find an appropriate min-tradeoff function for these channels, we note that the output $(\map_j\otimes\idmap_{R})(\omega_{R_{j-1}R})$ of channel $\map_j$ always satisfies
\begin{align}
H(X_j | \Yt_j R) = H(X_j) = 1,
\end{align}
because $X_j$ is produced by trusted randomness independent of $\Yt_j R$. Therefore, we can use the chain rule to write
\begin{align}
H(A_j \Bt_j X_j | \Yt_j R) 
&= H(X_j | \Yt_j R) + H(A_j \Bt_j | X_j \Yt_j \Fj R) \nonumber \\ 
&\geq 1 + (1-\gamma) \lin_\p(w) + \gamma \lin_0(w) \nonumber\\
&\defvar 1 + \tilde{\g}(w), 
\label{eq:chaintrick}
\end{align}
where the function $\tilde{\g}$ (in contrast to $\g$) does {not} have the factor of $1/2$ introduced by sifting, since Alice does not ``erase'' the outputs of any rounds. We can thus construct a new min-tradeoff function $\tilde{f}_\mathrm{min}$ in the same way as in Sec.~\ref{sec:fmin}, but using $1 + \tilde{\g}(w)$ in place of $\g(w)$.\footnote{We remark that if we think of this replacement as happening in two steps, first replacing $\g$ by $\tilde{\g}$ and then adding a ``constant offset'' of $1$, then the latter has no effect on $\Var_{\mathcal{Q}_{f}}(\fmin)$ or the difference $\Max(\fmin)-\Min_{\mathcal{Q}_{f}} (\fmin)$, and hence does not change the finite-size correction to the keyrate except indirectly via changing the system dimensions and the range of $\cperp$. However, the first step of replacing $\g$ by $\tilde{\g}$ does slightly increase the finite-size correction (since $\tilde{\g}$ has a somewhat larger range).} The rest of the proof then proceeds as before, leading to the following security statement:
\begin{theorem}\label{th:preshared}
Protocol~\ref{prot:preshared} 
has the same security guarantees as those described in Theorem~\ref{th:DIQKD},
except with the following changes (with $\tilde{\g}$ being defined in~\eqref{eq:chaintrick}):
\begin{itemize}
\item $\cperp$ is chosen to be in 
$[1+\tilde{\g}(0),1+\tilde{\g}(1)]$ 
instead.
\item $h_\mathrm{hon}$ is replaced by $\tilde{h}_\mathrm{hon}$ as specified in Eq.~\eqref{eq:ECpreshared}.
\item In Eq.~\eqref{eq:keylength} for $\lkey$, $\g(\wexp-\dtol)$ is replaced by $1 + \tilde{\g}(\wexp-\dtol)$, and the values of $V$ and $K_\alpha$ are replaced by 
\begin{align}
\begin{aligned}
\tilde{V} &\defvar \sqrt{\Var_{\mathcal{Q}_{f}}(\tilde{f}_\mathrm{min})+2} + \log
129,
\\
\tilde{K}_\alpha &\defvar \frac{2^{(\alpha-1)(2\log8 + \Max(\tilde{f}_\mathrm{min})-\Min_{\mathcal{Q}_{f}} (\tilde{f}_\mathrm{min}))} }{6(2-\alpha)^3\ln2}
\ln^3\left(2^{2\log8 + \Max(\tilde{f}_\mathrm{min})-\Min_{\mathcal{Q}_{f}} (\tilde{f}_\mathrm{min})} + e^2\right),
\end{aligned}
\end{align}
where $\tilde{f}_\mathrm{min}$ is a function that satisfies
\end{itemize}
\vspace*{-.5cm}
\begin{align}
\begin{gathered} 
\Max(\tilde{f}_\mathrm{min})
= 1 + \frac{1}{\gamma}\tilde{\g}(1) + \left(1-\frac{1}{\gamma}\right)\cperp, \qquad
\Min_{\mathcal{Q}_{f}} (\tilde{f}_\mathrm{min}) 
= 1 + \tilde{\g}(\wmin), \\
\Var_{\mathcal{Q}_{f}}(\tilde{f}_\mathrm{min}) \leq 
\frac{\wmin}{\gamma} 
\min \left\{\Delta_0^2, \Delta_1^2 \right\}
+ \frac{\wmax}{\gamma} 
\max \left\{\Delta_0^2, \Delta_1^2 \right\}
, \text{ where } \Delta_w \defvar \cperp - 1 - \tilde{\g}(w).
\end{gathered}
\end{align}
\end{theorem}

Overall, recalling that Protocol~\ref{prot:DIQKD} required $n$ bits of pre-shared key, we see that the \emph{net} gain of secret key bits in Protocol~\ref{prot:preshared} 
is larger than that of Protocol~\ref{prot:DIQKD} by a factor of approximately (ignoring the changes to the finite-size corrections)
\begin{align}
\frac{n \left(\tilde{\g}(\wexp-\dtol) - \tilde{h}_\mathrm{hon}\right)}{n \Big(\g(\wexp-\dtol) - h_\mathrm{hon}\Big)} \approx 2
,
\end{align}
since it avoids the sifting factor. 
Informally, by keeping $\str{X}$ secret and incorporating it in the privacy amplification step, we have ``recovered'' the entropy that was in the pre-shared key.

In practice, including the string $\str{X}$ in privacy amplification 
essentially doubles the input size for the hash function in that step, which raises its computational difficulty substantially (though not insurmountably).
One might wonder whether it would be possible to bypass this aspect --- for instance, by simply performing privacy amplification on $\str{A}$ and $\tilde{\str{A}}$ as before, then appending $\str{X}$ to the output. At first glance, this approach might appear plausible, since $\str{X}$ is not announced in Protocol~\ref{prot:preshared}. Unfortunately, it seems unclear how to certify that the publicly communicated error-correction string $\str{L}$ is independent of $\str{X}$ (in fact, it seems unlikely that this is true). Hence the idea of simply appending $\str{X}$ may not be secure. By instead incorporating it in privacy amplification in the specified manner, Protocol~\ref{prot:preshared} ensures that the entropy of $\str{X}$ is securely ``extracted'' into the final key.

As previously mentioned, the net increase in secret key given by one instance of Protocol~\ref{prot:preshared} is limited to a fraction of the amount of pre-shared key. 
However, it is possible in principle to recursively run Protocol~\ref{prot:preshared} in order to achieve unbounded key expansion --- one can use the key generated by one instance of Protocol~\ref{prot:preshared} to run it again with a longer pre-shared key and larger $n$ (since the security definition is composable, the soundness parameter 
will only increase 
additively in this process~\cite{arx_PR14,arx_PR21}). We stress that in doing so, one must always incorporate the seed into the privacy-amplification step exactly as specified in Protocol~\ref{prot:preshared} --- in particular, this means that the \emph{entire} key changes with every iteration, instead of simply having some new bits appended.
Some care is necessary regarding device memory across instances of this recursive procedure --- while it does not seem to be directly vulnerable to the memory attack of~\cite{BCK13}\footnote{This is because the only public communication in Protocol~\ref{prot:preshared} that can leak any information is the error-correction string 
(all other public communication is based on trusted randomness). In our security proof, we have bounded the min-entropy leakage at this step simply via the length of this string, without any assumptions about its structure, and hence we can still obtain a secure bound on the min-entropy of the input for privacy amplification in the final protocol instance. Note that this claim is strictly restricted to device reuse following the recursive process specified here --- once any key bits have been used for any other purpose, the attack again becomes a potential concern if the devices are reused.}, it is still important to ensure that the states measured in each instance of the protocol are independent of the key generated in the preceding instance, since this key is used to choose the device inputs (which must be independent of the state in order for our security arguments to hold). Again, this relies on the notion that the registers measured by the devices do not contain information about the key generated in the preceding instance.

There is another potential variant of this idea where a pre-shared key is instead used to generate \emph{both} input strings $\str{X}$ and $\str{Y}$, and the input-choice announcement is omitted entirely, with privacy amplification being performed on $\str{A}\str{X}\str{Y}$ and $\tilde{\str{A}}\str{X}\str{Y}$. This can be done by using $\kappa \binh(\gamma) n$ bits to choose the test rounds approximately according to the desired IID distribution of test rounds, then using $\kappa' \gamma n$ bits to set the value of $Y_j$ in the test rounds, where $\kappa,\kappa'>1$ are constants that can be chosen such that the approximations to the desired distributions are sufficiently accurate (see the randomness-expansion protocols in~\cite{ARV19,BRC20} for a more complete description of this process based on the \term{interval algorithm}), and $n$ bits to set the value of $X_j$ in all rounds.\footnote{We break up the use of the seed into separate processes because it allows for better efficiency as compared to directly approximating the desired distribution of $\str{X}\str{Y}$ --- with the approach we describe, the ``inefficiency'' prefactors $\kappa,\kappa'$ of the interval algorithm only appear on the $\binh(\gamma)n,\gamma n$ terms instead of the full entropy of $\str{X}\str{Y}$.} This would hence require $(1 + \kappa \binh(\gamma) + \kappa' \gamma)n$ bits of seed randomness. A similar argument as above could then be performed by noting that (for $X_jY_j$ generated according to the ideal distribution) we have $H(X_j Y_j) 
= 1 + \binh(\gamma) + \gamma
$, so most of the seed entropy can be ``recovered'', up to the losses from the $\kappa,\kappa'$ factors. However, tracking the effects of using the interval algorithm to approximate the ideal distribution is cumbersome (albeit possible), and it is unclear if this variant offers any immediate advantage over Protocol~\ref{prot:preshared} for DIQKD --- though it may be useful for protocols that use non-uniform input distributions.

On the other hand, it appears that this variant may have potential for the purposes of DIRE instead. The main reason why the random-key-measurement approach in~\cite{SGP+21} could not be easily generalized to DIRE is that in order for Alice to select a uniformly random input in every round, she requires a (local) source of $n$ random bits, which is a free resource in DIQKD but not in DIRE --- if a proposed DIRE protocol consumes more random bits than it produces, then it has failed to achieve randomness \emph{expansion} (though it may still be useful for DIRNG). However, the protocol proposed in this section has the property that it ``recovers'' the entropy contained in the seed, which means that one can afford to use much larger seeds while still obtaining a net increase in secret key. Explicitly, the application of this idea to DIRE would hence be as follows (also proposed independently in~\cite{arx_BRC21}): one begins with $2n$ random bits, which are then used as the input strings\footnote{For DIRE based on the CHSH inequality, Bob only requires two possible measurements instead of the four required for the DIQKD protocol here.} $\str{X}\str{Y}$ to the devices over $n$ rounds to obtain outputs $\str{A}\str{B}$. Modelling this process using EAT channels in a manner similar to above (see e.g.~\cite{LLR+21} for details), for each round we would have 
\begin{align}
H(A_j B_j X_j Y_j| R) 
&= H(X_j Y_j | R) + H(A_j B_j | X_j Y_j R) \nonumber \\ 
&= 2 + H(A_j B_j | X_j Y_j R),
\end{align}
which (given a bound on $H(A_j B_j | X_j Y_j R)$) allows one to bound the smoothed min-entropy of $\str{A}\str{B}\str{X}\str{Y}$ conditioned on $\allE$. By performing privacy amplification on $\str{A}\str{B}\str{X}\str{Y}$, one ``recovers'' all the entropy in the seed, due to the $H(X_j Y_j | R)$ term in the above equation. Overall, this proposed protocol allows one to use the improved entropy rate provided by the random-key-measurement approach~\cite{SGP+21}, in the context of DIRE instead of DIQKD.

\subsection{Collective attacks}
\label{sec:coll}

As a reference to compare our results against, we could consider whether a longer secure key could be obtained under the collective-attacks assumption. To this end, we derive the following theorem, with the proof given in Appendix~\ref{asec:collective}:
\begin{theorem}\label{th:collective}
Take any 
$\ecEC,\ecPE,\ePA,\eh,\es,\eIID \in (0,1]$, $\gamma\in(0,1)$, $\p\in[0,1/2]$, and $\dsou\in[0,\wexp - \dtol)$, such that 
\begin{align}\label{eq:eIID}
\eIID \geq 
\cdfBin{n}{1-(\wexp - \dtol - \dsou)\gamma}{\floor{(1-(\wexp-\dtol)\gamma)n}}.
\end{align}
Under the collective-attacks assumption, Protocol~\ref{prot:DIQKD} is $(\ecEC + \ecPE)$-complete and $(\max\{\eIID, \ePA + 2\es\} + 2\eh)$-sound when performed with any choice of $\cmax$ and $\dtol$ such that Eq.~\eqref{eq:comEC} and Eq.~\eqref{eq:ecPE} hold, and $\lkey$ satisfying
\begin{align}
\lkey \leq
n\g(\wexp-\dtol-\dsou) - \sqrt{n} \,(2\log5)\sqrt{\log\frac{2}{\es^2}}
- \cmax - \ceil{\log\left(\frac{1}{\eh}\right)} - 2\log\frac{1}{\ePA} + 2.
\label{eq:lkeycoll}
\end{align}
\end{theorem}
This bound asymptotically converges to the desired value by an analysis similar to that above~\eqref{eq:asympt}. Explicitly: set all the $\eps$ parameters to some constant values satisfying the desired completeness and soundness bounds, then take
$\gamma=2 n^{-1} \max\{ \dtol^{-2} \log(2/\ecPE), \dsou^{-2} \log(1/\eIID) \}$ to ensure that~\eqref{eq:ecPE} and~\eqref{eq:eIID} are satisfied (the latter claim follows from~\eqref{eq:chernoffeIID} in Appendix~\ref{asec:collective}). Now similarly set $\dtol,\dsou \propto 1/n^{1/3}$ to obtain $\dtol,\dsou,\gamma = O(1/n^{1/3})$, i.e.~the ``relaxation parameters'' $\dtol,\dsou$ and the test-round fraction $\gamma$ all go to zero as $n$ increases. This yields convergence to the asymptotic keyrate formula~\eqref{eq:asympt}. 

However, the above theorem is simply a statement for Protocol~\ref{prot:DIQKD} under the assumption of collective attacks, and that protocol does not fully exploit some implications of that assumption. For instance, in Theorem~\ref{th:collective} there is implicitly an $O(\gamma)$ subtractive penalty to the keyrates (which was also present in Theorem~\ref{th:DIQKD}\footnote{In fact, Theorem~\ref{th:DIQKD} has another $O(\gamma)$ subtractive penalty from the use of the bound~\eqref{eq:chain}, but this was due to the technical limitations of the EAT and the fact that the bounds~\eqref{eq:fmax} and \eqref{eq:gproof} are slightly suboptimal.}) caused by having to include the test-round data in the $\cmax$ term. While this is partly mitigated by the fact that the test rounds are also included in privacy amplification, the issue is that the test-round outputs are less strongly correlated as compared to generation rounds, and hence in most noise regimes there is an overall negative effect (on the order of $\gamma$). Yet under the collective-attacks assumption, the test rounds are completely independent of the generation rounds, which implies that the effect of $\gamma$ should instead be to reduce the keyrate by a \emph{multiplicative} factor of $(1-\gamma)$. Importantly, in the latter case it is possible to choose arbitrarily large test probabilities $\gamma$ without necessarily making the keyrates negative, which can dramatically improve the statistical bounds for parameter estimation. To formalize this idea, we consider Protocol~\ref{prot:optcoll} below, which attempts to minimize the finite-size correction as much as possible using the most optimistic assumptions that have been discussed thus far.
\begin{savenotes}
\begin{breakablealgorithm}
\caption{} 
\label{prot:optcoll}
This protocol proceeds the same way as Protocol~\ref{prot:DIQKD}, except for the following changes:
\begin{itemize}
\item Instead of independently choosing whether each round is a test or generation round, Alice chooses a uniformly random subset of size $m$ as test rounds before the protocol begins, 
and we define $\gamma$ as the value $m/n$. Alice also prepares the strings $\str{X}\str{Y}$ in advance, by choosing $X_j=Y_j\in\inYg$ uniformly at random in the generation rounds, and choosing $X_j\in\inX, Y_j\in\inYt$ uniformly at random in the test rounds.
\item In each round, Alice and Bob briefly store their received quantum states instead of immediately measuring them. Alice then publicly announces $X_jY_j$, which Alice and Bob then use as the inputs to their devices.\footnote{Here, in our attempt to minimize the finite-size effects, we are following the~\cite{ARV19} assumption mentioned previously: Alice and Bob can briefly store their received quantum states, in a manner such that the public communication cannot affect the stored states.}
\item In the error-correction step, Alice does not send error-correction data (and a corresponding hash) for the full string $\str{A}$, but rather only the subset of it consisting of the generation rounds, denoted as $\str{A}_g$. Bob's guess for this string will be denoted as $\tilde{\str{A}}_g$. The values of $\str{A}$ in the test rounds, denoted as $\str{A}_t$, are sent directly to Bob without compression or encryption, and Bob uses this string for parameter estimation.
\item Bob's accept condition is instead to check that $\hash(\str{A}_g) = \hash(\tilde{\str{A}}_g)$ and $\freq_{\str{c}_t}(1)\geq \wexp-\dtol$ 
hold, where $\str{C}_t$ denotes the substring of $\str{C}$ corresponding to the test rounds (in particular, this means the frequencies are computed with respect to a string of length $\gamma n$, not $n$).
\item Privacy amplification is performed only on the strings $\str{A}_g$ and $\tilde{\str{A}}_g$.
\end{itemize}
\end{breakablealgorithm}
\end{savenotes}

\begin{figure}
\centering
\subfloat[\cite{RBG+17} parameters, $\p=0$]{
\includegraphics[width=0.48\textwidth]{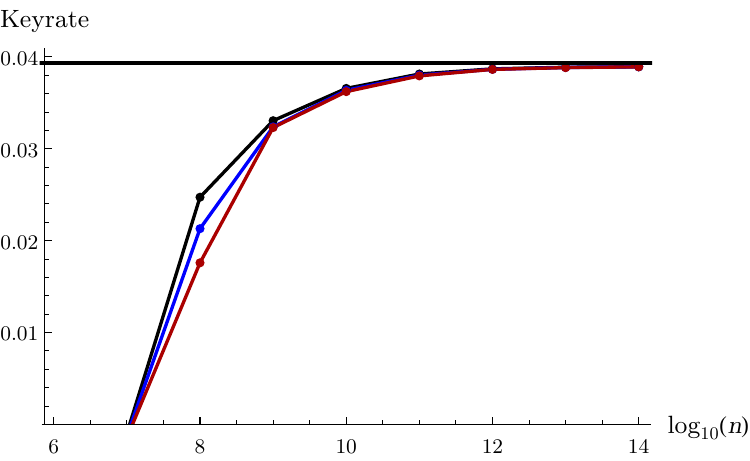}
} 
\subfloat[\cite{RBG+17} parameters, $\p=0.03$ (with heuristic $\lin_{\p}$)]{
\includegraphics[width=0.48\textwidth]{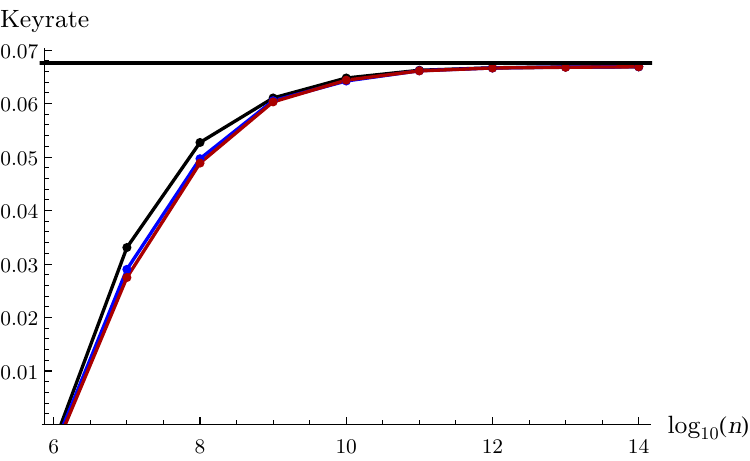}
}
\caption{(From~\cite{arx_TSB+20}) Finite-size keyrates as a function of number of rounds in Protocol~\ref{prot:optcoll} (see Theorem~\ref{th:optcoll}),
using honest devices following the estimated parameters in~\cite{MvDR+19} for the loophole-free Bell test in~\cite{HBD+15}, for $\p=0$ and $\p=0.03$ (the latter is a rough estimate of the choice of $\p$ which yields the highest asymptotic keyrate for these experimental parameters). Note that the latter graph is computed using a heuristic estimate of $\lin_{\p}$ rather than a certified bound.
The error-correction protocol is taken to satisfy Eqs.~\eqref{eq:optEC}~and~\eqref{eq:optcollAEP}. The colours correspond to soundness parameters of $\esound=10^{-3}$, $10^{-6}$, and $10^{-9}$ for black, blue, and red respectively, and the completeness parameter is $\ecom=10^{-2}$ in all cases.
The horizontal line denotes the asymptotic keyrate. 
All other parameters in Theorem~\ref{th:optcoll} were numerically optimized. The required number of rounds to achieve positive keyrate is substantially lower than Protocol~\ref{prot:DIQKD} (see Fig.~\ref{fig:experiments}).
}
\label{fig:protIID}
\end{figure}

For this protocol, the value of $\cmax$ is to be computed based only on the number of generation rounds, 
since error correction is performed on the string $\str{A}_g$ rather than $\str{A}$. Focusing on the best possible theoretical bounds from Sec.~\ref{sec:EC}, this means we take $\cmax$ to be given by Eq.~\eqref{eq:optEC} with
\begin{align}
\Hmax^{\ez}(\str{A}|\str{B}\str{X}\str{Y})_\mathrm{hon} \leq (1-\gamma)n 
h_\mathrm{hon}
+ \sqrt{(1-\gamma)n} \, (2\log5)\sqrt{\log\frac{2}{\ez^2}},
\label{eq:optcollAEP}
\end{align}
where
\begin{align}
h_\mathrm{hon} 
=
\sum_{z\in\inX} \frac{1}{2} H(A_j | B_j;X_j = Y_j = z)_\mathrm{hon}  ,
\end{align}
since the test rounds are excluded. With this value of $\cmax$ in mind, we can state the security guarantees of this protocol, with the proof given in Appendix~\ref{asec:optcoll}
(note that the dependence on 
several security parameters
here is somewhat different as compared to the previous theorems):
\begin{theorem}\label{th:optcoll}
Take any 
$\ecEC,\ecPE,\ePA,\eh,\es,\eIID \in (0,1]$, $\gamma\in(0,1)$, $\p\in[0,1/2]$, and $\dsou\in[0,\wexp - \dtol)$, such that 
\begin{align}
\eIID \geq 
\cdfBin{\gamma n}{1-\wexp+\dtol+\dsou}{\floor{(1-\wexp+\dtol)\gamma n}}. \label{eq:eIIDoptcoll}
\end{align}
Under the collective-attacks assumption, Protocol~\ref{prot:optcoll} is $(\ecEC + \ecPE)$-complete and $(\max\{\eIID, \ePA + 2\es\} + \eh)$-sound 
when performed with $\cmax$ defined in terms of $\ecEC$ as described above, 
and $\dtol,\lkey$ satisfying
\begin{align}
\ecPE &\geq \cdfBin{\gamma n}{\wexp}{\floor{(\wexp-\dtol)\gamma n}}, \label{eq:ecPEoptcoll}\\
\lkey &\leq
(1-\gamma)n\lin_\p(\wexp-\dtol-\dsou) - \sqrt{(1-\gamma)n} \,(2\log5)\sqrt{\log\frac{2}{\es^2}}
\nonumber \\ &\qquad 
- \cmax - \ceil{\log\left(\frac{1}{\eh}\right)} - 2\log\frac{1}{\ePA} + 2. \label{eq:lkeyoptcoll}
\end{align}
\end{theorem}

In Fig.~\ref{fig:protIID}, we plot the results of Theorem~\ref{th:optcoll}, focusing on the~\cite{RBG+17} experiment. This protocol has improved finite-size performance as compared to the original Protocol~\ref{prot:DIQKD} (under the collective-attacks assumption) due to at least two factors. Firstly, we can potentially use larger $\gamma$ values, as previously mentioned (some of the points shown in the figure correspond to values ranging up to $\gamma \approx 0.3$).
Secondly, we find that for fixed values of $\gamma,n,\wexp,\dtol,\dsou$, the binomial-distribution bounds for this protocol (Eqs.~\eqref{eq:eIIDoptcoll} and \eqref{eq:ecPEoptcoll}) are typically 
several orders of magnitude better
than their counterparts for Protocol~\ref{prot:DIQKD} (Eqs.~\eqref{eq:eIIDbound} and \eqref{eq:ecPEbound1}). Intuitively, this arises because in Protocol~\ref{prot:DIQKD}, the number of test rounds is itself a random variable, hence increasing the variance in e.g.~the number of rounds where $C_j=1$.
Practically speaking, this means Protocol~\ref{prot:DIQKD} requires noticeably larger values of $\dtol$ and $\dsou$ in order to achieve given completeness and soundness parameters, hence reducing the keyrate by a nontrivial amount.

However, we see that even with the optimistic assumptions that yield Theorem~\ref{th:optcoll}, the keyrate for the estimated experimental parameters we consider only becomes positive at fairly large $n$. This indicates that substantial further work is necessary in order to achieve a demonstration of positive finite-size keyrates.

\chapter{Advantage distillation}
\label{chap:AD}
\newcommand{\scenoneway}{i}
\newcommand{\scenHardyopt}{ii}
\newcommand{\scenMY}{iii}
\newcommand{\scenCHSHbasic}{iv}
\newcommand{\scenCHSHq}{v}
\newcommand{\scenCHSHeta}{vi}

\newcommand{\blk}[1]{\boldsymbol{#1}}
\newcommand{\prsymb}{\mathsf{P}}
\newcommand{\sz}{k}

\newcommand{\etat}{\eta_t} 
\newcommand{\qt}{\q_t} 

\newcommand{\class}[1]{\hat{#1}} 
\newcommand{\e}{\epsilon} 
\newcommand{\dn}{\delta_{{\sz}}} 
\newcommand{\prt}{\prnonoise[ab|xy]}
\newcommand{\flip}[1]{\overline{#1}} 
\newcommand{\supeps}{\delta}
\newcommand{\symmset}{\mathcal{S}}

\newcommand{\mixedsymm}{\rho}
\newcommand{\puresymm}{\rho'}
\newcommand{\psisymm}{\psi'}

We now turn to the somewhat less well-studied setting of advantage distillation in DIQKD, where the proof techniques we have discussed in previous chapters no longer apply straightforwardly. Because of this, in this chapter we will work under the simplifying assumption of collective attacks. We begin by defining in Sec.~\ref{sec:prelimAD} the specific advantage distillation protocol that we consider. Then in Sec.~\ref{sec:securityAD}, we derive a sufficient condition for this protocol to be secure in the context of DIQKD. Finally, in Sec.~\ref{sec:threshAD} we describe the resulting noise tolerances, and compare them to the values achieved by one-way protocols. 
The results in this chapter are based on~\cite{arx_TSB+20}, and the phrasing and presentation are essentially identical to that work.

\section{Preliminaries}
\label{sec:prelimAD}

We work under the collective-attacks assumption, so we follow the notation described in Sec.~\ref{sec:introdevices} and use $\rho_{\qA\qB\qE}$ to denote the single-round state. Given the IID structure, the outcome probabilities can be estimated to arbitrary accuracy given enough rounds, so we shall assume that the parameter estimation step basically only accepts states that produce exactly the same distribution $\pr{ab|xy}$ as the honest states. We focus only on protocols where the generation rounds always use the inputs $x=y=0$, so we are not considering the random-key-measurements technique. 
Also, we will only consider scenarios where all outputs are binary-valued, and we will not implement noisy preprocessing.

For convenience in the proofs, we assume a symmetrization step is implemented (as mentioned in Sec.~\ref{sec:qubit}), in which Alice generates a uniform random bit $F$ in each round and sends it to Bob, with both parties flipping their measurement outcome if and only if $F=1$\footnote{In this work, we take the symmetrization step to be applied to all measurements, which is possible because we focus on scenarios where all measurements have binary outcomes. In principle, one could instead symmetrize the key-generating measurements only, in which case the measurements not used for key generation can have more outcomes.
}. The bit $F$ can be absorbed into Eve's side-information $E$. (Again, this symmetrization step can be omitted in practice; see~\cite{TLR20}.) 
After this process, the generation measurements have symmetrized outcomes, in the sense $\pr{01|00}=\pr{10|00}=\e/2$ and $\pr{00|00}=\pr{11|00}=(1-\e)/2$ for some $\e < 1/2$ (if $\e>1/2$, simply swap Bob's outcome labels).\footnote{Note that in this context, $\e$ is equal to the QBER (for the generation inputs). When we consider depolarizing noise later, the parameter $\q$ is not necessarily equal to $\e$ in all the scenarios we consider --- this is somewhat unfortunate, but needs to be kept in mind.} 
Henceforth, $\prsymb_{\hat{A}\hat{B}|XY}$ refers to the distribution after symmetrization. 

We focus on the repetition-code protocol~\cite{Mau93,wolfthesis,rennerthesis,BA07}\footnote{We choose to focus on this protocol because in device-dependent QKD, thus far it has been the protocol that has achieved the highest noise tolerances~\cite{KL17}.} for advantage distillation, which proceeds as follows.
First, exploiting the collective-attacks assumption, we shall simply ignore the test rounds (apart from assuming they have been used for arbitrarily accurate estimation of $\pr{ab|xy}$), and focus only on the generation rounds. The generation rounds are divided into blocks of $\sz$ rounds each.
We shall denote Alice and Bob's output bitstrings in a block as $\str{A}_0$ and $\str{B}_0$ respectively (with the subscript reminding us that all generation rounds use the inputs $x=y=0$), and Eve's side-information in one block as $\blk{E}$ (this is slightly different from the all-rounds side-information $\str{E}$ that we used in previous discussions of collective attacks; we will not need that in this chapter). Alice privately generates a uniformly random bit $C$, and sends the message $\str{M} = \str{A}_0 \oplus (C,C,...,C)$ to Bob via a public authenticated channel. Bob replies with a bit $D$ that expresses whether to accept the block, with $D=1$ (accept) if and only if $\str{B}_0 \oplus \str{M} = (C',C',...,C')$ for some $C' \in \mathbb{Z}_2$. If the block is not accepted, Alice and Bob overwrite the values of $CC'$ with some null symbol $\perp$.
Alice and Bob will then try to distill a secret key from the bit pairs $(C,C')$ across many blocks, by using one-way error-correction followed by privacy amplification. (Here we do not consider the option of performing procedures such as noisy preprocessing, as it becomes very challenging to analyze.)

Under the collective-attacks assumption, we see that the system before applying the one-way error-correction protocol consists of many IID copies of the registers $CC'D\blk{E}\str{M}$, with Alice holding $C$, Bob holding $C'D$, and Eve holding $D\blk{E}\str{M}$. (The $\str{A}_0\str{B}_0$ registers can be excluded from consideration because Alice and Bob do not use them further in this protocol, and Eve has no direct access to them.) In that case, by the same arguments as described in Sec.~\ref{sec:proofsketch}, the asymptotic keyrate of this protocol (with respect to the number of protocol rounds) would be given by an analogous version of the expression~\eqref{eq:devwin}, i.e.~it would be
\begin{align}
\operatorname{rate}_\infty = \frac{1}{\sz} \left(\inf_{\mathcal{S}_{\mathrm{exact}}} H(C|D\blk{E}\str{M}) - H(C|DC')_{\mathrm{hon}}\right),
\end{align}
where the prefactor accounts for the fact that every block consists of $\sz$ rounds. (We have again taken the fraction of test rounds to go to zero asymptotically.) Furthermore, observe that the probability of accepting a block (i.e.~$D=1$) is $\e^{\sz} + (1-\e)^{\sz}$, and thus we have
\begin{align}
H(C|D\blk{E}\str{M}) = \sum_d \pr{D=d} H(C|\blk{E}\str{M};D=d) = (\e^{\sz} + (1-\e)^{\sz}) H(C|\blk{E}\str{M};D=1),
\end{align}
recalling that $C$ is deterministically set to a null value when $D=0$ (and thus we have $H(C|\blk{E}\str{M};D=1)=0$). Performing the same calculation for $H(C|DC')_{\mathrm{hon}}$, we can conclude that the asymptotic keyrate takes the form
\begin{align}
\operatorname{rate}_\infty = \frac{\e^{\sz} + (1-\e)^{\sz}}{\sz} \left( \inf_{\mathcal{S}_{\mathrm{exact}}} H(C|\blk{E}\str{M};D=1) - H(C|C';D=1)_{\mathrm{hon}} \right).
\label{eq:ADkeyrate}
\end{align}
From this expression, we see that key distillation would be asymptotically possible as long as
\begin{align}
\inf_{\mathcal{S}_{\mathrm{exact}}} H(C|\blk{E}\str{M};D=1) - H(C|C';D=1)_{\mathrm{hon}} > 0,
\label{eq_poskey}
\end{align}
and hence we shall just focus on this condition. For brevity, in the rest of this chapter we shall drop the $_{\mathrm{hon}}$ subscript on the second term.

\section{Security conditions}
\label{sec:securityAD}

\subsection{General scenarios}

We derive the following theorem:
\begin{theorem}
\label{th_fidbound}
For a DIQKD protocol as described above, a sufficient condition for Eq.~\eqref{eq_poskey} to hold for large $\sz$ is for all states accepted in parameter estimation to satisfy
\begin{align}
F(\rho_{E|00},\rho_{E|11})^2 > \frac{\e}{1-\e},
\label{eq_fidbound}
\end{align}
where $\rho_{E|a_0 b_0}$ is Eve's single-round state conditioned on outcomes $(a_0,b_0)$ being obtained for inputs $x=y=0$.
\end{theorem}

Note that in the above theorem, $\rho_{E|00}$ and $\rho_{E|11}$ refer to the states after the symmetrization step is implemented. The theorem can in fact be slightly broadened to encompass protocols without a symmetrization step, as long as $\prsymb_{\hat{A}\hat{B}|XY}$ has symmetrized outcomes, since the latter property is all that is required in the proof (see below). However, since in any case the symmetrization step can be omitted in practice~\cite{TLR20}, this does not seem to be a significant generalization.

The intuition behind the proof is that if Eve sees the message value $\str{M} = \str{m}$, then with high probability Alice and Bob's strings have the value $\str{A}_0 \str{B}_0 = \str{m} \str{m}$ or $\flip{\str{m}} \flip{\str{m}}$ (where $\flip{\str{m}}$ denotes the bitwise complement of $\str{m}$). Hence Eve essentially has to distinguish between these two cases, which can be quantified via the fidelity $F(\rho_{\blk{E}|\str{m} \str{m}},\rho_{\blk{E}|\flip{\str{m}} \flip{\str{m}}}) = F(\rho_{E|00},\rho_{E|11})^{\sz}$. To make this rigorous, we also need to prove that various ``low-probability'' contributions vanish asymptotically. We now give the full proof (see also Appendix~\ref{app:AD} for an alternative proof, which yields a potentially more general version of Theorem~\ref{th_fidbound} that was not stated in~\cite{TLR20}):

\begin{proof}
All states denoted in this proof are normalised. To bound $H(C|\blk{E}\str{M};D=1)$, we first observe that since $H(X|YZ)=\sum_z \pr[Z]{z} H(X|Y;Z=z)$ for classical $Z$, it suffices to bound $H(C|\blk{E};\str{M}=\str{m} \land D=1)$ for arbitrary messages $\str{m}$. Starting from the initial $\str{A}_0 \str{B}_0 \blk{E}$ state
\begin{align}
\rho_{\str{A}_0 \str{B}_0 \blk{E}} = \sum_{\str{a}_0, \str{b}_0} \pr[\str{A}_0 \str{B}_0]{\str{a}_0, \str{b}_0} \ketbra{\str{a}_0, \str{b}_0}{\str{a}_0, \str{b}_0} \otimes \rho_{\blk{E}|\str{a}_0 \str{b}_0},
\end{align}
a straightforward calculation shows that conditioned on the block being accepted and $\str{M}=\str{m}$, the $C \blk{E}$ state takes the form $\rho_{C \blk{E} | \str{M}=\str{m} \land D=1} = \sum_c (1/2) \ketbra{c}{c} \otimes \omega_c$ with
\begin{align}
\begin{gathered}
\omega_0 = \frac{\pr[\str{A}_0 \str{B}_0]{\str{m}, \str{m}} \rho_{\blk{E}|\str{m} \str{m}} + \pr[\str{A}_0 \str{B}_0]{\str{m}, \flip{\str{m}}} \rho_{\blk{E}|\str{m} \flip{\str{m}}}}{\pr[\str{A}_0 \str{B}_0]{{\str{m}}, {\str{m}}} + \pr[\str{A}_0 \str{B}_0]{{\str{m}}, \flip{\str{m}}}}, \\
\omega_1 = \frac{\pr[\str{A}_0 \str{B}_0]{\flip{\str{m}}, \flip{\str{m}}} \rho_{\blk{E}|\flip{\str{m}} \flip{\str{m}}} + \pr[\str{A}_0 \str{B}_0]{\flip{\str{m}}, \str{m}} \rho_{\blk{E}|\flip{\str{m}} \str{m}}}{\pr[\str{A}_0 \str{B}_0]{\flip{\str{m}}, \flip{\str{m}}} + \pr[\str{A}_0 \str{B}_0]{\flip{\str{m}}, {\str{m}}}}.
\end{gathered}
\end{align}
We now consider a state $\tilde{\rho}_{C \blk{E}}$ defined as follows: 
\begin{align}
\tilde{\rho}_{C \blk{E}} \defvar (1/2) (\ketbra{0}{0} \otimes \rho_{\blk{E}|\str{m} \str{m}} + \ketbra{1}{1} \otimes \rho_{\blk{E}|\flip{\str{m}} \flip{\str{m}}}).
\end{align}
With symmetrized IID outcomes, we have $\pr[\str{A}_0 \str{B}_0]{{\str{m}}, {\str{m}}} = \pr[\str{A}_0 \str{B}_0]{\flip{\str{m}}, \flip{\str{m}}} = (1-\e)^{\sz}/2^{\sz}$ and $\pr[\str{A}_0 \str{B}_0]{{\str{m}}, \flip{\str{m}}} = \pr[\str{A}_0 \str{B}_0]{\flip{\str{m}}, {\str{m}}} = \e^{\sz}/2^{\sz}$, so
\begin{align}
d(\tilde{\rho}_{C \blk{E}}, \rho_{C \blk{E} | \str{M}=\str{m} \land D=1}) \leq \dn \text{, where } \dn \defvar \frac{\e^{\sz}}{\e^{\sz} + (1-\e)^{\sz}}.
\end{align}
Applying a continuity bound for conditional von Neumann entropy~\cite{Win16} then yields
\begin{align}
H(C|\blk{E};\str{M}=\str{m} \land D=1) \geq H(C|\blk{E})_{\tilde{\rho}} - \dn - (1+\dn) \binh\left(\frac{\dn}{1+\dn}\right),
\label{eq_continuity}
\end{align}
where $\binh$ is the binary entropy function. The $H(C|\blk{E})_{\tilde{\rho}}$ term is bounded by~\cite{RFZ10}\footnote{The alternative proof in Appendix~\ref{app:AD} instead uses a bound based on min-entropy here, which appears suboptimal \textit{a priori} but can yield a slight generalization of Theorem~\ref{th_fidbound}.}
\begin{align}
H(C|\blk{E})_{\tilde{\rho}} \geq 1 - \binh\left( \frac{1-F(\rho_{\blk{E}|\str{m} \str{m}},\rho_{\blk{E}|\flip{\str{m}} \flip{\str{m}}})}{2} \right) 
= 1 - \binh\left( \frac{1-F(\rho_{E|00},\rho_{E|11})^{\sz}}{2} \right)
,
\label{eq_rogabound}
\end{align}
using the IID assumption. As for $H(C|C';D=1)$, it can be seen that $\pr{C \neq C'|D=1} = \dn$ for IID outcomes, so
\begin{align}
H(C|C';D=1) = \binh (\dn).
\label{eq_ECterm}
\end{align}

Combining these results, we conclude
\begin{align}
\frac{H(C|\blk{E} \str{M}; D=1)}{H(C|C';D=1)} &\geq \left(1 - \binh\left( \frac{1-F(\rho_{E|00},\rho_{E|11})^{\sz}}{2} \right) - \dn - (1+\dn) \binh\left(\frac{\dn}{1+\dn}\right) \right) \binh(\dn)^{-1} \nonumber\\
&\geq \left(\frac{F(\rho_{E|00},\rho_{E|11})^{2{\sz}}}{\ln4} - \dn - (1+\dn) \binh\left(\frac{\dn}{1+\dn}\right) \right) \binh(\dn)^{-1} 
,
\label{eq_ratioloose}
\end{align}
where the second line holds because $\binh((1-p)/2) \leq 1-p^2/\ln4$ (this inequality follows from the Taylor expansion of $\binh$).

It remains to find the behaviour of this expression in the large-$\sz$ limit. To do so, we first note that letting $\phi(\delta) \defvar \binh(\delta/(1+\delta))$, we have $\phi(\delta),\binh(\delta) \to 0$ and $\phi'(\delta),\binh'(\delta) \to \infty$ as $\delta \to 0^+$, so
\begin{gather}
\lim_{\delta \to 0^+} \frac{\delta}{\binh(\delta)} = \lim_{\delta \to 0^+} \frac{1}{\binh'(\delta)} = 0, \\
\lim_{\delta \to 0^+} \frac{\phi(\delta)}{\binh(\delta)} = \lim_{\delta \to 0^+} \frac{\phi''(\delta)}{\binh''(\delta)} = \lim_{\delta \to 0^+} \frac{(1 + \delta - 2 \delta \ln \delta)/(\delta(1+\delta)^3)}{1/(\delta(1-\delta))} = 1.
\end{gather}
Since $\dn \to 0^+$ as ${\sz} \to \infty$, the terms arising from the continuity bound hence have a finite limit,
\begin{align}
\lim_{\sz\to\infty} \left(- \dn - (1+\dn) \binh\left(\frac{\dn}{1+\dn}\right) \right) \binh(\dn)^{-1} = -1.
\label{eq_contfinite}
\end{align}

As for the term involving $F(\rho_{E|00},\rho_{E|11})$, let us write $\alpha \defvar F(\rho_{E|00},\rho_{E|11})^2$
and $\beta \defvar \e/(1-\e) \in [0,1)$, so the condition~\eqref{eq_fidbound} can be written as $\alpha > \beta$. 
Note that $\dn \leq \beta^{\sz}$, and for sufficiently large $\sz$ we have $\beta^{\sz} < 1/2$, so $\binh(\dn) \leq \binh(\beta^{\sz}) \leq 2 \beta^{\sz} \log (1/\beta^{\sz})$. This means that when~\eqref{eq_fidbound} (i.e.~$\alpha > \beta$) holds, we have
\begin{align}
\lim_{{\sz}\to\infty} \frac{\alpha^{\sz}}{\binh(\dn)} \geq \lim_{{\sz}\to\infty} \frac{\alpha^{\sz}}{2 {\sz} \beta^{\sz} \log (1/\beta)}
= \frac{1}{2 \log (1/\beta)} \lim_{{\sz}\to\infty} \frac{(\alpha/\beta)^{\sz}}{{\sz}}
= \infty,
\label{eq_limit}
\end{align}
which implies that the right-hand side of~\eqref{eq_ratioloose} limits to $\infty$ as ${\sz}\to\infty$ (since the continuity-bound contributions have a finite limit). Therefore,~\eqref{eq_poskey} will hold for sufficiently large $\sz$. 
\end{proof}

\begin{remark}
In the above analysis, the $H(C|\blk{E}\str{M};D=1)$ and $ H(C|C';D=1)$ terms are dominated by different contributions in some sense. Specifically, the analysis for the $H(C|\blk{E}\str{M};D=1)$ term mainly focuses on the contributions from the ``components'' where $\str{A}_0 = \str{B}_0$, while the ``components'' where $\str{A}_0 \neq \str{B}_0$ are absorbed into some roughly $O(\binh(\dn))$ continuity corrections and shown to be negligible in the large-$\sz$ limit. In contrast, the $H(C|C';D=1)$ term arises entirely from the contribution of the $\str{A}_0 \neq \str{B}_0$ case (which is of order $\binh(\dn)$ as well), because for the $\str{A}_0 = \str{B}_0$ case, we have $C=C'$ and hence the conditional entropy contribution is zero.

The relative ``sizes'' of these contributions impose some challenges when deriving converse counterparts of Theorem~\ref{th_fidbound} (see~\cite{arx_HT21}). Qualitatively, the goal there would be to obtain an upper bound in place of the lower bound on the ratio in~\eqref{eq_ratioloose}. However, the continuity correction to $H(C|\blk{E}\str{M};D=1)$ in the proof approach here is strictly larger than $\binh(\dn)$, while $H(C|C';D=1)$ is equal to $\binh(\dn)$. This makes it basically impossible to derive a nontrivial upper bound (specifically, a bound with value strictly less than $1$) if one merely follows an analogous approach to the above proof. Hence more detailed analysis of the $\str{A}_0 \neq \str{B}_0$ contribution to $H(C|\blk{E}\str{M};D=1)$ would be needed to derive such a result. (In~\cite{arx_HT21}, this contribution was handled by assuming a conjecture that is true for the device-dependent protocols in~\cite{BA07}, but has not been proven for DIQKD in general.)
\end{remark}

Given specific values or bounds for $F(\rho_{E|00},\rho_{E|11}), \e, {\sz}$, one can substitute them into the expressions~\eqref{eq_continuity}--\eqref{eq_ECterm} to get an explicit bound on the asymptotic keyrate. Currently, this inequality requires large values of $\sz$ to certify positive keyrates, and the resulting keyrates are thus extremely small (notice that~\eqref{eq:ADkeyrate} decreases superexponentially with $\sz$). However, small block sizes are sufficient for reasonable asymptotic keyrates in the device-dependent case~\cite{rennerthesis}, so there may be some possibility of tightening the bounds in the above proof.

This security proof has a key difference in structure compared to the proof techniques we have discussed in all the previous chapters regarding one-way QKD. Specifically, in this proof we must consider the ``security'' of both generation measurements $x=0,y=0$, while the proofs for one-way protocols only need to bound the security of Alice's generation measurement (in terms of the smoothed min-entropy of the bitstring it produces). The issue is that in the latter, we bound the information leakage to Eve simply by the length of the error-correction string. For the repetition-code protocol, however, a large number of bits are publicly communicated, and hence bounding the information leakage via the number of bits yields too crude a bound. 

The need to explicitly consider the security of Bob's measurement in this protocol can also be seen by considering an extreme example where Eve always knows the outcome of Bob's generation measurement, possibly at the cost of it being poorly correlated to Alice's measurement. In that case, regardless of how secure the output of Alice's measurement is, Eve will always know the value of $C'$, making it impossible to distill key from the $(C,C')$ pairs. Hence any security proof for this protocol must involve some kind of security argument regarding Bob's generation measurement, even if only indirectly via measuring its correlations with Alice's generation measurement. 

Theorem~\ref{th_fidbound} is similar to the condition obtained in~\cite{BA07} for device-dependent QKD, but it is derived here without detailed state characterization. However, it still remains to find bounds on $F(\rho_{E|00},\rho_{E|11})$ without device-dependent assumptions. We approach this task by combining the Fuchs--van de Graaf inequality with the operational interpretation of trace distance:
\begin{align}
F(\rho_{E|00},\rho_{E|11}) \geq 1-d(\rho_{E|00},\rho_{E|11}) = 2 (1-
\pg(\hat{A}_0|E)_{\sigma_{|\Omega}}
),
\label{eq_fidpg} 
\end{align}
where $\hat{A}_0$ is Alice's output register and $\sigma_{|\Omega}$ denotes the state conditioned on Alice and Bob measuring in the generation basis and getting the same outcome --- explicitly, we can write $\sigma_{\hat{A}_0 E|\Omega} = \sum_a (1/2) \ketbra{a}{a} \otimes \rho_{E|aa}$.
A DI method to bound such conditional (or ``postselected'') guessing probabilities based on the distribution $\prsymb_{\hat{A}\hat{B}|XY}$ was described in~\cite{TTB+16}, again using the NPA hierarchy~\cite{NPA08}. We could hence apply this method to find whether Eq.~\eqref{eq_fidbound} holds for various distributions.

\subsection{2-input 2-output scenarios}

However, Eq.~\eqref{eq_fidpg} is generally not an optimal bound. We observe that if $\rho_{E|00}$ and $\rho_{E|11}$ were both assumed to be pure, then the other side of the Fuchs--van de Graaf inequality would be saturated, i.e.~we would have
\begin{align}
F(\rho_{E|00},\rho_{E|11})^2 = 1-d(\rho_{E|00},\rho_{E|11})^2.
\label{eq_satfvdg}
\end{align}
While it seems difficult to justify such an assumption in general, we show that for 2-input 2-output protocols, one can almost replace Eq.~\eqref{eq_fidpg} with Eq.~\eqref{eq_satfvdg} after taking a particular concave envelope. Importantly, note that in this chapter, when we say that a protocol is 2-input 2-output, this \emph{includes} the generation rounds, unlike the case in Sec.~\ref{sec:qubit}.
\begin{theorem}
\label{th_22bound}
Consider a DIQKD protocol as described above, with $\mathcal{X}=\mathcal{Y}=2$ and all measurements having binary outcomes. Denoting the set of quantum distributions with $\pr{00|00}=\pr{11|00}$ as $\symmset$, let $f$ be a concave function on $\symmset$ such that for any $\vec{\constr} \in \symmset$, all states and measurements compatible with $\vec{\constr}$ satisfy $f(\vec{\constr}) \geq (1-\e)d(\rho_{E|00},\rho_{E|11})^2$. Let $\prsymb_{\hat{A}\hat{B}|XY}$ be the distribution accepted in parameter estimation. Then a sufficient condition for Eq.~\eqref{eq_poskey} to hold for large $\sz$ is
\begin{align}
{1-\frac{f(\prsymb_{\hat{A}\hat{B}|XY})}{1-\e}} > {\frac{\e}{1-\e}}.
\label{eq_22bound}
\end{align}
\end{theorem}

\begin{proof}
We use the qubit reduction mentioned before, although in greater detail (see~\cite{PAB+09,TLR20}): without loss of generality,\footnote{Strictly speaking, some care is needed here because our quantity of interest is now fidelity rather than entropy. However, essentially the same reduction carries through, with minor modifications to ensure certain states remain pure; see~\cite{TLR20}.} we can assume Eve's strategy in each round consists of generating a random variable $\Lambda$ and storing it in a classical register, then implementing some corresponding ``qubit strategy'', which is a strategy such that $\rho_{ABE|\lambda}$ is a $2\times2\times4$ pure state and all of Alice and Bob's measurements are rank-1 projective measurements. For a qubit strategy,  Eve's state conditioned on the joint outcome of both measurements is pure. Hence if Alice and Bob choose measurements $xy$ and store their results in classical registers $\class{A}_x \class{B}_y$, we can take the resulting single-round state $\rho_{\class{A}_x \class{B}_y E \Lambda}$ after symmetrization to be in the form\footnote{Technically, the notation $\rho_{E|\lambda a b}$ here corresponds to different states for different input pairs $xy$, though it is not crucial to keep track of this in the proof, since it only uses the conditional states produced by the key-generating measurements. }
\begin{align}
\rho_{\class{A}_x \class{B}_y E \Lambda} = \sum_{a, b} \sum_{\lambda} \pr{a b | x y\lambda} \, \pr{\lambda} \ketbra{a,b}{a,b} \otimes \rho_{E|\lambda a b} \otimes \ketbra{\lambda}{\lambda}, 
\label{eq_22state}
\end{align}
where all $\rho_{E|\lambda a b}$ are pure states. Note that the probabilities $\pr{a b | x y\lambda}$ satisfy $\sum_{\lambda} \pr{a b | x y\lambda} \pr{\lambda} = \pr{a b | x y}$, and for each $\lambda$ we have $\pr{0 0 | 0 0\lambda} = \pr{1 1 | 0 0\lambda} \defvar (1-\e_\lambda)/2$ for some $\e_\lambda \in [0,1]$, due to the symmetrization step. When Alice and Bob choose inputs $x=y=0$ and get the same outcome (which we shall denote here as $\gamma \in \{0,1\}$), Eve's conditional states are
\begin{align}
\rho_{E\Lambda|\gamma\gamma} = \sum_{\lambda} \widetilde{\prsymb}(\lambda) \rho_{E|\lambda\gamma\gamma} \otimes \ketbra{\lambda}{\lambda}, \text{ where } \widetilde{\prsymb}(\lambda) = \frac{\pr{\gamma\gamma| 00\lambda}}{\pr{\gamma\gamma|00}} \pr{\lambda} = \frac{1-\e_\lambda}{1-\e} \pr{\lambda}.
\label{eq_22condstate}
\end{align}
Note that $\widetilde{\prsymb}$ is itself a valid probability distribution over $\Lambda$, and independent of $\gamma$.

We now apply the same arguments as in the proof of Theorem~\ref{th_fidbound} up until Eq.~\eqref{eq_continuity}, though in this case Eve's side-information takes the form $\blk{E}\str{\Lambda}$, and so we consider the state $\tilde{\rho}_{C \blk{E}\str{\Lambda}} \defvar (1/2) (\ketbra{0}{0} \otimes \rho_{\blk{E}\str{\Lambda}|\str{m} \str{m}} + \ketbra{1}{1} \otimes \rho_{\blk{E}\str{\Lambda}|\flip{\str{m}} \flip{\str{m}}})$. Using the above, one can show that
\begin{gather}
\tilde{\rho}_{C \blk{E}\str{\Lambda}} = \sum_{\str{\lambda}} \widetilde{\prsymb}(\str{\lambda}) \frac{1}{2}  \left(\ketbra{0}{0} \otimes \rho_{\blk{E}|\str{\lambda}\str{m} \str{m}} + \ketbra{1}{1} \otimes \rho_{\blk{E}|\str{\lambda} \flip{\str{m}} \flip{\str{m}}} \right) \otimes \ketbra{\str{\lambda}}{\str{\lambda}},
\end{gather}
where $\widetilde{\prsymb}(\str{\lambda}) = \prod_j \widetilde{\prsymb}(\lambda_j)$ and $\rho_{\blk{E}|\str{\lambda}\str{\gamma} \str{\gamma}} = \bigotimes_j \rho_{E|\lambda_j \gamma_j \gamma_j}$. Hence we have 
\begin{align}
H(C|\blk{E}\str{\Lambda})_{\tilde{\rho}} = \sum_{\str{\lambda}} \widetilde{\prsymb}(\str{\lambda}) H(C|\blk{E}; \str{\Lambda} = \str{\lambda})_{\tilde{\rho}}.
\end{align}
Applying the bound from~\cite{RFZ10} together with $\binh((1-p)/2) \leq 1-p^2/\ln4$ then yields
\begin{align}
H(C|\blk{E}\str{\Lambda})_{\tilde{\rho}} 
&\geq \sum_{\str{\lambda}} \widetilde{\prsymb}(\str{\lambda}) \left(1 - \binh\left( \frac{1-F(\rho_{\blk{E}|\str{\lambda}\str{m} \str{m}},\rho_{\blk{E}|\str{\lambda}\flip{\str{m}} \flip{\str{m}}})}{2} \right)\right) \nonumber\\
&\geq \frac{1}{\ln4}\sum_{\str{\lambda}} \widetilde{\prsymb}(\str{\lambda}) F(\rho_{\blk{E}|\str{\lambda}\str{m} \str{m}},\rho_{\blk{E}|\str{\lambda}\flip{\str{m}} \flip{\str{m}}})^2.
\label{eq_22condentbound}
\end{align}
Using the IID structure,
\begin{align}
\sum_{\str{\lambda}} \widetilde{\prsymb}(\str{\lambda}) F(\rho_{\blk{E}|\str{\lambda}\str{m} \str{m}},\rho_{\blk{E}|\str{\lambda}\flip{\str{m}} \flip{\str{m}}})^2 
&= \sum_{\str{\lambda}} \prod_{j=1}^{\sz} \widetilde{\prsymb}(\lambda_j) F\!\left(\rho_{E|\lambda_j 00},\rho_{E|\lambda_j 11}\right)^2 \nonumber\\
&= \left(\sum_{\lambda} \widetilde{\prsymb}(\lambda) F\!\left(\rho_{E|\lambda 00},\rho_{E|\lambda 11}\right)^2 \right)^{\sz}.
\end{align}
Since the states $\rho_{E|\lambda 00},\rho_{E|\lambda 11}$ are pure, they satisfy $F(\rho_{E|00},\rho_{E|11})^2 = 1-d(\rho_{E|00},\rho_{E|11})^2$, and thus
\begin{align}
\sum_{\lambda} \widetilde{\prsymb}(\lambda) F\!\left(\rho_{E|\lambda 00},\rho_{E|\lambda 11}\right)^2 
&= 1-\sum_{\lambda} \frac{1-\e_\lambda}{1-\e} \pr{\lambda} \, d\!\left(\rho_{E|\lambda 00},\rho_{E|\lambda 11}\right)^2 \nonumber\\
&\geq 1- \frac{1}{1-\e}\sum_{\lambda} \pr{\lambda} f(\prsymb_{\hat{A}\hat{B}|XY\lambda}).
\end{align}
Finally, this is lower-bounded by $1-f(\prsymb_{\hat{A}\hat{B}|XY})/(1-\e)$ since $f$ is concave, so 
\begin{align}
\frac{H(C|\blk{E}\str{\Lambda} \str{M}; D=1)}{H(C|C';D=1)} \geq \left( \frac{1}{\ln4} \left(1- \frac{f(\prsymb_{\hat{A}\hat{B}|XY})}{1-\e} \right)^{\sz} - \dn - (1+\dn) \binh\left(\frac{\dn}{1+\dn}\right) \right) \binh(\dn)^{-1}.
\label{eq_22ratio}
\end{align}
The behaviour of this expression in the large-$\sz$ limit is given by the same analysis as the last part of the proof of Theorem~\ref{th_fidbound}, choosing $\alpha={1-f(\prsymb_{\hat{A}\hat{B}|XY})/(1-\e)}$ in this case. We conclude that when Eq.~\eqref{eq_22bound} holds, the right-hand side of Eq.~\eqref{eq_22ratio} limits to $\infty$ as ${\sz}\to\infty$. Therefore, Eq.~\eqref{eq_poskey} will hold for sufficiently large $\sz$. 
\end{proof}

Similar to Theorem~\ref{th_fidbound}, values can be substituted into Eq.~\eqref{eq_22ratio} to obtain explicit bounds on the keyrate. 
If $f$ is the \emph{optimal} concave upper bound on $(1-\e)d(\rho_{E|00},\rho_{E|11})^2$, in the sense that there always exists a mixture of qubit strategies such that $\sum_{\lambda} \pr{\lambda} (1-\e_\lambda) d\!\left(\rho_{E|\lambda 00},\rho_{E|\lambda 11}\right)^2 = f(\prsymb_{\hat{A}\hat{B}|XY})$, then the above analysis is essentially tight for large $\sz$. This is because in Eq.~\eqref{eq_22condentbound}, the first inequality~\cite{RFZ10} is in fact saturated because $\rho_{\blk{E}|\str{\lambda}\str{m} \str{m}},\rho_{\blk{E}|\str{\lambda}\flip{\str{m}} \flip{\str{m}}}$ are pure, and the second inequality is approximately saturated at large $\sz$ because $\binh((1-p)/2) = 1-p^2/\ln4-O(p^4)$.

Currently, we do not have a method for finding an optimal concave bound on the quantity $(1-\e)d(\rho_{E|00},\rho_{E|11})^2$ in Theorem~\ref{th_22bound}. However, we find a condition that is more restrictive than Eq.~\eqref{eq_22bound} but more tractable to verify:

\begin{corollary}
\label{co_22weakbound}
Consider a DIQKD protocol as described above, with $\mathcal{X}=\mathcal{Y}=2$ and all measurements having binary outcomes. Then a sufficient condition for Eq.~\eqref{eq_poskey} to hold for large $\sz$ is for all states accepted in parameter estimation to satisfy
\begin{align}
{1-d(\rho_{E|00},\rho_{E|11})} > {\frac{\e}{1-\e}}.
\label{eq_22weakbound}
\end{align}
\end{corollary}

\begin{proof}
Let $\tilde{f}$ be the optimal upper bound on $(1-\e) d(\rho_{E|00},\rho_{E|11})$, defined on the set of quantum distributions such that $\pr{00|00} = \pr{11|00}$. Since $d(\rho_{E|00},\rho_{E|11}) \leq 1$, we have\footnote{This is the main step in this argument which is not tight. It introduces a (multiplicative) gap on the order of $d(\rho_{E|00},\rho_{E|11})$.}
\begin{align}
\tilde{f}(\prsymb_{\hat{A}\hat{B}|XY}) \geq (1-\e) d(\rho_{E|00},\rho_{E|11}) \geq (1-\e) d(\rho_{E|00},\rho_{E|11})^2.
\label{eq_relaxedbound}
\end{align}
Also, $\tilde{f}$ must be concave, for essentially the same reason that optimal guessing-probability bounds must be concave, though modified to account for the ``postselection'' on the outcomes being $00$ or $11$. More specifically, denote the set of quantum distributions with $\pr{00|00} = \pr{11|00}$ as $\symmset$. Consider any probability distribution $\pr{\lambda}$ and any family of distributions $\prsymb_{\hat{A}\hat{B}|XY\lambda} \in \symmset$ indexed by $\lambda$, and take an arbitrary $\supeps>0$. For each $\lambda$, there exists a strategy for Eve that achieves the distribution $\prsymb_{\hat{A}\hat{B}|XY\lambda}$ and has $(1-\e_\lambda) d\left(\rho_{E|\lambda 00},\rho_{E|\lambda 11}\right) \geq \tilde{f}(\prsymb_{\hat{A}\hat{B}|XY\lambda}) - \supeps$, because $\tilde{f}$ is an optimal bound (i.e.\ there exist strategies arbitrarily close to saturating the bound). If Eve generates and stores a classical random variable $\Lambda$ according to the distribution $\pr{\lambda}$, then implements the corresponding strategy, the resulting states take the same form as in~\eqref{eq_22state}--\eqref{eq_22condstate}. This is a strategy that achieves probabilities $\pr{ab|xy} = \sum_{\lambda}\pr{\lambda}\pr{ab|xy\lambda}$, and therefore
\begin{align}
\tilde{f}(\prsymb_{\hat{A}\hat{B}|XY}) &\geq (1-\e) d(\rho_{E|00},\rho_{E|11}) \nonumber\\
&= (1-\e) d\left(\sum_{\lambda} \widetilde{\prsymb}(\lambda) \rho_{E|\lambda 00} \otimes \ketbra{\lambda}{\lambda},\sum_{\lambda} \widetilde{\prsymb}(\lambda) \rho_{E|\lambda 11} \otimes \ketbra{\lambda}{\lambda}\right) \nonumber\\
&= (1-\e) \sum_{\lambda} \widetilde{\prsymb}(\lambda) d\left( \rho_{E|\lambda 00} ,\rho_{E|\lambda 11}\right) \nonumber\\
&= \sum_{\lambda} \pr{\lambda} (1-\e_\lambda) d\left( \rho_{E|\lambda 00} ,\rho_{E|\lambda 11}\right) \nonumber\\
&\geq \left(\sum_{\lambda} \pr{\lambda} \tilde{f}(\prsymb_{\hat{A}\hat{B}|XY\lambda})\right) - \supeps.
\end{align}
Since $\supeps$ was arbitrary, we conclude that $\tilde{f}(\prsymb_{\hat{A}\hat{B}|XY}) \geq \sum_{\lambda} \pr{\lambda} \tilde{f}(\prsymb_{\hat{A}\hat{B}|XY\lambda})$, i.e.\ $\tilde{f}$ is concave on $\symmset$. Hence choosing $f=\tilde{f}$ satisfies the conditions of Theorem~\ref{th_22bound}. Since $\tilde{f}$ is an upper bound on $(1-\e) d(\rho_{E|00},\rho_{E|11})$, we conclude that when Eq.~\eqref{eq_22weakbound} holds, we have
\begin{align}
{1-\frac{\tilde{f}(\prsymb_{\hat{A}\hat{B}|XY})}{1-\e}} \geq {1-d(\rho_{E|00},\rho_{E|11})} > {\frac{\e}{1-\e}},
\end{align}
and the claim follows by Theorem~\ref{th_22bound}.
\end{proof}

As before, we can bound $d(\rho_{E|00},\rho_{E|11})$ by using the NPA hierarchy. Effectively, Corollary~\ref{co_22weakbound} improves over the combination of Theorem~\ref{th_fidbound} and Eq.~\eqref{eq_fidpg} by replacing $\left(1-d(\rho_{E|00},\rho_{E|11})\right)^2$ with $1-d(\rho_{E|00},\rho_{E|11})$. 

\section{Resulting noise thresholds} 
\label{sec:threshAD}

Using the above approaches, we study the noise tolerance of various DIQKD scenarios. More precisely, we consider a number of different choices of noiseless distribution $\prt$, and apply the noise models of depolarizing noise and limited detection efficiency. In this section it will again be convenient for us to describe the measurements using hermitian observables $A_0, A_1, ..., A_{\mathcal{X}-1}$ and $B_0, B_1, ..., B_{\mathcal{Y}-1}$ with eigenvalues $\{-1,+1\}$.

\begin{table*}
\caption{Noise thresholds for the repetition-code protocol, for various choices of noiseless distribution $\prt$. The value $\qt$ is the maximum depolarizing noise such that we can show positive keyrate is achievable using Theorem~\ref{th_fidbound} (for rows (i)--(iii)) or Corollary~\ref{co_22weakbound} (for rows (iv)--(vi)). Analogously, $\etat$ is the minimum efficiency which can be tolerated when we instead consider a limited-detection-efficiency model. Unless otherwise specified, the state used for $\prt$ is $\ket{\Phi^+} = (\ket{00} + \ket{11})/\sqrt{2}$.}
\def\arraystretch{1.5} 
\setlength\tabcolsep{1.5mm}
\begin{center}
\small
\begin{tabular}{C{45mm} C{79mm} C{9mm} C{10mm}}
\toprule
Description of $\prt$ & State and measurements for $\prt$ & $\qt$ & $\etat$ \\
\hline
\setcounter{tabrownumcounter}{0}
(\tabrownum) Achieves maximal CHSH value with the measurements $A_0,A_1,B_1,B_2$. & \makecell{$A_0=B_0=Z,\enspace A_1=X,$ \\ $B_1=(X+Z)/\sqrt{2},\enspace B_2=(X-Z)/\sqrt{2}$.} & 6.0\pct & 93.7\pct \\
\hline
(\tabrownum) Modification of a distribution exhibiting the Hardy paradox~\cite{Har93,RLS12} for improved robustness against limited detection efficiency. & $\ket{\psi} = \sqrt{\kappa}(\ket{01} + \ket{10}) + \sqrt{1-2\kappa} \ket{11}$ with $\kappa=(3-\sqrt{5})/2$; the $0$ outcomes correspond to projectors onto $\ket{a_0}=\ket{b_0}\propto\left(\sqrt{1+2\kappa}-\sqrt{1-2\kappa}\right)\ket{0} + 2\sqrt{\kappa}\ket{1}$, $\ket{a_1}=\ket{b_1}\approx 0.37972\ket{0} + 0.92510\ket{1}$, $\ket{a_2}=\ket{b_2}\approx 0.90821\ket{0} +  0.41851\ket{1}$. & 3.2\pct & 92.0\pct \\
\hline
(\tabrownum) Includes the Mayers-Yao self-test~\cite{MY98} and the CHSH measurements. & \makecell{$A_0=B_0= Z,\enspace A_1=B_1= (X+Z)/\sqrt{2},$\\ $A_2=B_2= X,\enspace A_3=B_3= (X-Z)/\sqrt{2}$.} & 6.8\pct & 92.7\pct \\
\hline
(\tabrownum) Achieves maximal CHSH value with the measurements $A_0,A_1,B_0,B_1$. & \makecell{$A_0=Z,\enspace A_1=X,$\\ $B_0=(X+Z)/\sqrt{2},\enspace B_1=(X-Z)/\sqrt{2}$.} & 7.7\pct & 91.7\pct \\
\hline
(\tabrownum) Similar to (\scenCHSHbasic), but with measurements optimised for robustness against depolarizing noise. & Measurements are in the $x$-$z$ plane at angles $\theta_{A_0} = 0.4187,\enspace \theta_{A_1} = 1.7900$, $\theta_{B_0} = 0.8636,\enspace \theta_{B_1} = 2.6340$. & 9.1\pct & 90.0\pct \\
\hline
(\tabrownum) Similar to (\scenCHSHbasic), but with states and measurements maximising CHSH violation for each value of detection efficiency $\eta$~\cite{Ebe93}. & $\ket{\psi} = \cos\Omega \ket{00} + \sin\Omega \ket{11}$ with $\Omega = 0.6224$; the 0 outcomes correspond to projectors onto states of the form $\cos(\theta/2)\ket{0} + \sin(\theta/2)\ket{1}$ with $\theta_{A_0} = -\theta_{B_0} = -0.35923,\enspace \theta_{A_1} = -\theta_{B_1} = 1.1538$. & 7.3\pct & 89.1\pct \\
\toprule
\end{tabular}
\end{center}
\def\arraystretch{1}
\label{tab_thresh}
\end{table*}

In Table~\ref{tab_thresh}, we present a selection of our results (see~\cite{TLR20} for the full list). 
From the table, we see that with appropriate choices of the noiseless distribution, the repetition-code protocol can tolerate depolarizing noise of $\qt \approx 9.1\pct$ or detection efficiencies of $\etat \approx 89.1\pct$. The depolarizing-noise threshold indeed outperforms the basic one-way protocol of~\cite{PAB+09}, which can tolerate $\qt \approx 7.1\pct$.
The detection-efficiency threshold is better than some early approaches for the~\cite{PAB+09} protocol which yielded thresholds of $\etat \approx 90.7\pct$; however, it is worse than the optimized threshold of $\etat \approx 88.4\pct$ for that protocol restricted to binary outputs, or $\etat \approx 86.5\pct$ if Bob's $\perp$ output is preserved (see Appendix~\ref{app:detthresh}).

As for the various improved one-way protocols we listed in Sec.~\ref{sec:depolthresh}, we see that the depolarizing-noise thresholds here at least outperform the approaches based on noisy preprocessing or random key measurements alone, though it is narrowly worse than the threshold of $\qt \approx 9.33\pct$ obtained by combining both approaches. 
Still, given that we have used the inequalities~\eqref{eq_fidpg} or~\eqref{eq_relaxedbound} to obtain the results here, it seems likely that these thresholds for advantage distillation are not tight, and hence there is still room for further improvement. We discuss this in Chapter~\ref{chap:conclusion}.

Note the \cite{PAB+09}~protocol (and some of its improved versions, e.g.~with noisy preprocessing or modified CHSH inequalities) uses the same $\prt$ as row~(\scenoneway) in Table~\ref{tab_thresh}, up to relabelling of the inputs. This is not a 2-input 2-output scenario, and so the noise thresholds we can prove for that specific setup are somewhat worse. However, row~(\scenCHSHbasic) is in fact the same scenario with one measurement \emph{omitted}, making it a 2-input 2-output scenario, thus we could use Corollary~\ref{co_22weakbound} to show that advantage distillation in this scenario can surpass the thresholds in~\cite{PAB+09}. Hence we have shown that for the scenario in~\cite{PAB+09}, advantage distillation achieves a higher noise tolerance even while ignoring one measurement. This is particularly surprising since the key-generating measurements in row~(\scenCHSHbasic) are not perfectly correlated. In fact, if the proof in~\cite{PAB+09} were applied to this scenario (by replacing the error-correction term $H(\class{A}_0|\class{B}_0) = \binh(\q)$ with $H(\class{A}_0|\class{B}_0) = \binh(\e)$), it would only tolerate noise up to $\qt \approx 3.1\pct$.
If we instead allow optimisation of the states and measurements for noise robustness, then the relevant rows are (\scenCHSHq) and (\scenCHSHeta), where the noise thresholds we find for advantage distillation also outperform one-way error correction.

In Table~\ref{tab_thresh}, the thresholds for scenarios with more than 2 inputs are generally worse, because for such scenarios we cannot apply Corollary~\ref{co_22weakbound}. The best results we have for such cases are listed in rows (\scenHardyopt) and (\scenMY). It would be of interest to find a way to overcome this issue, perhaps by finding more direct bounds on $F(\rho_{E|00},\rho_{E|11})$, or further study of when the analysis can be reduced to states satisfying Eq.~\eqref{eq_satfvdg}.
We observe that pure states are not the only states satisfying the equation --- for instance, if $\rho_{E|00}$ and $\rho_{E|11}$ are qubit states, the equality holds if and only if they have the same eigenvalues (see~\cite{TLR20}). 

We remark that an open question in quantum information theory is whether there exist entangled states that cannot be used for QKD~\cite{GW00,KL17}. There is a simple analogue to this in the context of DIQKD, namely whether there exist correlations which violate Bell inequalities but cannot be distilled into a secret key in a DI setting. This has recently been answered with an explicit example in~\cite{FBL21}, but the depolarizing-noise threshold derived in that work (at which DIQKD is impossible) still does not match what we have obtained here. In this vein, we also construct a potential attack on this protocol for $\q \gtrsim 12.8\pct$ if parameter estimation is based only on $\e$ and the CHSH value; see~\cite{TLR20}. It would be of interest to see whether further improvements in the noise threshold for advantage distillation could help to close the gaps between these results.

\chapter{Conclusion and future directions}
\label{chap:conclusion}

With the techniques presented in this work~\cite{TSG+21,HST+20,SBV+21,SGP+21,arx_TSB+20,TLR20}, as well as concurrent independent results~\cite{BFF21,WAP21}, we see that substantial progress has been made in improving the keyrates and noise tolerance of DIQKD. In particular, noisy preprocessing and random key measurements have made it possible in principle to achieve positive asymptotic keyrates in various experimental implementations. However, our finite-size analysis shows that for the NV-centre experiment in~\cite{HBD+15} and the cold-atom experiment in~\cite{RBG+17}, impractically large sample sizes (for those implementations) would still be needed in order to achieve a positive finite-size keyrate, even if one makes the optimistic assumption of collective attacks. A significant question that remains to be addressed is that of photonic experiments~\cite{SMC+15,GVW+15}, which achieve lower CHSH values but much larger sample sizes. 

Unfortunately, for photonic experiments the heuristic results suggest that the random-key-measurement approach is less useful in improving the keyrate. 
(A possible cause of the issue may be the fact that the experimental parameters achieving maximal CHSH value also seem to cause at least one of the error-correction terms $H(\hat{A}_0|\hat{B}_0)_\mathrm{hon},H(\hat{A}_1|\hat{B}_1)_\mathrm{hon}$ to be quite large, which decreases the keyrate.) 
Despite this challenge, we note that there is much freedom in parameter optimization for photonic experiments~\cite{MSS20}, and given the techniques we developed here, it is also possible to analyze variants such as choosing different amounts of noisy preprocessing for the two key-generating measurements. It hence seems promising to continue investigating photonic implementations as a prospective candidate for a DIQKD demonstration. (After preparation of this thesis, the approaches in~\cite{arx_BFF21,arx_MPW21} were developed, and they
also studied the possibility of photonic implementations. They indeed found promising results, as we shall briefly discuss below. See also the DIQKD experiments in~\cite{arx_NDN+21,arx_ZLR+21,arx_LZZ+21}.)

Apart from the question of improving parameters for experimental implementations, the various topics we have presented also raise some further questions to consider. We now list a few such possibilities.

\section*{Entropy bounds}

Recently, new approaches were developed in~\cite{arx_BFF21,arx_MPW21} for computing lower bounds on the optimization~\eqref{eq:mainopt} (the former for general nonlocality scenarios, the latter for 2-input 2-output scenarios). These approaches are both computationally efficient and arbitrarily tight, improving over the results presented in this thesis. They applied their methods to a simplified model of photonic Bell experiments, and found that DIQKD would be asymptotically possible with existing detection efficiencies. A critical question now would be to apply their methods to more detailed models, and find whether the same conclusion holds.

Another natural consideration would be computing analogous entropy bounds for the purposes of DIRNG/DIRE, to further improve the keyrates in the existing demonstrations. While some results were obtained in~\cite{TSG+21,BFF21,arx_BFF21}, they have yet to incorporate some subsequent ideas such as the possibility of using random key measurements in DIRNG/DIRE. As mentioned in earlier chapters, preliminary exploration of this idea was performed in~\cite{arx_BRC21}, but there are various details that remain to be finalized.

A somewhat separate point is the question of \emph{semi}-device-independent protocols. This refers to protocols which drop some of the assumptions on the devices in QKD or QRNG, but not as many as in fully DI protocols. For instance, one could assume that only one party's devices are performing uncharacterized measurements~\cite{BCW+12,TFK+13}. To some extent, this can be incorporated in the noncommutative-polynomial-based approaches we have mentioned, because the constraint that some measurements are trusted can be captured (possibly partially) by imposing algebraic relations on the measurement operators. In fact, in~\cite{TSG+21} we used this to compute keyrate bounds for a semi-DI version of the six-state QKD protocol.  Other possibilities include imposing dimension bounds on the system, or more ``physical'' constraints such as energy bounds~\cite{vHWC+17}. Exploring semi-DI protocols\footnote{Another paradigm that is somewhat related but takes an ``opposite'' perspective would be measurement-device-independent (MDI) protocols, in which Alice and Bob can generate trusted states, but measurements are performed by an untrusted third party. The security of such protocols in the basic case can be argued by a form of ``entanglement swapping'' analysis --- the third party's measurement can be viewed as generating entanglement between Alice and Bob. This does perhaps raise a question of whether MDI setups could potentially be ``promoted'' to fully DI versions, with the honest parties' preparation choices essentially corresponding to their input choices. However, to do so, one would need to take care in ensuring that the prepared states do not have ``side-channel'' information leakage regarding the preparation/input choices --- very informally, it seems such conditions on the preparations are a counterpart to the conditions on the measurements in standard DIQKD. Should this idea seem promising, a more rigorous analysis of these conditions would be necessary.} allows for making a tradeoff between the assumptions on the devices and the achievable keyrates, and it would be useful to find a good balance between these considerations.

\section*{Finite-size analysis}

Given the findings of~\cite{arx_BFF21,arx_MPW21} regarding the asymptotic keyrates in potential photonic DIQKD implementations, it would also be important to perform a finite-size analysis to see if an implementation would be possible with realistic sample sizes. Note that their findings indicate that the optimal Bell parameter to consider in these implementations is not necessarily the CHSH value.
While the finite-size analysis in Chapter~\ref{chap:finite} is for a protocol based on the CHSH game, it is not difficult to generalize it to other nonlocality scenarios, so this should not be a significant obstacle 
(essentially, it would just correspond to having a different bound $\lin_{\p}$). 

On a more theoretical level, as discussed in Sec.~\ref{sec:coll}, it is perhaps somewhat curious that in the current EAT-based security proofs, the test/generation decision is performed in an IID manner rather than picking a subset of fixed size as the test rounds. We found some indications that this leads to worse finite-size performance, due to the higher variances. It may hence be useful to consider whether the EAT can be applied in a different fashion to analyze a fixed number of test rounds, or whether new theoretical tools are needed. 
Another relevant question would be whether there is still any room for improvement in the finite-size bounds given by the EAT. As can be seen from Figs.~\ref{fig:experiments} and~\ref{fig:depol}, while they are fairly close to those given by the AEP, there is still something of a gap, which raises the question of whether the gap could be narrowed further or whether it is inherent in the non-IID situation.

\section*{Advantage distillation}

The approaches we used to compute the noise thresholds for the repetition-code protocol were not particularly tight, since they were based on the inequalities~\eqref{eq_fidpg} or~\eqref{eq_relaxedbound}. Recently, a follow-up work~\cite{arx_HT21} developed an algorithm to compute arbitrarily tight bounds on the fidelity in Theorem~\ref{th_fidbound}, hence providing a method to certify that condition arbitrarily well. The key insight is that given any pair of states, there always exists a measurement which preserves the fidelity between the states, which can be used to reduce the analysis to classical side-information for Eve. Still, even with this reduction, additional techniques were needed (based on an algorithm developed in~\cite{vHWC+17}, that again achieves a relaxation to the NPA hierarchy) to handle the fact that the dimension of Eve's side-information is \textit{a priori} unbounded.
 
Surprisingly, however, the noise threshold obtained by applying that approach (together with Theorem~\ref{th_fidbound}) to the~\cite{PAB+09} scenario narrowly failed to outperform the noise threshold for the one-way protocol of~\cite{PAB+09} itself. (This does not contradict our findings in Chapter~\ref{chap:AD}, since the improved thresholds we found there were instead based on the more specialized Corollary~\ref{co_22weakbound}.) Since that method is supposed to give arbitrarily tight bounds on the fidelity, this gives strong evidence that Theorem~\ref{th_fidbound} may in fact not be a very tight condition, in stark contrast to the situation in device-dependent QKD where it is both necessary and sufficient (for security of the repetition-code protocol).\footnote{There is a subtlety here, in that for device-dependent QKD, that condition could also be rewritten in alternative forms that are equivalent in that setting but not in DIQKD; see~\cite{arx_HT21} for details.} It is hence now important to consider how to improve on Theorem~\ref{th_fidbound} itself.

Furthermore, the repetition-code protocol is not the only advantage distillation protocol --- there are various other possibilities~\cite{arx_Myh11,KL17} that may be worth studying for DIQKD, even if only to verify whether the repetition-code protocol still seems to yield the best noise thresholds.
Also, one could consider combining advantage distillation with the other techniques such as noisy preprocessing and random key measurements (the former was indeed considered for device-dependent QKD; see e.g.~\cite{rennerthesis}). However, given the comparative lack of structure in the DI setting, these currently seem to be fairly challenging tasks.

Another significant goal (which is also related to the earlier topic of finite-size analysis) would be finding a security proof that applies for advantage distillation without the assumption of collective attacks. Existing techniques~\cite{VV14,PM13,NBS+18,JMS20,arx_Vid17,ARV19,ZKB18,ZFK20} for proving security against coherent attacks have focused on bounding the smoothed min-entropy of the device outputs, which is certainly sensible for one-way protocols where this is the central quantity of interest. However, when advantage distillation is used, one would instead have to consider various other ``intermediate'' registers, such as the bit $C$ generated in each block of the repetition-code protocol, which is a rather different object from the device outputs. A bound on the smoothed min-entropy of the device outputs does not seem to give us much useful information about these other registers, and hence different proof techniques would be necessary. Somewhat speculatively, appropriate de Finetti theorems~\cite{AR15,arx_JT21} could be a possible approach, though there remain many details to be resolved in applying them for DIQKD.

\appendix

\chapter{Reduction to projective measurements}
\label{app:proj}

\newcommand{\kr}{K} 

It is a well-known fact that any POVM measurement can be dilated to a projective measurement by a suitable embedding. However, there are some subtleties regarding this point. For instance, some versions of this statement only ensure that the dilated projective measurement produces the same output probabilities as the original POVM, without considering the post-measurement quantum states. 
This can be relevant in some cryptographic contexts, where the state after the measurement may be important (this is in some sense related to the question of memory effects in the system). 
Another issue that in nonlocality scenarios, each party has multiple measurements, which raises the question of whether they can all be \emph{jointly} dilated to projective measurements.\footnote{If one is only interested in the question of whether some nonlocal distribution $\pr{ab|xy}$ can be achieved by quantum systems, it is ``well-known'' that such a reduction to projective measurements indeed exists~\cite{HP16}, even in the scenario where Alice and Bob's measurements are only required to commute~\cite{HP16,PT15}. However, in a cryptographic context there are other considerations such as the post-measurement states, hence the argument for that situation may not immediately carry over.} Due to these points, there are in fact situations where it is not valid to assume that all measurements are projective --- for instance, it was found in~\cite{CTM+13} that when considering sequential measurements on a quantum system, there are some correlations that cannot be achieved using only projective measurements (without additional operations between the measurements).

In this appendix, we discuss a reduction to projective measurements that is relevant for device-independent cryptography. The existence of counterexamples such as that in~\cite{CTM+13} indicate that there cannot be a ``fully general argument'' that achieves a reduction to projective measurements in all possible scenarios, hence we shall need to be fairly careful and restrict ourselves to more specialized statements. Fortunately, these are sufficient to cover some number of DI protocols. (In fact, the analysis for DIQKD turns out to be rather simple, but we shall also present an analysis that works for some other protocols. Also, it is worth first discussing the general formalism for quantum measurements.)

\begin{remark}
Some recent works~\cite{BBS20,arx_CHM21} also investigated specialized reductions to projectors (in the contexts of sequential measurements and self-testing respectively), and their results may serve as alternate perspectives to some of the approaches discussed below. We refer the interested reader to those works for more details.
\end{remark}

\section{General form of a measurement}

We first consider the question of finding a suitable expression that captures the most general form that a measurement could take within a quantum framework (the result we shall prove below, Fact~\ref{fact:measrec}, is sometimes referred to as the~\term{quantum instrument} description of measurements). Of course, the exact set of postulates in such a framework is not universally agreed upon. Our goal here is to start from any fairly ``minimal'' framework, in which the structure of post-measurement states is not explicitly specified beforehand, and argue that we can validly restrict our attention to post-measurement states of a specific form. (At the end of this section, we briefly discuss some other perspectives.)
Basically, we can work with any framework where states are described by normalized PSD operators on Hilbert spaces, and dynamics are described by CPTP maps (which can be motivated by standard ``physical'' arguments without involving any statement of the structure of post-measurement states). 
We shall also need some notion of cq systems that is consistent with expressions such as~\eqref{eq:measrec} below. 

To begin, suppose that one measures a state $\rho$, and records the outcome $a$ in a separate classical register. Let us assume the set of possible outcomes is finite. Then most generally, we could say that conditioned on outcome $a$ being obtained, the (normalized) quantum state after the measurement is of the form $\hat{f}_a[\rho]$, for some functions $\hat{f}_a$. 
(We shall call these the post-measurement states; however, it should be understood that the idea here is to not strictly restrict the description to the states ``immediately'' after the measurement --- we can also incorporate any further operations that are performed, up until the next ``significant operation'' that we want to consider, informally speaking.)
We stress that \textit{a priori}, we make no assumptions about the nature of the functions $\hat{f}_a$, apart from requiring them to produce normalized states for any $\rho$ --- we do not even restrict them to be linear. Still, we shall at least take all the post-measurement states to be on a common Hilbert space, e.g.~by embedding all of them in a larger space.\footnote{This is just a technicality to ensure that~\eqref{eq:measrec} is a state on a tensor product of two Hilbert spaces, which is required to some extent when considering the map $\mathcal{P} \otimes \idmap$ in our subsequent proof of Fact~\ref{fact:measrec}.
Pedantically, this requirement may be nontrivial if e.g.~for every $\rho$, $\hat{f}_a[\rho]$ is a state on a different Hilbert space (since this might result in an uncountable family of Hilbert spaces to consider). One way to handle this would be to lean on the assumption in this work that all the Hilbert spaces are finite-dimensional, in which case we can just embed all of them in the sequence space $\ell^2$. Alternatively, the requirement of a common Hilbert space can be treated as an additional, hopefully minor, assumption.}
This process of measuring the state and recording the result can then be written as a function of the form
\begin{align}
\map[\rho] = \sum_a \pr[\rho]{a} \pure{a} \otimes \hat{f}_a[\rho] = \sum_a \pure{a} \otimes f_a[\rho],
\label{eq:measrec}
\end{align}
where we have explicitly denoted the $\rho$ dependence in the outcome probabilities $\pr[\rho]{a}$, and $f_a[\rho] \defvar \pr[\rho]{a} \hat{f}_a[\rho]$ (this implies that $\pr[\rho]{a} = \tr{f_a[\rho]}$ holds by construction; basically $f_a[\rho]$ is the subnormalized post-measurement state conditioned on outcome $a$). 
Thus far, we have imposed very few conditions on these functions.
However, if we now impose the requirement that $\map$ is a legitimate quantum channel (i.e.~a CPTP map), it can be proven that these functions must have a very specific form:
\begin{fact}\label{fact:measrec}
For any CPTP map $\map$ of the form~\eqref{eq:measrec},
the functions $f_a$ must be of the form
\begin{align}
f_a[\rho] = \tr[R]{\kr_a \rho \kr_a^\dagger},
\label{eq:postmeas}
\end{align}
for some ancillary register $R$ and operators $\kr_a$ satisfying $\sum_a \kr_a^\dagger \kr_a = \id$. 
\end{fact}
\begin{proof}
The proof is based on two observations. First, regardless of the nature of the functions $f_a$, a state of the form~\eqref{eq:measrec} is always invariant under the application of a pinching channel $\mathcal{P}[\sigma]\defvar\sum_a \pure{a} \sigma \pure{a}$ (in the ``classical basis'') to the first system, i.e.~we have $(\mathcal{P} \otimes \idmap) [\map[\rho]] = \map[\rho]$.
Second, given that $\map$ is a CPTP map, it has a Kraus representation, $\map[\rho] = \sum_r \widetilde{\kr}_r \rho \widetilde{\kr}_r^\dagger$ for some operators $\widetilde{\kr}_r$ satisfying $\sum_r \widetilde{\kr}_r^\dagger \widetilde{\kr}_r = \id$. Putting these points together allows us to write
\begin{align}
\map[\rho] = (\mathcal{P} \otimes \idmap) [\map[\rho]] 
&= (\mathcal{P} \otimes \idmap) \left[\sum_r \widetilde{\kr}_r \rho \widetilde{\kr}_r^\dagger\right] \nonumber\\
&= \sum_{ar} (\pure{a} \otimes \id) \widetilde{\kr}_r \rho \widetilde{\kr}_r^\dagger (\pure{a} \otimes \id) \nonumber\\
&= \sum_{a} \pure{a} \otimes \sum_{r} \widetilde{L}_{ar} \rho \widetilde{L}_{ar}^\dagger, 
\end{align}
where $\widetilde{L}_{ar} \defvar \sum_a (\bra{a} \otimes \id) \widetilde{\kr}_{r}$. By comparing the last expression to~\eqref{eq:measrec} (and using the orthogonality of the $\pure{a}$ factors), we find that\footnote{An alternative approach: first, using the linearity of $\map$ and orthogonality of the $\pure{a}$ factors, one can prove that all $f_a$ must be linear. Then using the fact that an operator of the form $\sum_{a} \pure{a} \otimes O_a$ is PSD if and only if all $O_a$ are PSD, one can prove that all $f_a$ are completely positive. This implies they can be written in the form~\eqref{eq:subnormkraus} for some operators $\widetilde{L}_{ar}$. Finally, the condition that $\map$ is trace-preserving implies that $\tr{O} = \tr{\map[O]} = \tr{\sum_{ar} \widetilde{L}_{ar}^\dagger \widetilde{L}_{ar} O}$ for any $O$, which imposes the condition $\sum_{ar} \widetilde{L}_{ar}^\dagger \widetilde{L}_{ar} = \id$.}
\begin{align}
f_a[\rho] = \sum_{r} \widetilde{L}_{ar} \rho \widetilde{L}_{ar}^\dagger,
\label{eq:subnormkraus}
\end{align}
furthermore, the operators $\widetilde{L}_{ar}$ have the property $\sum_{ar} \widetilde{L}_{ar}^\dagger \widetilde{L}_{ar} = \id$. 

Now simply define $\kr_a \defvar \sum_r \ket{r} \otimes \widetilde{L}_{ar}$, where the $\ket{r}$ states are on an ancillary system $R$. It is straightforward to verify that in that case,~\eqref{eq:subnormkraus} is equal to~\eqref{eq:postmeas}. Also, we have $\sum_a \kr_a^\dagger \kr_a = \sum_{ar} \inn{r}{r} \widetilde{L}_{ar}^\dagger \widetilde{L}_{ar} = \id$, as desired.
\end{proof}

Since we started by allowing the functions $f_a$ to be as general as reasonably possible, the above argument shows that we incur basically no loss of generality by restricting our analysis to post-measurement states of the form~\eqref{eq:postmeas}. In particular, this will automatically incorporate scenarios such as storing the measurement outcome in $f_a[\rho]$ itself, or erasing the states after measurement and replacing them with different states, or any of a host of other possibilities --- we do not have to worry about missing out \emph{any} possible ``physically allowed'' operation by only considering states of the form~\eqref{eq:postmeas}.

Given this result, any possible quantum measurement process is completely described by the operators $\kr_a$ (which are essentially Kraus operators of a sort, up to technicalities about how to view the role of $R$), with the subnormalized post-measurement states being $\tr[R]{\kr_a \rho \kr_a^\dagger}$ with corresponding probabilities $\pr[\rho]{a} = \tr{\tr[R]{\kr_a\rho\kr_a^\dagger}} = \tr{\kr_a\rho\kr_a^\dagger}$. 
If all the Kraus operators $\kr_a$ are projectors and the $R$ system is trivial, then it is a projective measurement.
In the subsequent discussion, we will mainly focus on the effects of measurements on pure states, since determining the action of a quantum operation on all pure states is sufficient to completely specify its action on mixed states.

We now turn to the setting of multiple possible measurements, indexed by $x\in\mathcal{X}$. In that case, for each such measurement we have its defining Kraus operators $\kr_{a|x}$, satisfying $\sum_a \kr_{a|x}^\dagger \kr_{a|x} = \id$. For convenience, let us take all the measurements to have the same outcome set, by padding out the Kraus operator tuples for each measurement with some trivial Kraus operators (i.e.~that map everything to the zero vector) as necessary.
Also, we shall take all of these operators to have the same codomain by using a suitable embedding, chosen in such a way that the $R$ system is the same for all $x$.\footnote{Here the embedding encounters no technical issues as long as $\mathcal{X}$ is finite, because that means we only have finitely many Hilbert spaces to consider --- if we denote the Kraus operator codomain for each $x$ as $\mathcal{H}_{R_x} \otimes \mathcal{H}_{A'_x}$ (recall that we have already previously taken all the post-measurement states for each $x$ to be on some common register), we can 
take the embedding space to be $\mathcal{H}_{R} \otimes \mathcal{H}_{A'}$ where $\mathcal{H}_{R} \defvar \bigoplus_x \mathcal{H}_{R_x}$, $\mathcal{H}_{A'} \defvar \bigoplus_x \mathcal{H}_{A'_x}$.
}
Let $RA'$ denote the registers corresponding to this codomain Hilbert space, and let $A$ denote the register for the state before measurement. 
With this, it is possible to obtain a partial reduction to projective measurements, as follows:
\begin{fact}\label{fact:pvm}
Consider a family of measurement processes as described above, with Kraus operators $\kr_{a|x}$ satisfying $\sum_a \kr_{a|x}^\dagger \kr_{a|x} = \id_A$ for all $x$. There exists an embedding of the pre-measurement Hilbert space $\mathcal{H}_A$ in a larger Hilbert space $\mathcal{H}_{\widetilde{A}}$, as well as unitaries $U_x$ and orthogonal projectors $P_a$ on $\mathcal{H}_{\widetilde{A}}$ satisfying $\sum_a P_a = \id_{\widetilde{A}}$, such that for any $\ket{\psi}\in\mathcal{H}_A$ and any $a,x$ we have
\begin{align}
P_a U_x \ket{\psi} = \ket{a}_{\hat{A}} \otimes \kr_{a|x} \ket{\psi},
\label{eq:pvmform1}
\end{align}
where $\hat{A}$ is another ancillary register. In turn, this implies there exist some other orthogonal projectors $P_{a|x}$ on $\mathcal{H}_{\widetilde{A}}$ satisfying $\sum_a P_{a|x} = \id_{\widetilde{A}}$ for all $x$, such that for any $\ket{\psi}\in\mathcal{H}_A$ and any $a,x$ we have
\begin{align}
U_x P_{a|x} \ket{\psi} = \ket{a}_{\hat{A}} \otimes \kr_{a|x} \ket{\psi}.
\label{eq:pvmform2}
\end{align}
\end{fact}
\begin{proof}
Under the assumption that the systems are finite-dimensional, the construction to obtain~\eqref{eq:pvmform1} is fairly standard (see e.g.~\cite{NC10}, though here we take a mildly different perspective on the embedding). Here we just outline the argument: take $\mathcal{H}_{\widetilde{A}} = \mathcal{H}_{A} \oplus \mathcal{H}_{\hat{A}RA'}$ so that $\mathcal{H}_{A}$ and $\mathcal{H}_{\hat{A}RA'}$ can be embedded in it (in the ``obvious'' way), define each unitary $U_x$ by requiring that $U_x \ket{\psi} = \sum_a \ket{a}_{\hat{A}} \otimes \kr_{a|x} \ket{\psi}$ for all $\ket{\psi}\in\mathcal{H}_A$ and then extending to a unitary\footnote{This is the step that needs separate handling in the infinite-dimensional case, where one cannot take for granted that an arbitrary isometry on a subspace can be extended to a unitary on the entire space. (For a simple example of this on the Hilbert space $\ell^2$ of square-summable sequences, 
consider the subspace $\{\vec{w}\in\ell^2 \mid w_1 = 0 \}$, and the ``shift isometry'' $V:(0,w_2,w_3,\dots) \to (w_2,w_3,w_4,\dots)$ on this subspace. Since this isometry is surjective on $\ell^2$, it cannot be extended any further.) For one approach to handle this and prove Fact~\ref{fact:pvm} for arbitrary Hilbert spaces, see e.g.~\cite{arx_KST20} where the idea is to introduce a yet larger Hilbert space in which to embed the domain and codomain, instead of trying to extend the isometry within the original Hilbert space.
} on $\mathcal{H}_{\widetilde{A}}$, and take $P_a = \pure{a}_{\hat{A}} \otimes \id_{RA'} \oplus \delta_{a,0} \id_{A}$ (the $\delta_{a,0} \id_{A}$ term is just a minor technicality to ensure that $\sum_a P_a = \id_{\widetilde{A}}$). 
As for the~\eqref{eq:pvmform2} statement, it follows immediately by taking $P_{a|x} = U_x^\dagger P_a U_x$. 
\end{proof}

Fact~\ref{fact:pvm} tells us that in place of thinking of the various measurement processes in terms of arbitrary Kraus operators $\kr_{a|x}$, we can describe them using only projective measurements and unitaries, performed in specific orders (and then tracing out $\hat{A}R$ afterwards). The statement~\eqref{eq:pvmform1} gives a description where a unitary (depending on the measurement choice $x$) is performed first followed by a projective measurement, whereas the statement~\eqref{eq:pvmform2} gives a description where a projective measurement that depends on $x$ is performed first, followed by a unitary that also depends on $x$.

When considering only a single possible choice of measurement, the above results are ``well-known'' reductions to projective measurements, often leading to claims that quantum measurements can be assumed projective without loss of generality (if there is only a single unitary $U_x$ to consider, it can usually be viewed as a basis change with limited physical significance). However, the presence of the $x$-dependent unitaries in Fact~\ref{fact:pvm} causes issues when multiple measurements are involved. For instance, it implies that in situations where we have multiple measurements that can be performed sequentially on a system, one cannot simply take for granted that all the measurements can be treated as projective, e.g.~in the sense of finding projectors such that $P_{a_2|x_2} P_{a_1|x_1} \ket{\psi}$ is isometrically equivalent to $\kr_{a_2|x_2} \kr_{a_1|x_1} \ket{\psi}$. 
Given only Fact~\ref{fact:pvm} by itself, we can only achieve a reduction where some unitaries are interleaved between the projectors, with these unitaries depending on the measurement chosen at each step.

Of course, this immediately raises the question of whether it is possible to derive a version of~\eqref{eq:pvmform2} that only has the projectors $P_{a|x}$ and not the $x$-dependent unitaries $U_x$.\footnote{Another perspective is to argue that the unitaries can be accounted for by including a classical memory register that records a list of all the measurements performed --- the projectors in subsequent measurements can then be conjugated with appropriate unitaries (conditioned on this record) in such a way that we arrive at e.g.~relations of the form $\widetilde{U}_{x_2 x_1} P_{a_2|x_2 x_1} P_{a_1|x_1} \ket{\psi} = \ket{a_2 a_1} \otimes \kr_{a_2|x_2} \kr_{a_1|x_1} \ket{\psi}$, which do not have explicit unitaries between the projective measurements. However, this comes at the cost of needing to introduce an explicit memory and/or a dependence on the measurement history in the projectors describing the measurements. This is a departure from the framework we have been aiming for here, where each measurement is completely described by Kraus operators $\kr_{a|x}$ that only depend on the choice of $x$, and any memory must be specifically accounted for in the post-measurement states produced by these Kraus operators. 
While it might intuitively seem that it must be possible to find projectors $\pvm_{a|x}$ that also store the value of $x$ in the post-measurement state, the approaches used here do not straightforwardly allow such a construction, and furthermore the counterexample in~\cite{CTM+13} places strong obstacles in the path of such a result.
Still, it is at least possible to derive a different (but again, partially restricted) form of ``reduction to projectors'' using this perspective --- see~\cite{BBS20}.}  
However, inspecting the proof of Fact~\ref{fact:pvm} does not seem to yield a straightforward way to do so. More fundamentally, as noted previously, results such as~\cite{CTM+13} impose insurmountable barriers towards such a result --- basically, that work investigated the set of correlations that can be produced when the post-measurement states have the form $P_{a_n|x_n} \dots P_{a_1|x_1} \ket{\psi}$ (with any ``memory'' having to be encoded in the post-measurement states rather than a separate classical register), and showed that it is in fact not even possible to achieve all classical (memory-assisted) correlations using such measurements. 
Given this result, it must be genuinely impossible to derive a ``fully general'' construction where e.g.~we always have projectors such that $P_{a_n|x_n} \dots P_{a_1|x_1} \ket{\psi}$ is isometrically equivalent to $\kr_{a_n|x_n} \dots \kr_{a_1|x_1} \ket{\psi}$.
Hence we need to be careful when making reductions to projective measurements in DI cryptography, and we shall discuss some approaches in the next section.

\subsection{Other possible approaches}

Before moving on, we briefly consider other perspectives that could be taken regarding the nature of measurements in a quantum framework. For example, some frameworks postulate that all measurements are ``fundamentally'' projective, i.e.~the state immediately after a measurement is of the form $\sum_a \pure{a} \otimes P_{a|x} \rho P_{a|x}$ for some projectors $P_{a|x}$, possibly with some embeddings in a larger Hilbert space (here we have 
allowed for multiple possible measurements, indexing the corresponding projector sets with $x$). However, even with this postulate, in the context of DI cryptography there seems to be no reason not to allow the devices to perform some operations before and after the ``fundamental'' projective part.
Since these additional operations are essentially uncharacterized, taking this projective-measurements postulate does not really end up imposing a useful restriction on the overall processes implemented in the devices --- the most convenient way to analyze them may be to simply observe that the resulting states are included within the family described by~\eqref{eq:measrec} (albeit with explicitly denoted dependence on measurement choice $x$), then run through the same line of reasoning as before. (In fact, by implementing unitaries $U_x$ on the post-measurement conditional states $P_{a|x} \rho P_{a|x}$ in a suitable embedding, one obtains expressions matching the LHS of~\eqref{eq:pvmform2}. Hence we can conclude that both perspectives essentially end up describing the same family of possible measurement processes, because Facts~\ref{fact:measrec}--\ref{fact:pvm} imply that the LHS of~\eqref{eq:pvmform2} suffices to describe all processes of the form~\eqref{eq:measrec}, apart from tracing out some registers.)
However, the main perspective we presented previously has the advantage that it did not need to start with the postulate that quantum measurements are ``fundamentally'' projective --- this postulate may not be universally agreed upon. 

Another framework that is sometimes used for quantum measurements is to first only consider POVM elements, i.e.~for each measurement $x$, one has a tuple of PSD operators $\povm_{a|x}$ satisfying $\sum_a \povm_{a|x} = \id$, and the probability of outcome $a$ is given by $\tr{\povm_{a|x} \rho}$. (Here, all the operators $\povm_{a|x}$ must have the same domain and codomain, namely the original Hilbert space $\mathcal{H}_A$.)
This has the advantage of a ``cleaner'' reduction to projectors --- it is always possible to 
embed $\mathcal{H}_A$ in a larger Hilbert space and find some projectors $P_{a|x}$ such that we simply have $\povm_{a|x}=P_{a|x}$ on the subspace $\mathcal{H}_A$ (it is this result that is sometimes referred to as the Naimark dilation). In this case, there are no additional unitaries to consider, and the embedding is ``the same'' across all the operator domains and codomains.
The drawback here is that there is no statement yet regarding the form of the post-measurement states. Some such frameworks postulate that they must be of the form $\sqrt{\povm_{a|x}} \rho \sqrt{\povm_{a|x}}$, where $\sqrt{\povm_{a|x}}$ denotes the unique PSD operator satisfying $\left(\sqrt{\povm_{a|x}}\right)^2 = \povm_{a|x}$. However, this seems almost as restrictive as the projective-measurements-only framework described above, and it is natural to again allow for operations before and after the ``fundamental'' measurement, arriving at similar conclusions. Alternatively, other such frameworks allow the post-measurement states to be of the form $\kr_{a|x} \rho \kr_{a|x}^\dagger$ for any $\kr_{a|x}$ satisfying $\kr_{a|x}^\dagger \kr_{a|x} = \povm_{a|x}$.\footnote{It can be shown using polar decomposition (even in the infinite-dimensional case, given conditions such as boundedness or closure) that any such $\kr_{a|x}$ are related to $\povm_{a|x}$ by $\kr_{a|x}
= V_{a|x}\sqrt{\povm_{a|x}}$ for some 
isometries or partial isometries 
$V_{a|x}$, so another way to view this version is that it simply allows a slightly restricted family of operations $V_{a|x}$ to be applied after $\sqrt{\povm_{a|x}}$.} 
This gives almost the same family of processes as described by~\eqref{eq:postmeas}, apart from the mostly minor issue of the ancillary register $R$. (One slight drawback of not having $R$ is that a pure initial state can then only produce pure post-measurement states; the register $R$ allows for some ``classical mixing'' to produce mixed states.)

\section{In DI cryptography}

We now turn to the question of finding reductions to projective measurements that are suitable for DI cryptography. As it turns out, once the framework we discussed above has been established, this is in fact a fairly simple task for the main protocols discussed in this thesis, such as DIQKD and DIRNG/DIRE. Broadly speaking, the common point in such protocols is that we have honest parties performing local measurements on their systems, with a single adversary Eve holding a separate side-information register. The key insight here is that if we focus on a single round\footnote{In principle, one might try to apply an analogous argument to the measurements producing the full output strings, instead of focusing on single rounds. However, this may not be the most efficient approach, because after reducing the analysis to projective measurements producing the full output strings, the security proofs we have described still involve some reductions to single rounds, which might not preserve the projective structure. Another problem in the sequential setting is that Alice's output might be allowed to depend on Bob's earlier inputs, which means that the operator describing Alice's output string is a rather complicated object that may not be amenable to the analysis described here (which somewhat relies on the tensor-product structure). Hence it seems better to \emph{first} reduce the analysis to single rounds via the various methods described in the rest of the thesis, \emph{then} reduce it to projective measurements.}, the security proofs for all these protocols can be phrased entirely in terms of the ccq states that store Alice and Bob's results and Eve's side-information (the post-measurement states in Alice and Bob's quantum registers are implicitly assumed to be securely discarded or otherwise inaccessible to Eve, since they could store a copy of the measurement outputs), and Alice and Bob's registers are only measured once. This gives us substantial freedom in performing additional operations on Alice and Bob's quantum registers, which will not affect Eve's side-information because quantum theory is non-signalling (in the calculations below, this is mathematically reflected by the cyclic property of partial trace).

More explicitly, let us focus on a single round in the sense of Sec.~\ref{sec:proofsketch} or Chapter~\ref{chap:singlernd} (or Chapter~\ref{chap:AD} for advantage distillation). 
Since determining the effect of a quantum operation on all pure states completely specifies its effect on arbitrary states (alternatively, by following the reductions to pure states in those sections), we consider only pure $\rho_{\qA \qB \qE}$, denoting the pure state vector as $\ket{\rho}_{\qA \qB \qE}$. 
In the sense discussed in the previous section, let the Kraus operators of the various measurements be $\kr_{a|x}$ and $\kr_{b|y}$, and denote the registers for their codomains as $R_A \qA'$ and $R_B \qB'$ respectively. 
Now observe that all quantities of interest for the single-round analysis are determined completely by the following ccq state (more precisely, a family of ccq states, indexed by the inputs $xy$): 
\begin{align}
\rho'_{\hat{A}_x \hat{B}_y \qE} &= \sum_{ab} \pure{ab}_{\hat{A}_x \hat{B}_y} \otimes \tr[R_A \qA' R_B \qB']{\left( \kr_{a|x}\otimes\kr_{b|y}\otimes\id_{\qE}  \ket{\rho}_{\qA \qB \qE} \right) (\text{h.c.})} , 
\label{eq:ccqpostmeas}
\end{align}
where for brevity we use ${\ket{v}(\text{h.c.})}$ to denote ${\pure{v}}$ (i.e.~appending the hermitian conjugate of $\ket{v}$). Specifically, the conditional entropies $H(\hat{A}_x|{E})$ (even after noisy preprocessing) and the probabilities $\pr{ab|xy}$, as well as all relevant states in the security analysis in Chapter~\ref{chap:AD}, can all be written entirely in terms of these states.

Hence to reduce the analysis to projective measurements, it would suffice to find some projective measurements and pre-measurement states such that we obtain the same ccq states after the measurement. This can be achieved using Fact~\ref{fact:pvm}: consider projectors $\pvm_{a|x},\pvm_{b|y}$ and unitaries $U_x,U_y$ as constructed in that statement, on registers $\widetilde{\qA},\widetilde{\qB}$ (for larger Hilbert spaces) respectively. Suppose the devices perform projective measurements described by this choice of $\pvm_{a|x},\pvm_{b|y}$, and take the pre-measurement state to be $\ket{\rho}_{\widetilde{\qA}\widetilde{\qB} \qE}$, i.e.~the state $\ket{\rho}_{\qA \qB \qE}$ with the Alice-Bob registers embedded in $\widetilde{\qA}\widetilde{\qB}$ (note that this means $\ket{\rho}_{\widetilde{\qA}\widetilde{\qB} \qE}$ lies within the subspace $\mathcal{H}_{\qA \qB \qE} $ of the embedding space $\mathcal{H}_{\widetilde{\qA} \widetilde{\qB} \qE}$).
Then the corresponding ccq states after measurement are
\begin{align}
& \sum_{ab} \pure{ab}_{\hat{A}_x \hat{B}_y} \otimes 
\tr[\widetilde{\qA}\widetilde{\qB}]{ \left(\pvm_{a|x}\otimes\pvm_{b|y}\otimes\id_{\qE} \ket{\rho}_{\widetilde{\qA}\widetilde{\qB} \qE}\right) (\text{h.c.})} \nonumber\\
=& \sum_{ab} \pure{ab}_{\hat{A}_x \hat{B}_y} \otimes 
\tr[\widetilde{\qA}\widetilde{\qB}]{ \left(U_x\pvm_{a|x}\otimes U_y\pvm_{b|y}\otimes\id_{\qE} \ket{\rho}_{\widetilde{\qA}\widetilde{\qB} \qE}\right) (\text{h.c.})} 
\nonumber\\
=& \sum_{ab} \pure{ab}_{\hat{A}_x \hat{B}_y} \otimes 
\tr[\widetilde{\qA}\widetilde{\qB}]{ \left(\ket{ab}_{\hat{A}\hat{B}} \otimes \left(\kr_{a|x}\otimes \kr_{b|y}\otimes\id_{\qE} \ket{\rho}_{\widetilde{\qA}\widetilde{\qB} \qE}\right)\right) (\text{h.c.})} 
\text{ since }  \ket{\rho} \in \mathcal{H}_{\qA \qB \qE}
\nonumber\\
=& \sum_{ab} \pure{ab}_{\hat{A}_x \hat{B}_y} \otimes 
\tr[\hat{A} R_A \qA' \hat{B} R_B \qB']{ \left(\ket{ab}_{\hat{A}\hat{B}} \otimes \left(\kr_{a|x}\otimes \kr_{b|y}\otimes\id_{\qE} \ket{\rho}_{\qA \qB \qE}\right)\right) (\text{h.c.})} ,
\end{align}
where in the last line we have explicitly written the relevant embedded registers within $\widetilde{\qA}\widetilde{\qB}$.
Evaluating the partial trace over $\hat{A}\hat{B}$ results in exactly the same expression as~\eqref{eq:ccqpostmeas}, as desired. 

\begin{remark}
As a slightly different approach, observe that the states~\eqref{eq:ccqpostmeas} can in fact be written entirely in terms of the POVM elements of the measurements, i.e.~$\povm_{a|x}= \kr_{a|x}^\dagger \kr_{a|x}$, $\povm_{b|y} = \kr_{b|y}^\dagger \kr_{b|y}$. Using a similar line of reasoning as above (or by directly using the Naimark dilation for POVMs), we can find a common embedding for all of Alice's POVM elements $\povm_{a|x}$ and projectors $\pvm_{a|x}$ on the larger space, such that $\povm_{a|x}=\pvm_{a|x}$ on $\mathcal{H}_A$. Doing the same for Bob, we again end up with projective measurements producing the same final ccq states. This approach has a small advantage of not needing to handle the Kraus operators directly, instead essentially only requiring a framework for POVM measurements that yields just the output probabilities. However, the scenarios discussed in the remainder of this section seem less amenable to this approach.
\end{remark}

Putting DIQKD aside, however, there are other DI protocols where the above argument may not suffice. For instance, a family of protocols with rather different structure from QKD or DIQKD would be the family of two-party ``distrustful'' protocols, in which one of the parties can be dishonest --- classic examples of such protocols are coin-flipping, bit commitment, and oblivious transfer. DI protocols in this setting typically centre around the two quantum registers held by the parties, and the dishonest party usually has access to the quantum states (on their register) after the measurement. This introduces some new challenges as compared to security proofs against a single third-party adversary, in that the structure of the post-measurement states may be more ``directly'' involved in the security proof (as compared to the preceding discussion, where they were basically discarded). 

We shall now describe a reduction to projectors that works for a particular family of such protocols. Specifically, we consider the DI certified-deletion protocols in~\cite{FM18,arx_KT20}, where the aim is for one party (say, Bob) to prove that they have deleted some information. The key observation that these protocols rely on is the following property of a nonlocal game known as the magic-square game (we describe the game in detail in Appendix~\ref{asec:MS}): if Alice and Bob implement the states and measurements~\eqref{eq:MSmeas} that perfectly win the magic-square game, then while Bob knows Alice's output bit $a_y$ exactly, he has no information about any other bits in Alice's string $a$ --- the measurement that produces Bob's output also ``deletes'' all information regarding those other bits. By combining this fact with the game's robust self-testing properties~\cite{WBM+16}, the following result was derived by~\cite{FM18}:\footnote{We highlight some technicalities in the phrasing of the result here, which is something of an ``intermediate'' statement between Proposition~2 and Corollary 3 of~\cite{FM18}. 
First, note that in this statement, there is no probability distribution for $x,y,y'$ --- rather, the part involving the magic-square winning probability is a ``hypothetical'' situation where the game is played with those states and measurements (which would involve random inputs, but not the values denoted in this statement), whereas the subsequent parts involve ``fixed'' values of $x,y,y'$, not random values.
Also, no explicit statement is made regarding whether Bob ``knows'' the values of $x,y'$ when producing the guess for $a_{y'}$ --- the bound is about \emph{any} guessing measurement on register $\qB$, so the question is ill-posed in some senses. (Though given the quantifier ordering in this phrasing, the optimal guessing measurement could be different for each $x,y'$, so in that sense Bob is allowed to ``know'' $x,y'$.)
However, in the context of applying this result in games or protocols, typically the situation of interest is along the lines of first having the parties play the standard magic-square game, then having Bob produce a guess for $a_{y'}$. In such cases, it is usually important to at least enforce that Bob does not know Alice's input $x$ {before} playing the magic-square game --- if he does, it is possible for him to learn her entire string by performing the same measurement as her on the shared state $\ket{\Phi^+}\ket{\Phi^+}$.
}
\begin{fact}\label{fact:certdel}
(See Proposition~2 and Corollary 3 of~\cite{FM18}) Suppose Alice and Bob have a state and projective measurements (on some registers $\qA \qB$) that can win the magic-square game with probability $1-\delta$. Take any values $x,y,y' \in \{0,1,2\}$ such that $y \neq y'$. Now suppose that Alice and Bob perform the aforementioned measurements for inputs $x,y$, obtaining outputs $a,b$. 
If Bob now implements any measurement on register $\qB$ to guess the value of $a_{y'}$, the probability of him guessing correctly is at most
$\frac{1}{2}+9\sqrt{\delta}$.
\end{fact}

We will not give the proof here; refer to~\cite{FM18} for the details. Instead, our focus is on explaining why a somewhat stronger statement also holds:
\begin{fact}\label{fact:certdelproj}
Fact~\ref{fact:certdel} holds even without the condition that the measurements are projective.
\end{fact}
\noindent The works~\cite{FM18,arx_KT20} use this as a building block to construct more complex DI protocols for various certified-deletion tasks. The fact that it holds for arbitrary measurements is important for this purpose, because the ``uncharacterized'' devices in DI protocols are not restricted to projective measurements. Again, our focus is not on the protocol details and security proofs, but rather we shall just explain how to prove Fact~\ref{fact:certdelproj}:
\begin{proof}
We can again suppose that the pre-measurement state is a pure state $\ket{\rho}_{\qA \qB}$ without loss of generality, e.g.~by giving a purifying register to Bob.
As before, let the Kraus operators of the measurements (for the magic-square part) be $\kr_{a|x}$ and $\kr_{b|y}$, and denote the registers for their codomains as $R_A \qA'$ and $R_B \qB'$ respectively. For each input pair $xy$, let $\omega^{xy}_{B' \land ab}$ denote Bob's subnormalized post-measurement state conditioned on outputs $ab$. In that case, we can write (ignoring the classical output registers, since the Kraus operators can anyway be defined such that they store a copy of the outputs):
\begin{align}
\omega^{xy}_{B' \land ab} = \tr[R_A \qA' R_B]{\left( \kr_{a|x}\otimes\kr_{b|y} \ket{\rho}_{\qA \qB} \right) (\text{h.c.})} .
\label{eq:postmeasBob}
\end{align}
Bob's optimal probability of guessing $a_{y'}$ can be written entirely as a function of the states $\omega^{xy}_{B' \land ab}$ (more precisely we only need to consider a ``coarse-graining'' of these states by summing over values of $a$ with the same value of $a_{y'}$, but our argument works even without this coarse-graining). This implies that to prove the desired claim, it would suffice to find some projective measurements achieving the same magic-square winning probability, such that there is some operation Bob can apply to the resulting subnormalized post-measurement states to obtain the same states $\omega^{xy}_{B' \land ab}$
(since guessing probability is nonincreasing under local operations\footnote{To apply this fact correctly, the operation cannot depend explicitly on any random variables Bob does not have direct access to, such as Alice's output $a$. However, by the phrasing of Fact~\ref{fact:certdel}, it can possibly depend on $x,y,y'$, though our construction will only involve $y$ (which Bob certainly has access to).} applied to the side-information system).

To do so, we again consider projectors $\pvm_{a|x},\pvm_{b|y}$ and unitaries $U_x,U_y$ as constructed in Fact~\ref{fact:pvm}, on registers $\widetilde{\qA},\widetilde{\qB}$ (for larger Hilbert spaces) respectively. 
Suppose the devices perform projective measurements described by this choice of $\pvm_{a|x},\pvm_{b|y}$, on the pre-measurement state  $\ket{\rho}_{\widetilde{\qA}\widetilde{\qB}}$, i.e.~$\ket{\rho}_{\qA \qB}$ embedded in $\widetilde{\qA}\widetilde{\qB}$. This yields the same nonlocal distribution $\pr{ab|xy}$ (as easily verified from~\eqref{eq:pvmform2}) and thus the same magic-square winning probability. Also, Bob's subnormalized post-measurement states would be
\begin{align}
\tr[\widetilde{\qA}]{\left( \pvm_{a|x}\otimes\pvm_{b|y} \ket{\rho}_{\widetilde{\qA}\widetilde{\qB}} \right) (\text{h.c.})} .
\end{align}
Now suppose Bob applies the unitary $U_y$ to his system, then traces out the registers $\hat{B} R_B$. The resulting states would be
\begin{align}
&\tr[\widetilde{\qA} \hat{B} R_B]{\left( \pvm_{a|x}\otimes U_y\pvm_{b|y} \ket{\rho}_{\widetilde{\qA}\widetilde{\qB}} \right) (\text{h.c.})} \nonumber\\
=& \tr[\widetilde{\qA} \hat{B} R_B]{\left( U_x \pvm_{a|x}\otimes U_y\pvm_{b|y} \ket{\rho}_{\widetilde{\qA}\widetilde{\qB}} \right) (\text{h.c.})} 
\nonumber\\
=& \tr[\widetilde{\qA} \hat{B} R_B]{\left(\ket{ab}_{\hat{A}\hat{B}} \otimes \left(\kr_{a|x}\otimes\kr_{b|y} \ket{\rho}_{\widetilde{\qA}\widetilde{\qB}}\right) \right) (\text{h.c.})} 
\text{ since }  \ket{\rho} \in \mathcal{H}_{\qA \qB} \nonumber\\
=& \tr[R_A \qA' R_B]{\left(\kr_{a|x}\otimes\kr_{b|y} \ket{\rho}_{{\qA}{\qB}}\right) (\text{h.c.})} .
\end{align}
These are exactly the same states as~\eqref{eq:postmeasBob}, as desired.
\end{proof}

A key point that allowed the above argument to work was the fact that in this context, we could allow Bob to perform some operations with explicit $y$ dependence (after the projective measurement). To somewhat highlight why this is nontrivial, we could consider what would happen if we had attempted a different approach by describing the process of Bob producing a guess $g$ as another measurement, with some Kraus operators $\kr_g$. Notice we do not need to include an explicit $y$ dependence in these Kraus operators, since the Kraus operators $\kr_{b|y}$ can be defined such that they store copies of $y$ (by the generality of Fact~\ref{fact:measrec}), allowing the $\kr_g$ operators to just ``read off'' the registers storing $y$. In that case, the distribution of the various classical values produced in this process would be given by $\tr{\left( \kr_{a|x}\otimes\kr_g \kr_{b|y} \ket{\rho}_{\qA \qB} \right) (\text{h.c.})}$. However, it would \emph{not} necessarily be possible to find projectors (and an embedding) such that $\tr{\left( \pvm_{a|x}\otimes\pvm_g \pvm_{b|y} \ket{\rho}_{\widetilde{\qA}\widetilde{\qB}} \right) (\text{h.c.})} = \tr{\left( \kr_{a|x}\otimes\kr_g \kr_{b|y} \ket{\rho}_{\qA \qB} \right) (\text{h.c.})}$, because of the unitaries $U_y$ that appear in~\eqref{eq:pvmform2}. Hence a reduction to projectors in this strong sense might not be possible. (One workaround would be to instead include an explicit $y$ dependence in the guessing-measurement projectors, but whether this is acceptable would depend on context --- for instance, another perspective on the above proof is that it works because this $y$ dependence is acceptable.)

Of course, ideally we would like to have systematic procedures for reductions to projectors for general DI protocols, but currently we appear to be restricted to more case-by-case approaches. (Yet another protocol where such a reduction is possible is in~\cite{arx_KST20}, exploiting a specific property of the protocol, 
but we will not describe it here.) Still, the ideas used in the above arguments should generalize to some families of DI protocols at least, which may be a useful starting point. Depending on the protocol, the (rather different) approaches described in~\cite{BBS20,arx_CHM21} may also be usable; refer to those works for further details.

As a closing remark, we note that the constructions in this section all relied fairly heavily on the tensor-product structure between Alice and Bob. Roughly, this allowed us to reduce the analysis to projective measurements independently for each of their systems, then combine them as a tensor product. In frameworks where Alice and Bob's measurements are only required to commute, one faces the more challenging task of preserving the commutation relations while constructing the projectors. As mentioned previously, this is possible when focusing solely on the Alice-Bob distributions $\pr{ab|xy}$~\cite{HP16,PT15}; however, we leave the question of whether this is possible for DIQKD to future work.

\chapter{Foundational perspective on DI cryptography}
\label{app:foundations}

(The following is a more detailed exposition of an argument described in e.g.~\cite{BCP+14}; see also an in-depth analysis in~\cite{woodheadthesisnodoi}.)
We begin by proving a small lemma:
\begin{lemma}\label{lem:NSprod}
Suppose a distribution $\pr{ab|xy}$ satisfies the NS conditions~\eqref{eq:NS} and is deterministic, i.e.~$\pr{ab|xy} \in \{0,1\}$ for all $a,b,x,y$. Then it must also be product, i.e.~$\pr{ab|xy}=\pr{a|x}\pr{b|y}$.
\end{lemma}
\begin{proof}
For any input pair $xy$, $\pr{ab|xy}$ is a valid probability distribution for $ab$. Given that $\pr{ab|xy}$ only takes values in $\{0,1\}$, this implies that for each input pair $xy$ there must be exactly one ``distinguished'' output pair $a^\star_{xy}b^\star_{xy}$ (the subscript denotes the input dependence) that occurs with probability $1$, with all others having probability $0$. 
We can summarize this as
\begin{align}
\pr{ab|xy} 
= \delta_{a,a^\star_{xy}} \delta_{b,b^\star_{xy}} .
\end{align}
Then we have (by the definitions of $\pr{a|x}$ and $\pr{b|y}$)
\begin{align}
\pr{a|x} 
= \sum_b \delta_{a,a^\star_{xy}} \delta_{b,b^\star_{xy}} 
= \delta_{a,a^\star_{xy}} ,\qquad
\pr{b|y}
= \sum_a \delta_{a,a^\star_{xy}} \delta_{b,b^\star_{xy}} 
= \delta_{b,b^\star_{xy}} .
\label{eq:detmarg}
\end{align}
Hence $\pr{ab|xy}=\pr{a|x}\pr{b|y}$, as claimed. (Note that putting together the NS conditions with~\eqref{eq:detmarg} implies that in fact $a^\star_{xy}$ is independent of $y$ and $b^\star_{xy}$ is independent of $x$, as we would intuitively expect in order to fulfill the NS property.)
\end{proof}

Now consider the following reasoning.
Suppose that a distribution $\pr{ab|xy}$ is described by an ``NS hidden-variable model'', in the sense that we can write it in the form $\pr{ab|xy} = \sum_\lambda \pr{ab|xy\lambda} \pr{\lambda}$ such that for each $\lambda$, $\pr{ab|xy\lambda}$ satisfies the NS conditions. 
Now note that if the hidden-variable model is deterministic (i.e.~all $\pr{ab|xy\lambda}$ are either $0$ or $1$), then by Lemma~\ref{lem:NSprod} we know that we must have $\pr{ab|xy\lambda} = \pr{a|x\lambda}\pr{b|y\lambda}$. This means that $\pr{ab|xy}$ admits an LHV model, and hence cannot violate a Bell inequality. 

Taking the contrapositive, we have the following statement for any distribution $\pr{ab|xy}$ described by an NS hidden-variable model: if the distribution violates a Bell inequality, then the hidden-variable model cannot be deterministic. Importantly, this means that in some sense there is ``fundamental randomness'' in the outputs: even if an adversary has perfect knowledge of the hidden variable $\lambda$, they cannot perfectly predict all the outputs. This is useful for cryptography, since randomness that cannot be predicted by an adversary is essentially ``secret'' randomness. Notice that this line of reasoning was extremely general, i.e.~it is not specific to quantum theory --- it only used the notion that the underlying model satisfies the NS conditions, which one might hope any ``reasonable'' physical theory satisfies (otherwise the theory would in principle allow setting up scenarios with instantaneous signalling between parties, unless perhaps some restriction is placed on how well the hidden variable $\lambda$ can be ``fixed''). There are, however, some points to be careful of, which we address at the end of this section.

\newcommand{\ga}{a'}
\newcommand{\gb}{b'}
To explicitly cast the quantum setting in this framework, let us consider a distribution produced by a state $\rho_{\qA\qB\qE}$ and measurements $\pvm_{a|x},\pvm_{b|y}$ as in~\eqref{eq:qprobs}. Now suppose that Eve performs a measurement with POVM operators $\pvm_{\ga\gb}$ on her system, obtaining a pair of outputs $\ga\gb$ which is intended to be a guess of Alice and Bob's outputs. 
Since this measurement commutes with Alice and Bob's measurements, we can suppose that it is instead performed first, in which case the state in Alice and Bob's devices conditioned on Eve getting output $\ga\gb$ is
\begin{align}
\rho_{\qA \qB | \ga\gb} = \frac{1}{\prnoscale{\ga\gb}} \tr[\qE]{\left(\id_{\qA}\otimes\id_{\qB}\otimes\pvm_{\ga\gb}\right) \rho_{\qA \qB \qE}},
\end{align}
where $\prnoscale{\ga\gb}  \defvar \tr{\pvm_{\ga\gb} \rho_{\qE}}$ is the distribution of Eve's outputs. 
In terms of this state, we can write 
(using the fact that $\sum_{\ga\gb} \pvm_{\ga\gb}=\id_{\qE}$):
\begin{align}
\pr{ab|xy} = \tr{\left(\pvm_{a|x}\otimes\pvm_{b|y}\otimes\sum_{\ga\gb} \pvm_{\ga\gb}\right) \rho_{\qA \qB \qE}} &= \sum_{\ga\gb} \tr{\pvm_{a|x}\otimes\pvm_{b|y} \, \rho_{\qA \qB | \ga\gb}} \prnoscale{\ga\gb} . 
\end{align}
The idea is to now view $\ga\gb$ as the hidden variable $\lambda$ in our above discussion. In that case, we see that the distribution $\pr{ab|xy\ga\gb} = \tr{\pvm_{a|x}\otimes\pvm_{b|y} \, \rho_{\qA \qB | \ga\gb}}$ indeed describes an NS hidden-variable model in the sense previously discussed. Hence that line of reasoning applies, i.e.~if $\pr{ab|xy}$ violates a Bell inequality, then it is impossible for $\pr{ab|xy\ga\gb}$ to be deterministic. This in turn implies that Eve's guess $\ga\gb$ cannot always be equal to $ab$.

However, while this line of reasoning can serve to motivate the possibility of DI cryptography, there are some technicalities to be aware of. In particular, we have not considered the possibility of Eve choosing her measurement depending on the inputs $xy$, whereas in many DI protocols, she does learn these inputs (after Alice and Bob's measurements), and should be able to choose her measurement depending on them\footnote{If Eve's measurement can depend on $xy$, then it produces a conditional Alice--Bob state of the form $\rho_{\qA \qB | \ga\gb xy}$, in which case the resulting distribution $\tr{\pvm_{a|x}\otimes\pvm_{b|y} \, \rho_{\qA \qB | \ga\gb xy}}$ is not of a form that straightforwardly satisfies the NS conditions, and our above reasoning does not directly apply.} (though for protocols using a fixed input pair in all generation rounds, this is not an issue, since Eve only needs to perform the optimal measurement for guessing the outputs from that input pair). Another issue is that although we have proven that (for $\pr{ab|xy}$ violating a Bell inequality and described by NS hidden variables) not \emph{all} the values $\pr{ab|xy\lambda}$ can be in $\{0,1\}$, this does not rule out the following possibility: 
there may be a distribution $\pr{ab|xy}$ that violates some Bell inequality, but has the property that for any specific input pair $x^\star y^\star$, there exists an NS hidden-variable model for $\pr{ab|xy}$ such that $\pr{ab|x^\star y^\star \lambda} \in \{0,1\}$ for all $a,b,\lambda$. (Note the quantifier ordering --- the NS hidden-variable models are necessarily different for each $x^\star y^\star$, otherwise our earlier conclusions would be violated.) In fact, explicit examples of such behaviour are known, for instance when considering NS 
(but not quantum, due to some self-testing results) 
distributions for a nonlocal game known as the magic-square game; we present the details in the next section. This could be a problem if a fixed input pair is used in all generation rounds, since Eve would be able to perfectly predict the outputs for that input pair --- regarding our above discussion of the quantum setting, she does not even need to choose her measurement as a function of $xy$ in this case, since she only needs to predict the outputs for that input pair. (Protocols that have multiple generation inputs, such as the ones in~\cite{SGP+21,JMS20,arx_Vid17}, may avoid this issue.)
To deal with these points, a more detailed analysis can be found in~\cite{woodheadthesisnodoi} --- as a quick summary, the idea is to define a superset of LHV models that they call \term{partially deterministic} models, and show that if a distribution $\pr{ab|xy}$ violates an inequality that holds for all partially deterministic models, then it must allow DI certification of ``secret randomness''.

\section{NS strategies for the magic-square game}
\label{asec:MS}

Here we quickly introduce the magic-square game, and present in more detail how it may be insecure against NS adversaries in a DI context.
The magic-square game is a nonlocal game defined as follows\footnote{There is in fact significant variability across different works considering the magic-square game (and not merely in the sense of using different input-output labellings and/or parity constraints). For instance, some works instead allow both devices to output strings of arbitrary parity, and incorporate the parity constraints into the win condition. Other works instead consider a setting where the inputs for both Alice and Bob can be rows or columns, making it a situation with 6 inputs per party instead of 3. Here we present the version that is easiest to describe for our purposes.} (intuitively, it can be visualized as having Alice fill in a row and Bob fill in a column of a $3\times3$ grid while satisfying some parity constraints, with the win condition being to have matching values on the square in common between the row and column):
\begin{itemize}
\item Alice and Bob supply uniformly random inputs $x,y\in\{1,2,3\}$.
\item Alice's device outputs a 3-bit string of even parity, and Bob's device outputs a 3-bit string of odd parity (i.e.~$a \in \{000,011,101,110\}$
and $b \in \{001,010,100,111\}$).
\item The devices win the game iff $a_y = b_x$.
\end{itemize}
The winning probability of this game is a Bell parameter in the sense mentioned in Sec.~\ref{sec:Bellineq}. Note that achieving a winning probability of $1$ on this game is equivalent to stating that for all input pairs $xy$, the outputs always satisfy the win condition $a_y = b_x$. (This equivalence also holds for any other input distribution with full support.)

Using LHV strategies, the magic-square game can only be won with probability at most $8/9$; however, there exists a simple quantum strategy that wins with probability $1$ (in other words, a nontrivial Bell inequality holds for the winning probability of the magic-square game). An example of such a quantum strategy would be as follows: the devices share the state $\ket{\Phi^+}_{AB}\ket{\Phi^+}_{A'B'}$, with $AA'$ for Alice and $BB'$ for Bob (i.e.~each party holds two qubits). On input $x$, Alice's device produces 3 output bits by performing the 3 measurements in row $x$ of the following table on her qubits; analogously, on input $y$, Bob's device performs the 3 measurements in column $y$ of the table on his qubits:
\begin{align}
\begin{array}{c c c}
Z \otimes \id & \id \otimes Z & Z \otimes Z \\
\id \otimes X & X \otimes \id & X \otimes X \\
- Z \otimes X & - X \otimes Z & Y \otimes Y
\end{array}
\label{eq:MSmeas}
\end{align}
(For compactness, here we specified the measurements as Pauli observables, with the eigenvalues $+1/-1$ being mapped to bit-valued outputs $0/1$ respectively. It can be verified in every row or column, the 3 measurements are compatible and can thus be performed simultaneously.)
This strategy produces a uniform distribution over all winning outputs, i.e.~the resulting distribution has $\pr{ab|xy}=1/8$ for all $abxy$ combinations that win the game (and $\pr{ab|xy}=0$ otherwise). 
We shall refer to these as the ``canonical'' quantum strategy and distribution for the magic-square game (of course, any unitarily equivalent states and measurements also produce the same distribution, but we shall use this term for brevity).

The magic-square game has a useful symmetry property, as follows. Consider any permutation $\pi$ on the set $\{1,2,3\}$, and let $f_\pi$ be the corresponding ``bitwise permutation'' on 3-bit strings 
$c\in\mathbb{Z}_2^3$ defined by 
$f_\pi(c)_{\pi(z)} = c_z$. Given any such $\pi$ and any distribution $\pr{ab|xy}$ that perfectly wins\footnote{The following symmetry argument also applies if the distribution only wins the game with probability $p<1$, i.e.~it produces another distribution that achieves the same winning probability $p$. However, in this case it would rely on the inputs $xy$ being uniform --- it may not immediately work for versions of the magic-square game with non-uniform inputs.} the magic-square game, we can construct another distribution $\prnew{ab|xy}$ that also perfectly wins the magic-square game, namely $\prnew{ab|xy} \defvar \pr{f_\pi(a) \, b | x \, \pi(y)}$. To see that this new distribution indeed always wins the magic-square game, observe that for any $abxy$ that occurs with nonzero probability under $\prnew{ab|xy}$, we have $a_y = f_\pi(a)_{\pi(y)} = b_x$ as desired, where 
the second equality holds because of the 
definition of $\prnew{ab|xy}$ in terms of the original distribution.
Furthermore, since this corresponds to simply a relabelling of Alice's outputs and Bob's inputs, it also preserves any LHV/quantum/NS properties of the original distribution. Clearly, similar relabellings can also be applied to Bob's outputs and Alice's inputs instead, i.e.~taking $\prnew{ab|xy} \defvar \pr{a \, f_\pi(b) | \pi(x) \, y}$.
(Note that $\prnew{ab|xy}$ is not necessarily distinct from $\pr{ab|xy}$, e.g.~when $\pr{ab|xy}$ is the canonical quantum distribution --- this is a corollary of the self-testing properties we shall now discuss. However, this is not an issue for our subsequent arguments based on these symmetries.) 

Another important property of this game is that it has self-testing properties in the quantum case --- the states and measurements that perfectly win the magic-square game are essentially unique, up to local isometries/ancillas. (Robust versions of this statement are also known~\cite{WBM+16}.) In terms of nonlocal distributions $\pr{ab|xy}$, this immediately implies that the only quantum distribution that perfectly wins the magic-square game is the canonical one. It also has useful implications for DI protocols, as follows. Recall that the canonical quantum strategy involves a pure state between Alice and Bob, and has $\pr{ab|xy} \neq 1$ for all $abxy$. Slightly informally, this ensures that devices that perfectly win the magic-square game will produce device-independently secure randomness against quantum adversaries, since they must essentially be implementing this pure Alice--Bob state (i.e.~the adversary cannot have any nontrivial purification/side-information) and none of the outputs are deterministic.

However, this statement no longer holds if we consider NS adversaries instead.\footnote{Another way to see that the magic-square game does not self-test NS strategies is to consider the main result of~\cite{RTHH16}, which showed that the \emph{only} extremal NS distributions that are quantum-realizable are the ``trivial'' ones that are already LHV-realizable. In particular, this means that since the canonical magic-square quantum distribution is not LHV-realizable, it cannot be an extremal NS distribution, so it cannot be self-tested with respect to the set of NS distributions. (Also, since this quantum strategy already reaches the maximum winning probability, we cannot try to obtain a self-testing result for NS strategies by ``raising the threshold'' on the required winning probability, as in e.g.~the scenario of self-testing PR boxes via CHSH. Similar conclusions should also hold for any other Bell inequality in which the maximal values for quantum and NS distributions coincide.) However, this does not directly seem to lead to a conclusion that it is ``fully insecure'' against NS adversaries in the sense we now describe (which requires not only a decomposition into other NS distributions, but also the property that the distributions in the decomposition allow the adversary to perfectly guess the output); we defer to~\cite{RTHH16} and follow-up works for further discussion.} More precisely, we have the following counterexample: 
\begin{fact}\label{fact:NSMS}
Consider the magic-square game. For any input pair $x^\star y^\star$, there exists a distribution $\pr{ab|xy}$ that perfectly wins the game and is described by an NS hidden-variable model (i.e.~we have $\pr{ab|xy} = \sum_\lambda \pr{ab|xy\lambda} \pr{\lambda}$ such that each $\pr{ab|xy\lambda}$ satisfies the NS conditions) with $\pr{ab|x^\star y^\star \lambda} \in \{0,1\}$ for all $a,b,\lambda$. Furthermore, this distribution can be chosen such that $\pr{ab|xy} \neq 1$ for all $a,b,x,y$ (in fact, it can be chosen such that $\pr{ab|xy} \leq 1/2$).
\end{fact}
\noindent 
This means that any magic-square DI protocol that uses a fixed input pair to generate output randomness would be insecure against an NS adversary, who could hold the hidden variable $\lambda$ and perfectly predict the outputs for that input pair. 
The possibility of ensuring that $\pr{ab|xy} \neq 1$ is relevant in this context because if the distribution $\pr{ab|xy}$ itself attains $\pr{ab|xy} = 1$ for some input-output combinations, then in principle Alice and Bob could detect the attack (given enough samples, possibly under an IID assumption) by estimating all the probabilities $\pr{ab|xy}$ individually instead of just the magic-square winning probability. By having all probabilities $\pr{ab|xy}$ far away from $1$, the adversary ``hides'' the attack, in that the deterministic behaviour only appears when conditioned on the adversary's side-information $\lambda$, not when considering the values $\pr{ab|xy}$ that Alice and Bob can estimate.

We now prove Fact~\ref{fact:NSMS}, by presenting explicit distributions with the claimed property.
\begin{proof}

\newcommand{\matMS}[5]
{\bgroup\renewcommand*{\arraystretch}{1}
$\begin{array}{c|c c c c}
#1 & 001 & 010 & 100 & 111 \\
\hline
000 & #2 \\
011 & #3 \\
101 & #4 \\
110 & #5
\end{array}$
\egroup}
\newcommand{\lose}{${\color{red}\xmark}$}

We shall first prove the claim for the specific case $x^\star y^\star = 11$, and without imposing the condition $\pr{ab|xy} \neq 1$ yet. We do so by simply presenting a list of all the probabilities $\pr{ab|xy}$ in an NS distribution such that $\pr{ab|11} \in \{0,1\}$ for all $ab$ --- this is a ``degenerate'' NS hidden-variable model in the sense that the hidden variable $\lambda$ is trivial. (This example was found by noting that the existence of such a distribution can be formulated as a linear program, which is easily solved with standard software packages.) In the following list, the ``outer layer'' of rows and columns are indexed by Alice and Bob's inputs respectively, while the ``inner layer'' of rows and columns in each the small sub-tables are indexed by Alice and Bob's outputs respectively. All non-numerical entries should be understood to be zero; 
we have used crosses to mark all $abxy$ combinations which do not satisfy the magic-square win condition. This makes it easy to check that the distribution indeed perfectly wins the magic-square game, since all nonzero probabilities only occur in positions not marked with a cross (note that the numerical entries are already correctly normalized, so there cannot be any nonzero values other than the numerical entries).

\def\arraystretch{1.5} 
\setlength\tabcolsep{2mm}
\centerline{
\begin{tabular}{c c c c}
& 
$y=1$ & $y=2$ & $y=3$ \vspace{2mm}\\
$x=1$ &
\matMS{}{1 & & \lose & \lose}{ & & \lose & \lose}{\lose & \lose & & }{\lose & \lose & & }&
\matMS{}{1/2 & 1/2 & \lose & \lose}{\lose & \lose & & }{ & & \lose & \lose}{\lose & \lose & & }&
\matMS{}{1/2 & 1/2 & \lose & \lose}{\lose & \lose & & }{\lose & \lose & & }{ & & \lose & \lose}
\vspace{2mm}\\
$x=2$ &
\matMS{}{1/2 & \lose & & \lose}{1/2 & \lose & & \lose}{\lose & & \lose & }{\lose & & \lose & }&
\matMS{}{1/2 & \lose & & \lose}{\lose & 1/2 & \lose & }{ & \lose & & \lose}{\lose & & \lose & }&
\matMS{}{1/2 & \lose & & \lose}{\lose & 1/2 & \lose & }{\lose & & \lose & }{ & \lose & & \lose}
\vspace{2mm}\\
$x=3$ &
\matMS{}{\lose & & & \lose}{\lose & & & \lose}{1/2 & \lose & \lose & }{1/2 & \lose & \lose & }&
\matMS{}{\lose & & & \lose}{ & \lose & \lose & }{\lose & 1/2 & & \lose}{1/2 & \lose & \lose & }&
\matMS{}{\lose & & & \lose}{ & \lose & \lose & }{1/2 & \lose & \lose & }{\lose & 1/2 & & \lose}
\end{tabular}
}
\def\arraystretch{1} 
\vspace{5mm}
\noindent To verify that this distribution is indeed non-signalling, note that for each $xy$, the marginal probabilities $\pr{a|xy}$ and $\pr{b|xy}$ can be computed from the corresponding sub-table in the above list by simply summing along the rows and columns respectively. 
A quick inspection confirms that for each value of $x$, the row-wise sums are the same for all $y$, and for each value of $y$, the column-wise sums are the same for all $x$; therefore, the NS conditions are satisfied.

Hence for $x^\star y^\star = 11$, this distribution proves the first part of the claim. To proceed, we make use of the symmetry property mentioned earlier, to construct another NS distribution $\prnew{ab|xy}$ such that $\prnew{ab|11}=1$ for a different value of $ab$. For instance, we can take the permutation that swaps Alice's input values $2$ and $3$ (and relabels Bob's outputs accordingly)\footnote{Technically, here we implicitly relied on the fact that this permutation acts nontrivially on the string $b=001$ (i.e.~the output such that $\pr{ab|11}=1$ in the original distribution). If we had instead considered applying the permutation on Alice's output and Bob's input, the argument would not have worked for this example at least, because all permutations of $a=000$ are identical.} --- note that this permutation does not change Alice's input $1$, hence the ``special input pair'' $x^\star y^\star$ that achieves deterministic outputs remains as $11$. 
By taking a convex combination of these two distributions (conditioned on a hidden variable $\lambda$), we obtain an NS hidden-variable distribution that perfectly wins the magic-square game and achieves $\pr{ab|11 \lambda} \in \{0,1\}$ for all $a,b,\lambda$, but also with $\pr{ab|xy}\neq1$ for all $abxy$ (since there is only one probability equal to $1$ in each of those distributions, and it corresponds to a different $ab$ in each case). Also, if each distribution is chosen with probability $1/2$, it is clear that the resulting distribution satisfies $\pr{ab|xy}\leq1/2$ for all $abxy$ (because in the original distribution described above, all probabilities are at most $1/2$ other than the special input-output combination that occurs with probability $1$, which has probability $0$ in the permuted distribution and vice versa).

Finally, to construct analogous distributions for any $x^\star y^\star$, we again exploit the symmetry property to relabel the $x^\star y^\star=11$ case to all other possible input pairs (e.g.~just by considering all possible permutations $\pi$ for each party). This completes the proof of the claim.
\end{proof}

We remark that the distributions constructed in the above proof are not the same as the canonical quantum distribution. As mentioned above, in principle a DI protocol could estimate all the probabilities $\pr{ab|xy}$ individually instead of just the magic-square winning probability, so a natural question is whether this version would rule out perfect attacks by NS adversaries. However, it turns out that this is also insufficient:

\begin{fact}
Consider the distribution $\pr{ab|xy}$ for the magic-square game where $\pr{ab|xy}=1/8$ for all winning combinations $abxy$ and $\pr{ab|xy}=0$ otherwise, i.e.~the canonical quantum distribution. For any input pair $x^\star y^\star$, there exists an NS hidden-variable model for this distribution that has $\pr{ab|x^\star y^\star \lambda} \in \{0,1\}$ for all $a,b,\lambda$.
\end{fact}
\begin{proof}
It suffices to prove the claim for some particular choice of $x^\star y^\star$, then apply the same symmetry argument as in the previous proof. Taking e.g.~$x^\star y^\star=11$, the existence of such an NS hidden-variable model can be formulated as a linear program, following the guessing-probability approach of~\cite{PAM+10,NPS14,BSS14} (since the $\pr{ab|x^\star y^\star \lambda} \in \{0,1\}$ condition is essentially equivalent to requiring that an adversary can perfectly guess $ab$). We find that this linear program can be solved in a matter of seconds, and indeed returns such an NS hidden-variable model. However, the solution is rather lengthy to present (in particular, observe that the $\pr{ab|xy}=1/8$ constraints together with the requirement that $\pr{ab|x^\star y^\star \lambda} \in \{0,1\}$ imply $\lambda$ must take on at least $8$ values), and we instead simply provide the code that computes this solution at \href{https://github.com/ernesttyz/magicsquareNS}{https://github.com/ernesttyz/magicsquareNS}. (The code also provides the option of computing guessing probabilities bounded by the NPA hierarchy instead. We found that NPA level 1 only reduces the adversary's guessing probability to $0.5000$, while NPA level 2 can bring it down to $0.1254$ --- for comparison, the true quantum value must be $1/8 = 0.125$, by self-testing. Higher NPA levels were too large to easily compute.)

Inspecting the solution suggests that all the conditional distributions $\pr{ab|xy\lambda}$ appear to correspond to some symmetry transformation applied to the distribution in the previous proof; however, it seems unclear what symmetry transformations yield some of those distributions.
\end{proof}

\chapter{Composable security}
\label{app:comp}

{

\newcommand{\A}{\mathrm{A}}
\newcommand{\B}{\mathrm{B}}

\newcommand{\clF}{\mathcal{F}}
\newcommand{\sfP}{\mathsf{P}}

\newcommand{\conv}{\chi}
\newcommand{\prot}{\Pi}
\newcommand{\simu}{\Sigma}
\newcommand{\freal}{\clF^\mathrm{real}}
\newcommand{\fideal}{\clF^\mathrm{ideal}}

In this appendix, we give a brief high-level description of some operational implications of the security definitions --- more detailed or pedagogical explanations are available in~\cite{MR11,arx_PR14,VPdR19,arx_PR21}. This description mostly follows the phrasing and presentation in~\cite{arx_KT20}, though the final version of that work contains significantly more details.

\section{Resources and converters}

We first very briefly state the concepts and definitions we require from the Abstract Cryptography framework (for more details, see~\cite{MR11,arx_PR14,VPdR19,arx_PR21}).
In this framework, a \term{resource} is an abstract system with an interface available to each party, to and from which they can supply some inputs and receive some outputs. We also have the notion of \term{converters} which can interact with a resource to produce a new resource. A converter is an abstract system with an inner and outer interface, with the inner interface being connected to the resource interfaces, and the outer interface becoming the new interface of the resulting resource. If $P$ is a subset of the parties and we have a converter $\conv^{P}$ that connects to their interfaces in a resource $\clF$, we shall denote this as $\conv^{P}\clF$
or $\clF\conv^{P}$ (the ordering has no significance except for readability).

As an important basic example, a protocol is essentially a tuple $\mathscr{P} = (\prot^{\A},\prot^{\B},\dots)$ of converters, one for each party. Each converter describes how that party interacts with its interfaces in $\clF$, producing a new set of inputs and outputs ``externally'' (i.e.~at the outer interface). If we have (for instance) a protocol with converters $\prot^{\A}$ and $\prot^{\B}$ for parties $\A$ and $\B$, for brevity we shall use $\prot^{\A\B}$ to denote the converter obtained by attaching both the converters $\prot^{\A}$ and $\prot^{\B}$.

The Abstract Cryptography framework also requires some metric that quantifies how ``close'' various resources are to each other. In general, any metric satisfying some abstract axioms is sufficient~\cite{MR11}. However, for the purposes of this work we shall focus on metrics based on an operational notion of distinguishability, which yields more easily understood operational implications. Specifically, given two resources $\clF$ and $\clF'$, a \term{distinguisher} is a system that interacts with the interfaces of these resources, and then produces a single bit $G$ (which can be interpreted as a guess of which resource it is interacting with). For a given distinguisher, let $\sfP_{G|\clF}$ be the probability distribution it produces on $G$ when supplied with $\clF$, and analogously for $\clF'$. Its \term{distinguishing advantage} $\eps$ between these two resources is defined to be 
\begin{align}
\eps \defvar \left|\sfP_{G|\clF}[0] - \sfP_{G|\clF'}[0]\right| = \frac{1}{2} \norm{\sfP_{G|\clF} - \sfP_{G|\clF'}}_1.
\label{eq:distadv}
\end{align}
This concept can be used to construct a metric on the set of resources~\cite{MR11}, though here we will only require the definition~\eqref{eq:distadv} of the distinguishing advantage rather than the full properties of the metric.

\section{Security definitions}

We can now discuss security definitions in this framework. 
Consider a situation with some number of parties, where some subset $Q$ (which may be the entire set of parties, depending on the relevant cryptographic situation) can be potentially dishonest. (For example, in the case of QKD or DIQKD, we would have three parties Alice/Bob/Eve, out of which only Eve is potentially\footnote{This statement might seem slightly odd, in that Eve is typically viewed as ``always dishonest'' in QKD/DIQKD, informally speaking. However, this discrepancy is merely a matter of terminology --- the cases of honest and dishonest Eve in the Abstract Cryptography framework just correspond respectively to the absence or presence of Eve in the informal sense, or put another way, whether she chooses to leave the honest implementation untouched, versus attempting to eavesdrop.} dishonest.)
We model such a situation as a tuple of resources $(\clF_P)_{P\subseteq Q}$, where $\clF_P$ denotes the resources available when parties $P$ are dishonest (which potentially have more functionalities than when they are honest). 
Suppose we have such a resource tuple $\left(\freal_P\right)_{P\subseteq Q}$ describing the ``real'' functionalities available to the various parties, and a protocol $\mathscr{P}$ which connects to the interfaces of $\freal$, with the informal goal of constructing a more idealized resource tuple $\left(\fideal_P\right)_{P\subseteq Q}$. This is formalized as follows (here we denote the complement of a set $P$ as $\overline{P}$):
\begin{definition}\label{def:comp}
For a scenario in which there is some set $Q$ of potentially dishonest parties, we say that \term{$\mathscr{P}$ constructs $\left(\fideal_P\right)_{P\subseteq Q}$ from $\left(\freal_P\right)_{P\subseteq Q}$ within distance $\eps$} if the following holds: for every $P \subseteq Q$, there exists a converter $\simu^{P}$ which connects to their interfaces, such that for every distinguisher, the distinguishing advantage between $\prot^{\overline{P}}\freal_P$ and $\fideal_P\simu^P$ is at most $\eps$.
The converters $\simu^{P}$ shall be referred to as \term{simulators}. 

In contexts where $\left(\fideal_P\right)_{P\subseteq Q}$ and $\left(\freal_P\right)_{P\subseteq Q}$ are unambiguously established, we may simply say for brevity that the protocol $\mathscr{P}$ is \term{$\eps$-secure}.
\end{definition}

We have stated the above definition slightly differently from~\cite{MR11}, in which an individual simulator is required for each dishonest party. 
We remark that a security proof satisfying Definition~\ref{def:comp} could be converted into one satisfying the~\cite{MR11} definition by modifying the choice of $\left(\fideal_P\right)_{P\subseteq Q}$ to one that (for every $P$ containing more than one party) explicitly includes quantum channels between the dishonest parties $P$, which would allow for individual simulators that communicate using these quantum channels in order to effectively implement the simulator $\simu^{P}$ in Definition~\ref{def:comp}. From the perspective of~\cite{MR11}, this would basically reflect the inability of a protocol to \emph{guarantee} that the dishonest parties cannot communicate with each other. For ease of description, in the rest of this discussion we shall follow Definition~\ref{def:comp} as stated, instead of the exact definition in~\cite{MR11} (in any case, for QKD/DIQKD there is only one potentially dishonest party).

In the case of QKD/DIQKD, the ideal functionality we aim to construct can be described as follows (see~\cite{arx_PR21} for more details): if Eve is honest, it always\footnote{Strictly speaking, the ideal functionality described in~\cite{arx_PR21} is a version that has some probability of producing an abort symbol even in the honest-Eve case. However, these versions of the ideal functionality only differ by $O(\ecom)$ (in terms of distinguishing advantage), so if $\ecom$ is small, this difference does not matter too much.} just gives Alice and Bob a perfect shared secret key of length $\ell$; if Eve is dishonest, she can input a single bit that specifies whether the functionality aborts (producing some ``null'' abort symbol for all parties) or gives Alice and Bob a perfect shared secret key of length $\ell$ (as in the honest case). Note that with the typical resources for QKD/DIQKD, it is impossible to produce an ideal functionality that {always} outputs a perfect shared secret key even when Eve is dishonest, because she can always set up states that force the protocol to abort with high probability --- in other words, a denial-of-service attack cannot be prevented. 

With this in mind, we can discuss the significance of the core security definition used in this work (Definition~\ref{def:secure}). The main point is that if Definition~\ref{def:secure} is satisfied in a QKD\footnote{Here we technically restrict the claim to (device-dependent) QKD rather than DIQKD. This is because for DIQKD, the device-reuse issues mentioned in Sec.~\ref{sec:assumptions} and~\ref{sec:securitydef} cause some difficulties in proving that Definition~\ref{def:comp} is satisfied --- if the device behaviour is correlated to ``secret data'' in past protocol instances, this gives a method to distinguish the real and ideal cases. Informally, it seems that it should be possible to construct a similar proof by imposing a condition that forbids such behaviour; however, this concept has not been fully formalized within the Abstract Cryptography framework.} protocol, then it can be proven~\cite{arx_PR21} that the protocol constructs the aforementioned ideal functionality within distance $\ecom+\esound$ in the sense of Definition~\ref{def:comp}. 
(More precisely, a closer inspection of the proof shows that for the dishonest-Eve case, the distinguishing advantage is at most $\esound$, independent of $\ecom$, which may be convenient if $\ecom$ is only heuristically estimated.)
Roughly, the intuition is that since Definition~\ref{def:secure} is based on trace distance, it 
ensures that the real and ideal functionalities are hard to distinguish, given appropriately chosen simulators.
We do not discuss the proof details here, but the main point is that Definition~\ref{def:secure} implies the slightly more abstract Definition~\ref{def:comp} within the Abstract Cryptography framework, which has consequences that we shall now explain.

\section{Operational implications}

In general (not just for QKD), Definition~\ref{def:comp} has an important operational implication, regarding the effects of composing protocols with each other.\footnote{A more abstract implication, which holds for more general metrics, is that if several protocols satisfying this definition are composed, the ``error'' $\eps$ of the resulting larger protocol can be bounded by simply by adding those of the sub-protocols~\cite{MR11}.} Namely, suppose we have a larger protocol that uses $\left(\fideal_P\right)_{P\subseteq Q}$ as a resource, and take any event that might be considered a ``failure'' in the larger protocol (we impose no restrictions on the nature of a failure, except that it be a well-defined event). Suppose we also have a proof that for any strategy by the dishonest parties, the probability of this failure event in the larger protocol is upper-bounded by some $p_\mathrm{max}$ when using $\left(\fideal_P\right)_{P\subseteq Q}$. In that case, one implication of Definition~\ref{def:comp} being satisfied is that if $\left(\fideal_P\right)_{P\subseteq Q}$ is replaced by the protocol $\mathscr{P}$ applied to $\left(\freal_P\right)_{P\subseteq Q}$, then the probability of this failure event is still upper-bounded by $p_\mathrm{max}+\eps$. This is a very useful operational statement, since it assures us that the probability of \emph{any} ``bad'' event cannot increase by too much when replacing the ideal functionality with the real protocol.

To prove this claim, we simply need to make two observations, taking an arbitrary $P \subseteq Q$: firstly, since the bound $p_\mathrm{max}$ for the functionality $\fideal_P$ holds for any strategy by the dishonest parties, it must in particular hold when they implement the simulator $\simu^P$, i.e.~it holds if they are using $\fideal_P\simu^P$ instead of $\fideal_P$. Secondly, since the distinguishing advantage between $\fideal_P\simu^P$ and $\prot^{\overline{P}}\freal_P$ is at most $\eps$ when Definition~\ref{def:comp} is satisfied, the probability of the failure event cannot differ by more than $\eps$ between them (otherwise the event would serve as a way to distinguish the two cases).

Note that there is a slight technicality here: in order for the argument to be valid, the bound $p_\mathrm{max}$ (for the larger protocol using $\fideal_P$) must be derived for a class of dishonest-party strategies that includes the simulator $\simu^P$, in order for the bound to hold for $\fideal_P\simu^P$ as well. 
This means that 
if a more ``powerful'' simulator is used in Definition~\ref{def:comp}, then the bound $p_\mathrm{max}$ must be proved against a more ``powerful'' class of strategies. 
However, this is perhaps more of a consideration for the larger protocol, rather than the protocol $\mathscr{P}$ in Definition~\ref{def:comp} itself.

}

\chapter{Parametrization with correlators}
\label{app:corrs}

\newcommand{\cset}{\mathcal{S}}

Here we give the details of a convenient parametrization of probability distributions in nonlocality scenarios, which we used in Sec.~\ref{sec:qubit}. It applies in scenarios where all measurements have outcomes in $\{-1,+1\}$, under the assumption that the NS conditions are satisfied.\footnote{Of course, more generally any 
NS distribution
with binary outcomes for all measurements 
is isomorphic to
such a scenario; this choice of output labelling is just what yields the desired parametrization most conveniently.} The main motivation behind such a parametrization is the fact that under the NS conditions, 
not all the probabilities can be chosen freely --- once some subset of them are specified, the rest are fixed by the NS constraints; however, it is perhaps not immediately clear how to choose a ``good'' subset to specify. The approach here describes a systematic set of correlator parameters that uniquely specifies an NS distribution.
The number of these parameters is somewhat smaller than the number of individual probability terms, which
is convenient when considering constraints to impose in optimizations such as those in Chapter~\ref{chap:singlernd}.
(This result has been previously used in various works on this topic, e.g.~\cite{PBS+11}; here we simply give a pedagogical presentation of it. See also~\cite{CG04} for the derivation of a different parametrization that generalizes to output sets of arbitrary size, as argued in~\cite{Pir05}.) 

While in this work we have mostly focused on two-party nonlocality scenarios, this parametrization also holds for $N$-party scenarios (with any number of inputs per party, as long as all outputs are in $\{-1,+1\}$). 
For completeness, the description here will cover this more general setting
(the NS conditions generalize ``intuitively'' to $N>2$ parties; see e.g.~\cite{Pir05} for a detailed description).

\section{Classical random variables}

We first give a proof of a useful fact about classical random variables, without involving nonlocal distributions for now. Let $Z_1, Z_2, \dots, Z_N$ be random variables taking values in $\{-1,+1\}$, following some joint distribution. There are various correlators between these random variables --- to take an explicit example, if for instance we consider the variables $Z_2, Z_3, Z_5$, their correlator $\expval{Z_2 Z_3 Z_5}$ is defined as the expectation value of their product:
\begin{align}
\expval{Z_2 Z_3 Z_5} \defvar \sum_{
z_2 z_3 z_5
} z_2 z_3 z_5 \pr{z_2 z_3 z_5} = \sum_{\vec{z} 
} z_2 z_3 z_5 \pr{\vec{z}},
\end{align}
where the last expression is written in a form that will be convenient later. 
More generally,
for any subset $\cset \subseteq\upto{N}$, we have a correlator for the corresponding subset of the random variables --- for ease of description, we shall introduce the notation $\expval{Z_{\cset}}$ to denote this:
\begin{align}
\label{eq:RVdefcorr}
\expval{Z_{\cset}} \defvar \sum_{\vec{z} 
} \bigg(\prod_{j\in\cset} z_j\bigg) \pr{\vec{z}}.
\end{align}
Note that we allow $\cset$ to be an \emph{arbitrary} subset of $\upto{N}$ in this expression --- in particular, it might only contain a single element, or even be empty. In such cases,~\eqref{eq:RVdefcorr} does not really describe a ``correlation'' between some variables, but it is still a well-defined expression (taking the value of the empty product to be $1$). More specifically, when $\cset$ contains only one element, $\expval{Z_{\cset}}$ is just the expectation value of one of the random variables $Z_j$, whereas when $\cset$ is the empty set, $\expval{Z_{\cset}}$ is simply the sum of all the probabilities $\pr{\vec{z}}$, and thus must equal $1$.

Observe that since there are $2^N$ subsets of $\upto{N}$, there are $2^N$ correlators $\expval{Z_{\cset}}$ (when including the ``empty correlator'' $\expval{Z_{\{\}}}$ as well). In comparison, the probability distribution of the random variables $Z_j$ also has $2^N$ terms $\pr{\vec{z}}$. It hence seems possible that given all $2^N$ correlators $\expval{Z_{\cset}}$, the probabilities $\pr{\vec{z}}$ may be uniquely specified. (The probabilities $\pr{\vec{z}}$ only have $2^N-1$ degrees of freedom due to the normalization constraint; however, we can treat this as being captured by the condition $\expval{Z_{\{\}}} = 1$ --- this is why it was convenient to include the $\cset=\{\}$ case when defining $\expval{Z_{\cset}}$.) We now prove that this is indeed the case (the proof also gives an explicit method to quickly compute the probabilities $\pr{\vec{z}}$ from the correlators $\expval{Z_{\cset}}$).

\begin{fact}\label{fact:RVcorrs}
Let $Z_1, Z_2, \dots, Z_N$ be random variables taking values in $\{-1,+1\}$, and consider the correlators $\expval{Z_{\cset}}$ as defined above. Given the values of all the
correlators $\expval{Z_{\cset}}$, the
probabilities $\pr{\vec{z}}$ are all uniquely specified.
\end{fact}

\begin{proof}
All the correlators $\expval{Z_{\cset}}$ are linear combinations of the probabilities $\pr{\vec{z}}$, as seen from~\eqref{eq:RVdefcorr}. Hence if we view the $2^N$-tuple of correlators $\expval{Z_{\cset}}$ as a column vector, and similarly for the $2^N$-tuple of probabilities $\pr{\vec{z}}$, we can obtain the former from the latter by multiplication with a suitable matrix. The expression~\eqref{eq:RVdefcorr} yields the explicit entries of this matrix --- labelling each of its rows using a subset $\cset \subseteq\upto{N}$, and labelling each of its columns using a string $\vec{z} \in \{-1,+1\}^N$, the entry in row $\cset$ and column $\vec{z}$ is simply $\prod_{j\in\cset} z_j$. 

We now argue that this matrix is an orthogonal matrix, up to a scaling factor of $1/2^N$. Therefore, it is invertible (proving the desired result); furthermore, the inverse matrix is obtained by simply transposing it and multiplying all entries by $1/2^{2N}$, which is easy to compute.

To prove this, consider any two rows of the matrix, labelled by distinct subsets $\cset$ and $\cset'$. We first show that the two rows are orthogonal. Since all the entries have values $\pm 1$, it is sufficient (in fact, also necessary) to show that the rows differ in sign on exactly half their entries. In turn, this is equivalent to stating that $\left(\prod_{j\in\cset} z_j\right)\left(\prod_{j'\in\cset'} z_{j'}\right) = -1$ for exactly half of the values $\vec{z} \in \{-1,+1\}^N$. Since $\left(\prod_{j\in\cset} z_j\right)\left(\prod_{j'\in\cset'} z_{j'}\right) = \prod_{k\in\cset\Delta\cset'} z_k$ where $\cset\Delta\cset'$ is the symmetric difference of $\cset$ and $\cset'$ (i.e.~the set of elements belonging to exactly one of the two sets), we see that indeed exactly half the values will be $-1$, as desired. 
(Here we technically used the fact that $\cset\Delta\cset'$ is non-empty because $\cset\neq\cset'$.)

Now observe that all the rows have norm $2^N$, since all entries have values $\pm 1$. Hence if we rescale all the matrix entries by $1/2^N$ (which does not affect the orthogonality of the rows), all the rows are orthonormal vectors, thus the rescaled matrix is orthogonal as claimed.
\end{proof}
Hence if we simply have $N$ classical random variables (with values $\{-1,+1\}$), we can use the correlators to completely parametrize the distribution.

\section{NS distributions}

We now return to the topic of nonlocality scenarios for $N$ parties, where a generalized form of the above result holds. The distributions in such scenarios are of the form $\pr{\vec{a}|\vec{x}} = \pr{a_1 a_2 \dots a_N | x_1 x_2 \dots x_N}$, with $a_1 a_2 \dots a_N$ being the outputs and $x_1 x_2 \dots x_N$ being the inputs. 
If the distribution satisfies the NS conditions, we can define some correlators as follows. For ease of explanation, we begin with an explicit example --- take for instance parties $2,3,5$, and specify some particular input values $x_2 x_3 x_5$ for these parties. Given this tuple of input values, we can define the correlator for the corresponding outputs:\footnote{A caveat: the notation in this setting is different from that which we used above for classical random variables, where the subscripts in e.g.~$\expval{Z_2 Z_3 Z_5}$ specified the subset of random variables to consider in the correlator. Here, the subscripts $x_2, x_3, x_5$ in the notation $\expval{A_{x_2} A_{x_3} A_{x_5}}$ instead refer to the input values we are considering, with the subset of parties being specified in turn by the subscripts on $x_2, x_3, x_5$. (This is basically similar to the correlator notation previously used in Sec.~\eqref{sec:qubit}, though here we specify Alice and Bob's terms using the index subscripts rather than the letters $A$ and $B$.) An analogous technicality arises in the notation $\expval{A_{x_{\cset}}}$ we shall shortly introduce.} 
\begin{align}
\expval{A_{x_2} A_{x_3} A_{x_5}} \defvar \sum_{a_2 a_3 a_5} a_2 a_3 a_5 \pr{a_2 a_3 a_5|x_2 x_3 x_5} = \sum_{\vec{a}} a_2 a_3 a_5 \pr{\vec{a}|\vec{x}},
\end{align}
where in the last expression, one can use any value of $\vec{x}$ that is compatible with the specified inputs $x_2 x_3 x_5$ --- crucially, the NS conditions ensure that any such choice of $\vec{x}$ yields the same value for this expression. (For that matter, it is only because of the NS conditions that we can sensibly speak of the correlations between these parties' outputs given these inputs only --- without those, one would need to specify the inputs used by the other parties before the output distribution of these parties can be determined. Put another way, the $\pr{a_2 a_3 a_5|x_2 x_3 x_5}$ terms in the middle expression are only well-defined because of the NS conditions.) This is an important property that will be vital in our subsequent arguments.

Having seen some basic properties from this example, we now turn to the general construction:
take any subset $\cset\subseteq\upto{N}$ of the parties, and specify some particular input values $(x_j)_{j\in\cset}$ for these parties, which we shall denote as $x_{\cset}$ for brevity. With this, we can define a correlator for the corresponding outputs, which we shall denote in this section as
\begin{align}
\label{eq:defcorr}
\expval{A_{x_{\cset}}} \defvar \sum_{\vec{a}} \bigg(\prod_{j\in\cset} a_j\bigg) \pr{\vec{a}|\vec{x}},
\end{align}
where in the $\pr{\vec{a}|\vec{x}}$ term one can use any value of $\vec{x}$ that agrees with $x_{\cset}$ on the subset $\cset$. 
(Again, this can only be reasonably interpreted as a correlation between some parties' outputs because of the NS conditions.) 
With this definition of correlators in this situation, we now prove an analogous statement to Fact~\ref{fact:RVcorrs} (again, the proof also yields a simple procedure to compute the probabilities $\pr{\vec{a}|\vec{x}}$):
\begin{fact}
Consider an $N$-party nonlocal distribution satisfying the NS conditions and with all outputs taking values in $\{-1,+1\}$, and let the correlators $\expval{A_{x_{\cset}}}$ be as defined above. Given the values of all the correlators $\expval{A_{x_{\cset}}}$, the 
probabilities $\pr{\vec{a}|\vec{x}}$ are all uniquely specified.
\end{fact}

\begin{proof}
It is convenient to first introduce a slight variant of the notation~\eqref{eq:defcorr}; namely,
\begin{align}
\label{eq:defcondcorr}
\expval{A_{x_{\cset}}}_{|\vec{x}} \defvar \sum_{\vec{a}} \bigg(\prod_{j\in\cset} a_j\bigg) \pr{\vec{a}|\vec{x}},
\end{align}
where $\vec{x}$ must be chosen such that it agrees with $x_{\cset}$ on the subset $\cset$ --- this notation simply makes explicit which value of $\vec{x}$ is chosen to compute the correlator (even though all valid choices yield the same final value). 

To prove the desired claim, first focus on some specific input tuple $\vec{x}^\star$. 
Given this value, there are $2^N$ probabilities $\pr{\vec{a}|\vec{x}^\star}$, and we shall argue that they can all be computed if we have all the correlators $\expval{A_{x^\star_{\cset}}}$, as follows. 
Since all the correlators are known, in particular, for each subset $\cset\subseteq\upto{N}$ we know the value of $\expval{A_{x^\star_{\cset}}}_{|\vec{x}^\star} = \expval{A_{x^\star_{\cset}}}$. 
Now observe that the relation~\eqref{eq:defcondcorr} between the correlators $\expval{A_{x^\star_{\cset}}}_{|\vec{x}^\star}$ and the probabilities $\pr{\vec{a}|\vec{x}^\star}$ is exactly the same as the relation~\eqref{eq:RVdefcorr} in the earlier discussion of classical random variables. Therefore, we can apply the same method to compute the probabilities $\pr{\vec{a}|\vec{x}^\star}$ from the correlators $\expval{A_{x^\star_{\cset}}}_{|\vec{x}^\star}$, as desired.

Finally, since $\vec{x}^\star$ was arbitrary, we can simply repeat the process for each possible value, thereby computing all the probabilities $\pr{\vec{a}|\vec{x}}$.
\end{proof}

\begin{remark}
The role played by the NS conditions in the above proof is ensuring that $\expval{A_{x_{\cset}}}_{|\vec{x}}$ is the same for any choice of $\vec{x}$ compatible with $x_{\cset}$, i.e.~they are well-defined just by writing $\expval{A_{x_{\cset}}}$. Or to take a ``reverse'' perspective: when implementing the above computation, many of the correlators 
will feature in the computations for more than one value of $\vec{x}^\star$; the fact that each correlator has the same value in every such computation serves to ensure that the resulting distribution satisfies the NS conditions. Similar ideas are behind the parametrization in~\cite{CG04,Pir05}, which is essentially based on the marginals $\pr{\vec{a}_{\cset}|\vec{x}_{\cset}}$ directly, using the fact that the NS conditions ensure that these terms are well-defined.
\end{remark}

To see that this approach indeed reduces the number of parameters (as compared to specifying all the $\pr{\vec{a}|\vec{x}}$ terms individually), we can compare the number of terms in each case. We first remark that in the earlier case of classical random variables (i.e.~Fact~\ref{fact:RVcorrs}), there was no advantage in this respect --- the number of correlators and the number of probabilities are both $2^N$ (and the fact that the ``empty correlator'' must have value $1$ is just the same statement as the normalization constraint on the probabilities). However, when it comes to nonlocal distributions $\pr{\vec{a}|\vec{x}}$, there can be a substantial reduction. If we denote the input set for party $j$ as $\mathcal{X}_j$, then the number of correlators $\expval{A_{x_{\cset}}}$ is the same as the number of possible choices of ${x_{\cset}}$, which is\footnote{To see this, let $\mathcal{X}'_j$ denote the set $\mathcal{X}_j$ with a ``null symbol'' $\perp$ added, which is distinct from all other input values. Then there exists a bijection between the set $
\vec{\mathcal{X}}' \defvar 
\mathcal{X}'_1 \times \mathcal{X}'_2 \times \dots \times \mathcal{X}'_N$ and the set of possible choices of $x_{\cset}$, as follows: for each element $\vec{x}' \in \vec{\mathcal{X}}'$, 
take $\mathcal{S}$ to be the set of indices on which $\vec{x}'$ does not take the value $\perp$, and set $x_{\cset}$ to be equal to $\vec{x}'$ on this set of indices. This will produce each possible choice of $x_{\cset}$ exactly once. (The idea here is basically the same as the standard 
counting argument for the size of a power set.)} $\prod_j (|\mathcal{X}_j| + 1)$ (if we omit the empty correlator using normalization, just subtract $1$ from this count). In contrast, the number of probabilities $\pr{\vec{a}|\vec{x}}$ is simply $2^N \prod_j |\mathcal{X}_j|$, and even if we account for the normalization constraints (which can be fairly straightforwardly implemented, e.g.~by leaving one value of $\pr{\vec{a}|\vec{x}}$ for each $\vec{x}$ to be determined from the other terms by normalization), we would still need to specify $(2^N-1) \prod_j |\mathcal{X}_j|$ terms. 
For any nonlocality scenario, $\prod_j (|\mathcal{X}_j| + 1) - 1$ is always smaller\footnote{Qualitatively, this must hold simply because in deriving the former, we made use of the NS and normalization constraints, whereas in the latter, only the normalization constraints were used. However, this can also be proven directly: first ignore the $-1$ term for the ``empty correlator'' in the former (since it only makes it smaller anyway), in which case the ratio of the two expressions is $
(2^N-1) \prod_j (|\mathcal{X}_j| / (|\mathcal{X}_j| + 1))$. As long as $|\mathcal{X}_j| \geq 2$ for all parties, this is lower bounded by $(2^N-1) (2/3)^N$, 
which
is greater than $1$ for all $N\geq2$.} than $(2^N-1) \prod_j |\mathcal{X}_j|$ (as long as the scenario is ``nontrivial'' in the sense that $N\geq 2$ and every party has more than one input), and the difference can be fairly significant.

Furthermore, $\prod_j (|\mathcal{X}_j| + 1) - 1$ is in fact the minimum number of independent parameters required to specify such a distribution. This follows from the fact that the set of NS distributions has dimension $\prod_{j\in\upto{N}} \left(\sum_{x_j \in \mathcal{X}_j} (|\mathcal{A}_{x_j}|-1) + 1\right) - 1$~\cite{Pir05}, where $\mathcal{A}_{x_j}$ is the set of output values for party $j$ given input $x_j$. In our case, we have $|\mathcal{A}_{x_j}|=2$ for all terms, in which case this expression indeed reduces to $\prod_j (|\mathcal{X}_j| + 1) - 1$. 
In other words, this parametrization is essentially optimal in terms of the number of independent parameters.

As a final remark, we note that while the correlators must clearly satisfy the constraints $\expval{A_{x_{\cset}}} \in [-1,+1]$, not every tuple of correlator values satisfying these constraints can be obtained from an NS distribution. 
In particular, this means that if one wants to generate some ``arbitrary'' NS distribution, one cannot simply choose some tuple of correlator values within those bounds and assume that this automatically yields a valid NS distribution.\footnote{Although, the set of correlator tuples that produce valid NS distributions is a polytope at least --- this follows from the fact that the set of NS distributions is a polytope, and the correlators are related to the probabilities by an affine transformation.}
An explicit example demonstrating this issue can be found by considering the 2-party Popescu-Rohrlich (PR) box~\cite{KT85,Ra85,PR94}: this is the NS distribution specified by the correlator values $\expval{A_0 B_0} = \expval{A_0 B_1} = \expval{A_1 B_0} = - \expval{A_1 B_1} = 1$ and $\expval{A_x} = \expval{B_y} = 0$ for all $x,y \in \mathbb{Z}_2$ (here we revert to the Alice-Bob correlator notation used in Sec.~\ref{sec:qubit}.) It turns out that this is the \emph{unique} NS distribution that achieves 
those values on the two-party correlators $\expval{A_x B_y}$.
This implies that, for instance, the tuple of correlator values $\expval{A_0 B_0} = \expval{A_0 B_1} = \expval{A_1 B_0} = - \expval{A_1 B_1} = 1$, $\expval{A_x} = \expval{B_y} = 1$ does not correspond to a valid NS distribution. (A different way to see that this correlator tuple does not match an NS distribution is to recall that by the arguments in Appendix~\ref{app:foundations}, any NS distribution with $\expval{A_x B_y}$ values violating the CHSH inequality cannot be deterministic. Therefore, it is not possible for such an NS distribution to attain extremal values on all the ``one-party correlators'' $\expval{A_x}, \expval{B_y}$. Note that this argument more generally rules out a different family of correlations as compared to the argument based on the uniqueness of the PR box.)

\chapter{Details of numerical algorithm}
\label{app:qubitnumerics}

We first describe how each minimization can be tackled
when treating the parameters in the other minimizations as constants, then summarize how all these algorithms can be put together in a consistent manner, and argue that this approach indeed yields arbitrarily tight bounds.

\section{Minimization over Alice's measurement}
\label{asec:minalice}

We tackle this minimization simply by applying a (uniform) continuity bound for $\thA$.
Specifically, for $\delta\in[0,\pi]$ we describe a monotone increasing function $\cont(\delta)$ that bounds the change in the objective function when $\thA$ is replaced by $\thA+\delta$ (treating $\thB$ and $\rho_{AB}$ as constants). 
Then for any set of intervals 
of the form $\{[\theta_j-\delta_j,\theta_j+\delta_j]\}_j$ that covers the interval $[0,\pi]$, we would have 
\begin{align}
\min_{\thA} F_\mathrm{obj}(\thA,\thB,\rho_{AB}) \geq \min_j 
F_\mathrm{obj}(\theta_j,\thB,\rho_{AB}) - \cont(\delta_j)
.
\end{align}
We apply this in practice by starting with a fairly ``coarse'' choice of intervals, then iteratively applying the process of deleting the interval that currently achieves the minimization over $j$, and replacing it with smaller intervals that cover the deleted interval.

To derive such a continuity bound, we first analyze the entropic term in the objective function, following~\cite{SBV+21} (with a minor modification to slightly improve the bound).
Take any pure initial state $\ket{\rho}_{ABE}$, and let $\sigma_{\hat{A}_1 BE}$ be the state obtained by performing a Pauli measurement along angle $\thA$ in the $X$-$Z$ plane on the $A$ register of $\ket{\rho}_{ABE}$, applying noisy preprocessing, then storing the result in the classical register $\hat{A}_1$ and tracing out $A$. 
Let $\sigma'_{\hat{A}_1 BE}$ be the analogous state with $\thA$ replaced by $\thA+\delta$ for some $\delta \in [0,\pi]$. Our goal would be to bound $\left|H(\hat{A}_1| E)_{\sigma} - H(\hat{A}_1| E)_{\sigma'}\right|$. 

We have $H(\hat{A}_1| E)_{\sigma} 
= \sum_{\hat{a}_1} \pr{\hat{a}_1} H(\sigma_{E|\hat{A}_1=\hat{a}_1}) + H(\hat{A}_1)_{\sigma} - H(E)_{\sigma}
$, and analogously for $\sigma'$. Since the operations on $A$ do not affect $E$, we have $H(E)_{\sigma}=H(E)_{\sigma'}$. Also, for states of the form~\eqref{eq:almostdiag}, we have $\rho_{A}=\id/2$ and hence 
$H(\hat{A}_1)_{\sigma} = H(\hat{A}_1)_{\sigma'} = 1$. This gives us
\begin{align}
\left|H(\hat{A}_1| E)_{\sigma} - H(\hat{A}_1| E)_{\sigma'}\right| 
\leq \max_{\hat{a}_1} \left| \left(H(\sigma_{E|\hat{A}_1=\hat{a}_1}) - H(\sigma'_{E|\hat{A}_1=\hat{a}_1})\right)\right|,
\end{align}
so it suffices to bound the difference in entropies of the conditional states on $E$.

Now observe that exactly the same state $\sigma_{\hat{A}_1 B E}$ would have been produced if the initial state had been $\left(e^{i\thA Y_A/2} \otimes \id_{BE}\right)\ket{\rho}_{ABE}$ and the initial Pauli measurement were replaced by a $Z$ measurement.
Furthermore, the fact that $\rho_A=\id/2$ implies we can write $\ket{\rho}_{ABE} = \sum_a \ket{a}_A \ket{a}_{BE} / \sqrt{2}$, where $\{\ket{a}_A\}_a$ is the $Z$-eigenbasis of $A$ and $\{\ket{a}_{BE}\}_a$ are two orthonormal states on $BE$. This implies 
\begin{align}
\left(e^{i\thA Y_A/2} \otimes \id_{BE}\right)\ket{\rho}_{ABE} 
= \left(\id_A \otimes \left(e^{i\thA Y_{BE}/2}\right)^T\right)\ket{\rho}_{ABE} 
= \frac{1}{\sqrt{2}}\sum_a \ket{a}_A \otimes 
\rot{\thA}
\ket{a}_{BE},
\end{align}
where 
$Y_{BE} \defvar i\ketbra{-1}{+1}_{BE}-i\ketbra{+1}{-1}_{BE}$ 
and $\rot{\thA}\defvar\left(e^{i\thA Y_{BE}/2}\right)^T$.
Performing the analogous analysis for
$\sigma'$,
we conclude that 
\begin{align}
\begin{aligned}
\sigma_{BE|\hat{A}_1=\hat{a}_1} &= 
(1-\p) \rot{\thA}\pure{\hat{a}_1}_{BE}\rot{\thA}^\dagger + \p \rot{\thA}\pure{-\hat{a}_1}_{BE}\rot{\thA}^\dagger, \\
\sigma'_{BE|\hat{A}_1=\hat{a}_1} &= 
(1-\p) \rot{\thA+\delta}\pure{\hat{a}_1}_{BE}\rot{\thA+\delta}^\dagger + \p \rot{\thA+\delta}\pure{-\hat{a}_1}_{BE}\rot{\thA+\delta}^\dagger.
\end{aligned}
\end{align}
Therefore, we have
\begin{align}
F(\sigma_{E|\hat{A}_1=\hat{a}_1},\sigma'_{E|\hat{A}_1=\hat{a}_1}) &\geq F(\sigma_{BE|\hat{A}_1=\hat{a}_1},\sigma'_{BE|\hat{A}_1=\hat{a}_1}) \nonumber \\
&\geq (1-\p) \left|\bra{\hat{a}_1}\rot{\thA}^\dagger\rot{\thA+\delta}\ket{\hat{a}_1}\right| + \p \left|\bra{-\hat{a}_1}\rot{\thA}^\dagger\rot{\thA+\delta}\ket{-\hat{a}_1}\right| \nonumber \\
&\geq \left|\cos\frac{\delta}{2}\right|,
\end{align}
where the second inequality holds by concavity of fidelity, 
and the third inequality is given by explicit calculation (see~\cite{SBV+21}).

This lets us apply a fidelity-based continuity bound~\cite{SBV+21}:
\begin{align}
\max_{\hat{a}_1} \left| \left(H(\sigma_{E|\hat{A}_1=\hat{a}_1}) - H(\sigma'_{E|\hat{A}_1=\hat{a}_1})\right)\right| \leq 4.023 \acos F(\sigma_{E|\hat{A}_1=-1},\sigma'_{E|\hat{A}_1=-1}) \leq 2.012 \delta,
\label{eq:Hcontbnd}
\end{align}
where in the second inequality we used the condition $\delta \in [0,\pi]$. 
(Numerical heuristics suggest that the true bound in Eq.~\eqref{eq:Hcontbnd} may simply be $\delta$, so there is some potential for improvement here, though the effect would be fairly small.)

As for the $\vec{\Gamma}$ term,
we note that 
\newcommand{\cs}[1]{\cos(#1)}
\newcommand{\sn}[1]{\sin(#1)}
\begin{align}
&\left|\tr{\left(\vec{\lambda}\cdot\vec{\Gamma}(\thA+\delta,\thB)-\vec{\lambda}\cdot\vec{\Gamma}(\thA,\thB)\right)\rho_{AB}}\right| 
\nonumber\\
\leq& \norm{\vec{\lambda}\cdot\vec{\Gamma}(\thA+\delta,\thB) - \vec{\lambda}\cdot\vec{\Gamma}(\thA,\thB)}_\infty \nonumber\\
=& \norm{\left((\cs{\thA+\delta}Z + \sn{\thA+\delta}X) - (\cs{\thA}Z + \sn{\thA}X) \right) \otimes (\lambda_{10} B_0+\lambda_{11} B_1)}_\infty \nonumber\\
\leq& (|\lambda_{10}| + |\lambda_{11}|)\norm{
(\cs{\thA+\delta}Z + \sn{\thA+\delta}X) - (\cs{\thA}Z + \sn{\thA}X)
}_\infty \nonumber\\
=& (|\lambda_{10}| + |\lambda_{11}|)\sqrt{2-2\cos(\delta)},
\end{align}
where the last line follows from an explicit eigenvalue calculation.

Overall, this means that for $\delta\in[0,\pi]$ we can choose 
\begin{align}
\cont(\delta) = 
2.012 \keyw_1
\delta + (|\lambda_{10}| + |\lambda_{11}|)\sqrt{2-2\cos(\delta)},
\label{eq:contbnd}
\end{align}
accounting for the $\keyw_1$ factor on the $H(\hat{A}_1| E)$ term. This bound is monotone increasing for $\delta\in[0,\pi]$, as required (so that it also bounds the change in entropy when the measurement angle is changed from $\thA$ to any value in the interval $[\thA,\thA+\delta]$).

\section{Minimization over Bob's measurement}
\label{asec:minbob}

The entropic term in the objective function has no dependence on Bob's measurement, so we only need to consider the $\vec{\Gamma}$ term.
In principle, this could be approached using the same argument as above, where we would arrive at the continuity bound
\begin{align}
&\left|\tr{\left(\vec{\lambda}\cdot\vec{\Gamma}(\thA,\thB+\delta)-\vec{\lambda}\cdot\vec{\Gamma}(\thA,\thB)\right)\rho_{AB}}\right| 
\leq
(|\lambda_{01}| + |\lambda_{11}|)\sqrt{2-2\cos(\delta)}.
\end{align}

However, some heuristic experiments indicate that the following approach (used in~\cite{SGP+21}) is more efficient: we can let $r_Z\defvar\cos(\thB)$ and $r_X\defvar\sin(\thB)$ and write
\begin{align}
\sum_{xy} \lambda_{xy} A_x \otimes B_y 
&= \left(\sum_{x} \lambda_{x0} A_x\right) \otimes Z + \left(\sum_{x} \lambda_{x1} A_x\right) \otimes (r_Z Z + r_X X),
\end{align}
in which case the minimization over $\thB\in[0,\pi]$ is equivalent to minimizing over $(r_Z,r_X)$ that lie on the set $\semicset \defvar \{(r_Z,r_X)\mid r_Z^2+r_X^2=1 \text{ and } r_X \geq 0\}$ (i.e.~a semicircular arc). Crucially, the objective function is affine with respect to the vector $(r_Z, r_X)$. Hence if $V$ is any set of points such that $\semicset$ is contained in their convex hull $\operatorname{Conv}(V)$, we immediately have 
\begin{align}
\min_{(r_Z,r_X)\in \semicset} F_\mathrm{obj}(\thA,(r_Z,r_X),\rho_{AB}) 
&\geq \min_{(r_Z,r_X)\in \operatorname{Conv}(V)} F_\mathrm{obj}(\thA,(r_Z,r_X),\rho_{AB}) \nonumber\\
&\geq \min_{(r_Z,r_X)\in V} F_\mathrm{obj}(\thA,(r_Z,r_X),\rho_{AB}),
\end{align}
because the minimum of an affine function over the convex hull of a set $V$ is always attained at an extremal point (which will be a point in $V$). 
To apply this result, we start with a simple choice of the set $V$ (for instance,
in our code we use 
$V = \{(1,0),(1,1),(-1,1),(-1,0)\}$) and find the point in $V$ that yields the minimum value. We then delete this point and replace it with two other points such that $\semicset$ is still contained in the convex hull, and iterate this process until a sufficiently tight bound is obtained (for instance, by checking that there is a feasible point of the optimization that is sufficiently close to the lower bound we have obtained). 

\section{Minimization over states}
\label{asec:minstate}

The minimization over states can be tackled by expressing it as a convex optimization and applying the \term{Frank-Wolfe algorithm}~\cite{FW56}, as was observed in~\cite{WLC18}. For completeness, we now describe the method, with minor 
modifications and clarifications for our specific scenario. 

We observe that (for fixed measurements, parametrized as described above) the minimization over states is the minimization of the convex function 
\begin{align}
f_\mathrm{obj}(\rho) \defvar
\left(\sum_{x\in\{0,1\}} \keyw_x D(\calG({\rho}) \Vert \calZ_x(\calG({\rho})))\right) - \vec{\lambda}\cdot \tr{\vec{\Gamma}(\thA,(r_Z,r_X)) \rho}.
\end{align}
This function is differentiable for all $\rho$ such that $\calG(\rho) > 0$, with~\cite{WLC18}\footnote{The notation here follows that used in~\cite{WLC18}, in which for a function $f$ defined on matrices parametrized by their matrix entries ($\rho=\sum_{jk}\rho_{jk}\ketbra{j}{k}$), its derivative is $\nabla f(\rho) \defvar \sum_{jk} (\partial f/ \partial\rho_{jk}) \ketbra{j}{k}$. 
Denoting the tangent to the graph of $f$ at $\rho$ as $t_{\rho}$, this definition of $\nabla f$ satisfies $t_{\rho}(\rho+\Delta) = f(\rho) + \tr{(\nabla f(\rho))^T \Delta}$.
}
\begin{align}
(\nabla f_\mathrm{obj}(\rho))^T = \left( \sum_{x\in\{0,1\}} \keyw_x \calG^\dagger \!\left(\log\calG(\rho) - \log\calZ_x(\calG(\rho))\right)\right) - \vec{\lambda}\cdot\vec{\Gamma}(\thA,(r_Z,r_X)),
\label{eq:fobjder}
\end{align}
where $\calG^\dagger$ is the adjoint channel of $\calG$. In practice, we will not need to explicitly compute $\calG^\dagger$, because all subsequent arguments rely only on ``inner products'' $\tr{(\nabla f_\mathrm{obj}(\rho))^T\sigma}$, which can be rewritten as
\begin{align}
\tr{(\nabla f_\mathrm{obj}(\rho))^T\sigma} = \tr{\left( \sum_{x\in\{0,1\}} \keyw_x \left(\log\calG(\rho) - \log\calZ_x(\calG(\rho))\right)\calG(\sigma)\right) - \vec{\lambda}\cdot\vec{\Gamma}(\thA,(r_Z,r_X))\sigma}.
\end{align}

We also note that the optimization domain is a set defined by PSD constraints (namely, $\rho$ of the form~\eqref{eq:almostdiag} with $\rho\geq0$ and $\tr{\rho}=1$), which means that optimizing any affine function over this set is an SDP, which can be efficiently solved and bounded~\cite{BV04v8} via its dual value. Together with the explicit expression~\eqref{eq:fobjder} for the derivative of the objective function, this makes the optimization problem here a prime candidate for the {Frank-Wolfe algorithm}~\cite{FW56}, which yields \emph{secure} lower bounds on the true minimum value of the optimization. 

The Frank-Wolfe algorithm is based on a simple geometric insight: for any point in the domain of a convex function $f_\mathrm{obj}$, the tangent hyperplane (or a supporting hyperplane, if $f_\mathrm{obj}$ is not differentiable at that point) to the graph of $f_\mathrm{obj}$ at that point yields an affine lower bound on $f_\mathrm{obj}$. Hence if we can minimize affine functions over the optimization domain, then any such tangent hyperplane lets us obtain a lower bound on the minimum of $f_\mathrm{obj}$ on the domain. In addition, it intuitively seems that taking the tangent at points closer to the true optimal solution should yield tighter lower bounds (though there are some technical caveats). 
Thus in principle, one could simply perform some heuristic computations to get an estimate of where the true minimum lies, then take the tangent at that estimate and obtain the corresponding lower bound. For this optimization, however, we found that the results were fairly sensitive to deviations from the true optimum. To get the bounds to converge, we found it more efficient to use the standard algorithm,
which (under mild assumptions) can be proven to converge in order $O(1/k)$ after $k$ iterations: 
\newcommand{\epsd}{\eps}
\begin{savenotes}
\begin{algorithm}[htp]
\caption*{\textbf{Frank-Wolfe algorithm} (As presented in~\cite{WLC18})}
Let the domain of optimization be $\optdom$
and the acceptable gap between the feasible values and the lower bounds be $\eps_\mathrm{tol}>0$.
\begin{algorithmic}[1]
\State Set $k=0$
and heuristically find $\rho_0 \defvar \arg\min_{\rho\in\optdom}f_\mathrm{obj}(\rho)$.
\State\label{step:FWSDP}Solve the SDP
\begin{align}
\min_{\Delta} \tr{(\nabla f_\mathrm{obj}(\rho_k))^T \Delta} \suchthat \; \rho_k + \Delta \in \optdom, \label{eq:FWSDP}
\end{align}
which returns a feasible point $\Delta^*$
as well as a dual value $-\epsd \leq 0$ 
that lower-bounds the true minimum of the SDP.
\State If $\epsd \leq \eps_\mathrm{tol}$ then stop and return $f_\mathrm{obj}(\rho_k)-\epsd$ as a lower bound on the optimization.
\State Otherwise, heuristically find 
$\mu^* \defvar \arg\min_{\mu\in(0,1)}f_\mathrm{obj}(\rho_k + \mu\Delta^*)$.
\State Set $\rho_{k+1} = \rho_k + \mu^* \Delta^*$, $k\leftarrow k+1$ and return to Step~\ref{step:FWSDP}.
\end{algorithmic}
\end{algorithm}
\end{savenotes}

\noindent Geometrically, the SDP in Eq.~\eqref{eq:FWSDP} corresponds to considering the tangent to $f_\mathrm{obj}$ at $\rho_k$ and computing the maximum amount by which it can decrease (as compared to its value at $\rho_k$) over the domain $\optdom$. 
As described earlier,
this bounds the maximum amount by which $f_\mathrm{obj}$ can decrease from $f_\mathrm{obj}(\rho_k)$ over $\optdom$.

In our application of the Frank-Wolfe algorithm, there is the technical issue that if $\calG(\rho_k)$ is singular (or has negative eigenvalues, from numerical imprecision), then the derivative at $\rho_k$ (Eq.~\eqref{eq:fobjder}) is ill-behaved. 
To cope with this, we used
the heuristic solution of simply replacing $\rho_k$ with $(1-\delta)\rho_k + \delta \idk$ in Step~\ref{step:FWSDP}, where $\delta 
\defvar 
\max(\operatorname{eigenvalues}(-\calG(\rho_k)) \cup \{ 10^{-14} \})
$. 
Note that this does not affect the \emph{security} of the result, since it merely corresponds to taking a tangent at a slightly different point, which still yields a valid lower bound. On the other hand, it could possibly affect the theoretical convergence rates, but in practice this did not appear to pose a significant problem in our setting. (In~\cite{WLC18}, this issue is addressed by analyzing a ``perturbed'' version of the optimization and applying a continuity bound, but we found that for the level of accuracy we desired in this work, the admissible perturbation values are too small to cope with the negative eigenvalues that occur.)
\clearpage

\section{Overall algorithm}
\label{asec:algorithmsummary}

Putting the above results together, we see that for any set of intervals $[\theta_j-\delta_j,\theta_j+\delta_j]$ such that $[0,\pi]\subseteq \bigcup_j[\theta_j-\delta_j,\theta_j+\delta_j]$, and set $V$ such that $\semicset \subseteq \operatorname{Conv}(V)$, we have 
\begin{align}
\min_{\rho_{AB}} \min_{\thB} \min_{\thA}
F_\mathrm{obj}(\thA,\thB,\rho_{AB}) 
&\geq \min_{\rho_{AB}} \min_{\thB} \min_j F_\mathrm{obj}(\theta_j,\thB,\rho_{AB}) - \cont(\delta_j) \nonumber \\
&= \min_j \min_{\rho_{AB}} \min_{(r_Z,r_X)\in \semicset} F_\mathrm{obj}(\theta_j,(r_Z,r_X),\rho_{AB}) - \cont(\delta_j) \nonumber \\
&\geq \min_j \min_{\rho_{AB}} \min_{(r_Z,r_X)\in V} F_\mathrm{obj}(\theta_j,(r_Z,r_X),\rho_{AB}) - \cont(\delta_j) \nonumber \\
&= \min_j \min_{(r_Z,r_X)\in V} \min_{\rho_{AB}}  F_\mathrm{obj}(\theta_j,(r_Z,r_X),\rho_{AB}) - \cont(\delta_j).
\label{eq:mincommute}
\end{align}
We can refine the intervals $[\theta_j-\delta_j,\theta_j+\delta_j]$ and set $V$ by the iterative processes described above, with the innermost minimization over $\rho_{AB}$ being handled by the Frank-Wolfe algorithm. It is clearly possible to swap the order of $\min_j$ and $\min_{(r_Z,r_X)\in V}$ in the last line; however, some heuristic plots of the objective function suggest that performing the optimizations in the order shown here is slightly faster. 

To see that this expression can converge to a tight bound, observe that the first inequality in~\eqref{eq:mincommute} becomes arbitrarily tight as we choose smaller values of $\delta_j$. For the second inequality, note that 
the described algorithm chooses $V$ in such a way that the minimum over $V$ approaches the minimum over $\semicset$, hence this inequality also becomes arbitrarily tight.

It would be convenient for future analysis if it were possible to develop closed-form expressions for $\breve{\lin}_\p$ (as was done in~\cite{HST+20,WAP21,SBV+21}), rather than the computationally intensive numerical approaches shown above. However, this appears to be rather challenging, as noted in~\cite{WAP21}. In particular, we found numerical evidence against the conjecture proposed in~\cite{SGP+21} that the minimum in~\eqref{eq:mainopt} (when restricted to qubit strategies) can always be attained by states such that $\rho_{AB}$ is of rank $2$. (A similar observation was reported in~\cite{WAP21}.) More precisely, in the process of heuristically solving the optimizations in order to estimate suitable choices of $\lambda$, we discovered that if we imposed the additional restriction that $\rho_{AB}$ has rank $2$, there was a small but numerically significant difference as compared to the results without this restriction. (The states that heuristically approach the minimum in the latter indeed tend to have two very small eigenvalues, but it appears that these eigenvalues cannot be reduced exactly to zero.) This suggests that the aforementioned conjecture is not true after all, which poses a challenge for closed-form analysis because the eigenvalues involved in computing the entropy are genuinely roots of a fourth-degree polynomial (in~\cite{HST+20,SBV+21}, a key element of the analysis was to argue that it suffices to consider rank-$2$ $\rho_{AB}$, simplifying the expression for the eigenvalues). 

\section{Optimality of depolarizing-noise threshold}
\label{asec:optthresh}

We now derive an explicit upper bound on the depolarizing-noise threshold for protocols of the form we consider. (We highlight that because we will do so by constructing an extremely generic attack, this upper bound holds for \emph{all} protocols of this form, regardless of choices of parameters such as the input distributions and noisy preprocessing. To some extent, it is surprising that there even exist any protocols that achieve thresholds close to the upper bounds implied by such a generic attack.) We first recall that the measurement statistics in the depolarizing-noise model can be viewed as being obtained via real projective measurements on the Werner state $\varrho_\q$ defined in~\eqref{eq:werner}. It is known~\cite{AGT06,Kri79} that these measurement statistics can be reproduced by a local-hidden-variable (LHV) model as soon as the state does not violate the CHSH inequality, i.e.~for $q\geq q_2 \defvar{(2-\sqrt{2})}/{4} \approx 14.6\pct$.\footnote{In fact, this is also true in a more general context where the honest parties have \emph{any} number of real projective measurements on the Werner state, via the value of the second Grothendieck constant $K_G(2)=\sqrt{2}$~\cite{AGT06,Kri79}. For general (i.e.~not necessarily real) projective measurements the analogous value is known to satisfy $q_3 \lesssim 15.9\pct$, as follows from the best known bound on the third Grothendieck constant $K_G(3)$~\cite{HQV+17}.} 

In the device-independent setting, the existence of an LHV model yields an attack for Eve that gives her full knowledge of the measurement outcomes.
This means that for any depolarizing-noise value $\q$, we can construct the following attack for Eve: first, she generates a classical ancilla bit which is equal to $0$ with probability $\pBell \defvar (\q_2-\q)/\q_2$. If the bit is equal to $0$, Alice and Bob's devices simply perform the honest measurements on the maximally entangled state $\Phi^+$. Otherwise, the devices implement an LHV model that yields the same statistics as the honest measurements performed on the Werner state $\varrho_{\q_2}$, but which gives Eve full knowledge of the outcomes. This attack indeed reproduces the statistics corresponding to depolarizing noise~$\q$, since $\pBell\pure{\Phi^+} + (1-\pBell) \varrho_{\q_2} = \varrho_{\q}$ and the second strategy produces the same statistics as the honest measurements performed on $\varrho_{\q_2}$.

In the two cases, the entropies of Alice's outputs after noisy preprocessing are $H(\hat{A}_x|\tilde{E})_\mathrm{triv}=1$ (i.e.~Eve's side-information is trivial since the Alice--Bob state is pure) and $H(\hat{A}_x|\tilde{E})_\mathrm{LHV}=\binh(\p)$ (i.e.~the uncertainty arises purely from the noisy preprocessing) respectively, where $\tilde{E}$ denotes Eve's side-information excluding the ancilla bit. Incorporating the ancilla bit into Eve's side-information $E$, this attack achieves
\begin{align}
H(\hat{A}_x|E) = \pBell H(\hat{A}_x|\tilde{E})_\mathrm{triv} + (1-\pBell) H(\hat{A}_x|\tilde{E})_\mathrm{LHV} 
= \frac{\q_2-\q}{\q_2} + \frac{\q}{\q_2} \binh(p).
\end{align}
This expression hence serves as an upper bound on the best possible lower bound we could derive on the conditional entropy against Eve. Also, the conditional entropy against Bob's output (for Bob's optimal key-generation measurements on Werner states) is simply
\begin{align}
H(\hat{A}_x|B_x) = \binh(\p +(1-2\p) \q).
\label{eq:entBob}
\end{align}
Recalling that the asymptotic keyrate is given by the Devetak-Winter expression~\eqref{eq:devwin} (and applying the same simplifications leading to~\eqref{eq:randomkeyrate}), this yields a simple upper bound on the critical value of $\q$ that still allows a positive asymptotic keyrate (for protocols of this form, i.e.~applying noisy preprocessing, random key measurements, and considering the full output distribution, but only using one-way error correction):
\begin{align}
\q_\textrm{att}(\p) \defvar& \max \bigg\{ \q 
\, \bigg| \, 
\sum_x \keyw_x (H(\hat{A}_x|E) - H(\hat{A}_x|B_x))  \geq 0 \bigg\} \nonumber\\
=& \max \left\{ \q 
\, \middle| \,
\frac{\q_2-\q}{\q_2} + \frac{\q}{\q_2} \binh(\p) -  \binh(\p +(1-2\p) \q) \geq 0 \right\},
\end{align}
where in the first line the entropies refer to those of the attack we have described.
We observe numerically that this upper bound is increasing with respect to $\p$, so we have $\q_\textrm{att}(\p)\leq \q_\textrm{att}(\p\to1/2)$. In the limit $\p\to 1/2$ we can write $\p=1/2-\delta$ and expand the expression in small $\delta$ to find
\begin{align}
\q_\textrm{att}\!\left(\p\to\frac{1}{2}\right) = \frac{1+ 4 \q_2 -\sqrt{8 \q_2+1}}{8 \q_2} = 1 - \frac{\sqrt{7+4\sqrt{2}} - 1}{2\sqrt{2}} \approx 9.57\pct.
\end{align}
Hence for protocols of the form described in this work, it is not possible for the depolarizing-noise threshold to exceed this value. For the parameter choice $\p=0.3$, we have $\q_\textrm{att}(0.3) \approx 9.51 \pct$, which is close to the threshold of $9.33 \pct$ for which we could certify a positive keyrate.\footnote{These two thresholds cannot match exactly, because the feasible points shown in Fig.~\ref{fig:entbnds} also imply that the tight bound $\breve{\lin}_{\p}({\constr})$ is not \emph{exactly} equal to the linear interpolation between $(2,\binh(\p))$ and $(2\sqrt{2},1)$, which essentially corresponds to the attack we describe here.}

Finally, it is worth noting that while the main focus of this section is depolarizing noise applied to the statistics corresponding to the ideal CHSH measurements on $\Phi^+$, the above analysis in fact generalizes to a substantially larger family of scenarios. Specifically, the same analysis applies for depolarizing noise applied to the statistics from any number of real projective measurements on $\Phi^+$, since the results of~\cite{AGT06,Kri79} yield the required LHV models.\footnote{A more general consideration would be the scenario of depolarizing noise applied to the statistics from an arbitrary two-qubit state, but it is not immediately obvious whether the thresholds in~\cite{AGT06,Kri79,HQV+17} apply to such states as well, so we leave this for future work.} (If Bob performs suboptimal generation measurements, then Eq.~\eqref{eq:entBob} still holds as a lower bound, so the argument carries through.) In addition, replacing $\q_2$ with $\q_3$ straightforwardly provides a threshold of $\q \lesssim 10.01 \pct$ for all possible projective measurements on that state, via the respective known bound in~\cite{HQV+17}. (However, in all such cases where there are more than 2 possible measurements, there is no qubit reduction to make it easier to derive corresponding lower bounds, hence the lower bounds have also not been very thoroughly explored.) 

The above argument relies on the fact that depolarizing-noise statistics can be obtained by simple mixtures of ``extremal'' strategies. This is not necessarily the case for other noise models, such as limited detection efficiency in photonic experiments. Therefore, the same argument cannot be directly applied to obtain upper bounds on the thresholds for such forms of noise.

\chapter{Detection efficiency thresholds}
\label{app:detthresh}

Here we briefly summarize some values that have been claimed as the detection efficiency threshold at which DIQKD is possible, using only the~\cite{PAB+09} bound (based on CHSH value) to compute the keyrate, and using Pauli measurements on two-qubit states for the noiseless distribution. Since these results preceded the random-key-measurements idea, they are all situations where Alice and Bob used a fixed input in all generation rounds.
Alice has 2 inputs and Bob has 3 inputs --- two of Bob's inputs are used only in test rounds (to compute the CHSH value), and his last input is used only in generation rounds (so it should ideally be well-correlated with Alice's generation input).
Grouping them roughly, we can describe the thresholds as follows:

\begin{itemize}
\item If the noiseless distribution is fixed as~\eqref{eq:noiselessCHSH}, and Bob's outputs are mapped to binary values, the threshold is $92.4\pct$. If Bob's $\perp$ outcome is preserved, this improves to $90.9\pct$~\cite{ML12}.

\item If the noiseless distribution is optimized to maximize the CHSH value (as in~\cite{Ebe93}), and Bob's generation measurement is simply set to be in the same basis as Alice's, and Bob's outputs are mapped to binary values, the threshold is $90.9\pct$. (The closeness to the~\cite{ML12} threshold appears to just be a numerical coincidence.) This can be slightly improved to $90.7\pct$ by optimizing Bob's generation measurement separately instead of just setting it to be in the same basis as Alice's. This was the value previously quoted as the reference in~\cite{TLR20}. 

\item If the noiseless distribution is instead optimized to maximize the keyrate (as computed via the~\cite{PAB+09} bound) directly, and Bob's outputs are mapped to binary values, the threshold is $88.4\pct$. (A slightly worse value of $89.3\pct$ was stated in~\cite{SBV+21} because that considers a situation where the distribution is symmetrized and the symmetrization bit is discarded afterwards, instead of being retained for error-correction purposes.) By preserving Bob's $\perp$ outcome, this threshold is improved to $86.5\pct$~\cite{BFF21,WAP21,SBV+21}.
\end{itemize}

\chapter{Detailed security proofs for collective attacks}
\label{app:colloverall}

\section{Proof of Theorem~\ref{th:collective}}
\label{asec:collective}

We observe that the completeness proof from Sec.~\ref{sec:com} directly applies, since we have not changed the honest behaviour. As for the soundness proof, we 
follow the argument in Sec.~\ref{sec:soundness} up until Eq.~\eqref{eq:secsplit}, upon which we again require a bound on the first term in that equation. To do so, we note that since the device behaviour is identical in every round, the true CHSH winning probability for the devices can be denoted as a constant value $\wtrue$ (though this value is not directly known to Alice and Bob). We can thus define the following exhaustive possibilities for the device behaviour (again, $\rho$ is the state at the end of Protocol~\ref{prot:virtual}, and the first case is just to deal with technicalities regarding the smoothing parameter):
{\setlist[enumerate]{leftmargin=2cm,labelindent=2cm}
\begin{enumerate}[label*=Case \arabic*:,ref=\arabic*]
\item \label{case:ECsmall2} For the state $\rho$, $\pr{\Og \land \Oh \land \OpPE} \leq \es^2$.
\item \label{case:wlow} $\wtrue < \wexp - \dtol - \dsou$.
\item \label{case:whigh} 
Neither of the above are true.
\end{enumerate}
}
\noindent In case~\ref{case:ECsmall2}, the first term of Eq.~\eqref{eq:secsplit} is immediately bounded by
\begin{align}
\pr{\Og \land \Oh \land \OpPE} \leq \es^2.
\end{align}
In case~\ref{case:wlow}, we observe that in each round, the probability of the CHSH game being played and the devices winning is $\wtrue\gamma$. Therefore, we have
\begin{align}
\pr{\freq_\ctstr(1) \geq (\wexp-\dtol)\gamma}
&=\pr{\freq_\ctstr(\neg 1) < 1-(\wexp-\dtol)\gamma} \nonumber\\
&\leq\pr{\freq_\ctstr(\neg 1)n \leq \floor{(1-(\wexp-\dtol)\gamma)n}} \nonumber\\
&= \cdfBin{n}{1-\wtrue\gamma}{\floor{(1-(\wexp-\dtol)\gamma)n}} \nonumber\\
&\leq \eIID, 
\label{eq:eIIDbound}
\end{align}
where the last inequality follows from the fact that we have $1-\wtrue\gamma > 1-(\wexp - \dtol - \dsou)\gamma$, and $\cdfBin{n}{p}{k} \leq \cdfBin{n}{p'}{k}$ for $p \geq p'$.
(Again, one could obtain a simpler 
expression by noting $\pr{\freq_\ctstr(1) \geq (\wexp-\dtol)\gamma} 
\leq \pr{\freq_\ctstr(1) \geq (\wtrue+\dsou)\gamma}$ and then upper bounding the latter expression via the Chernoff bound in the same way as the second line of~\eqref{eq:chernoffs} (as long as $\dsou$ is chosen small enough to ensure $\wtrue + \dsou \leq 1$, e.g.~by choosing it smaller than $1-w_\mathrm{max}$ where $w_\mathrm{max}$ is the maximum quantum CHSH winning probability):
\begin{align}\label{eq:chernoffeIID}
\pr{\freq_\ctstr(1) \geq (\wtrue+\dsou)\gamma} 
\leq e^{-\frac{n \gamma \dsou^2}{2}} ,
\end{align}
but this gives significantly poorer results.) 
Overall, this allows us to bound the first term of Eq.~\eqref{eq:secsplit} in this case by
\begin{align}
\pr{\Og \land \Oh \land \OpPE} \leq \pr{\OpPE} \leq \pr{\freq_\ctstr(1) \geq (\wexp-\dtol)\gamma} \leq \eIID.
\end{align}

As for case~\ref{case:whigh}, we can directly apply the AEP (Corollary~4.10 of~\cite{DFR20}) to obtain the following bound 
in place of Theorem~\ref{th:rawHmin}, with the slight difference that we have not conditioned on $\OpPE$ yet:
\begin{align}
\Hmin^{\es}(\str{A}|\str{X}\str{Y} \allE)_{\rho} >
n\g(\wexp-\dtol-\dsou) - \sqrt{n} \,(2\log5)\sqrt{\log\frac{2}{\es^2}}.
\label{eq:rawHminAEP}
\end{align}
The remainder of the analysis proceeds very similarly to 
Sec.~\ref{sec:soundness}. Specifically, since we have $\pr{\Og \land \Oh \land \OpPE} \geq \es^2$ in this case, we can conclude 
\begin{align}
\Hmin^{\es}(\str{A}|\str{X}\str{Y}\str{L} \allE)_{\rho_{\land \Og \land \Oh \land \OpPE}}
&\geq \Hmin^{\es}(\str{A}|\str{X}\str{Y}\str{L} \allE)_{\rho}
\nonumber \\
&\geq \Hmin^{\es}(\str{A}|\str{X}\str{Y} \allE)_{\rho} 
- \cmax - \ceil{\log\left(\frac{1}{\eh}\right)}.
\end{align}
Putting this together with Eq.~\eqref{eq:rawHminAEP}, we see that as long as we choose $\lkey$ satisfying Eq.~\eqref{eq:lkeycoll}, we will have
\begin{align}
\frac{1}{2}\left(\Hmin^{\es}(\str{A}|\str{X}\str{Y}\str{L} \allE)_{\rho_{\land \Og \land \Oh \land \OpPE}} - \lkey  + 2\right) \geq \log\frac{1}{\ePA}.
\end{align}
Again noting that $\pr{\Og \land \Oh \land \OpPE} \geq \es^2$ in this case, the Leftover Hashing Lemma then implies the first term of Eq.~\eqref{eq:secsplit} is bounded by
\begin{align}
&\frac{1}{2}\norm{\mPA(\rho_{\land \Og \land \Oh \land \OpPE})_{K_A E'} - \idk_{K_A} \otimes \mPA(\rho_{\land \Og \land \Oh \land \OpPE})_{E'}}_1 \nonumber\\
\leq{}& 2^{-\frac{1}{2} \left(\Hmin^{\es}(\str{A}|\str{X}\str{Y}\str{L} \allE)_{\rho_{\land \Og \land \Oh \land \OpPE}} - \lkey + 2\right)} + 2\es \nonumber\\
\leq{}& \ePA + 2\es.
\end{align}

Therefore, we finally conclude that under the collective-attacks assumption, the secrecy condition is satisfied by choosing
\begin{align}
\esecr = \max\{\es^2, \eIID, \ePA + 2\es\} + \eh = \max\{\eIID, \ePA + 2\es\} + \eh.
\end{align} 
As before, the protocol is $\eh$-correct, so we conclude that it is $(\max\{\eIID, \ePA + 2\es\} + 2\eh)$-sound when performed with $\lkey$ satisfying Eq.~\eqref{eq:lkeycoll}.

\section{Proof of Theorem~\ref{th:optcoll}}
\label{asec:optcoll}

To prove this theorem, we first observe that the claimed completeness bound holds because
\begin{align}
\pr{\freq_{\str{c}_t}(1) < \wexp-\dtol}_\mathrm{hon}
&\leq 
\pr{\freq_{\str{c}_t}(1)\gamma n \leq \floor{(\wexp-\dtol)\gamma n}}_\mathrm{hon} \nonumber\\
&= \cdfBin{\gamma n}{\wexp}{\floor{(\wexp-\dtol)\gamma n}},
\end{align}
and hence by applying the union bound, we see that the probability of the honest protocol aborting is at most $\ecEC + \ecPE$.

As for soundness, we follow the argument in Sec.~\ref{sec:finiteproof} up until Eq.~\eqref{eq:esecrnew}, though it is no longer necessary to introduce the virtual parameter estimation, since Bob has access to the exact values of $\str{A}_t$. We also need to slightly modify the event definitions (and it is no longer necessary to consider the event $\OpPE$): 
{\setlist[description]{leftmargin=2cm,labelindent=2cm}
\begin{description}
\item $\Og$: $\str{A}_g = \tilde{\str{A}}_g$ 
\item $\Oh$: $\hash(\str{A}_g) = \hash(\tilde{\str{A}}_g)$
\item $\OPE$: $\freq_{\str{c}_t}(1)\geq \wexp-\dtol$ 
\end{description}
}
\noindent Again, the accept condition for this protocol can be stated as the event $\Oh \land \OPE$. 
For this protocol, we can prove a bound of the form~\eqref{eq:esecrnew} directly, without splitting it into the terms in Eq.~\eqref{eq:secsplit}. To do so, denote the state after the parameter-estimation step in Protocol~\ref{prot:optcoll} as $\rho$, and consider the following exhaustive possibilities:
{\setlist[enumerate]{leftmargin=2cm,labelindent=2cm}
\begin{enumerate}[label*=Case \arabic*:,ref=\arabic*]
\item \label{case:ECsmall3} 
For the state $\rho$, $\pr{\Oh \land \OPE} \leq \es^2$.
\item \label{case:wlow3} $\wtrue < \wexp - \dtol - \dsou$.
\item \label{case:whigh3} Neither of the above are true.
\end{enumerate}
}
In case~\ref{case:ECsmall3}, the left-hand-side of Eq.~\eqref{eq:esecrnew} is immediately bounded by $
\es^2$.
In case~\ref{case:wlow3}, we observe that in each of the $\gamma n$ test rounds (note that now we focus only on the test rounds, instead of the full output strings as we did in the previous proofs), the probability of winning the CHSH game is $\wtrue
$. Therefore, we have
\begin{align}
\pr{\freq_{\str{c}_t}(1) \geq \wexp-\dtol}
&=\pr{\freq_{\str{c}_t}(
0
) < 1-\wexp+\dtol} \nonumber\\
&=\pr{\freq_{\str{c}_t}(0)\gamma n < (1-\wexp+\dtol)\gamma n} \nonumber\\
&\leq\pr{\freq_{\str{c}_t}(0)\gamma n \leq \floor{(1-\wexp+\dtol)\gamma n}} \nonumber\\
&= \cdfBin{\gamma n}{1-\wtrue}{\floor{(1-\wexp+\dtol)\gamma n}} \nonumber\\
&\leq \eIID,
\end{align}
which allows us to bound the left-hand-side of Eq.~\eqref{eq:esecrnew} in this case by
\begin{align}
\pr{\Oh \land \OPE} \leq \pr{\OPE} = \pr{\freq_{\str{c}_t}(1) \geq \wexp-\dtol} \leq \eIID.
\end{align}

Finally, for case~\ref{case:whigh3}, we note that since Eve's quantum side-information $\allE$ is genuinely in (IID) tensor-product form under the collective-attacks assumption, we can denote the subsets corresponding to test and generation rounds as $\str{E}_t$ and $\str{E}_g$ respectively. The AEP in~\cite{DFR20} then gives
\begin{align}
\Hmin^{\es}(\str{A}_g|\str{X}_g\str{Y}_g\str{E}_g)_{\rho} >
(1-\gamma)n\lin_\p(\wexp-\dtol-\dsou) - \sqrt{(1-\gamma)n} \,(2\log5)\sqrt{\log\frac{2}{\es^2}},
\end{align}
and 
we proceed similarly: since we have $\pr{\Oh \land \OPE} \geq \es^2$ in case~\ref{case:whigh3}, we can conclude
\begin{align}
\Hmin^{\es}(\str{A}_g|\str{A}_t\str{X}\str{Y}\str{L} \allE)_{\rho_{\land \Oh \land \OPE}}
&\geq \Hmin^{\es}(\str{A}_g|\str{A}_t\str{X}\str{Y}\str{L} \allE)_{\rho}
\nonumber \\
&\geq \Hmin^{\es}(\str{A}_g|\str{A}_t\str{X}\str{Y} \allE)_{\rho} 
- \cmax - \ceil{\log\left(\frac{1}{\eh}\right)} \nonumber \\
&= \Hmin^{\es}(\str{A}_g|\str{X}_g\str{Y}_g\str{E}_g)_{\rho}
- \cmax - \ceil{\log\left(\frac{1}{\eh}\right)},
\end{align}
where the last line holds because the test rounds are independent of the generation rounds
(this is intuitive, but for completeness we provide a proof; see Fact~\ref{fact:Hmintensor} below).
Hence as long as we choose $\lkey$ satisfying Eq.~\eqref{eq:lkeyoptcoll}, we will have
\begin{align}
\frac{1}{2}\left(\Hmin^{\es}(\str{A}_g|\str{A}_t\str{X}\str{Y}\str{L} \allE)_{\rho_{\land \Oh \land \OPE}} - \lkey  + 2\right) \geq \log\frac{1}{\ePA}.
\end{align}
Again noting that $\pr{\Oh \land \OPE} \geq \es^2$ in this case, the Leftover Hashing Lemma then implies 
\begin{align}
&\frac{1}{2}\norm{\mPA(\rho_{\land \Oh \land \OPE})_{K_A E'} - \idk_{K_A} \otimes \mPA(\rho_{\land \Oh \land \OPE})_{E'}}_1 \nonumber\\
\leq{}& 2^{-\frac{1}{2} \left(\Hmin^{\es}(\str{A}_g|\str{A}_t\str{X}\str{Y}\str{L}\allE)_{\rho_{\land \Oh \land \OPE}} - \lkey + 2\right)} + 2\es \nonumber\\
\leq{}& \ePA + 2\es.
\end{align}

Therefore, we finally conclude that under the collective-attacks assumption, the secrecy condition is satisfied by choosing
\begin{align}
\esecr = \max\{\es^2, \eIID, \ePA + 2\es\} = \max\{\eIID, \ePA + 2\es\}.
\end{align} 
As before, the protocol is $\eh$-correct, so we conclude that it is $(\max\{\eIID, \ePA + 2\es\} + \eh)$-sound when performed with $\lkey$ satisfying Eq.~\eqref{eq:lkeyoptcoll}.\\

For completeness, we provide a proof of an intermediate statement used in the above argument (the main technicality is verifying that it holds even with smoothing; also, for generality we prove the statement for subnormalized states):
\begin{fact}\label{fact:Hmintensor}
For any $\rho\in\dop{\leq}(ST), \sigma\in\dop{\leq}(R)$, we have $\Hmin^{\eps}(S|TR)_{\rho\otimes\sigma} = \Hmin^{\eps}(S|T)_{\tr{\sigma}\rho}$.
\end{fact}
\begin{proof}
We shall use the fact that the min-entropy 
can be equivalently defined~\cite{Tom16} as
\newcommand{\posop}{K}
\begin{align}
\Hmin(A|B)_\rho \defvar 
\max_{\posop_{B}} 
\left\{
-\log 
\tr{\posop_{B}} \;\middle|\; \id_{A} \otimes \posop_{B} \geq \rho_{AB} \text{ and } \posop_{B} \geq 0 \right\}.
\label{eq:Hmindef2}
\end{align}

We first show that $\Hmin^{\eps}(S|TR)_{\rho\otimes\sigma} \geq \Hmin^{\eps}(S|T)_{\tr{\sigma}\rho}$. To do so, we start by proving a version without smoothing, i.e.~$\Hmin(S|TR)_{\rho\otimes\sigma} \geq \Hmin(S|T)_{\tr{\sigma}\rho}$. 
Let $\posop'_{T}$ be an operator that attains the optimum in the definition~\eqref{eq:Hmindef2} for $\Hmin(S|T)_{\tr{\sigma}\rho}$. Then the operator inequality $\id_{S} \otimes \posop'_{T} \otimes \sigma_{R} \geq \tr{\sigma}\rho_{ST} \otimes \sigma_{R}$ holds, which implies $\tr{\sigma}^{-1} \posop'_{T} \otimes \sigma_{R}$ is a feasible point in the optimization that defines $\Hmin(S|TR)_{\rho\otimes\sigma}$. Hence the value of $\Hmin(S|TR)_{\rho\otimes\sigma}$ is at least $-\log\tr{\tr{\sigma}^{-1} \posop'_{T} \otimes \sigma_{R}} 
= -\log\tr{\posop'_{T}}
= \Hmin(S|T)_{\tr{\sigma}\rho}$, as claimed.

To prove the inequality with smoothing, we use the fact that the purified distance is monotone under CPTP (more generally, completely positive trace-nonincreasing) maps~\cite{Tom16}. Let $\tilde{\rho}\in\dop{\leq}(ST)$ be a state in the $\eps$-ball (in purified distance) around $\tr{\sigma}\rho$ such that $\Hmin(S|T)_{\tilde{\rho}} = \Hmin^{\eps}(S|T)_{\tr{\sigma}\rho}$. 
By considering the CPTP map that simply appends the normalized state $\hat{\sigma}\defvar\tr{\sigma}^{-1}\sigma$, we see that $\pd(\tilde{\rho}\otimes\hat{\sigma} , \tr{\sigma}\rho\otimes\hat{\sigma}) \leq \pd(\tilde{\rho} , \tr{\sigma}\rho) \leq \eps$, i.e.~$\tilde{\rho}\otimes\hat{\sigma}$ is in the $\eps$-ball around $\tr{\sigma}\rho\otimes\hat{\sigma}$. But $\tr{\sigma}\rho\otimes\hat{\sigma}$ is just $\rho\otimes\sigma$, hence the value of $\Hmin^{\eps}(S|TR)_{\rho\otimes\sigma}$ is at least $\Hmin(S|TR)_{\tilde{\rho}\otimes\hat{\sigma}} \geq \Hmin(S|T)_{\tilde{\rho}} = \Hmin^{\eps}(S|T)_{\tr{\sigma}\rho}$, where the first inequality holds by the unsmoothed version proven above.

We now show that the opposite inequality $\Hmin^{\eps}(S|T)_{\tr{\sigma}\rho} \geq \Hmin^{\eps}(S|TR)_{\rho\otimes\sigma}$ holds as well, hence the values must be equal. 
In fact, we shall prove a slightly stronger result, that $\Hmin^{\eps}(S|T)_{\omega} \geq \Hmin^{\eps}(S|TR)_{\omega}$ for any $\omega\in\dop{\leq}(STR)$, with the understanding that $\omega_{ST} \defvar \tr[R]{\omega_{STR}}$ for subnormalized states.
(Qualitatively, this inequality states that conditioning on additional systems never increases the smoothed min-entropy.)
We again start by proving the unsmoothed version, $\Hmin(S|T)_{\omega} \geq \Hmin(S|TR)_{\omega}$ for any $\omega\in\dop{\leq}(STR)$. To do so, let $\posop''_{TR}$ be an operator that attains the optimum in the definition~\eqref{eq:Hmindef2} for $\Hmin(S|TR)_{\omega}$. Then since operator inequalities are preserved by partial trace, we have $\id_{S} \otimes \posop''_{T} = \tr[R]{\id_{S} \otimes \posop''_{TR}} \geq 
\tr[R]{\omega_{STR}} = \omega_{ST}$, so $\posop''_{T}$ is a feasible point in the optimization that defines $\Hmin(S|T)_{\omega}$. Hence $\Hmin(S|T)_{\omega}$ is at least $-\log\tr{\posop''_{T}} = -\log\tr{\posop''_{TR}} = \Hmin(S|TR)_{\omega}$, as claimed.

To prove the smoothed version, let $\tilde{\omega}\in\dop{\leq}(STR)$ be a state in the $\eps$-ball around $\omega$ such that $\Hmin(S|TR)_{\tilde{\omega}} = \Hmin^{\eps}(S|TR)_{\omega}$.
By considering the CPTP map that traces out $R$, we see that $\pd(\tilde{\omega}_{ST},\omega_{ST}) \leq \pd(\tilde{\omega} , \omega) \leq \eps$, i.e.~$\tilde{\omega}_{ST}$ is in the $\eps$-ball around $\omega_{ST}$. Hence the value of $\Hmin^{\eps}(S|T)_{\omega}$ is at least $\Hmin(S|T)_{\tilde{\omega}} \geq \Hmin(S|TR)_{\tilde{\omega}} = \Hmin^{\eps}(S|TR)_{\omega}$, where the first inequality holds by the unsmoothed version proven above. This is the desired result.
\end{proof}

\chapter{Alternative security proof for advantage distillation}
\label{app:AD}

\newcommand{\ovlap}{G} 

Here we provide an alternative proof of Theorem~\ref{th_fidbound}, which also yields a potentially more general result. Essentially, the result holds for
any distinguishability measure $\ovlap(\rho,\sigma)$ (on normalized states $\rho,\sigma$) that shares the following properties of the fidelity:
\begin{gather}
\ovlap(\rho,\sigma) = \ovlap(\sigma,\rho) 
, \label{eq:symm} \\     
\ovlap(\rho \otimes \rho',\sigma \otimes \sigma') \geq \ovlap(\rho,\sigma) \ovlap(\rho',\sigma') 
, \label{eq:supmult} \\
\ovlap(\rho,\sigma)^2 + d(\rho,\sigma)^2 \leq 1 , \label{eq:FvdGupp} 
\end{gather}
The first line states that it is symmetric, the second is a ``one-sided'' multiplicativity property across tensor products, and the last is in the same form as one side of the Fuchs--van de Graaf inequality. For the purposes of considering \emph{necessary} conditions for security of the protocol~\cite{arx_HT21}, the other side of the inequality can be useful,
\begin{align}
\ovlap(\rho,\sigma) + d(\rho,\sigma) \geq 1 , \label{eq:FvdGlow} 
\end{align}
but we will not need it here.
(Since the trace distance has range $[0,1]$, the bound~\eqref{eq:FvdGupp} also implies that $\ovlap(\rho,\sigma) \leq 1$, and~\eqref{eq:FvdGlow} implies that $\ovlap(\rho,\sigma) \geq 0$. Again, we will only need the former property, not the latter.)

A concrete example of a distinguishability measure (other than fidelity) that has these properties is the \term{pretty-good fidelity}~\cite{LZ04,Aud14comp,IRS17}, defined as ${F}_\mathrm{pg}(\rho, \sigma) \defvar \tr{\sqrt{\rho} \sqrt{\sigma}}$. Another example is the \term{quantum Chernoff coefficient}~\cite{ACM+07,ANS+08,NS09} (also known under other names, e.g.~quantum Chernoff bound/divergence/information, although often with an additional logarithm), defined as
$Q(\rho, \sigma) \defvar \inf_{s\in(0,1)} \tr{\rho^s \sigma^{1-s}}$.
Both of these quantities satisfy~\eqref{eq:symm}--\eqref{eq:FvdGupp}, and in fact also the bound~\eqref{eq:FvdGlow} (which we do not need here). Technically, the fidelity and pretty-good fidelity satisfy~\eqref{eq:supmult} with equality as well; however, the quantum Chernoff coefficient does not, hence we state that property as an inequality only. 

For generality, we shall first present the new statement in terms of any distinguishability measure $\ovlap$ with the aforementioned properties, deferring the question of picking a ``good'' choice of $\ovlap$ until after the proof.
\begin{theorem}\label{th:gbound}
Let $\ovlap$ be a distinguishability measure with the properties~\eqref{eq:symm}--\eqref{eq:FvdGupp}. Then for a DIQKD protocol as described in Sec.~\eqref{sec:prelimAD}, a sufficient condition for Eq.~\eqref{eq_poskey} to hold for large $\sz$ is for all states accepted in parameter estimation to satisfy
\begin{align}
\ovlap(\rho_{E|00},\rho_{E|11})^2 > \frac{\e}{1-\e},
\label{eq:gbound}
\end{align}
where $\rho_{E|a_0 b_0}$ is Eve's single-round state conditioned on outcomes $(a_0,b_0)$ being obtained for inputs $x=y=0$.
\end{theorem}
\begin{proof}
We follow the proof of Theorem~\ref{th_fidbound} up until the inequality~\eqref{eq_continuity}. 
At this point, we instead observe that the $H(C|\blk{E})_{\tilde{\rho}}$ term is lower bounded by\footnote{Alternatively, we could use the slightly tighter bound $H(C|\blk{E})_{\tilde{\rho}} \geq 1-d(\rho_{\blk{E}|\str{m} \str{m}},\rho_{\blk{E}|\flip{\str{m}} \flip{\str{m}}})$ (Theorem~14 of~\cite{BH09}), directly relating the von Neumann entropy to the trace distance without going through the min-entropy. This yields a better lower bound of $1-\sqrt{1-\alpha^{\sz}} \geq \alpha^{\sz}/2$ in place of~\eqref{eq:taylorbnd}, but again only corresponds to having a different denominator and does not affect the final conclusion.}
$\Hmin(C|\blk{E})_{\tilde{\rho}} = -\log \pg(C|\blk{E})_{\tilde{\rho}} = -\log ((1+d(\rho_{\blk{E}|\str{m} \str{m}},\rho_{\blk{E}|\flip{\str{m}} \flip{\str{m}}}))/2)$, and thus another proof approach would be to upper-bound $d(\rho_{\blk{E}|\str{m} \str{m}},\rho_{\blk{E}|\flip{\str{m}} \flip{\str{m}}})$. This can be achieved by using the ``Fuchs--van de Graaf'' bound~\eqref{eq:FvdGupp}, followed by simplifying the expression using the properties~\eqref{eq:symm}--\eqref{eq:supmult}:
\begin{align}
d(\rho_{\blk{E}|\str{m} \str{m}},\rho_{\blk{E}|\flip{\str{m}} \flip{\str{m}}}) \leq \sqrt{1-\ovlap(\rho_{\blk{E}|\str{m} \str{m}},\rho_{\blk{E}|\flip{\str{m}} \flip{\str{m}}})^2} \leq \sqrt{1-\ovlap(\rho_{E|00},\rho_{E|11})^{2\sz}}.
\end{align}
This gives us the following bound in place of~\eqref{eq_rogabound}:
\begin{align}
H(C|\blk{E})_{\tilde{\rho}} \geq -\log \frac{1+\sqrt{1-\ovlap(\rho_{E|00},\rho_{E|11})^{2\sz}}}{2}.
\label{eq:hminbound}
\end{align}

Denoting $\alpha \defvar \ovlap(\rho_{E|00},\rho_{E|11})^2 \leq 1$ for brevity, we now further lower-bound this expression using the inequalities 
$\ln x \leq x-1$ and
$\sqrt{1-x} \leq 1-x/2$ 
(for $x>0$ and $x\leq 1$ respectively; these inequalities are easily derived using the fact that $\ln$ and $\sqrt{~}$ are concave):
\begin{align}
-\log \frac{1+\sqrt{1-\alpha^{\sz}}}{2} 
\geq \frac{1}{\ln 2} \left(1-\frac{1+\sqrt{1-\alpha^{\sz}}}{2}\right)
\geq \frac{\alpha^{\sz}}{\ln16}.
\label{eq:taylorbnd}
\end{align}
This gives us the following bound in place of~\eqref{eq_ratioloose}:
\begin{align}
\frac{H(C|\blk{E} \str{M}; D=1)}{H(C|C';D=1)} \geq \left(
\frac{\ovlap(\rho_{E|00},\rho_{E|11})^{2\sz}}{\ln16}
- \dn - (1+\dn) \binh\left(\frac{\dn}{1+\dn}\right) \right) \binh(\dn)^{-1}.
\label{eq:ratio2}
\end{align}
The behaviour of this expression in the large-$\sz$ limit is given by essentially the same analysis as the last part of the Theorem~\ref{th_fidbound} proof (the replacement of the $\ln4$ denominator by $\ln16$ does not change the fact that the limit is $\infty$), yielding the claimed result.
\end{proof}

The above proof basically goes through a lower bound on $\Hmin(C|\blk{E})_{\tilde{\rho}}$ rather than $H(C|\blk{E})_{\tilde{\rho}}$ directly, so we would expect it to be ``less tight'' in some sense than the previous proof of Theorem~\ref{th_fidbound}. Indeed, in the case where we take $\ovlap$ to be the fidelity, the previous bound~\eqref{eq_rogabound} is strictly (except at the extremal values) better than the new bound~\eqref{eq:hminbound}, and~\eqref{eq_ratioloose} is similarly better than~\eqref{eq:ratio2}.\footnote{In fact, the latter pair of bounds are the Taylor expansions of the former pair for small values of $\alpha^{\sz}$. 
This implies that the conversion from~\eqref{eq_rogabound} to~\eqref{eq_ratioloose} (or~\eqref{eq:hminbound} to~\eqref{eq:ratio2}) is essentially tight at large $\sz$.}
However, this difference is not sufficient to affect the conclusion based on the $\sz\to\infty$ limit, though it would be relevant if we aim to compute explicit lower bounds on the block size $\sz$ required for a positive keyrate. 
The advantage of the proof here is that it only uses the properties~\eqref{eq:symm}--\eqref{eq:FvdGupp}, hence yielding a result that holds for any distinguishability measure with those properties.\footnote{All of these statements also hold for the approach based on $H(C|\blk{E})_{\tilde{\rho}} \geq 1-d(\rho_{\blk{E}|\str{m} \str{m}},\rho_{\blk{E}|\flip{\str{m}} \flip{\str{m}}})$.}

Broadly speaking, a core step in this proof was using the ``Fuchs--van de Graaf'' inequality~\eqref{eq:FvdGupp} to bound $d(\rho_{\blk{E}|\str{m} \str{m}},\rho_{\blk{E}|\flip{\str{m}} \flip{\str{m}}})$ (and thus $H(C|\blk{E})_{\tilde{\rho}}$) in terms of the fidelity-like quantity $G$. In comparison, the previous Theorem~\ref{th_fidbound} proof used the inequality in~\eqref{eq_rogabound} (from~\cite{RFZ10}) to directly bound $H(C|\blk{E})_{\tilde{\rho}}$ in terms of fidelity.
This does raise the question of whether the latter inequality could still hold if the fidelity is replaced by some other distinguishability measure $\ovlap$. 
If so, then as long as $\ovlap$ also satisfies~\eqref{eq:symm}--\eqref{eq:supmult} (without requiring the ``Fuchs--van de Graaf'' bounds~\eqref{eq:FvdGupp}--\eqref{eq:FvdGlow}), the previous Theorem~\ref{th_fidbound} proof would immediately generalize to $\ovlap$ as well, yielding another way to derive Theorem~\ref{th:gbound}. 

Unfortunately, while Theorem~\ref{th:gbound} is in principle more general than Theorem~\ref{th_fidbound}, we do not currently know a choice of $\ovlap$ that allows the former to yield a concrete advantage over the latter. This is because the only examples we have listed are the pretty-good fidelity and the quantum Chernoff coefficient, which are both \emph{lower} bounds on the fidelity (in fact they obey the chain of inequalities 
$Q(\rho, \sigma) \leq {F}_\mathrm{pg}(\rho, \sigma) \leq F(\rho, \sigma)$~\cite{ACM+07}\footnote{As noted in~\cite{Aud14comp}, both inequalities can be simultaneously saturated by taking $\rho=\operatorname{diag}(1-t,t,0)$, $\sigma=\operatorname{diag}(1-t,0,t)$; 
a bit more generally, the second inequality is saturated if $\rho$ and $\sigma$ are simultaneously diagonalizable states.
(\cite{ACM+07} incorrectly claimed that $Q(\rho, \sigma) = F(\rho, \sigma)$ when either of the states is pure; what their argument in fact showed was that (e.g.~taking $\rho$ as the pure state) $Q(\rho, \sigma) 
= \bra{\rho}\sigma\ket{\rho} 
= F(\rho, \sigma)^2$, as noted in~\cite{ANS+08}. If \emph{both} states are pure, then furthermore this is equal 
to ${F}_\mathrm{pg}(\rho, \sigma)$.)}). 
Therefore, if we take $\ovlap$ to be either of these choices, then satisfying the condition~\eqref{eq:gbound} immediately implies satisfying the condition~\eqref{eq_fidbound}, hence the former has no advantage over the latter (in terms of being a sufficient condition for the protocol to be secure).\footnote{On the other hand, the fact that the fidelity is lower bounded by these choices of $\ovlap$ is useful when studying \emph{necessary} security conditions; see~\cite{arx_HT21} for details. To roughly summarize, the pretty-good fidelity yields a better result than the fidelity in that case; however, the quantum Chernoff coefficient is not covered by that analysis because it does not satisfy~\eqref{eq:supmult} when the inequality is flipped.} It remains to be seen whether there is any choice of $\ovlap$ that satisfies~\eqref{eq:symm}--\eqref{eq:FvdGupp} but is not a lower bound on the fidelity, in which case Theorem~\ref{th:gbound} would be a ``true'' improvement over Theorem~\ref{th_fidbound} (putting aside the question of whether $\ovlap(\rho_{E|00},\rho_{E|11})$ can be computed or bounded in a given protocol).

Another natural question raised by the above discussion is whether it suggests any approaches for further improvements over Theorem~\ref{th_fidbound}, perhaps by changing some of the inequalities used. Roughly speaking, it appears that the main point which has potential for improvement is the initial bound~\eqref{eq_rogabound} 
--- most of the subsequent steps are essentially equalities for large $\sz$. The inequality in~\eqref{eq_rogabound} is tight if nothing is known about the states other than the value of $F(\rho_{\blk{E}|\str{m} \str{m}},\rho_{\blk{E}|\flip{\str{m}} \flip{\str{m}}})$, but here there is additional structure in the states which could possibly be exploited. However, as observed earlier, replacing it by an inequality of essentially ``the same form'' (e.g.~just replacing the fidelity by $\ovlap$ exactly, with no other changes) would only yield a result in ``the same form'' as well (basically, Theorem~\ref{th:gbound}). Hence it would encounter the same issues as discussed above regarding Theorem~\ref{th:gbound}. Unless a better choice of $\ovlap$ is identified, it would seem that the only way to go beyond those issues would be to find some inequality to replace~\eqref{eq_rogabound} (possibly involving a very different quantity from the fidelity) that is in ``a different form'' in some sense.

The last question has since been resolved in a follow-up work after this thesis, where it was shown that a suitably different bound can be found in terms of the quantum Chernoff coefficient, yielding basically matching necessary and sufficient conditions for security of the repetition-code protocol (against collective attacks). We retain here the following discussion of some difficulties encountered in some prospective approaches we considered, in case it is helpful in other contexts (though it seems unlikely to be a novel observation). 

\section{Uniqueness of pure Uhlmann property}

When considering distinguishability measures, one particularly convenient property of the fidelity is that it satisfies Uhlmann's theorem, which substantially simplifies the proofs of inequalities such as the Fuchs--van de Graaf inequalities and the bound~\eqref{eq_rogabound} (by effectively reducing the analysis to pure states).
For clarity, we first highlight a subtle technicality --- in this section, we focus mainly on discussing whether a distinguishability measure $\ovlap(\rho_A,\sigma_A)$ can satisfy the following property: for any {purification} $\ket{\rho}_{AR}$ of $\rho_A$, we have (writing ``$\operatorname{opt}$'' to denote either supremum or infimum depending on the relevant distinguishability measure)
\begin{align}
\ovlap(\rho_A,\sigma_A) = \mathop{\operatorname{opt}}_{\ket{\sigma}_{AR}} \ovlap(\pure{\rho}_{AR},\pure{\sigma}_{AR}),
\label{eq:uhlmann}
\end{align}
where the optimization is taken over all \emph{purifications} $\ket{\sigma}_{AR}$ of $\sigma_A$. We will refer to this as a ``pure Uhlmann property''.
This should be contrasted against the subtly different statement that for any {extension} $\rho_{AR}$ of $\rho_A$, we have
\begin{align}
\ovlap(\rho_A,\sigma_A) = \mathop{\operatorname{opt}}_{\sigma_{AR}} \ovlap(\rho_{AR},\sigma_{AR}),
\label{eq:uhlmann2}
\end{align}
where the optimization is taken over all \emph{extensions} $\sigma_{AR}$ (not necessarily pure) of $\sigma_A$. The second statement is usually weaker in the sense that~\eqref{eq:uhlmann} implies~\eqref{eq:uhlmann2} whenever $\ovlap$ has an appropriate data-processing inequality, but the converse does not hold --- this follows from e.g.~the finding in~\cite{arx_MFSR22} that (taking the optimization to be an infimum) the max-relative entropy satisfies~\eqref{eq:uhlmann2}, but it cannot satisfy\footnote{To see this, note that the max-relative entropy has the property that it is infinite for any pair of distinct pure states, and therefore if we simply pick some distinct states $\rho_A \neq \sigma_A$ with a finite value of $D_\mathrm{max}(\rho_A\Vert\sigma_A)$, 
we would have $D_\mathrm{max}(\pure{\rho}_{AR}\Vert\pure{\sigma}_{AR}) = \infty$ for any purifications and hence~\eqref{eq:uhlmann} cannot hold.}~\eqref{eq:uhlmann}.

We now argue that the pure Uhlmann property (i.e.~\eqref{eq:uhlmann}) does not carry over straightforwardly to other distinguishability measures. In fact, given some mild conditions, one can claim that the \emph{only} distinguishability measure with this property is ``essentially'' the fidelity, in a sense we shall now elaborate on. We state the results in terms of both ``overlap-like measures'' and ``distance-like measures'' (i.e.~those which increase or decrease respectively as the states become ``more similar''). Note that we do not assume the distinguishability measure is symmetric in its arguments, so our analysis applies to quantities such as relative entropy. 

We begin with the following characterization of some properties that any distinguishability measure $\ovlap$ with a data-processing inequality must satisfy (we thank Christopher Chubb for pointing out this property). 
\begin{fact}\label{fact:DPI}
Let $\ovlap$ be a real-valued function on pairs of states, such that a data-processing inequality of the following form holds: $\ovlap(\map[\rho],\map[\sigma]) \geq \ovlap(\rho,\sigma)$ for any CPTP map $\map$.
Then it must be unitarily invariant (i.e.\ $\ovlap(U \rho U^\dagger,U \sigma U^\dagger) = \ovlap(\rho,\sigma)$ for any unitary $U$), which in turn implies that for pure states ($\rho=\pure{\rho}, \sigma=\pure{\sigma}$), we must have $\ovlap(\rho,\sigma) = g(\left|\inn{\rho}{\sigma}\right|)$ for some function $g:[0,1]\to\mathbb{R}$. Furthermore, the data-processing inequality implies the function $g$ must be nondecreasing.

If $\ovlap$ instead has a data-processing inequality in the direction $\ovlap(\map[\rho],\map[\sigma]) \leq \ovlap(\rho,\sigma)$, then all the above properties also hold, except that $g$ must be a nonincreasing function instead.
\end{fact}
\begin{proof}
The fact that $\ovlap$ is invariant under any unitary $U$ follows immediately by applying the data-processing inequality on both $U$ and its inverse.
With the fact that it is unitarily invariant, we are free to represent the pure states $\pure{\rho}, \pure{\sigma}$ in any choice of orthonormal basis when computing $\ovlap(\pure{\rho}, \pure{\sigma})$. Also, by the freedom to choose global phases when picking the vector representatives $\ket{\rho},\ket{\sigma}$ for the pure states, we can take $\inn{\rho}{\sigma} = \left|\inn{\rho}{\sigma}\right|$ without loss of generality. With this, one can always construct (e.g.~via Gram-Schmidt) an orthonormal basis $\left\{\ket{0},\ket{1},\dots\right\}$ in which $\ket{\rho} = \ket{0}$ and $\ket{\sigma} = x \ket{0} + \sqrt{1-x^2} \ket{1}$, where $x = \left|\inn{\rho}{\sigma}\right|$. When using this basis to compute $\ovlap(\pure{\rho}, \pure{\sigma})$, both states are completely described by the single parameter $x$, hence $\ovlap(\pure{\rho}, \pure{\sigma})$ can only be a function of this parameter.

We now focus only on the case where the data-processing inequality is in the direction $\ovlap(\map[\rho],\map[\sigma]) \geq \ovlap(\rho,\sigma)$; precisely analogous arguments apply when it is in the opposite direction.
To see that $g$ must be nondecreasing, we use a proof by contradiction: suppose to the contrary there exist $x,x'\in[0,1]$ such that $x>x'$ but $g(x)<g(x')$. One can always find a pair of states $\ket{\rho},\ket{\sigma}$ such that $\left|\inn{\rho}{\sigma}\right| = x$, and extend them to a pair of states $\ket{\rho}\otimes\ket{\rho'},\ket{\sigma}\otimes\ket{\sigma'}$ such that $\left|(\bra{\rho}\otimes\bra{\rho'})(\ket{\sigma}\otimes\ket{\sigma'})\right| = x'$ (simply by taking $\left|\inn{\rho'}{\sigma'}\right| = x'/x < 1$). Since all these states are pure, we have $\ovlap(\pure{\rho},\pure{\sigma}) 
= g(x)$ and $\ovlap(\pure{\rho}\otimes\pure{\rho'},\pure{\sigma}\otimes\pure{\sigma'}) = g(x')$. But since $g(x)<g(x')$, this would contradict the data-processing inequality (taking the map $\map$ to simply be the partial trace over the extension system). 
\end{proof}
\noindent To see examples of the above fact, notice that for pure states, the following distinguishability measures can all be written as a function of $\left|\inn{\rho}{\sigma}\right|$: fidelity, pretty-good fidelity, quantum Chernoff coefficient, trace distance, Jensen-Shannon divergence, relative entropy (and its R\'{e}nyi generalizations).\footnote{The relative entropy is an interesting edge case, as it does not ``obviously'' seem to be a function of $\left|\inn{\rho}{\sigma}\right|$ for pure states. However, recalling that $D(\rho\Vert\sigma)=\infty$ whenever the kernel of $\sigma$ is not contained in the kernel of $\rho$, we see that in fact for pure states we can indeed write $D(\rho\Vert\sigma)=g(\left|\inn{\rho}{\sigma}\right|)$ for a nonincreasing function $g:[0,1]\to\mathbb{R}\cup\{\infty\}$ (albeit one with range in the extended reals), in the sense that $D(\rho\Vert\sigma)=0$ if $\left|\inn{\rho}{\sigma}\right|=1$ and $D(\rho\Vert\sigma)=\infty$ for any other value of $\left|\inn{\rho}{\sigma}\right|$.} 

Furthermore, we now show that any distinguishability measure that is a function of $\left|\inn{\rho}{\sigma}\right|$ for pure states \emph{and} has a pure Uhlmann property must in fact just be the same function of the fidelity (for all states, not just pure states):
\begin{fact}
\label{fact:fidfunction}
Let $\ovlap$ be a real-valued function on pairs of states, such that for pure states ($\rho=\pure{\rho}, \sigma=\pure{\sigma}$), we have $\ovlap(\rho,\sigma) = g(\left|\inn{\rho}{\sigma}\right|)$ for some function $g:[0,1]\to\mathbb{R}$. Suppose $\ovlap$ also has a pure Uhlmann property of the form in~\eqref{eq:uhlmann},
where the optimization is a supremum or infimum (over all purifications $\ket{\sigma}_{AR}$ of $\sigma_A$).
Then the supremum/infimum is always attained, and $\ovlap$ must satisfy $\ovlap(\rho,\sigma) = g(F(\rho,\sigma))$ for all states $\rho,\sigma$.
\end{fact}
\begin{proof}
We focus on the case where the optimization is a supremum; precisely analogous arguments apply when it is instead an infimum.
We first note that in this case $g$ must again be a nondecreasing function: observe that the pure Uhlmann property~\eqref{eq:uhlmann} implies $\ovlap$ satisfies data-processing under partial trace on any pure initial states, i.e.~$\ovlap(\rho_A,\sigma_A) \geq \ovlap(\pure{\rho}_{AR},\pure{\sigma}_{AR})$, and this suffices to apply the same argument as in the second half of the Fact~\ref{fact:DPI} proof. 

Now to prove the main result, take any $\rho_A,\sigma_A$, and note that the standard Uhlmann theorem (for fidelity) states that there exists a purification $\ket{\sigma^\star}_{AR}$ such that $F(\rho_A,\sigma_A) = \left|\inn{\rho_{AR}}{\sigma^\star_{AR}}\right| = \max_{\ket{\sigma}_{AR}} \left|\inn{\rho_{AR}}{\sigma_{AR}}\right|$. Hence for any other purification $\ket{\sigma}_{AR}$, we immediately have
\begin{align}
\ovlap(\pure{\rho}_{AR},\pure{\sigma}_{AR}) 
= g\left(\left|\inn{\rho_{AR}}{\sigma_{AR}}\right|\right)
\leq g\left(\left|\inn{\rho_{AR}}{\sigma^\star_{AR}}\right|\right)
=\ovlap(\pure{\rho}_{AR},\pure{\sigma^\star}_{AR}),
\end{align}
where the inequality holds since $g$ is nondecreasing, and the equalities hold by hypothesis. This implies that the supremum in~\eqref{eq:uhlmann} is attained by $\pure{\sigma^\star}_{AR}$, and thus we have
\begin{align}
\ovlap(\rho_A,\sigma_A) = \ovlap(\pure{\rho}_{AR},\pure{\sigma^\star}_{AR})
= g\left(\left|\inn{\rho_{AR}}{\sigma^\star_{AR}}\right|\right) = g(F(\rho_A,\sigma_A)),
\end{align}
as claimed. (Note that the above argument invokes a subtle technicality that Uhlmann's theorem for fidelity also states that the supremum is in fact attained, which we essentially used to argue that the supremum/infimum in~\eqref{eq:uhlmann} is attained as well. This fact was needed in order to avoid considerations about whether $g$ is continuous.)
\end{proof}

By combining Facts~\ref{fact:DPI} and~\ref{fact:fidfunction}, we see that any distinguishability measure satisfying\footnote{There are various alternatives to the chain of reasoning shown above.
For instance, note that the pure Uhlmann property~\eqref{eq:uhlmann} implies that $\ovlap$ satisfies data-processing under partial trace on pure initial states. Hence we could instead start from e.g.~the conditions that $\ovlap$ has the pure Uhlmann property and is unitarily (or isometrically) invariant, which by a Stinespring-dilation argument would imply that it has a data-processing inequality for pure initial states $(\rho,\sigma)$, leading to the same conclusions.} a data-processing inequality and having a pure Uhlmann property (in the sense of~\eqref{eq:uhlmann}) must in fact just be ``secretly'' the fidelity (up to a rescaling by a monotone function $g$ --- while $g$ might not give a bijection between the ranges of $\ovlap$ and $F$, the fact that it is monotone implies that $\ovlap$ cannot be better than simply a ``coarse-graining'' of $F$). Hence it seems it would be difficult to find a ``useful'' distinguishability measure with a pure Uhlmann property, other than the fidelity. (While some distinguishability measures such as max-relative entropy satisfy the weaker property~\eqref{eq:uhlmann2}, that version is slightly less powerful in proofs, as it only allows us to reduce the analysis to scenarios where \emph{one} of the states is pure.)

\cleardoublepage
\phantomsection
\addcontentsline{toc}{chapter}{Bibliography}
\printbibliography

\end{document}